\documentclass[superscriptaddress, aps, prx, notitlepage, nofootinbib, longbibliography, 10pt, amsmath,amssymb,floatfix]{revtex4-2}

\pdfoutput=1
\usepackage[utf8]{inputenc}
\usepackage[T1]{fontenc}  

\usepackage{graphicx} 
\usepackage{bbold} 
\usepackage{amsthm}
\usepackage{array}[=2016-10-06]
\usepackage[caption=false]{subfig}
\usepackage{changepage}
\usepackage{import}
\usepackage{tikz}
\usepackage{relsize}
\usepackage{capt-of}
\usepackage{mathtools}
\usepackage{float}
\usepackage[section]{placeins}
\usepackage{mdframed}
\usepackage{listings}

\usetikzlibrary{decorations.pathreplacing}
\usepackage{braket}
\usepackage{xfrac}
\usepackage{algorithm, algpseudocode}
\usepackage{thmtools, thm-restate}
\usepackage{siunitx}
\usetikzlibrary{quantikz2, positioning, fit, shapes.geometric, arrows.meta}
\usepackage{pgfplots}
\pgfplotsset{compat=1.18} 

\usepackage{xcolor}
\usepackage{color,soul}
\usepackage{caption}
\usepackage{makecell}

\definecolor{NiceOrange}{rgb}{.9,0.55,0}
\definecolor{NicePurple}{rgb}{0.3,0.1,1}
\definecolor{DarkBlue}{rgb}{0.1,0.1,0.5}

\usepackage[colorlinks]{hyperref}
\hypersetup{
    unicode=false,          
    pdftoolbar=true,        
    pdfmenubar=true,        
    pdffitwindow=false,     
    pdfnewwindow=true,      
    colorlinks=true,        
    linkcolor=DarkBlue,     
    citecolor=NicePurple,   
    filecolor=NiceOrange,   
    urlcolor=DarkBlue,      
    pdftitle={Verifiable Quantum Advantage via Optimized DQI Circuits},
    pdfauthor={Khattar et al.},
}
\usepackage[all]{hypcap}
\usepackage[capitalise,nameinlink]{cleveref}

\makeatletter

\let\origtau\tau 
\renewcommand{\tau}{\scalebox{1.44}{$\origtau$}}
\DeclareFixedFont{\ttb}{T1}{txtt}{bx}{n}{8}
\DeclareFixedFont{\ttm}{T1}{txtt}{m}{n}{8}
\definecolor{deepblue}{rgb}{0,0,0.5}
\definecolor{deepred}{rgb}{0.6,0,0}
\definecolor{deepgreen}{rgb}{0,0.5,0}
\newcommand{\pythonlisting}[2]{\lstinputlisting[
    language=Python,
    basicstyle=\fontsize{#1}{#1}\ttfamily,
    commentstyle=\fontsize{#1}{#1}\ttfamily\bfseries\color{deepgreen},
    keywordstyle=\fontsize{#1}{#1}\ttfamily\bfseries\color{deepblue},
    xleftmargin=0.5cm,
    frame=tlbr,
    framesep=0.2cm,
    framerule=1pt,
    morekeywords={self,None,assert,venting\_into,venting\_into\_new\_table,alloc\_quint,quint,z,del\_measure\_x,alloc\_phase\_gradient,qpu,QPU,del\_phase\_gradient,del\_by\_equal\_to\_const,push\_uncompute\_info,pop\_uncompute\_info,ghz\_lookup},
    keywordstyle={[4]\fontsize{#1}{#1}\ttfamily\color{deepred}},
    morekeywords={[4]Q\_residue,Q\_result,Q\_total,Q\_e,Q\_grad,Q\_target,Q\_helper,Q\_unresult,Q\_l1,Q\_l0,Q\_l,Q\_k,Q\_exponent,Q\_acc,Q\_dlog,Q\_out},
]{#2}}

\theoremstyle{definition}
\newtheorem{definition}{Definition}[section]
\theoremstyle{definition}
\newtheorem{theorem}[definition]{Theorem}
\theoremstyle{definition}
\newtheorem{lemma}[definition]{Lemma}
\newtheorem{corollary}{Corollary}[section]
\newtheorem{remark}{Remark}[section]

\newcommand{\GL}{\mathrm{GL}}
\newcommand{\tr}{\mathrm{tr}}
\newcommand{\Aff}{\mathrm{Aff}}   
\newcommand{\F}{\mathbb{F}}   
\newcommand{\GF}{\text{GF}}   

\newcommand{\Google}{\affiliation{Google Quantum AI, Venice, CA 90291}}

\newcommand{\bigOtilde}{\widetilde{O}}

\def\l@subsubsection#1#2{} 
\def\l@f@section{}
\def\toc@@font{\normalsize \sffamily}

\renewcommand{\tableofcontents}{%
  \vspace{1em}
  \twocolumngrid
  \@starttoc{toc}
  \onecolumngrid
}
\newif\ifbibdone
\let\orig@bibliography\bibliography
\renewcommand{\bibliography}[1]{%
  \orig@bibliography{#1}%
  \bibdonetrue
}

\let\orig@appendix\appendix
\renewcommand{\appendix}{%
  \orig@appendix
  \crefalias{section}{appendix}%
  \crefalias{subsection}{appendix}%
  \crefalias{subsubsection}{appendix}%
  \ifbibdone
    \addtocontents{toc}{\protect\let\protect\l@subsection\protect\@gobble@tw@}
  \fi
}

\begin{document}

\title{Verifiable Quantum Advantage via Optimized DQI Circuits}

\author{Tanuj Khattar}
\email[Corresponding author: ]{tanujkhattar4@gmail.com}
\Google
\author{Noah Shutty}
\email[Corresponding author: ]{noahshutty@gmail.com}
\Google
\author{Craig Gidney}
\Google
\author{Adam Zalcman}
\Google
\author{Noureldin Yosri}
\Google
\author{Dmitri Maslov}
\Google
\author{Ryan Babbush}
\Google
\author{Stephen P.~Jordan}
\Google

\date{November 18, 2026}

\begin{abstract}
Recently, a quantum algorithm called Decoded Quantum Interferometry (DQI) was introduced that achieves an apparent exponential speedup for Optimal Polynomial Intersection (OPI) problem, which has previously been studied in the contexts of cryptography and error correcting codes. However, this left open the question of how many logical gates and logical qubits would be needed to solve a classically intractable instance of OPI. Here, we develop optimized implementations of DQI which greatly reduce its resource requirements. We establish that DQI for OPI is the first known candidate for verifiable quantum advantage with optimal asymptotic speedup: solving instances with classical hardness $O(2^N)$ requires only $\widetilde{O}(N)$ quantum gates, matching the theoretical lower bound. To realize this, we overcome the primary bottleneck of reversible Reed-Solomon decoding by introducing novel quantum circuits for the Extended Euclidean Algorithm (EEA) on degree-$n$ polynomials with $b$-bit coefficients that reduce the leading-order space complexity to $2nb$ qubits. These improvements are broadly applicable, including to Shor's algorithm for the discrete logarithm. We analyze OPI over binary extension fields $\mathrm{GF}(2^b)$, assess hardness against new classical attacks, and identify resilient instances. Our resource estimates show that classically intractable OPI instances (requiring $>10^{23}$ classical trials) can be solved with approximately 5.72 million Toffoli gates per shot and $1885$ logical qubits. This is roughly $100$ times fewer gates than required for factoring RSA-2048 and, remarkably, is also less than the leading interactive protocol for computational proof of quantumness, positioning DQI as a compelling candidate for practical, verifiable quantum advantage.
\end{abstract}

\maketitle
\emph{\textbf{Data availability:}} Code and assets created for this paper are 
\href{https://doi.org/10.5281/zenodo.24325102}{available on Zenodo} \cite{https://doi.org/10.5281/zenodo.24325102}.

\tableofcontents

\section{Introduction}

The advancing capabilities of quantum computing hardware drive a need for computational problems through which quantum advantage can be tested. As each major milestone in hardware development is reached, one would like to demonstrate capabilities that are feasible with the new quantum hardware but infeasible on classical supercomputers. Such computational problems provide computational certifications of quantumness, define concrete targets to guide quantum hardware development, and have been proposed as fundamental tests of quantum mechanics itself \cite{aaronson2004multilinear}.

Typically, there is a tradeoff between the desire to attack practical computational problems and the desire to achieve quantum advantage using near term hardware with minimal gate and qubit counts. At the more practical end, quantum computers can solve conventional computational problems such as integer factorization, where the problem is to find a bit string bearing a specific relation to the input (e.g. a list of prime numbers whose product is the input). Integer factorization has several appealing properties. First, it is a preexisting problem well-studied in classical computer science, particularly due to applications in cryptanalysis. Second, it is a problem that is hard to solve, but whose solutions can be easily checked for correctness using standard classical computers, since the purported factors can simply be multiplied. The drawback of factoring as a proof of quantum advantage is that the gate and qubit counts needed to solve a classically intractable instance are quite large. Specifically, the most highly optimized implementation of Shor's algorithm discovered so far \cite{gidney2025factor2048bitrsa} requires roughly $10^9$ logical Toffoli gates to solve a classically-intractable factoring problem\footnote{As is standard in quantum algorithms based on period finding, this cost is spread across multiple shots, which are separate quantum circuits without need to maintain coherence across shots. Except as noted otherwise we will always quote cumulative gate counts from all shots.}. This leaves a need for verifiable quantum advantage using lower resource requirements that can serve as milestone experiments for nearer-term quantum hardware.

One can reduce the resource requirements for factoring-based proofs of quantum advantage by restricting attention to a special class of instances. Specifically, $n$-bit integers of the form $P^2 Q$ with $\log_2 Q \simeq n^a$ for $2/3 < a < 1$, can be factored on a quantum computer using fewer qubits and quantum gates than would be required by Shor's algorithm to factor general integers of $n$ bits \cite{kahanamoku2025jacobi}. However, this restriction on the integers is not known to substantially reduce the classical cost of their factorization. This proposal thus greatly reduces the quantum resources needed for a factoring-based demonstration of quantum advantage, at the cost of solving a less well-motivated problem (in particular, one not well-motivated by cryptanalysis). The asymptotic scaling of this proposed verifiable proof of quantum advantage appears quite favorable, but to our knowledge there have not yet been any quantitative estimates published of the concrete logical qubit and gate requirements for an implementation.

Applying Shor's algorithm to discrete logarithm problems offers an alternative path to verifiable quantum advantage with lower resource requirements. In particular, as shown in \cite{babbush2026securingellipticcurvecryptocurrencies}, an optimized implementation of Shor's algorithm for discrete log can solve classically-intractable instances of the discrete logarithm problem for the Abelian groups associated with elliptic curves, using only $7 \times 10^7$ Toffoli gates and $1450$ logical qubits. 

An alternative approach to verifiable quantum advantage is to abandon connection to real world problems and instead use cryptographic trapdoors to construct interactive protocols by which quantum computational power can be demonstrated. In particular, building on prior work of \cite{brakerski2018certifiable}, an interactive protocol is proposed in \cite{kahanamoku2022classically} by which an interrogator can ask a sequence of questions to a quantum computer and, based on the responses, gradually accrue statistical evidence for quantum advantage over many rounds of interrogation. Specifically, it is shown in \cite{kahanamoku2022classically} that, assuming the existence of cryptographically secure claw free trapdoor functions, a quantum computer is able to correctly answer certain questions with greater reliability than is achievable by any classical polynomial-time algorithm. The verifier, unlike the quantum prover, possesses the secret key which enables efficient inversion of the trapdoor function, and is thereby able to verify the answers provided by the quantum computer without needing to perform any exponential time classical computations. In \cite{kahanamokumeyer2024fastquantumintegermultiplication}, concrete resource estimates are given, which improve upon the original estimates of \cite{kahanamoku2022classically}, and show that a space-optimized implementation of the protocol such that the computational cost of classical spoofing appears intractable can be achieved using 1583 logical qubits and $4 \times 10^7$ logical T gates per round, whereas a time-optimized implementation uses $2937$ logical qubits and $2 \times 10^7$ logical T gates per round\footnote{In \cite{kahanamokumeyer2024fastquantumintegermultiplication} the non-clifford gates used are Toffolis and controlled-rotations, which cost 4 T gates \cite{Jones2013} and 24 T gates \cite{kliuchnikov2023shorter}, respectively, at the required level of precision.}. This improves substantially on the resource requirements needed for quantum advantage in factoring, at the cost of moving to a purely contrived problem. The protocol of \cite{kahanamokumeyer2024fastquantumintegermultiplication, kahanamoku2022classically} is based on the CHSH version of a Bell test, and thus to achieve an expected p-value of $10^{-6}$ one would need at least 363 rounds. Thus, the gate-optimized protocol consists of $363$ separate shots, without a need for coherence between shots, that cumulatively contain $7 \times 10^{9}$ T gates. Thus, the gate requirements for the interactive protocol are larger than a verifiable proof of quantumness based on the discrete-logarithm implementation of \cite{litinski2023compute256bitellipticcurve}, but the required logical qubits and the required logical T gates per shot are lower for the interactive protocol.

To reduce resource requirements even further, one can abandon the requirement for efficient verification and move to problems of sampling from specified probability distributions over bit strings. This brings the demonstration of quantum advantage within the realm achievable by present-day hardware, at the cost of abandoning both connection to real-world computational problems and direct verification. Instead, because the verification has computational cost scaling exponentially in problem size, one runs small scale experiments whose output is classically tractable to replicate and verifies their correctness, then runs large scale experiments whose output is classically intractable to replicate and extrapolates that these are also correct.  In this manner, experimental evidence for quantum advantage has been obtained using random circuit sampling \cite{bouland2019complexity, arute2019quantum, zhu2022quantum, morvan2024phase}, and boson sampling \cite{zhong2020quantum, zhong2021phase, madsen2022quantum, deng2023gaussian}. Recently, a method for demonstrating quantum advantage has been proposed and experimentally tested, which is based on out of time ordered correlators (OTOC) \cite{google2025constructive, king2025simplified}. Quantum advantage based on OTOC is not thought to be efficiently classically verifiable, but since the output is a repeatable number rather than a sample from a distribution over exponentially many bit strings, the results are efficiently verifiable by other quantum computers.
Within the sampling paradigm, a classically verifiable task based on peaked circuits has also been proposed \cite{aaronson2024verifiable}, though without a systematic method to construct classically intractable instances.

Here, we construct and analyze an optimized implementation of the DQI algorithm from \cite{jordan2024optimizationdecodedquantuminterferometry}. We find that efficiently-verifiable quantum advantage can be demonstrated, using surprisingly modest quantum resources, on a problem of approximating the optimal choice of coefficients in a polynomial for fitting a set of data points over a finite field. Special cases of this problem, which we call Optimal Polynomial Intersection (OPI), have been previously studied in the coding theory literature, under the name \emph{list-recovery} \cite{guruswami1998improved}, and in the cryptography literature under the names \emph{noisy polynomial reconstruction/interpolation}~\cite{NP99,BN00} and polynomial approximation \cite{GM14, SW05, S05}. We find that instances of OPI which we believe to be classically intractable, requiring at least $10^{23}$ arithmetic operations to solve using any of the classical attacks known to us, can be solved using only $5.72 \times 10^6$ logical Toffoli gates per shot and $1885$ logical qubits. 
Each Toffoli gate can be implemented using four T gates \cite{Jones2013}. Thus, remarkably, these are not only far lower resource requirements than are needed for verifiable proofs of quantumness via integer factorization or discrete logarithms, but also lower total logical T gate counts than are needed for proposed interactive proofs of quantumness using contrived problems based on cryptographic trapdoors \cite{kahanamoku2022classically, kahanamokumeyer2024fastquantumintegermultiplication}. (See \cref{tbl:quantitative_advantage}.) 
Additionally, we establish that in terms of asymptotic scaling, the resource requirements for verifiable quantum advantage used by DQI are optimal, assuming that new classical attacks on the OPI are not discovered which achieve better asymptotic scaling than known attacks. (See \cref{tbl:verifiable_advantage_classical_hardness}.)
Finally, we estimate physical resources needed for a fault tolerant implementation of our algorithm on a 2D nearest-neighbor surface code architecture with planar degree four connectivity \cite{google2025quantum} under standard assumptions \cite{gidney2025factor2048bitrsa, Gidney2021, babbush2026securingellipticcurvecryptocurrencies, low2025fast} - a uniform gate error rate of 0.1\%, a surface code cycle time of 1 microsecond, and a control system reaction time of 10 microseconds. 
We provide hand optimized lattice surgery constructions for key subroutines and show that solving a classically intractable instance of OPI using DQI requires $\approx 800,000$ physical qubits and $\approx$ $6$ hours of runtime. 
We note that because our optimized DQI circuits trade a drastically reduced Toffoli gate count for a modestly higher logical qubit count, our algorithm is uniquely well suited to benefit from emerging architectures utilizing high rate quantum Low Density Parity Check (qLDPC) codes \cite{cain2026shorsalgorithmpossible10000, webster2026pinnaclearchitecturereducingcost, tripier2026faulttolerantquantumcomputingtrapped}.

\begin{table}[ht]
    \begin{center}
    \begin{tabular}{|c|c|c|c|c|c|c|}
    \hline
    \textbf{Problem} & 
    \textbf{Instance Size } & 
    \makecell{\textbf{logical} \\ \textbf{T gates}} & 
    \makecell{\textbf{logical} \\ \textbf{qubits}} & 
    \makecell{\textbf{number} \\ \textbf{of shots}} & 
    \makecell{\textbf{logical T} \\ \textbf{per shot}} & 
    \makecell{\textbf{classical} \\ \textbf{operations}} \\
    \hline
    \makecell{Factoring \\ \cite{gidney2025factor2048bitrsa}} & 
    $n=1024$ & $4 \times 10^9$  & 742 &  10  & $4 \times 10^8$ & $10^{24}$ \\
    \hline
    \makecell{Elliptic Curve \\ Discrete Log \cite{babbush2026securingellipticcurvecryptocurrencies}} & $n=256$ & 
    $3 \times 10^8$ & $1450$ & 1 & $3 \times 10^8$ & $10^{38}$\\
    \hline
    \makecell{Computational Bell Test \\ (gate-optimized) \cite{kahanamoku2022classically, kahanamokumeyer2024fastquantumintegermultiplication}} & $n=1024$ &
    $7 \times 10^9$ & 2937 & 363 &  $2 \times 10^7$ & $10^{24}$ \\
    \hline
    \makecell{Computational Bell Test \\ (qubit-optimized) \cite{kahanamoku2022classically, kahanamokumeyer2024fastquantumintegermultiplication}} & $n=1024$ &
    $1 \times 10^{10}$ & 1583 & 363 & $4 \times 10^7$ & $10^{24}$ \\
    \hline
    \makecell{OPI (gate-optimized) \\ (this work)}& $m, n, b, r=(4095,90,12,2016)$ &
    $4 \times 10^7$ & 2366 & $1.3$ & $3 \times 10^7$ & $> 10^{23}$\\
    \hline
    \makecell{OPI (qubit-optimized) \\ (this work)}& $m, n, b, r=(4095,70,12,2016)$ &
    $4 \times 10^8$ & 1885 & 17 & $2 \times 10^7$ & $> 10^{23}$\\
    \hline
    \end{tabular}
    \end{center}
    \caption{
    Quantitative costs for classically-verifiable quantum advantage. 
    \emph{logical T per shot} metric dictates the required fault-tolerant lifetime of the logical qubits.
    The problem parameters shown here are as follows. 
    The factoring instance is a product of two primes expressible using 1024 bits (i.e. breaking RSA1024, which is estimated to require roughly $2^{80}$ classical operations \cite{barker2009sp}). 
    The discrete logarithm instance is to compute the private key from an ECC public key of size 256 bits, using NIST-recommended parameters, as described in \cite{babbush2026securingellipticcurvecryptocurrencies}, which is estimated to required roughly $2^{128}$ classical operations \cite{chen2023recommendations}. 
    The OPI instance has $m=4095$ constraints, $n=70$ (or $n=90$) variables, operates over the field $\mathbb{F}_q$ where $q = 2^{b} = 2^{12}$, and each target set has size $r=2016$ (see \cref{sec:optimizePrange} for expected shot counts). 
}\label{tbl:quantitative_advantage}
\end{table}

\section{Main Results}

\subsubsection{Asymptotically Optimal Verifiable Quantum Advantage}

A fundamental question in the pursuit of verifiable quantum advantage concerns the efficiency of the quantum speedup itself. Given a problem instance with a target classical hardness of $O(2^N)$ (where $N$ is the security parameter), what is the required quantum runtime?
The theoretical lower bound is $\Omega(N)$ quantum gates \cite{Bravyi_2016}. However, existing candidates for verifiable superpolynomial advantage exceed this bound. For example, achieving $O(2^N)$ hardness in integer factorization requires Shor's algorithm \cite{Shor1999} to use $\bigOtilde(N^6)$ gates (against the General Number Field Sieve), while Elliptic Curve Cryptography (ECC) requires $\bigOtilde(N^2)$ gates (against Pollard Rho).

One corollary of the resource estimation performed in this work is to demonstrate that Decoded Quantum Interferometry (DQI) \cite{jordan2024optimizationdecodedquantuminterferometry}, when applied to the Optimal Polynomial Intersection (OPI) problem, is the first known candidate for verifiable quantum advantage that achieves the optimal asymptotic speedup (up to polylogarithmic factors), requiring only $\bigOtilde(N)$ quantum gates to solve instances with $O(2^N)$ classical hardness:
\begin{restatable}{theorem}{THMoptimalSpeedup}\label{thm:optimalSpeedup}
There is an NP-search / optimization problem where the runtime of the best-known classical algorithm for the problem is $2^N$ and which can be solved with a circuit of $\bigOtilde(N)$ quantum gates.
\end{restatable}
We prove this theorem in \cref{sec:optimal_speedup}. We compare this against prior speedups based on Shor's algorithm in \cref{tbl:verifiable_advantage_classical_hardness}.
\begin{table}[ht]
    \begin{center}
    \begin{tabular}{| p{4.35cm} | p{3.7cm}| p{1.8cm} |p{3.38cm}|}
    \hline
    \textbf{Problem} & \textbf{Classical Runtime} & \textbf{Quantum Gate Cost} & \textbf{Gate Cost (for classical runtime $2^N$)} \\
    \hline
    Factoring &  $2^{\bigOtilde(n^{1/3})}$ (GNFS) & $\bigOtilde(n^2)^\dagger$ &  $\bigOtilde(N^6)$ \\
    Jacobi Factoring & $2^{\bigOtilde(n^{1/3})}$ (GNFS) & $\bigOtilde(n)$ & $\bigOtilde(N^3)$\\
    Elliptic Curve Cryptography & $2^{\bigOtilde(n)}$ (Pollard Rho) & $\bigOtilde(n^2)$ & $\bigOtilde(N^2)$\\
    \hline 
    \textbf{OPI} & $2^{\bigOtilde(n)}$ (Extended Prange) & $\bigOtilde(n)$ & $\bigOtilde(N)$\\
    \textbf{Best Possible} & $2^n$ & $n$ & $\Omega(N)$ \\
    \hline
    \end{tabular}
    \end{center}
    \caption{
    Asymptotic costs for verifiable quantum advantage with classical hardness $2^N$.
    ${}^\dagger$By \cite{regev2025efficient}, one can factor $n$ bit integers by running $\sqrt{n}+4$ quantum circuits of $\bigOtilde(n^{3/2})$ gates. In this case, the overall circuit size is still $\bigOtilde(n^2)$, but even counting a single quantum circuit, the quantum gate cost for classical runtime $2^N$ is $N^{4.5}$.
}\label{tbl:verifiable_advantage_classical_hardness}
\end{table}

\subsubsection{Background: DQI and the OPI Problem}

To understand how DQI achieves this speedup, we must formalize the class of problems it addresses. DQI provides a framework for approximating solutions to constraint satisfaction problems, specifically max-LINSAT, where its efficiency is tied to the efficiency of reversibly decoding a related error-correcting code.

\begin{definition}[max-LINSAT \cite{jordan2024optimizationdecodedquantuminterferometry}]
Let $\F_q$ be a finite field. Given an $m\times n$ matrix $B$ over $\F_q$ (with $m>n$), and for each constraint $i = 1,2,\dots, m$, a subset $F_i \subset \F_q$. The max-LINSAT problem is to find an assignment $x \in \F_q^n$ that maximizes the number of satisfied constraints, where the $i$-th constraint is satisfied if $\sum_{j=1}^n B_{ij}x_j \in F_i$.
\end{definition}

The DQI algorithm addresses max-LINSAT by defining a related objective function, $f(x)$, as the number of satisfied constraints minus the number of unsatisfied constraints. This can be expressed as:
\begin{equation}
    f(x) = \sum_{i=1}^{m} f_i\left(\sum_{j=1}^{n} B_{ij}x_j\right), \quad \text{where} \quad f_i(y) = 
    \begin{cases} 
        +1 & \text{if } y \in F_i \\
        -1 & \text{if } y \notin F_i 
    \end{cases}.
\end{equation}

Maximizing $f(x)$ is equivalent to maximizing the number of satisfied constraints. The DQI algorithm works by preparing a quantum state $|P(f)\rangle = \sum_x P(f(x))|x\rangle$, where $P$ is a polynomial of degree $\ell$ designed to enhance the amplitudes of states $\ket{x}$ where $f(x)$ is large. The core insight of DQI is that the preparation of this state can be reduced to a decoding problem on the dual code $C^\perp = \{d \in \mathbb{F}_{q}^{m}: B^T d = 0\}$. The algorithm involves creating a superposition of errors $e$ and their syndromes $B^T e$. To achieve the necessary interference that amplifies good solutions, the error register $|e\rangle$ must be coherently uncomputed, which necessitates a reversible quantum implementation of a decoder for $C^\perp$.
The performance of DQI is directly tied to the error-correction capability of this decoder. The degree $\ell$ of the enhancing polynomial $P$ corresponds to the maximum number of errors the decoder must correct. A key result of the DQI framework \cite{jordan2024optimizationdecodedquantuminterferometry} is the semicircle law, which provides a rigorous performance guarantee based on the decoding capability.

\begin{theorem}[DQI Semicircle Law (Informal) \cite{jordan2024optimizationdecodedquantuminterferometry}]
Given a max-LINSAT instance where the allowed sets $F_i$ have size $r$ over a field of size $q$. If the dual code $C^\perp$ can be efficiently decoded up to $\ell$ errors, DQI can sample solutions that satisfy an expected fraction of constraints $\langle s \rangle/m$ approaching:
\begin{equation}
\left(\sqrt{\frac{\ell}{m} \left(1-\frac{r}{q}\right)} + \sqrt{\left(1-\frac{\ell}{m}\right)\frac{r}{q}}\right)^2.
\end{equation}
\end{theorem}

This theorem formalizes the intuition that a better decoder (a larger correctable error fraction $\ell/m$) leads to a better optimization result. Consequently, the efficiency of the reversible decoder dominates the resource requirements of the entire DQI algorithm.

In this work, we focus on constructing efficient quantum circuits for DQI applied to the Optimal Polynomial Intersection (OPI) problem, identified in \cite{jordan2024optimizationdecodedquantuminterferometry} as a candidate for superpolynomial quantum speedup.

\begin{definition}[Optimal Polynomial Intersection (OPI)]
Let $\mathbb{F}_q$ be a finite field and $n < q-1$. Given $m=q-1$ subsets $F_y \subset \mathbb{F}_q$ for each $y \in \mathbb{F}_q^*$, find a polynomial $Q \in \mathbb{F}_q[y]$ of degree at most $n-1$ that maximizes the objective function: $f_\text{OPI}(Q) = |\{y \in \mathbb{F}_q^* : Q(y) \in F_y\}|$.
\end{definition}

When OPI is cast as max-LINSAT, the constraint matrix $B$ is a Vandermonde matrix, implying the dual code $C^\perp$ is a Reed-Solomon (RS) code. For OPI, DQI can achieve approximation ratios that appear beyond the reach of known polynomial-time classical algorithms \cite{jordan2024optimizationdecodedquantuminterferometry}. RS codes possess efficient classical decoders, such as the Berlekamp-Massey decoder \cite{berlekamp2015algebraic} or the Extended Euclidean Algorithm (EEA) \cite{SUGIYAMA197587}. 
To achieve the optimal asymptotic speedup and the concrete resource requirements reported above, we introduce several novel algorithmic and compilation techniques. We formalize our primary technical contributions below.

\subsubsection{Classical Hardness of OPI over Binary Extension Fields}

The DQI algorithm solves the Optimal Polynomial Intersection (OPI) problem by manipulating polynomials whose coefficients lie in a finite field. Consequently, the primary computational bottleneck of the quantum circuit is the cost of implementing reversible arithmetic operations over this field. To significantly reduce these arithmetic gate costs, we shift the analysis of DQI from prime fields to binary extension fields $\text{GF}(2^b)$.
To analyze classical hardness in this new regime, we formalize a novel classical attack tailored for extension fields, which we call the Extended Prange's algorithm (XP). 

\begin{definition}[Extended Prange's Algorithm (XP) --- see \cref{sec:extendedPrange}]
Consider the search problem of finding a vector $\mathbf{x} \in \mathbb{F}_{2^b}^n$ to satisfy constraints of the form $c_i \in F_i$ for $i \in \{1, \dots, m\}$, where $c_i$ is the $i$-th coordinate of the vector $B\mathbf{x}$ for a given matrix $B \in \mathbb{F}_{2^b}^{m \times n}$, and $F_i \subset \mathbb{F}_{2^b}$ are target sets. 
Viewing $\mathbb{F}_{2^b}$ as a $b$-dimensional vector space over the base field $\mathbb{F}_2$, the unknown vector $\mathbf{x}$ affords exactly $n b$ degrees of freedom. 
The standard Prange's algorithm \cite{prange1962use} generates candidate solutions by completely determining $n$ coordinates of $B\mathbf{x}$ (imposing exactly $b$ linearly independent $\mathbb{F}_2$-equations per selected coordinate). 
The Extended Prange's Algorithm (XP) distributes this global budget of $nb$ base-field linear equations across all $m$ coordinates.
Specifically, XP computes an integer allocation vector $\mathbf{s} = (s_1, \dots, s_m) \in \{0, \dots, b\}^m$ subject to the global budget constraint $\sum_{i=1}^m s_i \le nb$. 
For the $i$-th coordinate, the algorithm imposes $s_i$ independent $\mathbb{F}_2$-linear equations, restricting $c_i$ to an affine subspace $A_i \subset \mathbb{F}_{2^b}$ of $\mathbb{F}_2$-codimension $s_i$. 
Assuming $c_i$ is uniformly distributed within $A_i$, the marginal probability that the constraint $c_i \in F_i$ is satisfied is exactly the fractional overlap $|A_i \cap F_i|/|A_i|$. To maximize this local success probability for a given cost $s_i$, the algorithm chooses an optimal $A_i$ such that:
\[
    P_i[s_i] = \max_{\substack{A \subset \mathbb{F}_{2^b} \\ \operatorname{codim}_{\mathbb{F}_2}(A) = s_i}} \frac{|A \cap F_i|}{|A|}.
\]
Candidate solutions $\mathbf{x}$ are then generated by solving the resulting $\mathbb{F}_2$-linear system (where $c_i \in A_i$ for all $i$).
\end{definition}

The overall allocation $\mathbf{s}$ is optimized based on the desired objective: one can employ a linear program (LP) to maximize the \emph{expected} number of satisfied constraints, or a dynamic programming (DP) knapsack solver to maximize the \emph{tail probability} of satisfying at least a target number of constraints. 
Crucially, the quality of these generated candidates is entirely bottlenecked by the local probabilities $P_i[s_i]$. 
If the target sets $F_i$ are explicitly structured to have inherently small intersections with all affine subspaces, $P_i[s_i]$ will remain small regardless of how the equation budget is allocated. 
This stifles the probability of generating a high-quality candidate, forcing an exponentially large number of classical trials to find a valid solution.
To exploit this and establish classical intractability, we introduce the Twisted Bent Target OPI (TBT-OPI) problem, which constructs target sets from Maiorana-McFarland bent functions to minimize these affine intersections.

\begin{definition}[Twisted Bent Target OPI (TBT-OPI) --- see \cref{definition:TBTOPI}]
Let $S_k = \{(x_1, \dots, x_{2k}) \in \mathbb{F}_2^{2k} : \sum_{i=1}^k x_i x_{i+k} = 1 \pmod 2\}$ be a Maiorana-McFarland bent function set. Let $m = 2^{2k}-1$. For each evaluation point $\alpha \in \mathbb{F}_{2^{2k}}^*$, let $A \sim \text{GL}_{2k}(\mathbb{F}_2)$ be a random independent invertible matrix, and define the target set $F_\alpha = \phi^{-1}(A S_k)$, where $\phi$ is the natural bijection from $\mathbb{F}_{2^{2k}}$ to $\mathbb{F}_2^{2k}$. The $(R, \mu)$-TBT-OPI problem is to find a polynomial $f(X)$ of degree at most $\lceil R \cdot m \rceil$ such that $f(\alpha) \in F_\alpha$ for at least $\lfloor \mu \cdot m \rfloor$ distinct points $\alpha$.
\end{definition}

\begin{theorem}[Classical Hardness of TBT-OPI --- see \cref{thm:classicalHardness}]
\label{thm:classical_hardness_main}
There exists a choice of constant rate $R \approx 1/10$ such that, setting $\mu = \mu_{DQI} := \frac{1}{2} + \sqrt{\frac{R}{2}(1-\frac{R}{2})}$, the expected number of trials required by the Extended Prange's algorithm to solve the $(R, \mu)$-TBT-OPI problem over $\text{GF}(2^b)$ scales exponentially as $\Omega(\exp(0.02m))$.
\end{theorem}

\subsubsection{Space-Efficient DQI via Sequential Measurement-Based Uncomputation}

In the original DQI construction \cite{jordan2024optimizationdecodedquantuminterferometry}, the algorithm requires the simultaneous instantiation of an $mb$-qubit error register to hold a superposition of error patterns $|e\rangle$, alongside an $nb$-qubit syndrome register to hold the syndrome values $|B^\intercal e\rangle$. 
This results in a space complexity of $(m + n)b$ qubits, plus additional ancilla overhead for reversible decoding. 
In the typical regime for OPI where $m \gg n$ (e.g., $m=4095$ constraints and $n=70$ variables), the multiplicative $mb$ overhead dominates the logical qubit requirements and makes the circuit prohibitively expensive.

To eliminate this bottleneck, we introduce an optimized circuit architecture that reduces the total space requirement to $2nb + \mathcal{O}(\log_2 n)$ qubits. 
This reduction relies on two key innovations: encoding error locations logarithmically using Sparse Dicke States, and sequentially generating and uncomputing the error register using Measurement-Based Uncomputation (MBU) \cite{Jones2013, Gidney2018}. 

\begin{theorem}[Sparse Dicke State Preparation --- see \cref{sec:sparse_dicke_state}]
\label{thm:sparse_dicke_main}
An equal-weight superposition over all $\binom{m}{\ell}$ combinations (a Dicke State) representing the locations of $\ell$ errors in an $m$-length codeword can be prepared using $\ell \cdot \lceil \log_2 m \rceil \approx \ell \cdot b$ qubits, rather than the standard $m$ qubits.
\end{theorem}

Rather than materializing the entire $mb$-qubit error state $|e\rangle$ at once, our architecture iteratively translates this Sparse Dicke State into the target syndrome state $|B^\intercal e\rangle$ by processing one constraint at a time. 

\begin{theorem}[Sequential Syndrome Computation via MBU --- see \cref{sec:dqi_quantum_circuit}]
\label{thm:sequential_mbu}
The $nb$-qubit DQI syndrome register $|B^\intercal e\rangle$ can be prepared without an $mb$-qubit error register by iterating over the $m$ constraints sequentially using a single $b$-qubit ancilla. In each iteration $i \in \{1, \dots, m\}$:
\begin{enumerate}
    \item The $i$-th error term $|e_i\rangle$ is temporarily generated conditioned on the Sparse Dicke State.
    \item The term is coherently added to the $nb$-qubit syndrome register.
    \item The ancilla $|e_i\rangle$ is immediately measured in the X-basis (MBU), clearing it for the next iteration.
\end{enumerate}
While MBU introduces a classical measurement outcome $c_i$ and a known phase kickback $(-1)^{c_i \cdot e_i}$ onto the quantum state, these phase errors can be deterministically corrected later during the reversible decoding step. By reusing the $b$-qubit ancilla across all $m$ iterations and utilizing the Sparse Dicke State, the space complexity of the state preparation and syndrome computation is bounded by $\approx 1.5nb$ qubits. Together with the reversible decoding step for Reed Solomon codes, the total space cost of our DQI algorithm is bounded by $2nb + \mathcal{O}(\log_2 n)$ qubits.
\end{theorem}

\subsubsection{Space-Optimal Extended Euclidean Algorithm (EEA) Circuits}

The Extended Euclidean Algorithm (EEA) is the primary computational bottleneck in reversible Reed-Solomon decoding (as well as in other prominent quantum algorithms, such as Shor's algorithm for Elliptic Curve Cryptography; see \cref{sec:broader_applicability}). Given two input polynomials of degree up to $n$ over $\text{GF}(2^b)$, requiring $2nb$ qubits of input space, historical quantum EEA implementations have faced a strict space-time trade-off (see \cref{tbl:prior_art_eea}). 
Implementations based on classic Euclid's GCD algorithm \cite{kaye2004optimizedquantumimplementationelliptic, proos2004shorsdiscretelogarithmquantum} are stepwise reversible, resulting in low qubit counts, but high gate counts due to unrolling polynomial long division via complex synchronization schemes. 
Conversely, modern constant-time, division-free GCD algorithms (such as Binary GCD \cite{Stein1967} or Bernstein-Yang GCD \cite{cryptoeprint:2019/266}) reduce gate counts by relying on simple shift-and-add operations. However, because each iteration is not stepwise reversible, they generate continuous garbage, which balloons the space requirement.
We resolve this historical trade-off by introducing two novel compilation strategies. Both strategies deterministically achieve a leading-order space complexity of $2nb$ qubits while simultaneously improving upon state-of-the-art gate counts.

\begin{theorem}[Space-Time Complexity of Reversible EEA --- see \cref{sec:eea_circuit}]
\label{thm:eea_complexity_main}
Given two polynomials of maximum degree $n$ over $\text{GF}(2^b)$, the Extended Euclidean Algorithm can be computed reversibly using exactly $2nb + \mathcal{O}(\log_2 n)$ qubits under two distinct computational models:
\begin{enumerate}
    \item \textbf{Explicit Access (Synchronized EEA):} When the full B\'ezout coefficient polynomials must be explicitly materialized in memory, we utilize a unified-cycle synchronized architecture. By dynamically utilizing guaranteed padding within shared registers for in-place quotient storage, we eliminate heuristic auxiliary registers. This explicitly computes the B\'ezout coefficients using $\approx 6n^2$ quantum-quantum multiplications.
    
    \item \textbf{Implicit Access (Dialog Representation):} When explicit materialization is not needed, we introduce a linear representation called the \textbf{Dialog}. The Dialog is a compact execution trace of a constant-time division-free EEA (e.g., Bernstein-Yang). Because the number of active polynomial coefficients strictly decreases as the algorithm progresses, we dynamically repurpose freed memory to store the growing Dialog. This entirely in-place architecture requires only $\approx 2n^2$ quantum-quantum multiplications, after which the B\'ezout coefficients exist implicitly and can be applied on-the-fly.
\end{enumerate}
\end{theorem}

\subsubsection{Reversible Reed-Solomon Decoder Construction}

The computational core of the DQI algorithm is the coherent uncomputation of the error register. 
Because our space-efficient MBU state preparation leaves known phase errors $(-1)^{c_i \cdot e_i}$ on the quantum state, we must decode the $nb$-qubit syndrome register $|S(z)\rangle$ to recover the error values and correct these phases. 
Syndrome decoding for Reed-Solomon codes requires three stages: (1) solving the fundamental key equation $\sigma(z)S(z) \equiv \Omega(z) \pmod{z^{2\ell}}$ to find the error-locator polynomial $\sigma(z)$ and error-evaluator polynomial $\Omega(z)$; (2) identifying error locations via Chien search \cite{Chien1964}; and (3) computing error values via Forney's algorithm \cite{Forney1965}. 

\begin{theorem}[Reversible Reed-Solomon Decoding --- see \cref{sec:reversible_rs_decoder}]
\label{thm:rs_decoder_main}
A complete, reversible syndrome decoder for a Reed-Solomon code of length $m$ and dimension $n$ over $\text{GF}(2^b)$, encompassing the Half-EEA, Chien search, and Forney's algorithm, can be implemented using $2nb + \mathcal{O}(\log_2 n)$ qubits. The dominant gate complexity depends strictly on the chosen EEA access model:
\begin{itemize}
    \item \textbf{Explicit Access:} Evaluates $\sigma(z)$ and $\Omega(z)$ explicitly, requiring $\mathcal{O}(n^2)$ quantum-quantum multiplications for the EEA, and $\mathcal{O}(mn)$ quantum-classical multiplications for the Chien search and Forney's algorithm.
    \item \textbf{Implicit Access:} Evaluates the polynomials implicitly via Dialog playback, requiring $\mathcal{O}(n^2)$ quantum-quantum multiplications for the EEA, and $\mathcal{O}(mn)$ quantum-quantum multiplications to evaluate the Dialogs during Chien search and Forney's algorithm.
\end{itemize}
\end{theorem}

For DQI over $\text{GF}(2^b)$, quantum-classical multiplication requires only Clifford gates, whereas quantum-quantum multiplication requires expensive Toffoli gates. Because the $\mathcal{O}(mn)$ decoding stages dominate the overall runtime ($m \gg n$), the Explicit Access model is more efficient for this application (see \cref{tbl:decoder_costs})

\subsubsection{Logical and Physical Resource Estimates}

To contextualize the practical feasibility of DQI and provide concrete targets for near-to-intermediate term hardware development, we perform a comprehensive resource analysis of the complete algorithm. We differentiate between idealized logical resources and the fully fault-tolerant physical resources required to execute the algorithm reliably under realistic hardware noise. 

\begin{theorem}[Concrete Resource Estimates for Verifiable Quantum Advantage --- see \cref{sec:results}]
\label{thm:concrete_resources}
Consider a classically intractable TBT-OPI instance over $\text{GF}(2^{12})$ with $m=4095$ constraints and $n=70$ variables. Solving this instance requires $> 10^{23}$ classical Extended Prange trials. 
\begin{itemize}
    \item \textbf{Logical Resources:} Synthesis in the Qualtran compilation framework~\cite{harrigan2024expressinganalyzingquantumalgorithms} yields a complete DQI circuit requiring $1885$ logical qubits and $5.72 \times 10^6$ logical Toffoli gates per shot. 
    \item \textbf{Physical Resources:} Assuming a fault-tolerant superconducting architecture with a uniform physical gate error rate of $0.1\%$, a $1\mu\text{s}$ surface code cycle time, and utilizing 2D Yoked Surface Codes \cite{gidney2025yoked} for cold storage with standard surface code patches (distance $d=21$) for active computation, this instance can be solved using approximately $800,000$ physical qubits in roughly $6$ hours of runtime. 
\end{itemize}
\end{theorem}

These resource requirements compare highly favorably to other leading candidates for verifiable quantum advantage. For instance, the per-shot logical Toffoli count required to solve this intractable OPI instance is roughly $100\times$ smaller than that required to factor a 2048-bit RSA integer \cite{gidney2025factor2048bitrsa}, $10\times$ smaller than that required to break a $256$-bit ECDLP problem \cite{babbush2026securingellipticcurvecryptocurrencies} and also smaller than state-of-the-art interactive protocols based on cryptographic trapdoors \cite{kahanamoku2022classically, kahanamokumeyer2024fastquantumintegermultiplication}.

\section{Methods for Optimized Implementation}\label{sec:methods}

We begin by presenting an improved, space-efficient construction of the DQI quantum circuit. We then present new techniques for compiling the Extended Euclidean Algorithm (EEA) for the two regimes---first, where one needs to explicitly compute the Bézout coefficients in memory, and second, where an implicit access to the Bézout coefficients is sufficient. For both scenarios, our goal is to find a construction that minimizes the qubit counts.  For the first case, we make several improvements to the classical reversible EEA construction by \cite{kaye2004optimizedquantumimplementationelliptic, proos2004shorsdiscretelogarithmquantum}, resulting in lower qubit and gate counts. 
For the second case, we first formalize the idea of having implicit access to the Bézout coefficients and show how one can take advantage of the low gate counts for constant-time division-free variants of EEA \cite{Stein1967, cryptoeprint:2019/266} while avoiding the high ancilla overhead. In both cases, we achieve a leading-order space complexity of $2nb + \mathcal{O}(\log_{2}(n))$. 
We believe these improved techniques for compiling the Extended Euclidean Algorithm will be useful beyond DQI, especially in the context of Elliptic Curve Cryptography \cite{häner2020improvedquantumcircuitselliptic, litinski2023compute256bitellipticcurve} (see \cref{sec:broader_applicability}). We then specialize the discussion to the Optimal Polynomial Intersection (OPI) problem over binary extension fields, motivated by the fact that arithmetic circuits over binary extension fields are significantly cheaper to compile. In the end, we show how the decoding problem for Reed-Solomon codes can be solved using the Extended Euclidean Algorithm, along with other subroutines like Chien Search \cite{Chien1964} and Forney's algorithm \cite{Forney1965}. We present optimized quantum circuits to account for the cost of these subroutines. \cref{tbl:variables} provides a reference for the key parameters used throughout our analysis.

\begin{table}[H]
    \centering
    \caption{Key parameters and their dual roles in the DQI framework \cite{jordan2024optimizationdecodedquantuminterferometry}, connecting the optimization problem with the corresponding coding theory problem. 
    For Reed Solomon codes, the field size $q \geq m$. For our resource estimates, $m = q - 1$ and $\ell \approx n/2$.
    }
    \label{tbl:variables}
    \renewcommand{\arraystretch}{1.4} 
    \begin{tabular}{| c | p{8cm} | p{8cm} |}
    \hline
    \textbf{Symbol} & \textbf{Role in Optimization (max-LINSAT/OPI)} & \textbf{Role in Coding Theory ($C^\perp$)} \\
    \hline
    \hline
    $q$ & Size of the finite field $\F_q$ over which the problem is defined. & Alphabet size of the code. \\
    \hline
    $b$ & Bit-length of the field elements, where $b = \lceil\log_2 q\rceil$. & Bit-length of the code symbols. \\
    \hline
    $m$ & Number of constraints in the optimization problem. & Block length of the dual code $C^\perp$. \\
    \hline
    $n$ & Number of variables in the optimization problem. & Dimension of the primal code $C$; length of the syndrome of dual code $C^\perp$. \\
    \hline
    $\ell$ & Maximum number of errors DQI is configured to handle, which sets the degree of the enhancing polynomial $P$. & The error-correction capability (number of correctable errors) of the decoder for $C^\perp$. \\
    \hline
    $r$ & Size of the allowed sets $F_i$ for each constraint. & (Not a standard coding parameter, but influences the problem instance). \\
    \hline
    $B$ & The $m \times n$ constraint matrix defining the instance. & Generator matrix of the primal code $C = \{x B : x \in \F_q^n\}$. \\
    \hline
    $B^\intercal$ & The $n \times m$ matrix used to compute the syndrome. & Parity-check matrix of the dual code $C^\perp = \{c \in \F_q^m : B^T c = 0\}$. \\
    \hline
    \end{tabular}
\end{table}

\subsection{The DQI Quantum Circuit: An Improved Construction for Qubit Efficiency}\label{sec:dqi_quantum_circuit}

The goal of the Decoded Quantum Interferometry (DQI) circuit is to efficiently prepare the state $|P(f)\rangle = \sum_x P(f(x))|x\rangle$, where $P$ is an appropriately normalized degree-$\ell$ polynomial. 
In the original construction\cite[Section-8]{jordan2024optimizationdecodedquantuminterferometry}, the algorithm requires simultaneous instantiation of an $mb$-sized error register, to hold a superposition of error patterns $e \in \F_q^m$, and an $nb$-sized syndrome register to hold the syndrome values $s = B^T e$. 
This leads to a total space complexity of $(m + n)b$ plus the ancilla overhead due to reversible decoding. 
Note that the decoding problem is defined on an input of size $nb$ (the length of the syndrome), and hence for the regime where $m \gg n$, the multiplicative $mb$ qubit overhead in the DQI circuit is prohibitive. 

We present an optimized circuit construction that significantly reduces qubit overhead by reformulating the Dicke state preparation, syndrome computation, and reversible decoding steps (stages 1, 2 and 3). 
The Dicke states in the DQI circuit are used to encode the locations of $\ell < \frac{n}{2}$ errors in the length $m$ codeword. Instead of using a dense representation consisting of $m$ qubits, we use a sparse representation consisting of $\ell \cdot \lceil \log_2{m} \rceil \leq \ell \cdot b$ qubits to encode this information. We show how to efficiently prepare Sparse Dicke States in \cref{sec:sparse_dicke_state}.
Next, instead of generating the entire $mb$-qubit error state $|e\rangle$ to encode the values of errors at each of the $m$ locations in the codeword, we employ a sequential approach utilizing Measurement-Based Uncomputation (MBU) \cite{Jones2013, Gidney2018}. 
We iterate through the $m$ constraints, generate one error symbol $e_i$ (of size $b$ qubits) at a time, update the syndrome register with its contribution, and immediately uncompute $e_i$ via measurement in the X-basis. 
This leads to phase errors that we later fix as part of the reversible decoding step, where we sequentially generate each decoded error term $e_i$ and apply a phase fixup using the measurement result $c_i$. This allows us to reuse a single $b$-qubit ancilla register for all error terms during state preparation. 
Using Sparse Dicke state preparation and sequential computation of the error terms, the ancilla cost for preparing the $nb$-qubit syndrome register reduces to $nb$ qubits. We will later show how to perform the reversible decoding on the $nb$-qubit syndrome register using $nb + \mathcal{O}(\log_{2}(n))$ ancilla qubits, thus achieving a total space cost of $2nb + \mathcal{O}(\log_{2}(n))$ qubits.

The following description of the DQI quantum circuit is for the general max-LINSAT case over a Galois field $\F_q$. The construction proceeds through four main stages. 
The evolution of the quantum state across these stages is summarized below, utilizing an error locator register ($\ell \cdot b$ qubits), an error value register ($b$ qubits), and a syndrome register ($n \cdot b$ qubits). We omit normalization factors for clarity.

\begin{align*}
&|0\rangle_{\ell b} |0\rangle_{b} |0\rangle_{nb} \\
&
\xrightarrow[\text{Sparse Dicke State Preparation}]{\text{Stage 1}} 
\sum_{k=0}^{\ell} \hat{w_k} \left(\sum_{\substack{1 \le j_1 < j_2 < \dots < j_{k} \le m \\ j_{k+1} \dots j_{\ell} = 0}} \ket{j_1}_b \ket{j_2}_b \cdots \ket{j_\ell}_b\right) |0\rangle_b |0\rangle_{nb}
\left(
\text{ where }\hat{w}_k = w_k \binom{m}{k}^{-1/2}
\right)
\\
&\left\{
\begin{aligned}
&\xrightarrow{\text{2a: Check if $i \in [j_1, \dots j_\ell]$ and store the result in qubit $y_i$}} 
\sum_{k=0}^{\ell} \hat{w}_k \ket{SD^{m}_{k}}_{lb} |y_i\rangle\sum_{e_i\in\mathbb{F}_{q}}\tilde{g}_{i}(e_i)\ket{e_i}_b |s\rangle_{nb} \\
&\xrightarrow{\text{2b: Apply } G_i \text{ on qubit } y_i \text{ to generate error term } e_i} 
\sum_{k=0}^{\ell} \hat{w}_k \ket{SD^{m}_{k}}_{lb} \sum_{e_i\in\mathbb{F}_{q}}\tilde{g}_{i}(e_i)\ket{e_i}_b |s\rangle_{nb} \\
&\xrightarrow{\text{2c: Update syndrome register } \ket{s} \text{ with } B^T_i \cdot e_i} 
\sum_{k=0}^{\ell} \hat{w}_k \ket{SD^{m}_{k}}_{lb} \sum_{e_i\in\mathbb{F}_{q}}\tilde{g}_{i}(e_i)\ket{e_i}_b  |s + B^T_i \cdot e_i\rangle_{nb}\\
&\xrightarrow{\text{2d: MBU on error register } \ket{e_i} \text{ (record } c_i \text{)}} 
\sum_{k=0}^{\ell} \hat{w}_k \ket{SD^{m}_{k}}_{lb} \sum_{e_i\in\mathbb{F}_{q}}\tilde{g}_{i}(e_i)(-1)^{e_i \cdot c_i}  |s + B^T_i \cdot e_i\rangle_{nb}
\\
\end{aligned}
\right\} \text{Repeat } i=1 \dots m \\
&\xrightarrow[\text{Syndrome Computation}]{\text{Stage 2}} 
\sum_{k=0}^{\ell} \hat{w}_k \sum_{|e|=k} \left(\prod_{i=1}^{m} \tilde{g}_{i}(e_i)\right) (-1)^{c \cdot e} |B^T e\rangle_{nb} \\
&\xrightarrow[\text{Reversible Decoding}]{\text{Stage 3}} 
\sum_{k=0}^{\ell} \hat{w}_k \sum_{|e|=k} \left(\prod_{i=1}^{m} \tilde{g}_{i}(e_i)\right) |B^T e\rangle_{nb} \equiv |\widetilde{P(f)}\rangle_{nb} \\
&\xrightarrow[\text{IQFT + Measurement}]{\text{Stage 4}} \sum_x P(f(x))|x\rangle_{nb} \equiv |P(f)\rangle_{nb}
\end{align*}

\begin{itemize}
    \item \textbf{Stage 1: Sparse Dicke State Preparation:}
    We prepare a superposition $\sum_{k=0}^{l} w_k |k\rangle$, using classically pre-computed coefficients $w_k$ which define the optimal enhancing polynomial $P$ \cite{gosset2024quantumstatepreparationoptimal, Low2024, berry2025rapid}. 
    This is used to prepare the $\ell \cdot b$-qubit error locator register into the corresponding superposition of Sparse Dicke states:
    
    \begin{equation}
        \sum_{k=0}^{l}w_k \binom{m}{k}^{-1/2}\left(\sum_{\substack{1 \le j_1 < j_2 < \dots < j_{k} \le m \\ j_{k+1} \dots j_{\ell} = 0}} \ket{j_1}_b \ket{j_2}_b \cdots \ket{j_\ell}_b\right)
    \end{equation}

    In \cref{sec:sparse_dicke_state}, we describe a way to prepare Sparse Dicke States using $2\cdot k \cdot \log_2{m}$ qubits and $\mathcal{O}(m.k + k^2b^2)$ gates. The sparse construction is useful in reducing the qubit counts for instances where $n \ll m$.

    \item \textbf{Stage 2: Sequential Syndrome Computation and MBU:}
    This stage translates the error locator register into a superposition of syndromes, using our space-efficient sequential approach.

    \textit{Defining the Constraint Encoding Gates $G_i$:}
    The operations in this stage depend on the specific constraints of the optimization problem. 
    Recall the objective function $f(\mathbf{x}) = \sum_{i=1}^m f_i(\mathbf{b}_i \cdot \mathbf{x})$ with $f_i:\mathbb{F}_q \to \{+1,-1\}$. 
    As detailed in \cite[Section 8.2.1]{jordan2024optimizationdecodedquantuminterferometry}, we find it convenient to work in terms of $g_i$, which we define as $f_i$ shifted and rescaled such that:
    \begin{equation}
        g_i(x) := \frac{f_i(x) - \bar{f}}{c},
    \end{equation}
    where $\bar{f}$ is the average value of $f_i$ over $\mathbb{F}_q$ and $c$ is a normalization constant. 
    This standardization ensures that the Fourier transform of $g_i$, denoted $\tilde{g}_i(e)$, vanishes at $e=0$ and is normalized ($\sum_e |\tilde{g}_i(e)|^2 = 1$).
    We then define an isometry $G_i$, which acts on a single qubit $y_i$ storing whether an error occurs at index $i$ or not, and maps it to a superposition over $b$-qubit error terms that encode these Fourier coefficients:
    \begin{equation}
        G_i|0\rangle = |0\rangle_b, \quad G_i|1\rangle = \sum_{e \in \mathbb{F}_q, e\neq 0} \tilde{g}_i(e)|e\rangle_b.
    \end{equation}

    \textit{Sequential Syndrome Computation:}
    We now iterate $i$ from $1$ to $m$, utilizing a single $b$-qubit ancilla register:
    \begin{itemize}
        \item[(2a)] \textbf{Check if marked in Sparse Dicke State:} Check whether index $i$ is contained in the Sparse Dicke state and record the outcome in a qubit $y_i$. 
        Naively, this requires performing a $b$-bit comparison with each of the $\ell$ registers per iteration, requiring a total of $m \cdot \ell \cdot b$ gates across $m$ iterations. 
        However, since we are sequentially comparing each of the $\ell$ registers in the Sparse Dicke State with consecutive classical constants $i=1, 2, \dots, m$, we can use a unary iteration \cite{Babbush_2018, Khattar_2025} circuit to merge adjacent comparisons and get a total cost of $m \cdot \ell$ gates across all $m$ iterations.

        \item[(2b)] \textbf{Error Generation:} Apply $G_i$ on the $i$-th qubit of the mask register $\ket{y_i}$ to temporarily generate the $i$'th error term $\ket{e_i}_b$ using the $b$-qubit ancilla register. 
        \item[(2c)] \textbf{Syndrome Update:} Compute the contribution of $e_i$ to the syndrome register $\ket{s}$ by coherently adding $b_i \cdot e_i$ to $\ket{s}$ (where $b_i$ is the $i$-th column of $B^T$).
        \item[(2d)] \textbf{Measurement-Based Uncomputation:} Immediately uncompute the ancilla $\ket{e_i}$ using MBU \cite{Jones2013, Gidney2018}. The $b$ qubits are measured in the X basis, recording the classical outcome $c_i$. This resets the ancilla for the next iteration but introduces a known phase kickback $(-1)^{c_i \cdot e_i}$.
    \end{itemize}
    At the end of the iteration, we simply do a measurement-based uncomputation of the Sparse Dicke state and perform phase fixups corresponding to the error locations during the reversible decoding step.
    By reusing a single ancilla register across all $m$ iterations, we avoid the $mb$-qubit overhead of the original DQI construction. The final state of the system is
    \begin{equation}
        \sum_{k=0}^{l} \frac{w_k}{\sqrt{\binom{m}{k}}} \sum_{|e|=k} \left(\prod_{i=1}^{m} \tilde{g}_{i}(e_i) 
        (-1)^{c_i \cdot e_i}\right)
        |B^T e\rangle_{nb}.
    \end{equation}

    \item \textbf{Stage 3: Phase Fixups via Reversible Decoding:}
    In the original DQI construction \cite{jordan2024optimizationdecodedquantuminterferometry}, a reversible decoder was used to uncompute the explicit error register $|e\rangle$. 
    In our construction, the error register was already uncomputed via MBU in Stage 2. 
    However, we must now correct the phase errors $(-1)^{c\cdot e}$ introduced by the measurements. 
    Here, we perform reversible decoding using the syndrome register $|s\rangle = |B^T e\rangle$. 
    We sequentially compute each decoded error term $e_{i}$ and use the measurement results $c_i$ to fix the phase error $(-1)^{c_i \cdot e_i}$. 
    For the Optimal Polynomial Intersection (OPI) problem, this requires a reversible implementation of a decoder for Reed-Solomon codes, like the Berlekamp-Massey algorithm \cite{berlekamp2015algebraic} or the Extended Euclidean Algorithm-based decoder \cite{SUGIYAMA197587, sarwate2009modifiedeuclideanalgorithmsdecoding}. 
    The final state of the system after this step is:
    \begin{equation}
        \sum_{k=0}^{l} \frac{w_k}{\sqrt{\binom{m}{k}}} \sum_{|e|=k} \left(\prod_{i=1}^{m} \tilde{g}_{i}(e_i)\right) 
        |B^T e\rangle_{nb}
        \equiv |\widetilde{P(f)}\rangle_{nb}        
    \end{equation}

    \item \textbf{Stage 4: Final Transformation and Measurement:}
    After the successful phase fixups, the syndrome register is left in a state that is the Quantum Fourier Transform of the desired output state. To obtain the final state, an Inverse Quantum Fourier Transform (IQFT) is applied to the $n$ qudits of the syndrome register. This produces the final DQI state:
    \begin{equation}
        |P(f)\rangle = \sum_x P(f(x)) |x\rangle.
    \end{equation}
\end{itemize}

\subsection{Efficient Quantum Circuits for Extended Euclidean Algorithm}\label{sec:eea_circuit}
The Extended Euclidean Algorithm (EEA) is a cornerstone of computational number theory and a critical subroutine in algorithms for both classical and quantum computing \cite{cormen2022introduction}. 
While often introduced for integers, its principles generalize to other Euclidean domains like univariate polynomials with coefficients in a finite field, which is relevant to Reed-Solomon decoding. 
Given two polynomials, $A(z)$ and $B(z)$, Euclid's algorithm iteratively generates a sequence of remainder polynomials $r_i(z)$ using the recurrence $r_{i-2}(z) = r_{i-1}(z) q_i(z) + r_{i}(z)$, where $r_{-1}(z)=A(z)$, $r_{0}(z) = B(z)$ and the quotient $q_i(z)$ is chosen such that $\deg(r_i(x)) < \deg(r_{i - 1}(x))$. The EEA is an extension to Euclid's algorithm, which, alongside this sequence, also maintains two corresponding cofactor sequences, $u_i(z)$ and $v_i(z)$, that are updated at each step to preserve a crucial linear relationship known as Bézout's identity. At every iteration $i$, the triplet $(r_i(z), u_i(z), v_i(z))$ satisfies the invariant:
\begin{equation}
    A(z)u_i(z) + B(z)v_i(z) = r_i(z).
\end{equation}
This process is guaranteed to terminate because the degree of the remainder polynomial, $\deg(r_i(z))$, is a non-negative integer that strictly decreases with each iteration. The algorithm stops when the remainder $r_m(z)$ becomes the zero polynomial. At this point, the previous remainder, $r_{m-1}(z)$, is the $\gcd(A(z), B(z))$, and the corresponding cofactors $u_{m-1}(z)$ and $v_{m-1}(z)$ are the final Bézout coefficients.

Each step of the EEA can be expressed as an application of a $2\times2$ transition matrix $\tau_{i}$ on the vector $[r_{i-2}, r_{i-1}]^T$. The product of the first $i$ transition matrices store the Bézout coefficients $u_i$, $u_{i-1}$ and $v_i$, $v_{i-1}$. This matrix representation of the EEA will be useful as we describe our new constructions below.

\begin{align}
\begin{pmatrix}
r_{i-1}\\
r_{i}
\end{pmatrix}=
\underbrace{
\begin{pmatrix}
0 & 1\\
1 & -q_{i}
\end{pmatrix}
}_{\tau_{i}}
\begin{pmatrix}
r_{i-2}\\
r_{i-1}
\end{pmatrix}
&&
\begin{pmatrix}
v_{i-1}\\
v_{i}
\end{pmatrix}=
\tau_\text{i}
\begin{pmatrix}
v_{i-2} \\
v_{i-1}
\end{pmatrix}
&&
\begin{pmatrix}
u_{i-1}\\
u_{i}
\end{pmatrix}=
\tau_\text{i}
\begin{pmatrix}
u_{i-2} \\
u_{i-1}
\end{pmatrix}
\end{align}

\begin{align}
\tau_{\text{total}}
= 
\tau_{m}
\cdot
\tau_{m - 1}
\dots
\tau_{1}
=
\begin{pmatrix}
u_{m-1} &  v_{m-1}\\
u_{m} & v_{m}
\end{pmatrix}
&&
\begin{pmatrix}
\gcd(A, B)\\
0
\end{pmatrix}=
\tau_\text{total}
\begin{pmatrix}
A \\
B
\end{pmatrix}
\end{align}

One important application of the Extended Euclidean Algorithm is to compute multiplicative inverses in a finite field $\mathbb{F}_q \cong \mathbb{F}_p[z]/P(z)$, where $P(z)$ is an irreducible polynomial.  
To find the inverse of a polynomial $A(z)$, one applies the EEA to $P(z)$ and $A(z)$. This yields cofactor polynomials $u(z)$ and $v(z)$ such that $P(z)u(z) + A(z)v(z) = \gcd(P(z), A(z)) = 1$. Taking this equation modulo $P(z)$ gives $A(z)v(z) \equiv 1 \pmod{P(z)}$, revealing that the Bézout coefficient $v(z)$ is the modular inverse. 
This operation is the principal computational bottleneck in Shor's algorithm for Elliptic Curve Cryptography (ECC), a cornerstone of modern public-key infrastructure. 
As a result, a significant body of research, summarized in \cref{tbl:prior_art_eea} and \cref{tbl:prior_art_eea_division}, has been dedicated to designing resource-efficient quantum circuits for the EEA.

Early work by \textcite{proos2004shorsdiscretelogarithmquantum, kaye2004optimizedquantumimplementationelliptic} showed that the standard Euclidean Algorithm based on polynomial long division is stepwise reversible and can be compiled using $3nb + \mathcal{O}(\log{n})$ qubits using a register-sharing technique, where the registers holding the remainders and the Bézout coefficients are overlapped to use a total of $2nb$ qubits and another $nb$ ancilla qubits are used to store the quotient in each iteration. 
By tolerating a small error, they show how to further reduce the qubit counts down to $2nb + \mathcal{O}(\log{n})$, by using a $\mathcal{O}(\log{n})$ sized register to store the quotients since large quotients occur relatively rarely in the EEA.
Their circuit design relies on a complex synchronization scheme, and they do not analyze the constant factors for gate counts.
More recent works \cite{roetteler2017quantumresourceestimatescomputing, häner2020improvedquantumcircuitselliptic, cryptoeprint:2020/1296, kim2023newspaceefficientquantumalgorithm} focus on using constant-time division-free variants of the Extended Euclidean Algorithm, like Binary GCD \cite{Stein1967} and Bernstein-Yang GCD \cite{cryptoeprint:2019/266}, to reduce the gate counts but suffer from significantly higher qubit counts because each iteration of these constant-time EEA algorithms is not stepwise reversible and generates additional garbage that leads to a higher ancilla overhead.

In this work, we provide two improved constructions for compiling the EEA, the first where the EEA explicitly computes the Bézout coefficients in memory, and the second where it suffices to have an implicit representation of Bézout coefficients. 
For both constructions, we reduce the qubit counts to $2nb + \mathcal{O}(\log{n})$, with gate counts lower than the previous state of the art.

\begin{table}[H]
    \centering
    \caption{
    Cost of compiling the Extended Euclidean Algorithm for two degree-$n$ polynomials over $\F_q$ with $b=\lceil\log_2{q}\rceil$.
    Gate cost is specified in terms of the number of calls to the dominant subroutine of quantum-quantum multiplication of elements in the underlying field $\F_q$.
    }
    \begin{tabular}{|l|c|c|c|}
    \hline
    \textbf{Source} & \textbf{EEA Technique} & \textbf{Qubit Cost} & \textbf{Dominant gate cost} \\
    \hline
    \textcite{kaye2004optimizedquantumimplementationelliptic} & Euclids GCD & $3nb + \mathcal{O}(\log_2{n})$ & $12n^2$ multiplications \\
    \textcite{roetteler2017quantumresourceestimatescomputing} & Binary GCD & $4nb + 2n + \mathcal{O}(\log n)$ & $4n^2$ multiplications \\
    \textcite{cryptoeprint:2020/1296} & Bernstein-Yang GCD & $4nb + n+\mathcal{O}(\log n)$ & $4n^2$ multiplications \\
    \textcite{kim2023newspaceefficientquantumalgorithm} & Bernstein-Yang GCD & $4nb + 0.5n + \mathcal{O}(\log n)$ & $4n^2$ multiplications \\
    \hline
    \textbf{This Work (Explicit Bézout)} & Euclids GCD & $2nb + \mathcal{O}(\log_2{n})$ & $6n^2$ multiplications\\
    \textbf{This Work (Implicit Bézout)} & Bernstein-Yang GCD & $2nb + \mathcal{O}(\log_2{n})$ & $2n^2$ multiplications\\
    \hline
    \end{tabular}
    \label{tbl:prior_art_eea}
\end{table}

\subsubsection{Explicit versus Implicit Access to Bézout Coefficients}

A key theme of our work is the optimization of quantum circuits by carefully considering how the results of the EEA are used by the broader algorithm. We distinguish between two computational models: one requiring \emph{explicit} access to the Bézout coefficients, and another where \emph{implicit} access is sufficient. This distinction allows us to select the most resource-efficient EEA implementation for the task at hand.

As outlined at the beginning of this section, any execution of the EEA can be viewed as generating a sequence of $2 \times 2$ transition matrices, $\tau_1, \tau_2, \dots, \tau_m$. The product of these matrices, $\tau_{\text{total}} = \tau_m \cdot  \tau_{m-1} \cdots \tau_1$, contains the final Bézout coefficients as its entries.

\begin{itemize}
    \item \textbf{Explicit Access:} This is the conventional model, where the primary goal of the EEA circuit is to compute and store the full product matrix $\tau_{\text{total}}$. In the context of polynomial inputs, this means computing the coefficients of the final Bézout polynomials (e.g., $v_{m-1}(z)$) and storing them explicitly in a quantum register. This approach is necessary when the full polynomial is required for subsequent algebraic manipulations.

    \item \textbf{Implicit Access.} In this model, we avoid the costly step of multiplying the sequence of transition matrices. Instead, we only compute and store a compact representation of each individual matrix $\tau_i$. This metadata is sufficient to reconstruct and apply each $\tau_i$ on the fly. The full Bézout coefficients are never materialized in memory but exist \emph{implicitly} as the result of applying the sequence of stored transformations. Any operation requiring the Bézout coefficients is simply reformulated as a sequential application of the $\tau_i$ matrices. 

\end{itemize}

We illustrate the power of this implicit approach with two concrete examples:

\begin{enumerate}
    \item \textbf{Modular Polynomial Division:} Consider the task of computing $A(z)B(z)^{-1} \pmod{P(z)}$, where $P(z)$ is an irreducible polynomial.
    \begin{itemize}
        \item The \emph{explicit} ``invert-then-multiply'' approach first runs $\text{EEA}(P, B)$ to explicitly compute and store the inverse polynomial $v_{m-1}(z) = B(z)^{-1} \pmod{P(z)}$. 
        It then performs a separate quantum-quantum multiplication of $A(z)$ by this inverse polynomial.
        \item The \emph{implicit} ``direct division'' approach instead runs $\text{EEA}(P, B)$ to generate the sequence of matrices $\tau_1, \tau_2, \dots, \tau_m$. To compute the division, it applies a sequence of matrix-vector multiplications to a different input vector, $[0, A(z)]^T$. The final state after applying $\tau_{\text{total}}$ is precisely $[A(z)B(z)^{-1}, 0]^T \pmod{P(z)}$, avoiding the need to ever store the full inverse polynomial.
    \end{itemize}
    We compare the costs of these approaches for modular division in \cref{tbl:prior_art_eea_division}.

    \item \textbf{Polynomial Evaluation:} Consider evaluating the Bézout coefficient polynomial $v_{m-1}(z)$ at a specific point $\gamma \in \F_q$.
    \begin{itemize}
        \item With \emph{explicit} access, one would use the stored coefficients of $v_{m-1}(z)$ and perform a standard polynomial evaluation (e.g., via Horner's method), requiring a series of quantum-classical multiplications.
        
        \item With \emph{implicit} access, we first evaluate the simple polynomial entries within each transition matrix $\tau_i$ at the point $\gamma$. 
        This yields a sequence of constant $2 \times 2$ matrices with entries in $\F_q$. 
        We then apply this sequence of constant matrices to an initial vector, such as $[0, 1]^T$. The result of these operations over $\F_q$ is the desired value, $v_{m-1}(\gamma)$. Once again, we did not ever store the full Bézout coefficient polynomial $v_{m - 1}(z)$ in memory.
    \end{itemize}
\end{enumerate}

Our work provides optimized circuits for both models. For explicit access, we improve upon the synchronized algorithm by \textcite{proos2004shorsdiscretelogarithmquantum}, while for implicit access, we use the constant-time division-free variants of EEA \cite{Stein1967, cryptoeprint:2019/266} as basis for the \emph{Dialog} representation detailed in the following subsections.

\begin{table}[H]
    \centering
    \caption{
    Cost of compiling modular division using Extended Euclidean Algorithm for two degree $n$ polynomials over $\F_q$ with $b=\lceil\log_2{q}\rceil$. 
    Gate cost is specified in terms of the number of calls to the dominant subroutine of quantum-quantum multiplication of elements in the underlying field $\F_q$.
    We wish to achieve $\ket{A(z)}_b\ket{B(z)}_b \rightarrow \ket{\text{junk}}\ket{\frac{A(z)}{B(z)} \pmod{p(z)}}_b$. 
    When accounting for qubit costs, we assume multiplication of two elements in $\F_q$ does not consume any ancilla.
    }
    \begin{tabular}{|l|c|c|c|}
    \hline
    \textbf{Source} & \textbf{Approach} & \textbf{Qubit Cost} & \textbf{Gate cost} \\
    \hline
    \textcite{kaye2004optimizedquantumimplementationelliptic} & invert-then-multiply & $4nb + \mathcal{O}(\log_2{n})$ & $13n^2$ multiplications \\
    \textcite{roetteler2017quantumresourceestimatescomputing} & invert-then-multiply & $6nb + 2n + \mathcal{O}(\log n)$ & $5n^2$ multiplications \\
    \textcite{cryptoeprint:2020/1296} & invert-then-multiply & $6nb + n+\mathcal{O}(\log n)$ & $5n^2$ multiplications \\
    \textcite{kim2023newspaceefficientquantumalgorithm} & invert-then-multiply & $6nb + 0.5n + \mathcal{O}(\log n)$ & $5n^2$ multiplications \\
    \hline
    \textbf{This Work (Explicit Bézout)} & invert-then-multiply & $4nb + \mathcal{O}(\log_2{n})$ & $7n^2$ multiplications\\
    \textbf{This Work (Implicit Bézout)} & direct-division & $4nb + \mathcal{O}(\log_2{n})$ & $4n^2$ multiplications\\
    \hline
    \end{tabular}
    \label{tbl:prior_art_eea_division}
\end{table}

\subsubsection{Synchronized Reversible EEA for Explicit Bézout coefficients}\label{sec:zalka_eea_optimized}
We advance the synchronized register-sharing framework of Zalka et al. \cite{proos2004shorsdiscretelogarithmquantum, kaye2004optimizedquantumimplementationelliptic} through two key innovations, resulting in an implementation that is both rigorously space-optimal and time-efficient. 

First, we achieve the $2nb$ space bound deterministically via \emph{in-register quotient storage}. Prior synchronized implementations utilized an auxiliary register for the quotient $q$. To maintain low space overhead, this register was heuristically bounded to $O(\log n)$ qubits, incurring a fidelity loss for large quotients. We eliminate this auxiliary register entirely by utilizing the guaranteed available padding within the shared register (storing the pair $(u_i, r_i)$) to temporarily store $q_{i+1}$. This makes the $2nb + O(\log n)$ space complexity rigorous, without relying on heuristics. \cref{fig:zalka_eea_improved_flow} shows the evolution of the state of the system for each logical iteration of the synchronized Extended Euclidean Algorithm. 

Second, we introduce a \emph{unified cycle architecture} to reduce time complexity constants. The baseline synchronized approach cycles through four distinct circuit blocks (Division, Normalization, Bézout Update, Swap), leading to redundant arithmetic operations. For example, both the Division and Bézout Update step would need $n$ quantum-quantum multiplications and controlled additions each, even though for any step of the EEA these two cycles act on different parts of the shared registers storing the remainders and Bézout polynomials. In our unified architecture, a single, optimized circuit executes every cycle, with its behavior controlled by the synchronization counter. For example, our architecture integrates the arithmetic for Division and Bézout Update steps, such that a total of $n$ quantum-quantum multiplications and controlled additions are performed, effectively halving the dominant gate costs. Furthermore, we optimize the access to the in-register quotient by integrating localized CSWAPs within the sequential arithmetic loop, avoiding expensive generalized quantum addressing mechanisms. 

In \cref{sec:improved_zalka_eea_iterations} we show that our synchronized reversible EEA requires $6n$ iterations to terminate in the worst case. \cref{fig:optimized_zalka_eea} in \cref{sec:ref_python_impl} gives a complete Python implementation that uses $2nb + \mathcal{O}(\log_2{n})$ space and $n$ quantum-quantum multiplications per iteration, resulting in the leading order gate cost of $6n^2$ quantum-quantum multiplications.

\begin{figure}
    \centering
    \includegraphics[width=0.9\linewidth]{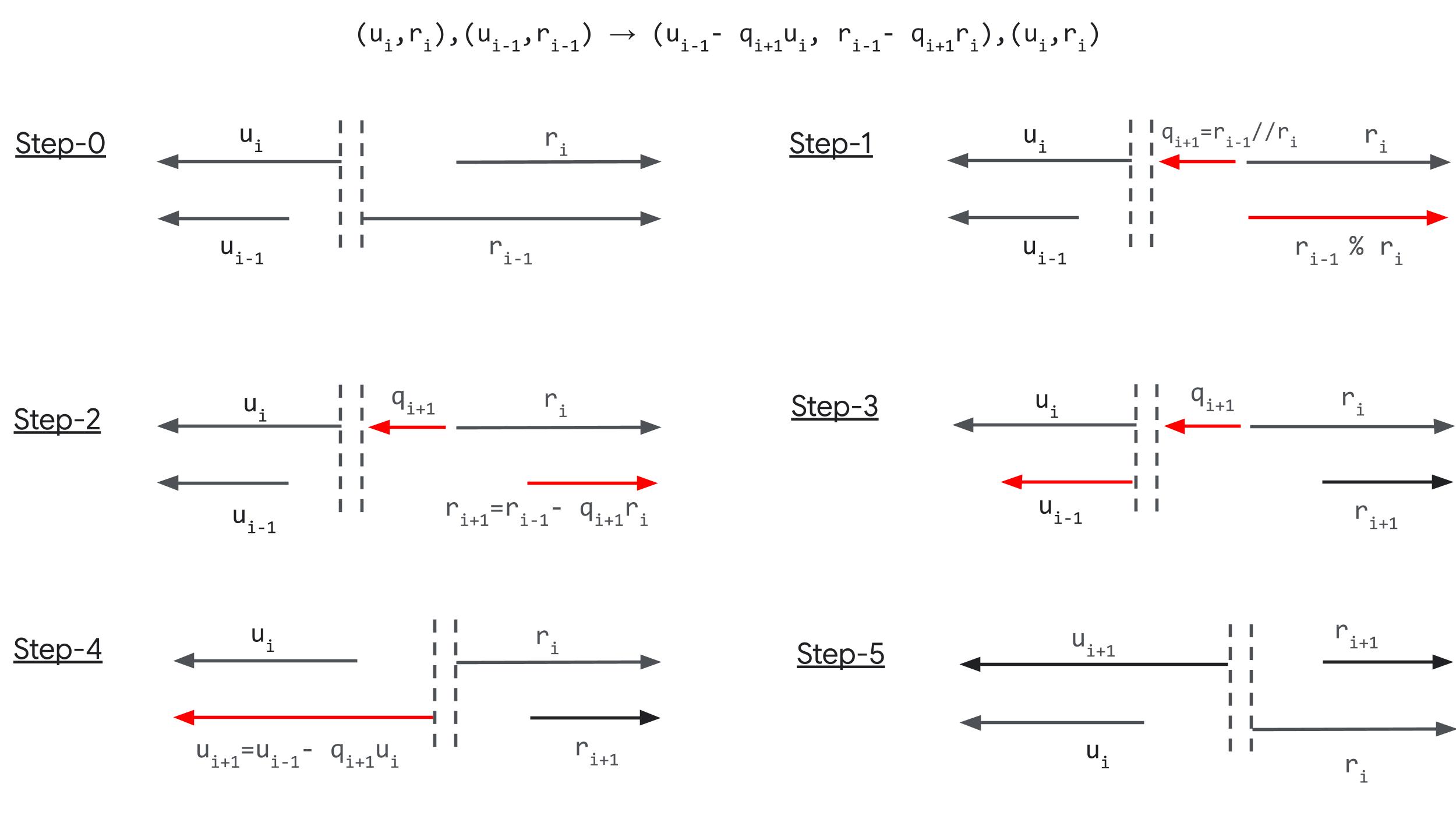}
    \caption{Evolution of the state of the system for our improved construction of Zalka's reversible EEA \cite{kaye2004optimizedquantumimplementationelliptic, proos2004shorsdiscretelogarithmquantum} with in-register quotient storage which deterministically reduces the qubit counts to $2nb + \mathcal{O}(\log{n})$.
    \underline{Step-0} shows the initial configuration of the system. Empty space corresponds to qubits in the $\ket{0}$ state. 
    The direction of the arrow denotes the order in which the coefficients of the polynomials are stored, where the tip of the arrow stores the highest degree coefficients and the tail stores the lowest degree coefficient.  
    At each step, three key invariants are satisfied - (a) $\deg(u_i) + \deg(r_i) \leq n$, (b) $\deg(r_{i}) < \deg(r_{i-1})$ and $\deg(u_{i}) > \deg(u_{i-1})$ and (c) $\deg(u_{i}) + \deg(r_{i-1}) = n$.
    \underline{Step-1} iteratively computes each term of the quotient $q_{i+1} = \lfloor r_{i-1} / r_{i} \rfloor$ where $\deg(q_{i + 1}) = \deg(r_{i - 1}) - \deg(r_{i}) = n - \deg(u_i) - \deg(r_{i})$, and thus the quotient is stored in-place within the shared register storing $u_{i}$ and $r_{i}$. 
    \underline{Step-2} performs right shift until the leading order coefficient of $r_{i+1}$ is non-zero, such that at the end of step 2, the polynomial long division is finished and we have successfully computed both $r_{i+1}$ and $q_{i+1}$.
    \underline{Step-3} iteratively right shifts $u_{i-1}$ until $\deg(u_{i-1}) = \deg(u_{i})$.
    \underline{Step-4} iteratively computes $u_{i+1} = u_{i - 1} - q_{i + 1} u_{i}$ by 
    iteratively performing $u_{i - 1} = (u_{i} \times q_{i+1, j} - u_{i - 1})\times x$ for all $j \text { in } [0, \dots, \deg(q_{i + 1})]$. The multiplication by $x$ corresponds to a right shift for $u_{i - 1}$.
    \underline{Step-5} Swaps the two registers and finishes one logical iteration of the EEA. 
    }
    \label{fig:zalka_eea_improved_flow}
\end{figure}

\subsubsection{Dialog Representation for Implicit Bézout coefficients}\label{sec:dialog_representation}

To understand the power of our implicit access model, it is useful to think not just about algorithms, but about different ways to represent mathematical objects.
The familiar binary representation of an integer, for instance, is a list of $0$ and $1$ symbols.
It is a \emph{representation} (known since the square-and-multiply exponentiation algorithm) because the integer can be recreated from the symbols by initializing an accumulator to $0$, and then iterating through the list of symbols and applying operations depending on the symbol (with `0' meaning \emph{multiply accumulator by 2} and `1' meaning \emph{multiply accumulator by 2 and add 1}).
The binary representation has many strengths.
For example, digital logic circuits can implement the addition of binary integers in logarithmic depth.
However, binary representation is not always the optimal choice.
For example, carry-save adders use a slightly larger representation that enables performing addition in constant depth~\cite{carrysaveadderpatent}.

In this paper, we employ a different representation of numbers and polynomials.
We call this representation the \emph{Dialog representation}.
Concretely, the Dialog is the sequence of transition matrices, $\tau_1, \tau_2, \dots, \tau_m$, generated by executing the Extended Euclidean Algorithm on the input pair $A, B$.
For constant-time division-free GCD algorithms \cite{Stein1967, cryptoeprint:2019/266}, this corresponds to a recording of the conditional branch decisions taken in each iteration.
Abstractly, the Dialog is a decomposition of a pair of inputs (numbers or polynomials) into a series of small invertible matrix multiplications.
This sequence is a valid representation because it provides a complete recipe to recover the original inputs from the output. Specifically, the total transformation $\tau_{\text{total}} = \tau_m \cdot
\tau_{m - 1} \cdots \tau_1$ maps the input vector $[A, B]^T$ to the output vector $[\gcd(A, B), 0]^T$. Consequently, applying the inverse transformation, $\tau_{\text{total}}^{-1}$, to the output recovers the original inputs.

$$
\begin{pmatrix}
\gcd(A, B)\\
0
\end{pmatrix}=
\tau_\text{total}
\begin{pmatrix}
A \\
B
\end{pmatrix}
=
\underbrace{
\tau_{m}
\cdot
\tau_{m - 1}
\dots
\tau_{1}
}_\text{Dialog Representation}
\begin{pmatrix}
A \\
B
\end{pmatrix}
$$

\begin{align*}
\begin{pmatrix}
u_{m-1}\\
u_{m}
\end{pmatrix}=
\tau_\text{total}
\begin{pmatrix}
1 \\
0
\end{pmatrix}
&&
\begin{pmatrix}
v_{m-1}\\
v_{m}
\end{pmatrix}=
\tau_\text{total}
\begin{pmatrix}
0 \\
1
\end{pmatrix}
&&
\begin{pmatrix}
A\\
B
\end{pmatrix}=
\tau_\text{total}^{-1}
\begin{pmatrix}
\gcd(A, B)\\
0
\end{pmatrix}
\end{align*}

\begin{align*}
\begin{pmatrix}
u_{m-1} \cdot C\\
u_{m} \cdot C
\end{pmatrix}=
\tau_\text{total}
\begin{pmatrix}
C \\
0
\end{pmatrix}
&&
\begin{pmatrix}
v_{m-1} \cdot C\\
v_{m} \cdot C
\end{pmatrix}=
\tau_\text{total}
\begin{pmatrix}
0 \\
C
\end{pmatrix}
&&
\begin{pmatrix}
A \cdot C\\
B \cdot C
\end{pmatrix}=
\tau_\text{total}^{-1}
\begin{pmatrix}
\gcd(A, B) \cdot C \\
0
\end{pmatrix}
\end{align*}
Every symbol in the dialog representation must correspond to a small invertible matrix.
When a representation has this property, we call it a \emph{linear representation}.
Note that the binary representation \emph{is not} linear, because the `1' symbol needs to increment the accumulator, and this action does not correspond to a matrix multiplication.
The linearity of the dialog representation is a powerful property that allows us to perform complex algebraic manipulations by simply operating on the dialog itself. For example:

\begin{itemize}
    \item \textbf{Direct Modular Division and Multiplication:} 
    Since the transformation from inputs to outputs is linear, we can apply it to different vectors. 
    To compute $C(z)B(z)^{-1} \pmod{P(z)}$ where $\gcd(P(z), B(z)) = 1$, we first compute the Dialog representation of $(P(z), B(z))$. 
    We can then apply the forward transformation $\tau_{\text{total}}$ to the scaled vector $[0, C(z)]^T$. 
    By linearity, this directly yields the result $[C(z)B(z)^{-1}, 0]^T \pmod{P(z)}$ without ever materializing the inverse. Similarly, we can also perform modular multiplication by applying $\tau_{\text{total}}^{-1}$ to the vector $[C(z), 0]^T$ to obtain $[0, B(z)C(z)]^T \pmod{P(z)}$. 

    \item \textbf{Evaluating the Bézout Polynomial:} 
    To evaluate the Bézout Polynomial $v_{m-1}(z)$ at a point $\gamma \in \F_q$, we can distribute the evaluation across the composition. 
    We simply evaluate the polynomial entries of each transformation matrix $\tau_i(z)$ at $\gamma$ to get a sequence of constant matrices, and then multiply these constant matrices together.

\end{itemize}

This linear representation is most powerful when the transformation matrices $\tau_i$ are themselves simple. For this reason, we construct the Dialog not from the classic EEA, but from constant-time, division-free algorithms like Bernstein-Yang \cite{cryptoeprint:2019/266}, whose steps correspond to simple, constant-time matrix operations such as these:

\begin{itemize}
    \item \textbf{Doubling:} The operation $b(z) \gets b(z)/z$ corresponds to:
        
        $$M_{\text{double}} = \begin{pmatrix} 1 & 0 \\ 0 & 1/z \end{pmatrix}$$

    \item \textbf{Inversion:} The operation $(a(z), b(z)) \gets (b(z), a(z))$ corresponds to the matrix:
    
        $$M_{\text{invert}} = \begin{pmatrix} 0 & 1 \\ 1 & 0 \end{pmatrix}$$

    \item \textbf{Addition:} The operation $b(z) \gets b(z) - c \cdot a(z)$, where $c \in \F_q$, corresponds to the matrix:
    
        $$M_{\text{add}}(c) = \begin{pmatrix} 1 & 0 \\ -c & 1 \end{pmatrix}$$

\end{itemize}

In this case, the \textbf{Dialog} is simply the list of coefficients $c_i$ and the swap decisions from each step, such that it defines the ordered list of transformation matrices $\tau_{i}$. The name Dialog itself is a mnemonic for the core arithmetic operations involved: \textbf{D}oubling (or shifting), \textbf{I}nversion (or swap), and \textbf{A}ddition/subtraction, and the representation is similar to a logarithm because it makes multiplication and division easy but addition and subtraction hard. For two input polynomials of degree up to $n-1$, the Dialog consists of $2n$ field elements.

 A na\"ive quantum circuit for constructing the dialog would require $2nb$ qubits to store and manipulate the remainder polynomials $r_{-1}=A(z), r_{0}=B(z), r_{1}, \dots r_{m}$ and another $2nb$ qubits to store the dialog, resulting in overall qubit count of $4nb + \mathcal{O}(\log_2{n})$. In order to reduce the qubit counts to $2nb + \mathcal{O}(\log_2{n})$, we once again use a register sharing technique by observing a key invariant of the constant-time EEA \cite{cryptoeprint:2019/266}:
\textit{In each iteration, as one field element is computed and appended to the Dialog, the total number of coefficients required to represent the active polynomials, $r_{i-1}(z)$ and $r_i(z)$, decreases by one.}
This inverse relationship between the size of the Dialog and the size of the remaining polynomial data allows for an elegant quantum circuit that dynamically reclaims and repurposes qubits. 
Our construction uses a shared quantum register, \emph{poly}, which stores the coefficients of the remainder polynomials $r_{i-1}(z), r_i(z)$ and another quantum register, \emph{dialog}, that stores the growing Dialog. 
At any point in time, we have the invariant: \lstinline{len(poly) + len(dialog) = 2n} such that the total space occupied is $2nb + \mathcal{O}(\log_2{n})$.
The architecture can be described as follows:

\begin{enumerate}
    \item \textbf{Initialization:} The \lstinline{poly} register is initialized to hold the coefficients of the two input polynomials, $a(z)$ and $b(z)$, stored from opposite ends of the register. This creates a layout with a central padded region:
    \begin{center}
    $[a_0, a_1, \dots, a_{n_a}, 0, 0, 0, b_{n_b}, \dots, b_1, b_0]$
    \end{center}

    \item \textbf{Constant-Time Iteration Loop:} The circuit proceeds for a fixed $2n$ iterations, where $n$ is the maximum degree of the input polynomials. This ensures the algorithm completes for any input. Inside each iteration \lstinline{i=}$1,2, \dots, 2n$:
    \begin{itemize}
        \item The control flow (e.g., the decision to swap the logical polynomials) is managed by the integer \lstinline{delta} and the constant term of $b(z)$ (i.e., \lstinline{poly[-1]}). A swap is performed by efficiently reversing the entire \lstinline{poly} register, an operation implementable with a sequence of \lstinline{len(poly)// 2} CSWAP gates.
        \item The field element \lstinline{coeff} for the update step is calculated from the constant terms of the active polynomials (\lstinline{poly[0]} and \lstinline{poly[-1]}).
        \item The polynomial subtraction, $b(z) \gets b(z) - \text{coeff} \cdot a(z)$, is performed in-place. 
        As per \cite[Theorem A.1]{cryptoeprint:2019/266}, after $i$ iterations, we have 
        \begin{align*}
            &2\deg(a_i) \leq 2n - 1 - i + \delta_{i} \\
            &2\deg(b_i) \leq 2n - 1 - i - \delta_{i}
        \end{align*}
        When $\text{coeff} \neq 0$, after the conditional swap, we always have $\delta_{i} \leq 0$ such that $\deg(b(z)) > \deg(a(z))$. 
        Thus, in the shared register representation, we perform 
        \lstinline{poly[-j - 1] -= coeff * poly[j]} $\forall j \in[0, \frac{\text{len(poly)} + \delta_{i}}{2})$.
        \item As a result of the subtraction and a subsequent logical right-shift of $b(z)$, the rightmost $b$ qubits of the register (previously holding $b_0$) become free.
        \item This newly freed block of $b$ qubits is immediately removed from the \lstinline{poly} register and appended at the end of the \lstinline{dialog} register, thus immediately re-purposing to store the computed \lstinline{coeff}, which is the next element of the Dialog.
    \end{itemize}
\end{enumerate}

Crucially, at any intermediate step \lstinline{i} of the algorithm, the \lstinline{poly} and \lstinline{dialog} registers contain a complete, holistic representation of the EEA's state. The growing \lstinline{dialog} register holds the partially constructed Dialog (from step \lstinline{1 to i}), while the shrinking \lstinline{poly} register holds the coefficients of the current remainder polynomials, $r_{i-1}(z)$ and $r_{i}(z)$. From these two registers of total size $2nb$, one can always reconstruct the full mathematical state: the remainders are available directly, and the corresponding Bézout coefficients, $u_i(z)$ and $v_i(z)$, can be computed by playing back the partial Dialog.
After the full $2n$ iterations, the original polynomials have been completely reduced, and the Dialog has expanded to occupy the entire $2nb$ qubits. The final state of the system is the complete Dialog, from which the gcd and the final Bézout coefficients can be derived implicitly. This in-place architecture, detailed in \cref{fig:in-place-dialog-divstep} of \cref{sec:ref_python_impl}, successfully constructs the full Dialog representation using a leading-order qubit cost of only $2nb$.

\subsection{Optimal Polynomial Intersection (OPI) over binary extension fields}
While the original DQI paper focused on the Optimal Polynomial Intersection (OPI) problem over prime fields $\mathbb{F}_p$, we now study the problem defined over binary extension fields, $\GF(2^b)$. The primary motivation for this is the prospect of substantially more efficient quantum circuits. Arithmetic in fields of characteristic 2 is often less resource-intensive to implement on a quantum computer. For instance, addition is a simple bitwise XOR, implementable with a linear number of CNOT gates.
Quantum-Classical multiplication and Squaring are linear reversible circuits that can also be implemented using only the CNOT gates.
Quantum-quantum multiplication based on Karatsuba algorithm \cite{vanhoof2020spaceefficientquantummultiplicationpolynomials} uses $\mathcal{O}(b^{\log_2{3}})$ Toffoli gates, $\mathcal{O}(b^2)$ CNOT gates, and no ancilla.
Inversion can be reduced to only $\mathcal{O}(\log_2{b})$ quantum-quantum multiplications, resulting in low Toffoli counts \cite{amento2012efficientquantumcircuitsbinary, ITOH198921}.
\cref{sec:improved_gf2_arithmetic} details our optimized implementations of arithmetic circuits for $\GF(2^b)$.
This reduction in the cost of the underlying field arithmetic translates directly to lower resource requirements for the overall DQI algorithm.

The OPI problem and the DQI algorithm generalize naturally to this setting, as discussed in \cite[Section 14]{jordan2024optimizationdecodedquantuminterferometry}. The problem can be stated analogously:

\begin{definition}[OPI over $\GF(2^b)$]
Let $q=2^b$. Given an integer $n < q-1$, an instance of the Optimal Polynomial Intersection problem over $\GF(2^b)$ is as follows. For each non-zero element $y \in \mathbb{F}_q^*$, let $F_y$ be a subset of the finite field $\mathbb{F}_q$. The problem is to find a polynomial $Q \in \mathbb{F}_q[y]$ of degree at most $n-1$ that maximizes the objective function:
\begin{equation*}
f_{OPI}(Q) = |\{y \in \mathbb{F}_q^* : Q(y) \in F_y\}|
\end{equation*}
\end{definition}

For a general OPI instance, the subsets $F_y$ can be arbitrary. As described in \cref{sec:dqi_quantum_circuit}, implementing the DQI algorithm requires constructing constraint-encoding gates $G_i$ that prepare a superposition related to these sets. 
For an arbitrary set $F_y$, this step would require a generic quantum state preparation, costing roughly $\mathcal{O}(q)$ gates for each of the $m$ constraints.

To reduce the gate cost of the constraint encoding step while maintaining the classical hardness of the problem, we focus our resource analysis on a specific, structured variant of OPI where the sets $F_y$ are chosen to allow for highly efficient implementations of the $G_i$ gates. 
Specifically, we introduce and analyze the \textbf{Twisted Bent Target OPI} (\cref{definition:TBTOPI}). 
In this variant, the sets $F_y$ are constructed from Maiorana-McFarland bent functions. This structure ensures that each gate $G_i$ can be implemented with a gate cost that scales as $\mathcal{O}(\log q)$ rather than $O(q)$. 

\subsection{Decoding Reed Solomon Codes using the Extended Euclidean Algorithm}\label{sec:reversible_rs_decoder}
Having developed efficient quantum circuits for the EEA under both explicit and implicit access models, we now apply them to the central challenge of our work: constructing a complete, reversible decoder for Reed-Solomon (RS) codes. This decoder is the computational core of the DQI algorithm when applied to the OPI problem. 
Syndrome decoding of Reed-Solomon codes addresses the problem of recovering an error pattern $e = [e_0, e_1, \dots, e_{m-1}]$, where each $e_i \in \F_q$ and the number of non-zero errors $|e|$ is at most $\ell$. The input to the decoder is a list of known syndromes, which are represented as the coefficients of a syndrome polynomial, $S(z) = s_0 + s_1z + \dots + s_{n-1}z^{n-1}$.
The full decoding procedure involves three main stages: first, solving the \emph{key equation} using the EEA; second, finding the error locations using a Chien search; and third, finding the error values using Forney's algorithm.

\paragraph{Solving the Key Equation:} The core of the decoding process is to solve the fundamental key equation, a polynomial congruence that relates the known $S(z)$ to two unknown polynomials: the error-locator polynomial $\sigma(z)$ and the error-evaluator polynomial $\Omega(z)$:

\begin{equation} \label{eq:key_equation}
\sigma(z)S(z) \equiv \Omega(z) \pmod{z^{2\ell}}.
\end{equation}
Here, $\ell = \frac{n}{2}$ is the maximum number of errors that the code can correct, $\sigma(z)$ is called the error-locator polynomial because its roots identify the error positions, and $\Omega(z)$ is called the error-evaluator polynomial and is used to find the error magnitudes \cite{Forney1965, berlekamp2015algebraic}. A unique solution exists if $\deg(\sigma(z)) = |e| \le \ell$ and $\deg(\Omega(z)) < \ell$.

This polynomial congruence can be solved using several algorithms, most notably the Berlekamp-Massey algorithm \cite{berlekamp2015algebraic} or the Extended Euclidean Algorithm (EEA) \cite{SUGIYAMA197587}. 
To solve it using EEA, we apply the EEA to the two specific polynomials $A(z) = z^{2\ell}$ and $B(z) = S(z)$. As established in the previous section, the EEA produces a sequence of remainders $r_i(z)$ and corresponding Bézout coefficients $u_i(z)$ and $v_i(z)$ such that:
\begin{equation}
    A(z)u_i(z) + B(z)v_i(z) = r_i(z).
\end{equation}
Taking this identity modulo $A(z) = z^{2\ell}$, the first term vanishes, leaving:
\begin{equation}
    S(z)v_i(z) \equiv r_i(z) \pmod{z^{2\ell}}.
\end{equation}

This directly matches the form of the key equation \cref{eq:key_equation}, with $\sigma(z) \propto v_i(z)$ and $\Omega(z) \propto r_i(z)$. 
The algorithm is halted at the first iteration $i$ where $\deg(r_i(z)) < \ell$, which guarantees that the degree constraints for a valid solution are met. 

Our quantum circuits for EEA perform this first step using either the synchronized algorithm of \cref{sec:zalka_eea_optimized} for explicit access to both $\sigma(z)$ and $\Omega(z)$, or the Dialog-based method of \cref{sec:dialog_representation} for implicit access to $\sigma(z)$ and explicit access to $\Omega(z)$. 
The efficiency of the subsequent decoding stages depends on which model is used.

\paragraph{Chien Search and Forneys algorithm:} Once the key equation has been solved via the EEA, yielding the error-locator polynomial $\sigma(z)$ and the error-evaluator polynomial $\Omega(z)$, the decoder proceeds in two final stages. First, the Chien search \cite{Chien1964} finds the error locations by identifying the roots of $\sigma(z)$. This is achieved by evaluating $\sigma(z)$ at every non-zero element of the field, $\gamma_j \in \F_q^*$. An error is located at position $j$ if $\sigma(\gamma_j^{-1})=0$. Second, For each error location $j$ found by the Chien search, Forney's algorithm \cite{Forney1965} computes the corresponding error value $e_j$ using the formula:
\begin{equation}
e_j = - \frac{\Omega(\gamma_j^{-1})}{\sigma'(\gamma_j^{-1})},
\end{equation}
where $\sigma'(z)$ is the formal derivative of $\sigma(z)$. This requires evaluating both $\Omega(z)$ and $\sigma'(z)$ at the point $\gamma_j^{-1}$, followed by a modular division.

The implementation of these two steps differs significantly between the access models. 
In the \emph{explicit model}, the EEA provides the coefficients of $\sigma(z)$ and $\Omega(z)$ directly. The Chien search and Forney's algorithm then proceed by evaluating these polynomials (and the easily computed derivative $\sigma'(z)$) at each point $\gamma_j^{-1}$ using $n$ quantum-classical multiplications per evaluation. 
In contrast, the \emph{implicit model} provides the Dialog of length $n = 2\ell$ for $\sigma(z)$, while still providing $\Omega(z)$ explicitly as the final remainder. 
Here, evaluating $\sigma(\gamma_j^{-1})$ and $\sigma'(\gamma_j^{-1})$ is achieved by playing back the Dialog on an augmented state vector. 
Since differentiation is a linear operator, we can track the derivative through the sequence of linear transformations by applying the chain and product rules.
This requires $2n$ quantum-quantum multiplications and $n/2$ quantum-classical multiplications per evaluation. 

The quantum resource costs for these stages depend heavily on whether the explicit or implicit access model was used for the initial EEA step. The trade-off is summarized in \cref{tbl:decoder_costs}. The dominant costs arise from sequences of Galois field multiplications, which are categorized as either quantum-quantum (QQ), where both operands are in quantum registers, or quantum-classical (QC), where one operand is a known classical value.

\begin{table}[H]
    \centering
    \caption{
    Leading-order gate costs for the reversible EEA-based RS decoder, comparing the explicit (\cref{sec:zalka_eea_optimized}) and implicit (\cref{sec:dialog_representation}) access models. 
    Gate Costs are dominated by quantum-quantum (QQ) and quantum-classical (QC) multiplications.
    GF Inverse refers to the number of field inversions.
    The explicit access model also has a much higher constant factor coming from other operations like \lstinline{CSWAPS} and \lstinline{CADDS}, to manage the state of the system as part of the complex synchronization scheme. 
    }
    \label{tbl:decoder_costs}
    \renewcommand{\arraystretch}{1.4}
    \begin{tabular}{|l||c|c|c||c|c|c|}
    \hline
    \multicolumn{1}{|c||}{} & \multicolumn{3}{c||}{\textbf{Explicit Access Model}} & \multicolumn{3}{c|}{\textbf{Implicit Access Model}} \\
    \cline{2-7}
    \multicolumn{1}{|c||}{\textbf{Decoder Stage}} & \textbf{QQ Mults} & \textbf{QC Mults} & \textbf{GF Inverse} & \textbf{QQ Mults} & \textbf{QC Mults} & \textbf{GF Inverse} \\
    \hline
    \hline
    Step 1: EEA (Key Equation) & $3n^2$ & --- & $6n$ & $n^2$ & --- & $n$ \\
    \hline
    Steps 2\&3: Chien Search \& Forney Algorithm & --- & $mn$ & $m$ & $2mn$ & $mn/2$ & $m$ \\
    \hline
    \textbf{Total (Leading Order)} & $\boldsymbol{3n^2}$ & $\boldsymbol{mn}$  & $\boldsymbol{m + 6n}$ & $\boldsymbol{2mn + n^2}$ & $\boldsymbol{mn/2}$ & $\boldsymbol{m + n}$ \\
    \hline
    \end{tabular}

\end{table}

As the table shows, the choice of model creates a clear trade-off. In the typical regime for OPI where $n \ll m$, the complexity of both models is dominated by the $O(mn)$ terms from the Chien search and Forney's algorithm. For the implicit model, the dominant cost is $2mn$ quantum-quantum multiplications needed to evaluate $\sigma(z)$ and its derivative via Dialog playback. For the explicit model, the dominant cost is $mn$ quantum-classical multiplications to evaluate the three required polynomials.

For circuits over binary extension fields, a quantum-classical multiplication in $\GF(2^b)$ can be implemented very efficiently using only Clifford gates (a sequence of CNOTs), whereas a quantum-quantum multiplication requires non-Clifford resources (Toffoli gates). Therefore, for the parameter regimes we study, the \emph{explicit access approach is substantially cheaper}, as its dominant cost has a much lower Toffoli complexity. However, this trade-off should be carefully re-evaluated for other applications, such as those over prime fields, where quantum-classical multiplication is also expensive.

\section{Classical Attacks on OPI}\label{sec:classical_attacks}
Suppose we have a max-LINSAT problem over $\mathbb{F}_q$ with allowed sets $F_1, \ldots, F_m$ of size $|F_i|=r$.
Prange's algorithm \cite{prange1962use} can be applied to this problem as follows.
First, we subsample any $n$ constraints to enforce via a linear system on the $n$ variables. This system can then be solved with Gaussian elimination. Since we have disregarded the remaining $m-n$ constraints, we expect that roughly $r/q \cdot (m-n)$ will be satisfied by chance since each constraint is satisfied by $r$ of the $q$. Thus, we get, in expectation, a solution that satisfies $n + r/q \cdot (m-n)$ clauses.
Prange's algorithm can be repeated to achieve any desired fraction of satisfied constraints, albeit potentially with an exponential overhead in the number of repetitions. This is the most efficient classical attack on OPI that we are aware of.
In \cref{sec:optimizePrange}, we comment on how Prange's algorithm could be optimized and show how to estimate the runtime required to match DQI's performance on a large supercomputer.
Then in \cref{sec:extendedPrange}, we explain an improvement to Prange's algorithm over extension fields and the impact on the runtime estimates in \cref{sec:results}.

\subsection{Optimizing Prange's Algorithm}\label{sec:optimizePrange}
We first comment on how Gaussian elimination in Prange's algorithm \cite{prange1962use} can be optimized to avoid an $n^3$ cost per attempt.
First, we can bring the generator matrix $G\in \mathbb{F}_q^{n\times m}$ into row echelon form such that the leftmost $n$ columns form the identity matrix. We let $\mathbf{e}^{(1)}, \mathbf{e}^{(2)}, \ldots, \mathbf{e}^{(n)}\in \mathrm{Span}(G)$ denote the rows of this row echelon form.
We can then represent any codeword $\mathbf{c}$ as
\begin{equation}
    \mathbf{c} = \sum_{i=1}^{n}c_i\mathbf{e}^{(i)}.
\end{equation}
To iterate over the codewords that satisfy the first $n$ clauses, we can view the problem as iterating over the strings where the $c_i \in F_i$ above.
Assume each $|F_i|$ has size at least 2. Then we can iterate over strings that satisfy the first $n$ clauses using a Gray code that only updates two of the first $n$ coordinates per iteration. Over $\mathbb{F}_{2^b}$, this could be implemented as two vectorized SIMD operations per update.

The semicircle law of \cite{jordan2024optimizationdecodedquantuminterferometry} states that DQI will satisfy $t$ clauses in expectation in the asymptotic limit, where
\begin{equation}
    t / m = \left(\sqrt{\frac{\ell}{m} \left(1-\frac{r}{q}\right)} + \sqrt{\left(1-\frac{\ell}{m}\right)\frac{r}{q}}\right)^2.
\end{equation}

Suppose we apply the classical attack to find a solution satisfying at least $t$ of the $m$ clauses, where $t > n$. The number of satisfied clauses is distributed binomially, and the probability of sampling such a solution is:
\begin{equation}
p(t):= \sum_{s=t}^{m}\binom{m-n}{s-n}\left(\frac{r}{q}\right)^{s-n}\left(1-\frac{r}{q}\right)^{m-s}
 .\label{eq:binomSum}
\end{equation}
To match this with the classical attack will take, in expectation, $1/p(t)$ trials.\footnote{At finite $\ell$ and $m$, the single-shot expected score (from the $(\ell+1)\times(\ell+1)$ tridiagonal matrix of \cite{jordan2024optimizationdecodedquantuminterferometry}) is slightly below the asymptotic semicircle value $t$---e.g., $2359.75$ versus $2392.97$ for $(4095,70,12,2016)$ at $\ell=35$---but this drop in the mean is compensated by finite-size broadening of the DQI score distribution. Wrapping the same circuit in a classical repeat-until-success loop reaches the rounded semicircle target ($2393$) in $\approx 35$ expected shots using expectation-optimal tridiagonal weights $w_k$, or $\approx 17$ expected shots ($14\text{--}19$ across all rows of \cref{tbl:resource_estimates}) when the Stage-1 weights $w_k$ (\cref{sec:dqi_quantum_circuit}) are chosen to maximize tail probability. Similarly, for $(4095,90,12,2016)$, tail-optimized weights sample a $>10^{23}$-XP-trial solution in $\approx 1.3$ expected shots (\cref{tbl:quantitative_advantage}).}
We estimate the classical resources required for one trial, using the Frontier supercomputer as an example \cite{frontier2023epyc}.
Frontier contains 37,632 AMD Instinct MI250X accelerators (i.e., GPUs). Each MI250x has 220 compute units (CUs) and runs at 1.7 GHz. Each CU has 4 SIMD lanes. So, we can imagine that with an optimized implementation, Frontier could potentially run this many trials of the heuristic in one day:
\begin{equation}
    24 \times 3600 \times \frac{1}{2} \times 37,632 \times 220 \times 4 \times 1.7 \times 10^9 \approx 2.43 \times 10^{21}\label{eq:FrontierTime}
\end{equation}
Note that no networking is required for our classical attacks (i.e., they are embarrassingly parallel), so we could easily distribute the workload over the Internet to many independent machines that run at different times, in odd hours, etc.

\subsection{Extended Prange's Algorithm over Extension Fields}\label{sec:extendedPrange}
The eXtended Prange's algorithm (XP) is a folklore algorithm alluded to in \cite{clz21}, although we do not know of a reference that fully formalizes it as a cryptanalytic attack.
For simplicity, suppose $F_1=F_2=\ldots=F_m=F$. Suppose $q = p^b$ for prime $p$. Let
\begin{align}
    \phi: &\mathbb{F}_{p^b} \leftrightarrow \mathbb{F}_p^b \nonumber \\
    \phi(\mathbf{x}) &= \left(\tr\left(\zeta x\right) : \zeta \in \mathcal{B}\right) \label{eq:definePhiEmbedding}
\end{align}
denote the natural bijection from elements in $\mathbb{F}_{p^b}$ to vectors in $\mathbb{F}_p^b$, where $\mathcal{B}$ is any basis for $\mathbb{F}_{p^b}$ over $\mathbb{F}_p$ and $\tr(\alpha) = \sum_{i=0}^{b-1}\alpha^{p^i}$ is the field trace (See Ch. 2 of \cite{mullen2013handbook}). Let $\Aff\subseteq 2^{\mathbb{F}_p^b}$ denote the set of nonempty affine subspaces of $\mathbb{F}_{p}^{b}$.
For each codimension $s\in \{0, 1, \ldots, b\}$, let $A_s\in \Aff$ denote the $(b-s)$-dimensional affine subspace which has the {\it largest intersection} $|A\cap \phi(F)|$ with the set $\phi(F)$, and let $P[s]$ denote its fractional overlap by $P[s] := \frac{|A\cap \phi(F)|}{|A|}$. The general idea is that if we ``pay'' $s$ affine constraints over the base field on one of the coordinates $i\in [m]$, then we get to flip a $P[s]$-biased coin that determines whether the $i$th constraint is satisfied. We have a budget of $n\cdot b$ base field constraints to spend. We call any budget-respecting strategy which assigns a value of $s$ to each coordinate $i$ an ``XP Allocation Strategy''. We consider two such strategies below.
\subsubsection{Expectation-Optimal XP Allocation Strategies via Linear Programming}
We can use a randomized procedure to decide $s$ for each clause. For each clause, this strategy samples $s$ independently according to the discrete probability distribution $p_s$.
This distribution is constrained such that, on average, we use $bn/m$ $\mathbb{F}_p$-linear constraints, meaning we are exactly within-budget. It can be optimized to get the maximum expected number of satisfied constraints by solving the following linear program (LP):
\begin{align}
    \mathrm{maximize}\,& \sum_{s=0}^{b}p_s P[s] \nonumber \\
    \mathrm{subject\, to}\,&\sum_{s=0}^{b}p_s = 1 \nonumber \\ 
    \mathrm{and}\,&\sum_{s= 0}^{b} s p_s \leq \sfrac{b n}{m} \nonumber \\ 
    \mathrm{where}\,\,\,&p_s\in [0,1]\,\,\,\,\,\forall s\in \{0, \ldots, b\}.\label{eq:linearProgramExtPrange}
\end{align}
The optimal value of the above LP gives the optimal expected fraction of satisfied constraints achieved by the XP.
We recover the truncation heuristic when we replace $\Aff$ with the set $\mathcal{T} := \{\{e\} : e \in \mathbb{F}_q\}\cup\{\mathbb{F}_q\}$.
There are exponentially many affine subspaces in $b$, the degree of the field extension (e.g. when $p=2$, $|\Aff|\sim 7.37\cdot 2^{(b+1)^2/4}$). But for $p=2$ and $b\leq 10$, we have $|\Aff|<10^{10}$ and we can explicitly enumerate these to compute the $P[s]$ table and solve eq~\eqref{eq:linearProgramExtPrange}.
Moreover, it seems unlikely we would really need to enumerate all the affine subspaces, since most of them will have small overlap with the set $F$. Perhaps, using better heuristics or Fourier-analytic techniques, we could directly compute the optimal allocation in time $\mathrm{poly}(q)$.
The XP can sometimes achieve a significantly higher expected value than Prange's algorithm. For example, over $\mathbb{F}_4$ with $|F|=2$ and $\sfrac{n}{m} = 1/2$, XP has an expected value of $1$, whereas Prange's algorithm would give an expected value of $\sfrac{3}{4}$.
\subsubsection{Probability-Optimal XP Allocation Strategies via a Knapsack Solver}
If $t$ is larger than the optimal expected fraction of satisfied constraints from eq~\eqref{eq:linearProgramExtPrange}, we may need to run the XP many times to obtain one sample that satisfies $\geq t$ clauses. Rather than maximizing the expectation value, we should instead maximize the probability of obtaining a sample that satisfies $\geq t$ clauses. The optimal-probability allocation is the solution to the following problem:
\begin{align}
    \mathrm{maximize}\,\,\,& \mathbb{P}\left(\sum_{i = 1}^m X_i \geq t\right)     \nonumber \\
    \mathrm{subject\, to}\,\,\,&X_i \sim \mathrm{Ber}(P[s_i]) \nonumber \\ 
    &\sum_{i=1}^m s_i \leq B \nonumber \\
    &0 \leq s_i \leq b \label{eq:probabilityOptimal}
\end{align}
This problem is non-convex, and the solution may not be the same as the allocation strategy from eq~\eqref{eq:linearProgramExtPrange}.
To obtain good approximate solutions to equation~\eqref{eq:probabilityOptimal}, we use a dynamic programming algorithm explained in \cref{sec:dynamic_programming}.

\subsubsection{Search for good target sets for OPI over binary extension fields}
The choice of good target sets seems to be a nontrivial problem. For a fixed set size $|F|$, we want to make the overlaps $P[s]$ as small as possible for each $s\in \{0, \ldots, b\}$. This way, the XP allocation strategies of eq~\eqref{eq:linearProgramExtPrange} and eq~\eqref{eq:probabilityOptimal} will have lower optimal values, giving a larger classical runtime. Identifying the optimal such set $F$ or even the best attainable list of overlaps $P[s]$ appears to be an open problem. Some recent insight on a good way to choose $F$ came from \cite{briaud2025quantum}, who performed algebraic cryptanalysis of the binary OPI problem in the case that $F$ is chosen as the solution set of a multivariate polynomial equation (i.e., an algebraic variety). They concluded that for their application, cubic or higher-degree varieties were required. This is because quadratic varieties contain many dimension $b/2$ affine subspaces. They operated in a high-rate regime (specifically, $n/m>7/8$) so that they could obtain an exact solution with their quantum algorithm. The problem that then arises is that these affine subspaces are so large that the number of linear equations sufficient to enforce them as constraints is small enough that a rate $1/2$ code would already be able to satisfy every clause using XP. However, in our scenario, we propose to use DQI, which is an {\it approximate optimization} algorithm. This allows us to keep the rate $n/m$  significantly lower than in \cite{briaud2025quantum} (in fact, the space cost scaling is dominated by $n$, so the lower the rate of the code, the better).
Although we do not obtain exact solutions, we obtain better approximate optima than can be obtained in polynomial time by any algorithm known to us, which is sufficient to obtain the large speedups here.

In fact, we find that a particular quadratic variety is an excellent choice to get a quantum advantage. Specifically, we chose bent functions from the Maiorana-McFarland family. In more detail, let $n = 2k$. Let
\begin{equation}
    S_k = \left\{(x_1, ..., x_n) \in \mathbb{F}_2^n : \sum_{i=1}^{k} x_i x_{i+k} = 1\mod 2\right\}.\label{eq:quadraticVariety}
\end{equation}
As we show in \cref{sec:bounds_on_a_intersect_s}, for any affine subspace $A\subset\mathbb{F}_2^n$ of dimension $d=\dim A$, we have
\begin{equation}
    \label{eq:upper_bound_on_a_intersection_sk}
    |A\cap S_k| \leq \begin{cases}
    2^d & d < k\\
    2^{d-1} + 2^{k-2} & k \leq d < 2k-1\\
    2^{2k-2} & d=2k-1\\
    2^{2k-1} - 2^{k-1} & d=2k.
    \end{cases}
\end{equation}
This can be fed as input to the knapsack solver described in \cref{sec:dynamic_programming}.

\subsubsection{Upper bounds on affine intersections of Maiorana-McFarland target sets}
\label{sec:bounds_on_a_intersect_s}

We have found a technical proof of the upper bounds in \eqref{eq:upper_bound_on_a_intersection_sk}, although we believe a simpler proof may be possible. All the key concepts involved are present in the proof for affine spaces with $k\leq \dim A < 2k-1$. Therefore, we defer the proofs of upper bounds for $d < k$, $d=2k-1$, and $d=2k$ to \cref{sec:proofs_of_bounds_on_a_intersect_s} and focus here on $k\leq d < 2k-1$.

Let $k$ be a positive integer. We use the standard dot product $\langle x, y\rangle:=\sum_{i=1}^kx_iy_i\in\mathbb{F}_2$ where $x,y\in\mathbb{F}_2^k$ to define
\begin{align}
    R_k := \left\{(x,y) \in \mathbb{F}_2^k\times\mathbb{F}_2^k : \langle x,y\rangle = 0\right\} \\
    S_k := \left\{(x,y) \in \mathbb{F}_2^k\times\mathbb{F}_2^k : \langle x,y\rangle = 1\right\}
\end{align}
so that $R_k\cup S_k=\mathbb{F}_2^{2k}$ and $R_k\cap S_k=\emptyset$. We will use $\sqcup$ to make simultaneous assertions about the union and intersection of two sets. Namely, we define $X \sqcup Y = Z \iff (X\cup Y = Z) \,\land\, (X \cap Y = \emptyset)$ for any three sets $X$, $Y$, and $Z$. Thus, $R_k\sqcup S_k=\mathbb{F}_2^{2k}$. In arguments below, we will make frequent use of the implication $X\sqcup Y=Z \implies |X|+|Y|=|Z|$.

We begin by deriving \textit{recursive structure formulas} that express the sets $A\cap R_k$ and $A\cap S_k$ as the disjoint union $\sqcup$ of the sets of the form $B\cap R_{k-1}$ and $B\cap S_{k-1}$ for some affine subspaces $B\subset\mathbb{F}_2^{2k-2}$. To that end, we use two linear projections $\pi_0:\mathbb{F}_2^{2k} \to \mathbb{F}_2^{2k-2}$ and $\pi_1:\mathbb{F}_2^{2k} \to \mathbb{F}_2^2$ defined by
\begin{align}
    \pi_0(xayb) = xy, \quad \pi_1(xayb) = ab
\end{align}
where $x,y\in\mathbb{F}_2^{k-1}$ and $a,b\in\mathbb{F}_2$ and where juxtaposition denotes concatenation. For a $w\in\mathbb{F}_2^{2k}$, we will refer to $\pi_0(w)$ as the \textit{root} of $w$ and to $\pi_1(w)$ as the \textit{suffix} of $w$. The map $(\pi_0,\pi_1):\mathbb{F}_2^{2k}\to\mathbb{F}_2^{2k-2}\times\mathbb{F}_2^2$ is invertible. Its inverse is $\otimes:\mathbb{F}_2^{2k-2}\times\mathbb{F}_2^2\to\mathbb{F}_2^{2k}$ defined as
\begin{align}
    xy \otimes ab = xayb
\end{align}
where $x,y\in\mathbb{F}_2^{k-1}$ and $a,b\in\mathbb{F}_2$ as before. For $X\subset\mathbb{F}_2^{2k-2}$, $Y\subset\mathbb{F}_2^2$, and $s\in\mathbb{F}_2^2$, we will write $X\otimes s = \{x\otimes s\,|\,x\in X\}$ and $X\otimes Y = \{x\otimes y\,|\,x\in X,y\in Y\}$. Note that $|X\otimes Y|=|X|\cdot|Y|$.

The recursive structure of $A\cap R_k$ and $A\cap S_k$ arises from
\begin{align}
    \label{eq:recursion_on_dot_product}
    \langle x,y\rangle = \langle \pi_0(x),\pi_0(y)\rangle + \langle \pi_1(x),\pi_1(y)\rangle
\end{align}
where we slightly abused notation by using $\langle.,.\rangle$ to refer to the dot product in three different vector spaces, $\mathbb{F}_2^{2k}$, $\mathbb{F}_2^{2k-2}$, and $\mathbb{F}_2^2$. We now state and prove the recursive structure formulas.

\begin{lemma}[Recursive Structure Formulas]
    \label{lm:recursive_structure_formulas}
    \,\newline
    For integers $d\geq k>1$, let $A$ be a $d$-dimensional affine subspace of $\mathbb{F}_2^{2k}$. Then
    \begin{align}
        \label{eq:structure_of_a_cap_r}
        A \cap R_k = \bigsqcup_{\sigma\in F} \begin{cases}
            (A_\sigma'\cap S_{k-1})\otimes\sigma & \text{if}\quad\sigma = 11 \\
            (A_\sigma'\cap R_{k-1})\otimes\sigma & \text{if}\quad\sigma \ne 11
        \end{cases} \\
        \label{eq:structure_of_a_cap_s}
        A \cap S_k = \bigsqcup_{\sigma\in F} \begin{cases}
            (A_\sigma'\cap R_{k-1})\otimes\sigma & \text{if}\quad\sigma = 11 \\
            (A_\sigma'\cap S_{k-1})\otimes\sigma & \text{if}\quad\sigma \ne 11
        \end{cases}
    \end{align}
    where $F=\pi_1(A)$ is an affine subspace of $\mathbb{F}_2^2$ and the sets $A_\sigma'\subset\pi_0(A)$ are affine subspaces of $\mathbb{F}_2^{2k-2}$ that arise as translations of the same linear subspace $W'\subset\mathbb{F}_2^{2k-2}$ with $\dim W'=\dim A - \dim F$.
\end{lemma}
\begin{proof}
We can write $A = a + V$ where $a\in\mathbb{F}_2^{2k}$ is a fixed bit string and $V$ is a $d$-dimensional linear subspace of $\mathbb{F}_2^{2k}$. Let $\mathcal{B}=\{b_1,\ldots,b_d\}$ be an arbitrary basis of $V$. Suppose $b_1,b_2\in \mathcal{B}$ are two distinct bit strings with suffix $01$. Then $b_1+b_2$ has suffix $00$, so $\{b_1+b_2\}\cup \mathcal{B}\setminus\{b_2\}$ is another basis of $V$ with the number of elements with suffix $01$ reduced by one. By iterating this procedure, we can reduce the number of elements with suffix $01$ to at most one.

Similarly, we can reduce the number of basis elements with suffixes $10$ and $11$ to at most one. Moreover, if $\mathcal{B}$ contains three elements with distinct nonzero suffixes, then we can replace one of them with an element with suffix $00$ by adding the other two elements to it. Thus, $V$ has a basis $\mathcal{B}$ that takes one of the following three forms
\begin{align}\label{eq:basis_structure}
    \mathcal{B} = \begin{cases}
        \mathcal{B}_{00} \\
        \mathcal{B}_{00} \sqcup \{b\} \\
        \mathcal{B}_{00} \sqcup \{b, c\}
    \end{cases}
\end{align}
where $\mathcal{B}_{00}$ consists of bit strings with suffix $00$ and $b,c$ are two bit strings with distinct nonzero suffixes.

Define $F:=\pi_1(A)$ as the set of suffixes of the elements of $A=a+V$. The projector $\pi_1$ is a linear map, so $F$ is an affine subspace of $\mathbb{F}_2^2$ and therefore it has either one, two or four elements. Indeed,
\begin{align}
    \mathcal{B}=\mathcal{B}_{00} & \iff F = \{\sigma\} \\
    \mathcal{B}=\mathcal{B}_{00} \sqcup \{b\} & \iff F = \{\sigma,\origtau\} \\
    \mathcal{B}=\mathcal{B}_{00} \sqcup \{b,c\} & \iff F = \{00,01,10,11\}
\end{align}
where $\sigma=\pi_1(a)$ and $\origtau=\pi_1(a+b)$.

For every suffix $\sigma\in F$, let $A_\sigma := A \cap \pi_1^{-1}(\{\sigma\})$ be the set of elements in $A$ with suffix $\sigma$ and $A_\sigma' := \pi_0(A_\sigma)$ the set of roots of those elements of $A$. A singleton set is an affine space, the intersection of affine subspaces is an affine subspace, and $\pi_0$, $\pi_1$ are linear maps, so $A_\sigma$ and $A_\sigma'$ are affine subspaces of $\mathbb{F}_2^{2k}$ and $\mathbb{F}_2^{2k-2}$, respectively. We can express $A$ in terms of $A_\sigma'$ for $\sigma\in F$ as
\begin{align}
    \label{eq:structure_of_a}
    A = \bigsqcup_{\sigma\in F} A_\sigma = \bigsqcup_{\sigma\in F} \pi_0(A_\sigma) \otimes\sigma = \bigsqcup_{\sigma\in F} A_\sigma'\otimes\sigma.
\end{align}
Define $W:=\mathrm{span}(\mathcal{B}_{00})$, so that $A_\sigma = u_\sigma + W$ where $u_\sigma\in\mathbb{F}_2^{2k}$ is any bit string in $A$ with suffix $\pi_1(u_\sigma)=\sigma$. Then, $A_\sigma' = \pi_0(u_\sigma) + W'$ where $W':=\pi_0(W)$. But $\pi_0$ is a bijection on $W$, so $\dim A_\sigma'=\dim W' = \dim W = |\mathcal{B}_{00}|$.

Finally, intersecting both sides of \eqref{eq:structure_of_a} with $R_k$ (respectively, $S_k$), we obtain
\begin{align}
    A \cap R_k = \bigsqcup_{\sigma\in F} \begin{cases}
        (A_\sigma'\cap S_{k-1})\otimes\sigma & \text{if}\quad\sigma = 11 \\
        (A_\sigma'\cap R_{k-1})\otimes\sigma & \text{if}\quad\sigma \ne 11
    \end{cases} \\
    A \cap S_k = \bigsqcup_{\sigma\in F} \begin{cases}
        (A_\sigma'\cap R_{k-1})\otimes\sigma & \text{if}\quad\sigma = 11 \\
        (A_\sigma'\cap S_{k-1})\otimes\sigma & \text{if}\quad\sigma \ne 11
    \end{cases}
\end{align}
where we switch between $R_k$ and $S_k$ when $\sigma=11$ as needed in light of \eqref{eq:recursion_on_dot_product}.
\end{proof}

\begin{remark}
\label{rm:five_cases}
The recursive structure formulas take five forms depending on $0\leq\dim F\leq 2$ and on whether $11\in F$.
\begin{enumerate}
    \item If $F=\{\sigma\}$ with $\sigma\ne 11$, then
    \begin{align}
        A\cap R_k &= (A_\sigma'\cap R_{k-1})\otimes\sigma \\
        A\cap S_k &= (A_\sigma'\cap S_{k-1})\otimes\sigma.
    \end{align}
    \item If $F=\{11\}$, then
    \begin{align}
        A\cap R_k &= (A_{11}'\cap S_{k-1})\otimes 11 \\
        A\cap S_k &= (A_{11}'\cap R_{k-1})\otimes 11.
    \end{align}
    \item If $F=\{\sigma,\origtau\}$ with $11\notin F$ and $\sigma\ne\origtau$, then
    \begin{align}
        A\cap R_k &= \left((A_\sigma' \cap R_{k-1})\otimes\sigma\right) \sqcup \left((A_{\origtau}' \cap R_{k-1})\otimes\origtau \right)\\
        A\cap S_k &= \left((A_\sigma'\cap S_{k-1})\otimes\sigma\right) \sqcup \left((A_\origtau' \cap S_{k-1})\otimes\origtau \right).
    \end{align}
    \item If $F=\{\sigma,11\}$ with $\sigma\ne 11$, then
    \begin{align}
        A\cap R_k &= \left((A_\sigma'\cap R_{k-1})\otimes\sigma\right) \sqcup \left((A_{11}'\cap S_{k-1})\otimes 11 \right)\\
        A\cap S_k &= \left((A_\sigma'\cap S_{k-1})\otimes\sigma\right) \sqcup \left((A_{11}'\cap R_{k-1})\otimes 11 \right).
    \end{align}
    \item If $F=\{00,01,10,11\}$, then
    \begin{align}
        A\cap R_k = ((A_{11}'\cap S_{k-1})\otimes 11) \sqcup \bigsqcup_{\sigma\in\{00,01,10\}} (A_\sigma'\cap R_{k-1})\otimes \sigma \\
        A\cap S_k = ((A_{11}'\cap R_{k-1})\otimes 11) \sqcup \bigsqcup_{\sigma\in\{00,01,10\}} (A_\sigma'\cap S_{k-1})\otimes \sigma.
    \end{align}
\end{enumerate}
\end{remark}

Lemma \ref{lm:recursive_structure_formulas} enables inductive proof of the upper bound on $|A\cap R_k|$ and $|A\cap S_k|$ for affine spaces with dimension $\dim A=2k-1$, which we defer to Lemma \ref{lm:case_3}. Below, we prove that $|A\cap R_k| \leq 2^{d-1}+2^{k-1}$ and $|A\cap S_k| \leq 2^{d-1}+2^{k-2}$ for affine $A$ with $k \leq \dim A < 2k-1$ which requires more sophisticated tools in addition to Lemma \ref{lm:recursive_structure_formulas}. The right-hand sides $2^{d-1}+2^{k-1}$ and $2^{d-1}+2^{k-2}$ indicate that the cardinality of $A\cap S_k$ (respectively, $A\cap R_k$) can exceed $\frac12|A|$ by at most $2^{k-2}$ (respectively, $2^{k-1}$). Thus, $|A\cap R_k|$ can exceed $\frac12|A|$ by twice as much as $|A\cap S_k|$, so if a recursive structure formula for $A\cap S_k$ contains both an $A\cap R_{k-1}$ term and an $A\cap S_{k-1}$ term, then a naive inductive argument fails. We refer to such formulas as \textit{mixed} recursive structure formulas. They have either two or four terms
\begin{align}
    \label{eq:hard_recursive_structure_formula_sr}
    A\cap S_k =& \left((A_\sigma'\cap S_{k-1})\otimes \sigma \right) \sqcup \left((A_{11}'\cap R_{k-1})\otimes 11 \right) \\
    \label{eq:hard_recursive_structure_formula_sssr}
    A\cap S_k =& \left((A_{00}'\cap S_{k-1})\otimes 00 \right) \sqcup \left((A_{01}'\cap S_{k-1})\otimes 01 \right) \sqcup \left((A_{10}'\cap S_{k-1})\otimes 10 \right) \sqcup \left((A_{11}'\cap R_{k-1})\otimes 11 \right).
\end{align}

These recursive structure formulas do not cause issues in the proof of Lemma \ref{lm:case_3} where we work around the challenge by exploiting the fact that when $\dim A=2k-1$ not all affine subspaces on the right hand side are distinct. To face a more general situation below, we will exploit hidden affine structures in certain combinations of sets of the form $A\cap S_{k-1}$ and $A\cap R_{k-1}$ that arise from the cancellation of quadratic terms in the indicator functions of these sets as they are combined together.

Let $W$ be a linear subspace of $\mathbb{F}_2^{2k}$. If a function $f$ from $W$ to the underlying field of scalars $\mathbb{F}_2$ is linear, then we will refer to $f$ as a \textit{linear functional}. Let $a\in\mathbb{F}_2^{2k}$ and let $A:=a+W$ be an affine subspace of $\mathbb{F}_2^{2k}$. If a function $g$ from $A$ to the underlying field of scalars $\mathbb{F}_2$ can be written as $g(a+w)=c+f(w)$ for some $c\in\mathbb{F}_2$ and some linear functional $f$, then we will refer to $g$ as an \textit{affine functional}. The significance of affine functionals to our arguments derives from the way they partition their domain. If $g$ is a non-constant affine functional, then $f$ is a non-constant linear functional and $f^{-1}(1)$ is a coset of $\ker f=f^{-1}(0)$. Therefore, $|f^{-1}(0)|=|f^{-1}(1)|$. But then also $|g^{-1}(0)|=|g^{-1}(1)|$. Thus, if $g:A\to\mathbb{F}_2$ is an affine functional, then each $g^{-1}(0)$ and $g^{-1}(1)$ is either empty or affine. Moreover, if both $g^{-1}(0)$ and $g^{-1}(1)$ are non-empty, then $\dim g^{-1}(0) = \dim g^{-1}(1) = \dim A - 1$.

In order to express the value of the indicator function $\mathbb{1}_{S_k}:\mathbb{F}_2^{2k}\to\mathbb{F}_2$ of $S_k$
\begin{align}
    \mathbb{1}_{S_k}(x) := \sum_{i=1}^kx_ix_{k+i} = \langle\pi_L(x),\pi_R(x)\rangle
\end{align}
on bit strings in an affine space $A:=a+W$, we introduce the (skew-)symmetric bilinear form $\omega:\mathbb{F}_2^{2k}\times\mathbb{F}_2^{2k}\to\mathbb{F}_2$
\begin{align}
    \omega(x,y) := \sum_{i=1}^k \left(x_iy_{k+i} + x_{k+i}y_i\right) = \langle\pi_L(x),\pi_R(y)\rangle + \langle\pi_L(y),\pi_R(x)\rangle
\end{align}
where $\pi_L,\pi_R:\mathbb{F}_2^{2k}\to\mathbb{F}_2^k$ are linear projectors defined via
\begin{align}
    \pi_L(ab)=a,\quad\pi_R(ab)=b
\end{align}
for any $a,b\in\mathbb{F}_2^k$. Indeed, we can express $\mathbb{1}_{S_k}(a+x)$ for any $x\in W$ in terms of $\omega$ as
\begin{align}
    \label{eq:q(a+w)}
    \mathbb{1}_{S_k}(a+x) &= \langle a_L + x_L, a_R + x_R\rangle = \mathbb{1}_{S_k}(a) + \omega(a,x) + \mathbb{1}_{S_k}(x)
\end{align}
where $a_L:=\pi_L(a)$, $a_R:=\pi_R(a)$, $x_L:=\pi_L(x)$, and $x_R:=\pi_R(x)$. This way of writing $\mathbb{1}_{S_k}$ allows us to construct affine subspaces from sets of the form $A\cap R_k$ and $A\cap S_k$ by arranging for the cancellation of the quadratic terms $\mathbb{1}_{S_k}(x)$.

We can use these objects and their properties to prove Lemma \ref{lm:hard_rsf_2terms} and \ref{lm:hard_rsf_4terms} that enable us to express the right hand side of \eqref{eq:hard_recursive_structure_formula_sr} and \eqref{eq:hard_recursive_structure_formula_sssr} in terms of cardinalities of sets involving $S_{k-1}$, but no $R_{k-1}$.

\begin{lemma}[Proxy expression for mixed recursive structure formulas with two terms]
    \label{lm:hard_rsf_2terms}
    \,\newline
    Let $k$ be a positive integer and $a_1,a_2\in\mathbb{F}_2^{2k}$. For any two affine spaces $A_1:=a_1+W$ and $A_2:=a_2+W$ arising as translations of the same linear subspace $W\subset\mathbb{F}_2^{2k}$, there exist two sets $B_1$ and $B_2$ each of which is either empty or affine and such that
    \begin{align}
        |A_1\cap R_k| + |A_2\cap S_k| = |B_1| + 2\cdot|B_2\cap S_k|
    \end{align}
    and $|B_1|+|B_2|=|W|$.
\end{lemma}
\begin{proof}
We begin by partitioning $W$ into four disjoint sets
\begin{align}
    \label{eq:w_rr}
    W_{rr} &:= \{x\in W\,|\,(a_1+x\in R_k)\,\land\,(a_2+x\in R_k)\} \\
    \label{eq:w_rs}
    W_{rs} &:= \{x\in W\,|\,(a_1+x\in R_k)\,\land\,(a_2+x\in S_k)\} \\
    \label{eq:w_sr}
    W_{sr} &:= \{x\in W\,|\,(a_1+x\in S_k)\,\land\,(a_2+x\in R_k)\} \\
    \label{eq:w_ss}
    W_{ss} &:= \{x\in W\,|\,(a_1+x\in S_k)\,\land\,(a_2+x\in S_k)\}
\end{align}
so that $W = W_{rr}\sqcup W_{rs} \sqcup W_{sr} \sqcup W_{ss}$. Consider the sets
\begin{align}
    \label{eq:rs_sets}
    W_0 := W_{rr} \sqcup W_{ss} \quad
    W_1 := W_{rs} \sqcup W_{sr}.
\end{align}
The indicator function of $W_1$ restricted to $W$ can be written as
\begin{align}
    \mathbb{1}_{W_1}(x) &= \mathbb{1}_{R_k}(a_1+x)\cdot\mathbb{1}_{S_k}(a_2+x) + \mathbb{1}_{S_k}(a_1+x)\cdot\mathbb{1}_{R_k}(a_2+x) \\
    &= \left(1+\mathbb{1}_{S_k}(a_1+x)\right)\cdot\mathbb{1}_{S_k}(a_2+x) + \mathbb{1}_{S_k}(a_1+x)\cdot\left(1+\mathbb{1}_{S_k}(a_2+x)\right) \\
    &= \mathbb{1}_{S_k}(a_1+x) + \mathbb{1}_{S_k}(a_2+x) \\
    &= \omega(a_1+a_2,x) + \mathbb{1}_{S_k}(a_1) + \mathbb{1}_{S_k}(a_2)
\end{align}
where we dropped the terms $\mathbb{1}_{S_k}(a_1+x)\cdot\mathbb{1}_{S_k}(a_2+x)$ and $\mathbb{1}_{S_k}(x)$, because $\mathbb{F}_2$ has characteristic $2$. But $\omega$ is bilinear, so $\mathbb{1}_{W_1}$ restricted to $W$ is an affine functional on $W$. Therefore, $W_1$ is either empty or an affine subspace of $W$. Similarly, the indicator function of $W_0$ restricted to $W$ is $\mathbb{1}_{W_0}(x) = \mathbb{1}_{W_1}(x) + 1$ and hence also an affine functional on $W$. Therefore, $W_0$ is either empty or an affine subspace of $W$. Moreover, $W=W_0\sqcup W_1$, so if both $W_0$ and $W_1$ are non-empty, then $\dim W_0=\dim W_1=\dim W - 1$. 
Translation is bijective, so definitions \eqref{eq:w_rr}-\eqref{eq:w_ss} imply that
\begin{align}
    \label{eq:aras_in_terms_of_ws}
    |A_1\cap R_k| + |A_2\cap S_k| = |W_{rr}| + 2\cdot|W_{rs}| + |W_{ss}|.
\end{align}
Finally, define $B_1:=a_1+W_0$ and $B_2:=a_2+W_1$, so that $|B_1|+|B_2|=|W_0|+|W_1|=|W|$ and rewrite \eqref{eq:aras_in_terms_of_ws} as
\begin{align}
    |A_1\cap R_k| + |A_2\cap S_k| = |B_1| + 2\cdot|B_2\cap S_k|
\end{align}
where we used $|B_1|=|W_{rr}|+|W_{ss}|$ and $|W_{rs}|=|B_2\cap S_k|$.
\end{proof}

In order to facilitate the use of Lemma \ref{lm:hard_rsf_2terms} in an inductive proof, we note the following conditional upper bound it implies.

\begin{corollary}[Bound for mixed recursive structure formulas with two terms]
    \label{cr:hard_rsf_2terms}
    \,\newline
    Let $e$ and $k>1$ be positive integers. If $|B\cap S_{k-1}| \leq 2^{\dim B-1} + 2^{k-3}$ for every affine subspace $B$ of $\mathbb{F}_2^{2k-2}$ with $\dim B\in\{e-1,e\}$, then
    \begin{align}
        |A_1\cap R_{k-1}| + |A_2\cap S_{k-1}| \leq 2^e + 2^{k-2}
    \end{align}
    for every pair of $e$-dimensional affine subspaces $A_1:=a_1+W$ and $A_2:=a_2+W$ of $\mathbb{F}_2^{2k-2}$ that arise as translations of the same linear subspace $W$ of $\mathbb{F}_2^{2k-2}$.
\end{corollary}
\begin{proof}
By Lemma \ref{lm:hard_rsf_2terms}, there exist two sets $B_1$ and $B_2$, each of which is either empty or affine and such that
\begin{align}
    \label{eq:ar_as=b+2bs}
    |A_1\cap R_{k-1}| + |A_2\cap S_{k-1}| = |B_1| + 2\cdot|B_2\cap S_{k-1}|
\end{align}
and $|B_1|+|B_2|=|W|$. If $B_1=\emptyset$, then $B_2$ is an affine subspace of $\mathbb{F}_2^{2k-2}$ of dimension $e$. Therefore,
\begin{align}
    |A_1\cap R_{k-1}| + |A_2\cap S_{k-1}| \leq 0 + 2\cdot(2^{e-1} + 2^{k-3}) = 2^e + 2^{k-2}.
\end{align}
If $B_2=\emptyset$, then $B_1$ is an affine subspace of $\mathbb{F}_2^{2k-2}$ of dimension $e$. Therefore,
\begin{align}
    |A_1\cap R_{k-1}| + |A_2\cap S_{k-1}| \leq 2^e + 2\cdot0 < 2^e + 2^{k-2}.
\end{align}
Finally, if neither $B_1$ nor $B_2$ is empty, then both are affine and $\dim B_1=\dim B_2=e-1$. Therefore,
\begin{align}
    |A_1\cap R_{k-1}| + |A_2\cap S_{k-1}| \leq 2^{e-1} + 2\cdot(2^{e-2} + 2^{k-3}) = 2^e + 2^{k-2}
\end{align}
completing the proof of the Corollary.
\end{proof}

A similar, although technically more complicated, proof establishes an analogous result for mixed recursive formulas with four terms, see Lemma \ref{lm:hard_rsf_4terms} and Corollary \ref{cr:hard_rsf_4terms}. Equipped with Corollary \ref{cr:hard_rsf_2terms} above for dealing with two-term mixed recursive structure formulas and Corollary \ref{cr:hard_rsf_4terms} for dealing with four-term mixed recursive structure formulas, we are now ready to prove the upper bounds on $|A\cap R_k|$ and $|A\cap S_k|$ for affine spaces with $k\leq \dim A < 2k-1$.

\begin{lemma}[Upper bound on $|A\cap R_k|$ and $|A\cap S_k|$ when $k\leq \dim A < 2k-1$]
    \label{lm:case_2}
    \,\newline
    Let $k,d$ be two positive integers and let $A \subset \mathbb{F}_2^{2k}$ be a $d$-dimensional affine subspace of $\mathbb{F}_2^{2k}$. If $k \leq d < 2k-1$, then
    \begin{align}
        |A\cap R_k| & \leq 2^{d-1} + 2^{k-1} \\
        |A\cap S_k| & \leq 2^{d-1} + 2^{k-2}.
    \end{align}
\end{lemma}
\begin{proof}
We prove the bounds by induction on $k$. If $k=1$, then there are no affine subspaces of dimension $1\leq \dim A < 1$, so the upper bounds are vacuously satisfied. Fix $k>1$ and suppose that for any affine subspace $B\subset\mathbb{F}_2^{2k-2}$ whose dimension $e:=\dim B$ satisfies $k-1 \leq e < 2k-3$, it is the case that
\begin{align}
    \label{eq:inductive_hypothesis_rk}
    |B\cap R_{k-1}| & \leq 2^{e-1} + 2^{k-2} \\
    \label{eq:inductive_hypothesis_sk}
    |B\cap S_{k-1}| & \leq 2^{e-1} + 2^{k-3}.
\end{align}
We will also use \eqref{eq:inductive_hypothesis_rk} and \eqref{eq:inductive_hypothesis_sk} when $e=2k-3$ and $e=2k-2$ which follow from Lemma \ref{lm:case_3} and Corollary \ref{cr:case_4}, respectively. We will show that for every affine subspace $A \subset \mathbb{F}_2^{2k}$ of dimension $d=\dim A$ with $k \leq d < 2k-1$, we have
\begin{align}
    |A\cap R_k| & \leq 2^{d-1} + 2^{k-1} \\
    |A\cap S_k| & \leq 2^{d-1} + 2^{k-2}.
\end{align}
By Remark \ref{rm:five_cases}, we have five cases to consider
\begin{align}
    \dim F = 0 &\quad\land\quad 11\notin F\\
    \dim F = 0 &\quad\land\quad 11\in F\\
    \dim F = 1 &\quad\land\quad 11\notin F\\
    \dim F = 1 &\quad\land\quad 11\in F\\
    \dim F = 2 &.
\end{align}
In the first case, where $F=\{\sigma\}$ with $\sigma\ne 11$, the recursive structure formulas \eqref{eq:structure_of_a_cap_r} and \eqref{eq:structure_of_a_cap_s} become
\begin{align}
    A \cap R_k &= (A_\sigma' \cap R_{k-1}) \otimes \sigma \\
    A \cap S_k &= (A_\sigma' \cap S_{k-1}) \otimes \sigma
\end{align}
where the affine subspace $A_\sigma'$ has dimension $\dim A_\sigma'=d-\dim F=d$. Using the inductive hypothesis, we obtain
\begin{align}
    |A \cap R_k| &= |A_\sigma' \cap R_{k-1}| \leq 2^{d-1} + 2^{k-2} < 2^{d-1} + 2^{k-1} \\
    |A \cap S_k| &= |A_\sigma' \cap S_{k-1}| \leq 2^{d-1} + 2^{k-3} < 2^{d-1} + 2^{k-2}.
\end{align}
In the second case, where $F=\{11\}$, the recursive structure formulas \eqref{eq:structure_of_a_cap_r} and \eqref{eq:structure_of_a_cap_s} become
\begin{align}
    A \cap R_k &= (A_{11}' \cap S_{k-1}) \otimes 11 \\
    A \cap S_k &= (A_{11}' \cap R_{k-1}) \otimes 11
\end{align}
where the affine subspace $A_{11}'$ again has dimension $\dim A_{11}'=d-\dim F=d$. Using the inductive hypothesis, we find
\begin{align}
    |A \cap R_k| &= |A_{11}' \cap S_{k-1}| \leq 2^{d-1} + 2^{k-3} < 2^{d-1} + 2^{k-1} \\
    |A \cap S_k| &= |A_{11}' \cap R_{k-1}| \leq 2^{d-1} + 2^{k-2} = 2^{d-1} + 2^{k-2}.
\end{align}
In the third case, where $F=\{\sigma,\origtau\}$ with $\sigma\ne 11$, $\origtau\ne 11$ and $\sigma\ne\origtau$, the recursive structure formulas read
\begin{align}
    A \cap R_k &= ((A_\sigma' \cap R_{k-1}) \otimes \sigma) \sqcup ((A_\origtau' \cap R_{k-1}) \otimes \origtau) \\
    A \cap S_k &= ((A_\sigma' \cap S_{k-1}) \otimes \sigma) \sqcup ((A_\origtau' \cap S_{k-1}) \otimes \origtau).
\end{align}
Using the inductive hypothesis and $\dim A_\sigma'=\dim A_\origtau'=d-\dim F = d - 1$, we find
\begin{align}
    |A \cap R_k| & \leq 2\cdot(2^{d-2} + 2^{k-2}) = 2^{d-1} + 2^{k-1} \\
    |A \cap S_k| & \leq 2\cdot(2^{d-2} + 2^{k-3}) = 2^{d-1} + 2^{k-2}.
\end{align}
In the fourth case, where $F=\{\sigma,11\}$ with $\sigma\ne 11$, the recursive structure formulas \eqref{eq:structure_of_a_cap_r} and \eqref{eq:structure_of_a_cap_s} become
\begin{align}
    A \cap R_k &= ((A_{11}' \cap S_{k-1}) \otimes 11) \sqcup ((A_\sigma' \cap R_{k-1}) \otimes \sigma) \\
    A \cap S_k &= ((A_{11}' \cap R_{k-1}) \otimes 11) \sqcup ((A_\sigma' \cap S_{k-1}) \otimes \sigma)
\end{align}
where $\dim A_\sigma'=\dim A_{11}'=d-\dim F=d-1$. Consequently,
\begin{align}
    |A \cap R_k| &= |A_{11}' \cap S_{k-1}| + |A_\sigma' \cap R_{k-1}| \\
    |A \cap S_k| &= |A_{11}' \cap R_{k-1}| + |A_\sigma' \cap S_{k-1}|.
\end{align}
We obtain the upper bound on $|A \cap R_k|$ from the inductive hypothesis 
\begin{align}
    |A \cap R_k| & \leq 2^{d-2} + 2^{k-3} + 2^{d-2} + 2^{k-2} < 2^{d-1} + 2^{k-1}
\end{align}
and the upper bound on $|A \cap S_k|$ using Corollary \ref{cr:hard_rsf_2terms} with $e=d-1$
\begin{align}
    |A \cap S_k| &= |A_{11}'\cap R_{k-1}| + |A_\sigma'\cap S_{k-1}| \leq 2^{d-1} + 2^{k-2}
\end{align}
where the Corollary's assumption follows from the inductive hypothesis and Lemma \ref{lm:case_1}.

Finally, in the fifth case, where $F=\{00,01,10,11\}$, the recursive structure formulas \eqref{eq:structure_of_a_cap_r} and \eqref{eq:structure_of_a_cap_s} become
\begin{align}
    \label{eq:recursive_structure_formulas_case_5_rk}
    A\cap R_k =& \left((A_{00}'\cap R_{k-1})\otimes 00 \right) \sqcup \left((A_{01}'\cap R_{k-1})\otimes 01 \right) \sqcup \left((A_{10}'\cap R_{k-1})\otimes 10 \right) \sqcup \left((A_{11}'\cap S_{k-1})\otimes 11 \right) \\
    \label{eq:recursive_structure_formulas_case_5_sk}
    A\cap S_k =& \left((A_{00}'\cap S_{k-1})\otimes 00 \right) \sqcup \left((A_{01}'\cap S_{k-1})\otimes 01 \right) \sqcup \left((A_{10}'\cap S_{k-1})\otimes 10 \right) \sqcup \left((A_{11}'\cap R_{k-1})\otimes 11 \right).
\end{align}
There are two possibilities. Either all four affine subspaces $A_\sigma'$ of $\mathbb{F}_2^{2k-2}$ are distinct or they are not.

Suppose first that the affine subspaces $A_\sigma'$ are not all distinct, so that there are $A_1'$ and $A_2'$ such that
\begin{align}
    A_1' := A_{11}'=A_\rho',\quad A_2':=A_\sigma'=A_\origtau'.
\end{align}
In this case, we use the fact that $R_{k-1}\sqcup S_{k-1}=\mathbb{F}_2^{2k-2}$ to write
\begin{align}
    |A \cap R_k| &= |A_1'| + 2\cdot |A_2'\cap R_{k-1}| \\
    |A \cap S_k| &= |A_1'| + 2\cdot |A_2'\cap S_{k-1}|
\end{align}
which, using our inductive hypothesis and $\dim A_\sigma'=d-\dim F=d-2$ for all $\sigma\in\mathbb{F}_2^2$, implies
\begin{align}
    |A \cap R_k| &\leq 2^{d-2} + 2\cdot(2^{d-3} + 2^{k-2}) = 2^{d-1} + 2^{k-1} \\
    |A \cap S_k| &\leq 2^{d-2} + 2\cdot(2^{d-3} + 2^{k-3}) = 2^{d-1} + 2^{k-2}.
\end{align}

Suppose now that the four affine subspaces $A_\sigma'$ of $\mathbb{F}_2^{2k-2}$ are all distinct. By Lemma \ref{lm:recursive_structure_formulas}, the subspaces arise as translations of the same linear subspace $W' \subset \mathbb{F}_2^{2k-2}$ of dimension $\dim W'=d-\dim F=d-2$.

By combining $A_{00}'$ and $A_{01}'$ into $(d-1)$-dimensional affine space $A_{00,01}'$, we can rewrite \eqref{eq:recursive_structure_formulas_case_5_rk} as
\begin{align}
    A\cap R_k =& \left((A_{00,01}'\cap R_{k-1})\otimes 00 \right) \sqcup \left((A_{10}'\cap R_{k-1})\otimes 10 \right) \sqcup \left((A_{11}'\cap S_{k-1})\otimes 11 \right)
\end{align}
which leads to
\begin{align}
    |A\cap R_k| = |A_{00,01}'\cap R_{k-1}| + |A_{10}'\cap R_{k-1}| + |A_{11}'\cap S_{k-1}|.
\end{align}
If $d=k$, then $|A \cap R_k| \leq 2^{d-1} + 2^{k-1} = 2^d$. For affine spaces with $d>k$, we prove the upper bound on $|A \cap R_k|$ using Corollary \ref{cr:hard_rsf_2terms} with $e=d-2$
\begin{align}
    |A\cap R_k| & \leq 2^{d-2} + 2^{k-2} + 2^{d-2} + 2^{k-2} = 2^{d-1} + 2^{k-1}
\end{align}
where the Corollary's assumption follows from the inductive hypothesis and Lemma \ref{lm:case_1}. The upper bound on $|A\cap S_k|$ follows from Corollary \ref{cr:hard_rsf_4terms} with $e=d-2$
\begin{align}
    |A\cap S_k| &= |A_{00}'\cap S_{k-1}| + |A_{01}'\cap S_{k-1}| + |A_{10}'\cap S_{k-1}| + |A_{11}'\cap R_{k-1}| \leq 2^{d-1} + 2^{k-2}
\end{align}
where the Corollary's assumption follows from the inductive hypothesis and Lemma \ref{lm:case_1}.
\end{proof}

\begin{theorem}
For any positive integer $k$ and any affine subspace $A\subset \mathbb{F}_2^{2k}$ of dimension $d=\dim A$, we have
\begin{equation}
    \label{eq:bounds_on_intersection}
    |A\cap S_k| \leq \begin{cases}
    2^d & d < k\\
    2^{d-1} + 2^{k-2} & k \leq d < 2k-1\\
    2^{2k-2} & d=2k-1\\
    2^{2k-1} - 2^{k-1} & d=2k
    \end{cases}
\end{equation}
where
\begin{equation}
    S_k := \left\{(x_1, ..., x_{2k}) \in \mathbb{F}_2^{2k} : \sum_{i=1}^{k} x_i x_{i+k} = 1\mod 2\right\}.\label{eq:setSkquadricDef}
\end{equation}
\end{theorem}
\begin{proof}
The four upper bounds in \eqref{eq:bounds_on_intersection} follow from Lemma \ref{lm:case_1}, Lemma \ref{lm:case_2}, Lemma \ref{lm:case_3}, and Corollary \ref{cr:case_4}, respectively.
\end{proof}

\subsection{Asymptotic Classical Hardness of Twisted Bent Target OPI}
\begin{definition}\label{definition:TBTOPI}
    Let $R, \mu \in (0, 1)$ be constants. For each positive integer $k$, we define the Twisted Bent Target OPI (TBT-OPI) problem as follows. Let $S_k$ be the set of \cref{eq:setSkquadricDef}.
    For each of the $m=2^{2k}-1$ evaluation points $\alpha\in \mathbb{F}_{2^{2k}}^*$, let $A\sim \GL_{2k}(\mathbb{F}_2)$ be a random independent invertible matrix and set
\begin{equation}
    F_{\alpha} = \phi^{-1}(A S_k),
\end{equation}
where $\phi$ is the bijection of eq~\eqref{eq:definePhiEmbedding}.
The $(R, \mu)$-TBT-OPI problem is to find a polynomial $f(X)$ of degree at most $\lceil R\cdot m\rceil$ such that $f(\alpha)\in F_\alpha$ for at least $\lfloor \mu \cdot m \rfloor$ distinct $\alpha$.
\end{definition}
\begin{lemma}
    The TBT-OPI sets $F_{\alpha}$ are efficiently constructible per requirement (2) of \cref{thm:quantumEasiness}.
\begin{proof}
    Note that if $A = I$, then applying a phase to $S_k$ can be achieved with a single layer of transversal CZ gates. We can then change the basis to account for $A$ by conjugating this with a circuit of $O(k^2)=\bigOtilde(1)$ CNOT gates.
\end{proof}
\end{lemma}
\begin{theorem}\label{thm:classicalHardness}
There exists a choice of constant $R \approx 1/10$ so that, letting $\mu = \mu_{DQI}:=\frac{1}{2}+\sqrt{\frac{R}{2}\left(1-\frac{R}{2}\right)}$, the number of trials of XP required to solve $(R,\mu)$-TBT-OPI (\cref{definition:TBTOPI}) is at least $\exp(0.02 m)$.

\begin{proof}
The number of trials that XP needs to beat DQI is $\frac{1}{\gamma}$ where $\gamma$ is the objective function of \eqref{eq:dynamicprogrammingobjective}. This objective function is upper bounded through a direct application of the Hoeffding inequality \cite{Hoeffding} by \eqref{eq:hoeffding}:
\begin{equation} \label{eq:hoeffding}
    \mathbb{P}\left(\sum_{i=1}^m X_i \geq t\right) \leq \exp\left(-2\frac{(t - \mathbb{E}[\sum_{i=1}^m X_i])^2}{m}\right)
\end{equation}
In \cref{subsection:objective_upper_bound} we prove that $\mathbb{E}[\sum_{i=1}^m X_i] < (\frac{1}{2} + \frac{n}{m})m = (\frac{1}{2} + R)m$. Asymptotically $\lim_{b\to \infty} t = (\frac{1}{2} + \sqrt{\frac{R}{2} (1 - \frac{R}{2})})m$ leading to the asymptotic scaling
\begin{equation} \label{eq:asymptotic_scaling}
    \lim_{b\to\infty} \mathbb{P}\left(\sum_{i=1}^m X_i \geq t\right) \leq \exp\left(-2\left(\sqrt{\frac{R}{2}\left(1 - \frac{R}{2}\right)} - R\right)^2 m\right)
\end{equation}

In the regime where $R = \frac{n}{m} \approx 0.10557$, we get an exponential lower bound on the number of trials of 

\begin{equation}
    \#trials \geq e^{0.02786 m}
\end{equation}
\end{proof}
\end{theorem}

\section{Asymptotically Optimal Quantum Speedup}\label{sec:optimal_speedup}
To achieve the speedup of Theorem~\ref{thm:optimalSpeedup}, we have to use methods for the circuit implementation that are not necessarily optimal at the smaller problem sizes we envision for near-term applications. We describe these alternative implementations in \cref{sec:asymptoticNiceCircuits}. Then in \cref{sec:optimalityPuttingItTogether} we observe \cref{thm:optimalSpeedup} as a corollary of this implementation (on the quantum side) along with the lower bounds on runtime of XP that was shown in \cref{thm:classicalHardness}.

\subsubsection{Asymptotically Optimal Implementation}\label{sec:asymptoticNiceCircuits}
While \cref{sec:methods} focused on optimizing constant factors for finite-size implementations, we now analyze the asymptotic complexity of DQI applied to the OPI problem. We demonstrate that under certain conditions, the DQI algorithm can be implemented using a quantum circuit with a number of gates that is nearly linear in the problem size $m$. We use the notation $\bigOtilde(m)$ to hide factors polylogarithmic in $m$.

We first establish the complexity of the key algebraic subroutines required. These complexities are well-established in classical computer algebra and apply over arbitrary finite fields $\mathbb{F}_q$.

\begin{lemma}[Fast Algebraic Subroutines]\label{lem:fast_algebra}
Let $m$ be the problem size parameter. The following tasks can be computed reversibly using $\bigOtilde(m)$ field operations in $\mathbb{F}_q$:
\begin{enumerate}
    \item Multipoint Polynomial Evaluation (MPE): Evaluating a polynomial of degree $<m$ at $m$ points.
    \item Reed-Solomon (RS) Decoding: Decoding an RS code of length $m$ up to half its minimum distance.
    \item Syndrome Computation: Computing $B^T e$ for $e \in \mathbb{F}_q^m$ and a fixed $m \times n$ ($n \leq m$) Vandermonde matrix $B$.
\end{enumerate}
\begin{proof}
$ $\newline \vspace{-1.5em}
\begin{enumerate}
    \item Fast algorithms for MPE utilize a divide-and-conquer strategy. The complexity is $O(M(m)\log m)$ field operations, where $M(m)$ is the cost to multiply two degree $m$ polynomials \cite[Chapter 10]{vonzurGathen2013moderncomputer}. Since $M(m) = \bigOtilde(m)$ using fast multiplication algorithms \cite{Schnhage1971, cantor1991fast, harvey2014fasterpolynomialmultiplicationfinite}, the total cost is $\bigOtilde(m)$ field operations.
    \item The key equation for RS decoding can be solved using the Fast Extended Euclidean Algorithm (Fast EEA), which runs in $O(M(m)\log m) = \bigOtilde(m)$ field operations \cite[Chapter 11]{vonzurGathen2013moderncomputer} (with data-independent control flow via the constant-time variant of \cite{cryptoeprint:2019/266}). Subsequent steps (root finding and error evaluation) also take $\bigOtilde(m)$ field operations using fast MPE \cite{justesen2006complexity}.
    \item Multiplying by $B^T$ (RS syndrome computation) is the transpose of MPE and likewise runs in $O(M(m)\log m) = \bigOtilde(m)$ field operations \cite{justesen2006complexity, sarwate1977complexity, bostan2003tellegen}.
\end{enumerate}

These algorithms are algebraic (with data-independent control flow) and can be implemented reversibly on a quantum computer with $O(1)$ gate overhead \cite{bennett1973logical} (see, e.g., \cite{Nie2023} for explicit Sch\"onhage--Strassen quantum circuits).
\end{proof}
\end{lemma}

We now analyze the overall complexity of DQI. The result depends critically on the efficiency of the objective function implementation and the size of the field relative to $m$.

\begin{theorem}\label{thm:quantumEasiness}
Consider the Optimal Polynomial Intersection (OPI) problem over $\mathbb{F}_q$ with $m$ constraints and $n=Rm$ variables (where $R \in (0,1)$ is constant). The DQI algorithm can be implemented by a quantum circuit using $\bigOtilde(m)$ elementary quantum gates, provided that:
\begin{enumerate}
    \item The field size $q$ satisfies $\log q = O(\mathrm{polylog}(m))$ (i.e., the bit length $b = \bigOtilde(1)$).
    \item The subsets $F_i$ are structured such that an oracle $U_i$ for objective function $f_i$, defined as $U_i\ket{x} = f_i(x)\ket{x}$, can be implemented by a quantum circuit using $\bigOtilde(1)$ gates.
\end{enumerate}
\end{theorem}

\begin{proof}
We analyze the complexity of the DQI algorithm stages (as described in \cref{sec:dqi_quantum_circuit}). The conditions ensure that basic arithmetic operations (addition, multiplication) in $\mathbb{F}_q$ require $\bigOtilde(b) = \bigOtilde(1)$ gates. The degree of the enhancing polynomial is $l = O(n) = O(m)$.

\begin{itemize}
    \item \textbf{Stage 1: Dicke State Preparation.}
Standard dense Dicke state preparation requires $O(lm) = O(m^2)$ gates~\cite{Bartschi2022}. 
To achieve near linear time, we describe a new Dicke State Preparation method in \cref{sec:fast_dicke_state}, which employs a reversible implementation of a Divide and Conquer unranking strategy, accelerated by Binary Splitting. As analyzed in \cref{sec:sparse_dicke_state}, the total gate complexity for this step is $\bigOtilde(m)$.

    \item \textbf{Stage 2: Syndrome Computation.}
Sequential updates $B^T_i \cdot e_i$ as in \cref{sec:dqi_quantum_circuit} require $n$ field operations per constraint ($\Theta(mn)$ total). Instead, following \cite{jordan2024optimizationdecodedquantuminterferometry}, we retain the error register $\ket{e}_{mb}$ and compute all $n$ syndromes in a single batch:

\begin{itemize}
    \item \textit{Constraint Encoding and Error Generation (Apply $G_i$):} The gate $G_i$ can be implemented using QFTs over $\mathbb{F}_q$ and one call to the oracle $U_i$ \cite[Section 14.4]{jordan2024optimizationdecodedquantuminterferometry}. The cost of the QFT is $\bigOtilde(b) = \bigOtilde(1)$. By assumption, the cost of $U_i$ is $\bigOtilde(1)$. Thus, the cost of $G_i$ is $\bigOtilde(1)$, and applying $G_1, \dots, G_m$ to the dense Dicke state prepares $\ket{e}_{mb}$ in $\bigOtilde(m)$ gates.
    
    \item \textit{Syndrome Computation:} Since $B$ is a Vandermonde matrix, mapping $\ket{e}_{mb}\ket{0}_{nb} \mapsto \ket{e}_{mb}\ket{B^T e}_{nb}$ takes $\bigOtilde(m)$ field operations by item (3) of \cref{lem:fast_algebra}, costing $\bigOtilde(m \cdot b) = \bigOtilde(m)$ gates.
\end{itemize}
 
\item 
\textbf{Stage 3: Reversible Decoding.}
We reversibly decode the Reed-Solomon code $C^\perp$. By Lemma~\ref{lem:fast_algebra}, this requires $\bigOtilde(m)$ field operations, resulting in $\bigOtilde(m \cdot b) = \bigOtilde(m)$ gates.

\item \textbf{Stage 4: Final Transformation (IQFT).}
Applying the inverse QFT over $\mathbb{F}_q^n$. This takes $\bigOtilde(n \cdot b) = \bigOtilde(m)$ gates.

\end{itemize}

\textbf{Total Complexity:}
Since each stage requires $\bigOtilde(m)$ gates, the total gate complexity of the DQI algorithm under the stated conditions is $\bigOtilde(m)$.
\end{proof}

The conditions required for \cref{thm:quantumEasiness} are met, for instance, when $q$ is small (e.g., the binary extension fields analyzed in this paper, provided $b$ is polylogarithmic in $m$) and when the OPI instance has some special structure that allows for efficient constraint encoding. For example, in the Polynomial Approximation problem, the subsets $F_i$ are contiguous intervals \cite{GM14}, allowing $U_i$ to be implemented efficiently via comparators in $\bigOtilde(b)$ gates.

It is important to contrast this with the general OPI case studied in Section 5 of \cite{jordan2024optimizationdecodedquantuminterferometry}, where $m \approx q$ and the subsets $F_i$ are arbitrary.

\begin{corollary}
Consider the general OPI problem over $\mathbb{F}_q$ where $m=q-1$ and the subsets $F_i$ are arbitrary (e.g., the balanced case where $|F_i| \approx q/2$). The DQI algorithm requires $\bigOtilde(m^2)$ elementary quantum gates.
\end{corollary}
\begin{proof}
In this case, $b = O(\log m) = \bigOtilde(1)$. However, implementing the oracle $U_i$ or the gate $G_i$ for an arbitrary function defined by a subset of size $O(q)$ generally requires $\bigOtilde(q) = \bigOtilde(m)$ gates (e.g., using QROM or generic state preparation \cite{gosset2024quantumstatepreparationoptimal}). The total cost of Stage 2 (Constraint Encoding) becomes $m \cdot \bigOtilde(m) = \bigOtilde(m^2)$, dominating the overall complexity.
\end{proof}

\subsubsection{Asymptotic Optimality}\label{sec:optimalityPuttingItTogether}
Combining \cref{thm:quantumEasiness} and \cref{thm:classicalHardness} we have \cref{thm:optimalSpeedup}:
\THMoptimalSpeedup*
Since any quantum circuit of $n$ gates can be simulated in time $O(2^n)$, our speedup provides the largest separation possible between classical and quantum runtimes.
This is somewhat extraordinary since OPI is essentially a contrived problem and seems (to us) no more ``quantum" or ``natural" than e.g. factoring or period-finding.
We speculate informally on the outlook for this situation below.
\begin{itemize}
    \item {\bf Faster Classical Algorithms for OPI:} 
    Although the best attacks we are currently aware of for OPI take exponential time in the number of evaluation points, it is certainly possible that better algorithms exist. 
    It is interesting to note that none of the algorithms we know about leverage the low-degree algebraic structure and instead treat it as an arbitrary max-LINSAT problem. 
    We see no fundamental barrier ruling out an algorithm that exploits the low degree, but existing algebraic techniques do not succeed in our parameter regime. 
    Over $\mathbb{F}_q$, interpolation-based list-recovery and noisy polynomial reconstruction algorithms \cite{guruswami1998improved, koetter2003algebraic, BN00} use the low degree of the polynomial directly, but their interpolation step requires the fraction of satisfied constraints to exceed roughly $\sqrt{rn/m}$, which is larger than $1$ unless $r \lesssim m/n$, whereas our target sets have size $r \approx q/2 \approx m/2$. 
    Over $\mathbb{F}_2$, TBT-OPI asks for an assignment to $nb$ binary unknowns satisfying at least a $\mu_{DQI}$ fraction of $m$ quadratic constraints (\cref{definition:TBTOPI}). 
    As discussed in \cref{sec:extendedPrange}, the $(b/2)$-dimensional affine subspaces inside each quadratic target set break exact, high-rate instances \cite{briaud2025quantum}, but XP already captures this subspace structure in our low-rate, approximate regime and still requires exponential time (\cref{thm:classicalHardness}). 
    Beyond restricting to affine subspaces, standard Gr\"obner-basis and linearization methods do not handle a constant fraction of violated equations without an exponential blowup in degree \cite{arora2011new}, and the independent random twists in \cref{definition:TBTOPI} leave no obvious shared structure across the $m$ quadrics. 
    Exploiting the low-degree polynomial structure in this regime would therefore require genuinely new algorithmic ideas.
    \item {\bf Better Shor circuits:} Assuming the speedup for OPI stands, we wonder if further ideas along the line of thinking of \cite{regev2025efficient,kahanamoku2025jacobi} could culminate in an asymptotically-optimal speedup based on Shor's algorithm.
    \item {\bf Peeling off log factors:} Although we omit the polylogarithmic factors here, there is room to optimize them and we wonder whether such effort would reveal a more practically-usable speedup or merely complicate matters without much gain at useful problem scales.
\end{itemize}

\section{Resource Estimates}\label{sec:results}
\subsection{Logical Costs}
To provide a concrete assessment of the resources required to solve an instance of OPI using DQI, we performed a detailed logical cost analysis of the complete, end-to-end algorithm. 
We implemented the key subroutines—including the RS Decoders using the Synchronized and Dialog based EEA methods (as detailed in \cref{sec:zalka_eea_optimized} and \cref{sec:dialog_representation}), and the Sparse Dicke State Preparation from \cref{sec:sparse_dicke_state}; as a \lstinline{Bloq} in the Qualtran quantum compilation framework \cite{harrigan2024expressinganalyzingquantumalgorithms} and used it to analyze the precise qubit and gate costs.

The results of this analysis are summarized in \cref{tbl:resource_estimates} for several representative instances of the Optimal Polynomial Intersection (OPI) problem. The table presents the total Toffoli and Clifford gate counts, the number of logical qubits required for the reversible decoder, and the estimated classical intractability of each instance. 
For $\GF(2^b)$, since multiplication by a constant field element requires only Clifford gates, for small $n / m$ it turns out to be better to spend a higher gate cost to run Zalka's EEA and have explicit access to the Bézout coefficients such that the Chien search subroutine that evaluates the Bézout coefficient polynomial $\sigma(x)$ at $m$ different constants, uses only Clifford gates. 
Using Dialogs, the EEA part is significantly cheaper than Zalka's EEA, but because we only have implicit access to the polynomial $\sigma(x)$, each evaluation of $\sigma(x)$ and $\sigma^\prime(x)$ now requires applying the Dialog and uses quantum-quantum multiplication \cite{vanhoof2020spaceefficientquantummultiplicationpolynomials} of field elements instead of quantum-classical multiplication. 
For both the approaches, the qubit counts scale as $2nb + \mathcal{O}(\log_2(n))$

These resource estimates become particularly noteworthy when contextualized against other well-studied, classically intractable and verifiable optimization problems. Integer factorization using Shor's algorithm serves as the canonical benchmark for large-scale quantum computation. As shown in the state-of-the-art estimates in \cite{gidney2025factor2048bitrsa}, factoring a 2048-bit RSA integer is estimated to require approximately $6.5 \times 10^9$ Toffoli gates and $1399$ logical qubits, while factoring a 1024-bit RSA integer requires $1.1 \times 10^9$ Toffoli gates and $742$ logical qubits. In comparison, several of our OPI instances that are well into the classically intractable regime (requiring $>10^{23}$ trials) exhibit Toffoli counts that are about two to three orders of magnitude smaller, with logical qubit counts that are about two to three times higher. 

\begin{table}
    \centering
    \begin{tabular}{|c|c|c|c|c|c|}
    \hline
    \hline
    ($m$, $n$, $b$, $r$) & Toffoli & Clifford & Qubits & $\#$ Prange Trials & $\#$ XP Trials\\
    \hline
    \hline
    $(1023, 60, 10, 496)$ & \num{2764580} & \num{38945109} & $1364$ & \num{5.4935525387784946e+19} & \num{1915882803738476.8} \\
    \hline
    $(1023, 70, 10, 496)$ & \num{3593114} & \num{49078873} & $1569$ & \num{1.256406251307753e+22} & \num{4.641439887538182e+16} \\
    \hline
    $(1023, 80, 10, 496)$ & \num{4521586} & \num{60425691} & $1769$ & \num{4.2964767808546385e+24} & \num{1.224179182654277e+18} \\
    \hline
    $(1023, 90, 10, 496)$ & \num{5561906} & \num{72998037} & \num[exponent-mode=fixed]{1970} & \num{1.0704385285673214e+27} & \num{2.078358397648132e+19} \\
    \hline
    $(1023, 100, 10, 496)$ & \num{6710658} & \num{86786415} & \num[exponent-mode=fixed]{2170} & \num{1.74941809707523e+29} & \num{2.562701796685802e+20} \\
    \hline
    \hline
    $(4095, 60, 12, 2016)$ & \num{4640644} & \num{127334931} & 1640 & \num{2.019633906949013e+23} & \num{4.019800669718791e+20} \\
    \hline
    $(4095, 70, 12, 2016)$ & \num{5717729} & \num{151649309} & $1885$ & \num{4.7509334068170893e+26} & \num{1.965720586103349e+23} \\
    \hline
    $(4095, 80, 12, 2016)$ & \num{6922158} & \num{178643837} & 2125 & \num{9.479001846779738e+29} & \num{7.994544407999735e+25} \\
    \hline
    $(4095, 90, 12, 2016)$ & \num{8265800} & \num{208309955} & $2366$ & \num{1.413037121295554e+33} & \num{2.265453777773324e+28} \\
    \hline
    $(4095, 100, 12, 2016)$ & \num{9745253} & \num{240637673} & $2606$ & \num{2.101371145129246e+36} & \num{5.912123905073406e+30} \\
    \hline
   \hline
\end{tabular}
\caption{Resource estimates for solving the Optimal Polynomial Intersection (OPI) problem using Decoded Quantum Interferometry (DQI). 
$m$ is the number of constraints, $n$ is the number of variables, the problem is defined over binary extension field $\F_q$ where $q = 2^b$, $r$ is the size of the target set $F_y, \forall y \in \F^*_q$.
We also report the expected number of classical trials needed to sample a DQI grade solution using Extended Prange's algorithm \cite{prange1962use}. 
Gate counts are per shot, while classical trial counts use the asymptotic semicircle target ($\approx 14\text{--}19$ expected DQI shots at finite size; see \cref{sec:optimizePrange}).
We believe problem instances requiring more than $10^{23}$ trials can be classically intractable. 
Here, we use the quadric set in eq~\eqref{eq:quadraticVariety} throughout. 
We have observed that some increase in the $\#$ XP trials is possible by optimizing the choice of set. 
Therefore we expect the classical hardness can be increased, though we likely cannot match $\#$ Prange Trials even with an optimal choice of set. 
We do not know the ultimate limits nor do we understand how to choose the set in a principled way.
}
\label{tbl:resource_estimates}
\end{table}

\subsection{Physical Costs}
For estimating physical costs, we assume one large rectangular grid of physical qubits with nearest neighbor connections, a uniform gate error rate of 0.1\%, a surface code cycle time of 1 microsecond, and a control system reaction time of 10 microseconds.
With these assumptions, the imagined physical layout of the algorithm will consist of a compute region using hot storage and a memory region using cold storage. 
The compute region will store logical qubits ``normally", as distance $d$ surface code patches using $2(d + 1)^2$ physical qubits. 
The cold storage memory region will store qubits more densely, by using 2D yoked surface codes \cite{gidney2025yoked}. 
Since our Toffoli gate counts are low enough, magic state cultivation \cite{gidney2024magicstatecultivationgrowing} will suffice to prepare $\ket{T}$ states with low enough error rates, and thus we will not need any dedicated space for magic state factories. 

As an example, let us look at compiling the instance $(m=4095, n=70, b=12)$. 
The dominant subroutines in the decomposition of the circuit are primitives to do arithmetic over Galois Field, like quantum-quantum multiplication (\lstinline{GF2Mul}) and synthesizing a linear reversible circuit (\lstinline{SynthesizeLR}) for quantum-classical multiplication and squaring.
For these two primitives, we provide hand optimized lattice surgery layouts in \cref{fig:gf2_mul} and \cref{fig:gf2_mul_k} with magic state cultivation in-place. 

Referring to \cite[Figure 6]{gidney2025factor2048bitrsa}, note that a distance 21 patch is sufficient for normal surface code patches to reach a per-patch per-round logical error rate of $10^{-13}$. 
So our hot patches will use $2 \cdot (21 + 1)^2 = 968$ physical qubits per logical qubit. 
For the cold storage, again referring to \cite[Figure 6]{gidney2025factor2048bitrsa}, yoking with a 2D parity check code reaches a logical error rate of $10^{-13}$ when using $350$ physical qubits per logical qubit. 
So cold logical qubits will be roughly triple the density of hot logical qubits. For $(m=4095, n=70, b=12)$, the algorithm uses $1885$ logical qubits. 
We propose to allocate a $40 \times 48$ cold storage region with $1920$ logical qubits using 2D Yoked Surface codes. 
We also propose a thin $3\times 40 = 120$ qubit hot storage region. 
Our mockup for GF multiplication for $b=12$ fits in a $3 \times (2.5b + 9) = 3\times 39$ region, including ancilla patches and space needed for cultivation. 
Thus, our layout uses $120 \times 968 + 1920 \times 350 \approx 800k$ physical qubits. A mockup of the high-level spatial layout is shown in \cref{fig:qubit_storage_layout}.

Now let's come to time. In the decomposition of logical circuit for $(m=4095, n=70, b=12)$, there are $58215$ calls to \lstinline{GF2Mul} and $687564$ calls to \lstinline{SynthesizeLR}. 
We define the Parity Control Toffoli gate as a generalization of the Toffoli gate that computes the Boolean AND of the XOR of a subset of controls and updates a subset of targets \cite[Figure 1]{gidney2025classicalquantumadderconstantworkspace}; we employ such gates since they can be implemented efficiently using lattice surgery \cref{fig:parity_control_toffoli_layout}. 
Our \lstinline{GF2Mul} compilation from \cref{fig:gf2_mul} requires $1.5d$ rounds per Parity Control Toffoli gate (PCTOF). 
For $b=12$, using optimized GF2 arithmetic circuits from \cref{sec:improved_gf2_arithmetic}, there are $51$ PCTOF gates. 
Therefore, it takes $1.5 \cdot 51 \approx 77d$ rounds for one quantum-quantum multiplication of field elements.  
For \lstinline{SynthesizeLR}, we present a hand-optimized layout in \cref{fig:gf2_mul_k} that takes $10d$ rounds and $3\times 23$ surface code patches for $b=12$. 
Therefore, with $d=21$, these two combined give us $58215 \cdot 77 \cdot 21 +  687564 \cdot 10 \cdot 21 = 2.4 \times 10^8$. 
To provide a conservative estimate of the total runtime, we multiply this number by a factor of four:
doubling it once to account for overhead from other smaller operations and routing, and doubling it a second time to account for the time needed to move qubits between cold and hot storage. Thus, we estimate a total of $1 \times 10^9$ rounds.
We need to protect $1920 + 120 = 2040$ logical qubits for a total of $1\times 10^9$ rounds. 
With an LER of $10^{-13}$, this gives us a no-logical-error shot rate of $p_\text{lattice surgery} = (1 - 10^{-13})^{2040 \cdot 1\cdot 10^9} = (1 - 10^{-13})^{2\times 10^{12}} \approx 81.5\%$ and a per shot runtime of ~$16$ minutes. 
To account for error due to cultivation, we cultivate $\ket{T}$ states with a logical error rate of $2\times 10^{-9}$. 
For a per-shot Toffoli count of $5.72 \times 10^6$, this gives us a $p_\text{cultivation} = (1 - 2\times 10^{-9})^{4 \times 5.72 \times 10^6} \approx 95.5\%$. 
Combining the two success probabilities above, we have a per-shot success rate of $p_\text{cultivation} \times p_\text{lattice surgery} \approx 77.9\%$. Since $\approx 17$ fault-free shots are needed in expectation at finite size (\cref{sec:optimizePrange}), this requires $17 / 0.779 \approx 22$ physical shots, for a total runtime of $\approx 6$ hours.

To summarize, we estimate that solving a classically intractable instance of the OPI problem using DQI requires $\approx 800k$ physical qubits and $\approx 6$ hours of runtime. This estimate assumes a quantum computer with a surface code cycle time of 1 microsecond, a control system reaction time of 10 microseconds, a square grid of qubits with nearest neighbor connectivity, and a uniform depolarizing noise model with a noise strength of 1 error per 1000 gates.

While our physical resource estimates assume a conservative 2D nearest neighbour surface code architecture, we note that the logical resource reductions achieved by our algorithm translate to substantial physical savings across any fault-tolerant architecture. 
Recently, alternative architectures based on neutral atoms and high-rate quantum LDPC (qLDPC) codes have been proposed \cite{yoder2025tourgrossmodularquantum,webster2026pinnaclearchitecturereducingcost, cain2026shorsalgorithmpossible10000}, which theoretically offer massive reductions in physical qubit overhead for algorithms like integer factorization. 
Because our optimized DQI circuits trade a drastically reduced Toffoli gate count for a modestly higher logical qubit count (relative to Shor's algorithm), our algorithm is uniquely well-suited to benefit from these high-rate qLDPC architectures.

\begin{figure}
    \centering
    \subfloat[
    ZX graph for synthesizing a linear reversible circuit (\lstinline{SynthesizeLR}) on $10$ qubits using LU-decomposition. 
    Depending on the entries of the LU-decomposition of the matrix, only a subset of CNOT gates will be present in the circuit.
    For the task of estimating the worse case spacetime overhead, we assume all possible pairs of CNOTs can be present.
    ]{\includegraphics[width=\linewidth]{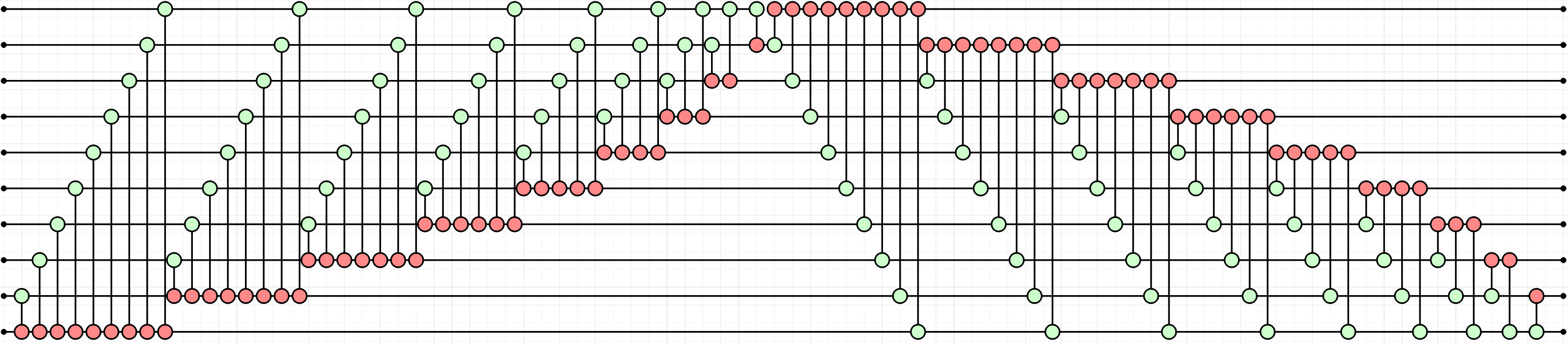}\label{fig:gf2_mul_k_zx_graph}
    }

    \subfloat[
    Lattice surgery layouts (front and back view) for synthesizing a linear reversible circuit on $10$ qubits, compiled using ZX graph shown in \cref{fig:gf2_mul_k_zx_graph}. 
    The layout fits in a $3 \times 19$ rectangle of surface code patches and requires $8d$ rounds.    
    ]{
        \includegraphics[width=0.5\linewidth]{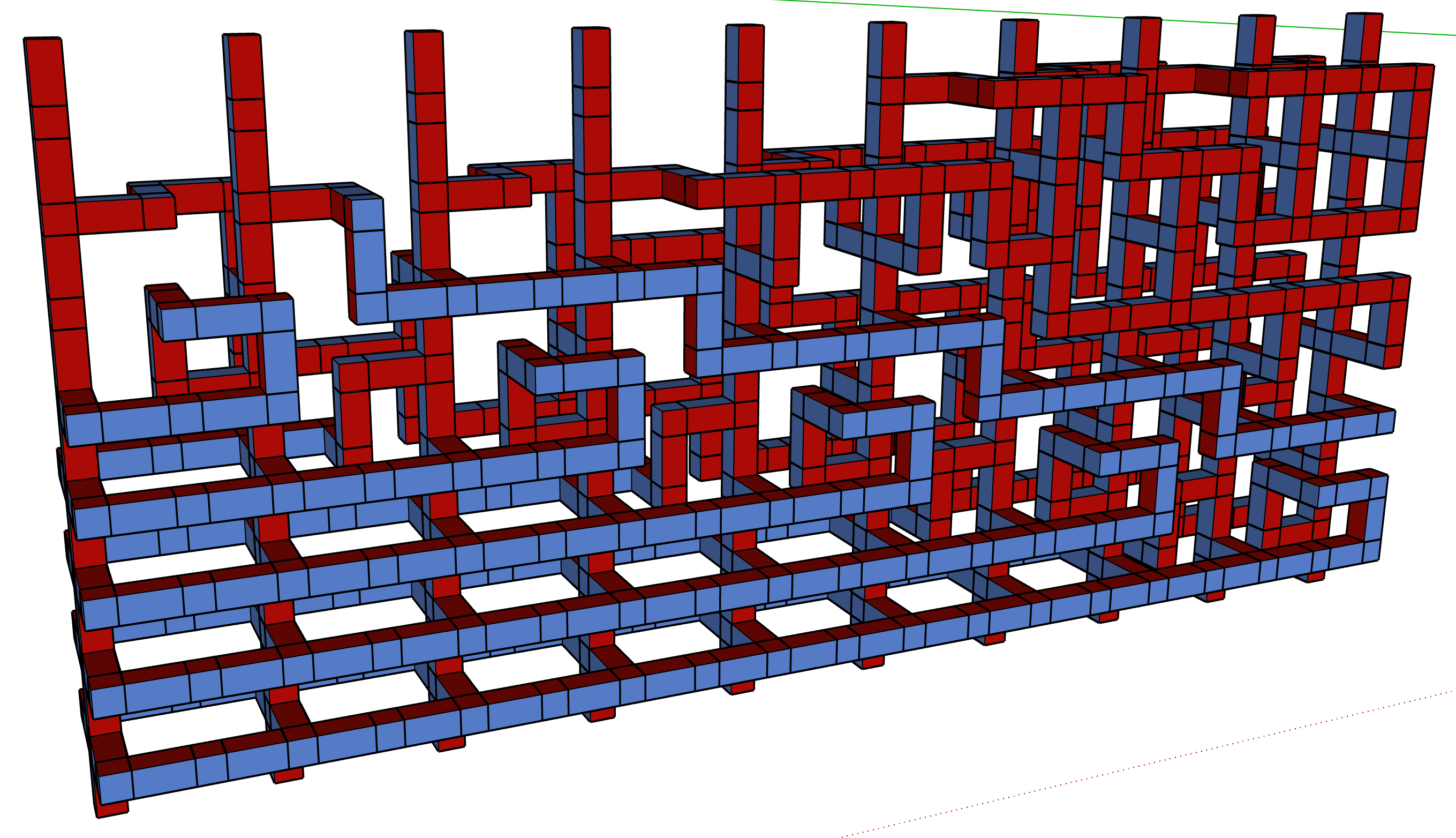}
        \includegraphics[width=0.5\linewidth]{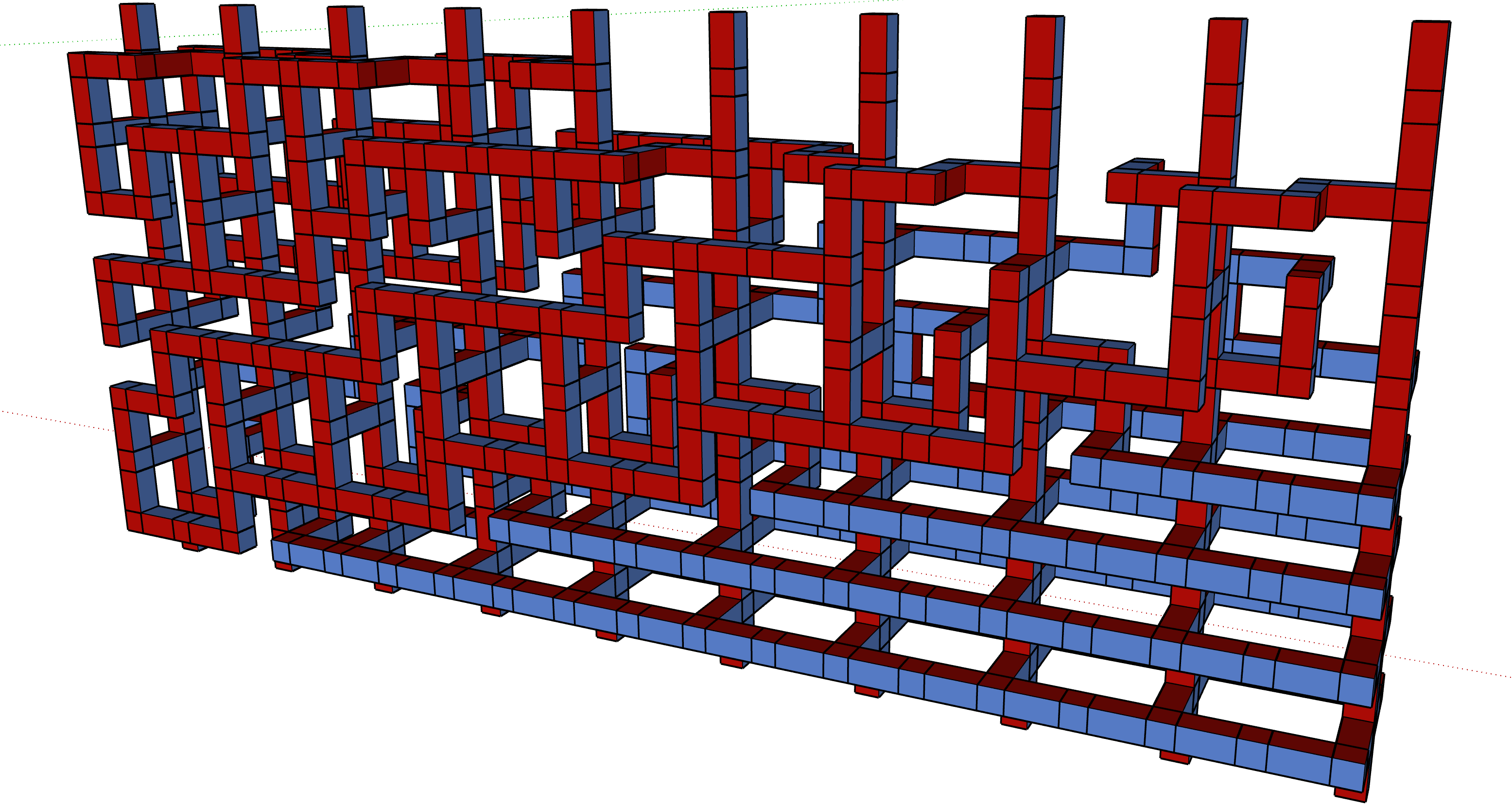}    
    }

    \subfloat[
    Lattice surgery layouts (front and back view) for synthesizing a linear reversible circuit on $12$ qubits, compiled using ZX graph shown in \cref{fig:gf2_mul_k_zx_graph}. 
    The layout fits in a $3 \times 23$ rectangle of surface code patches and requires $10d$ rounds.    
    ]{
    \includegraphics[width=0.5\linewidth]{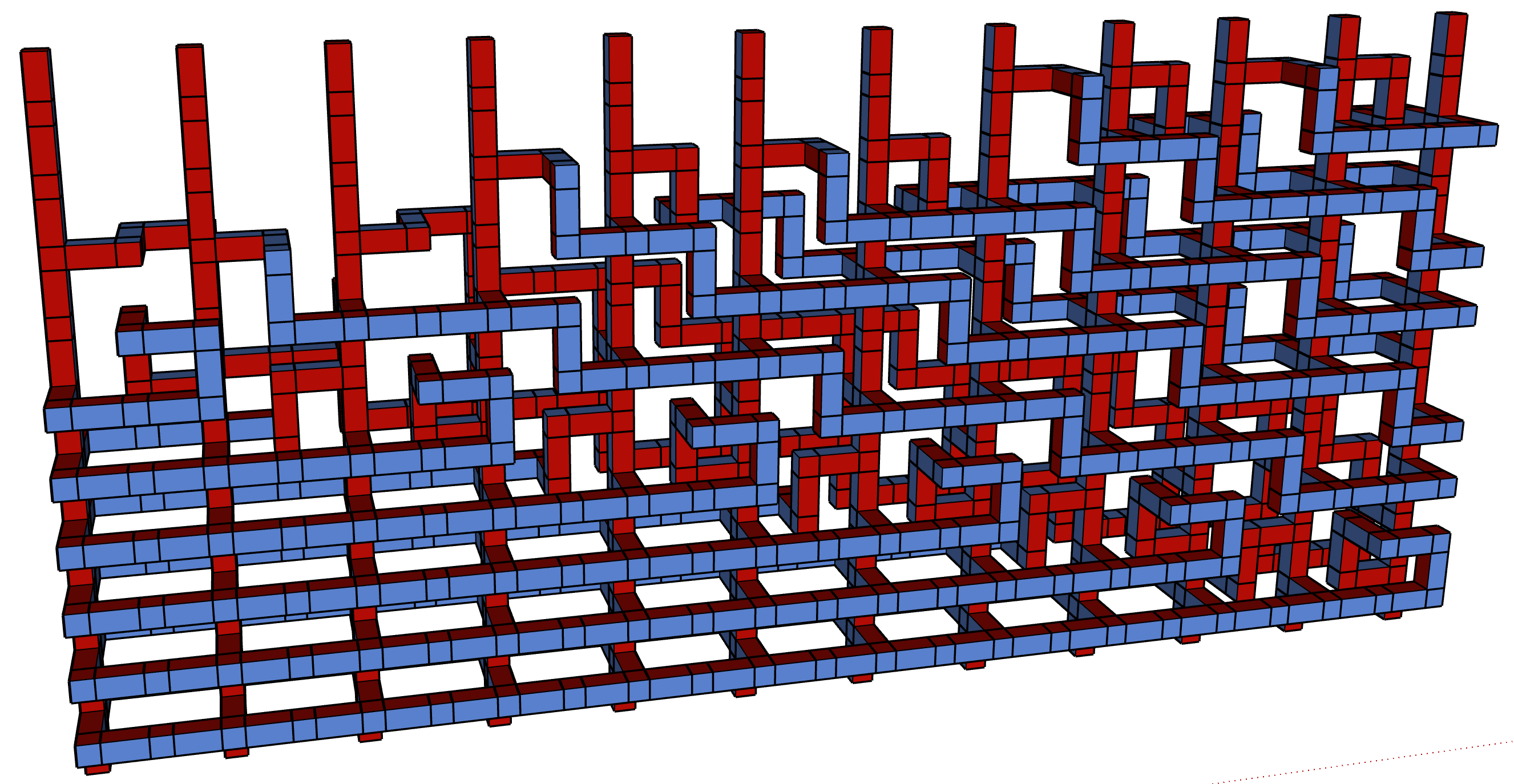}
    \includegraphics[width=0.5\linewidth]{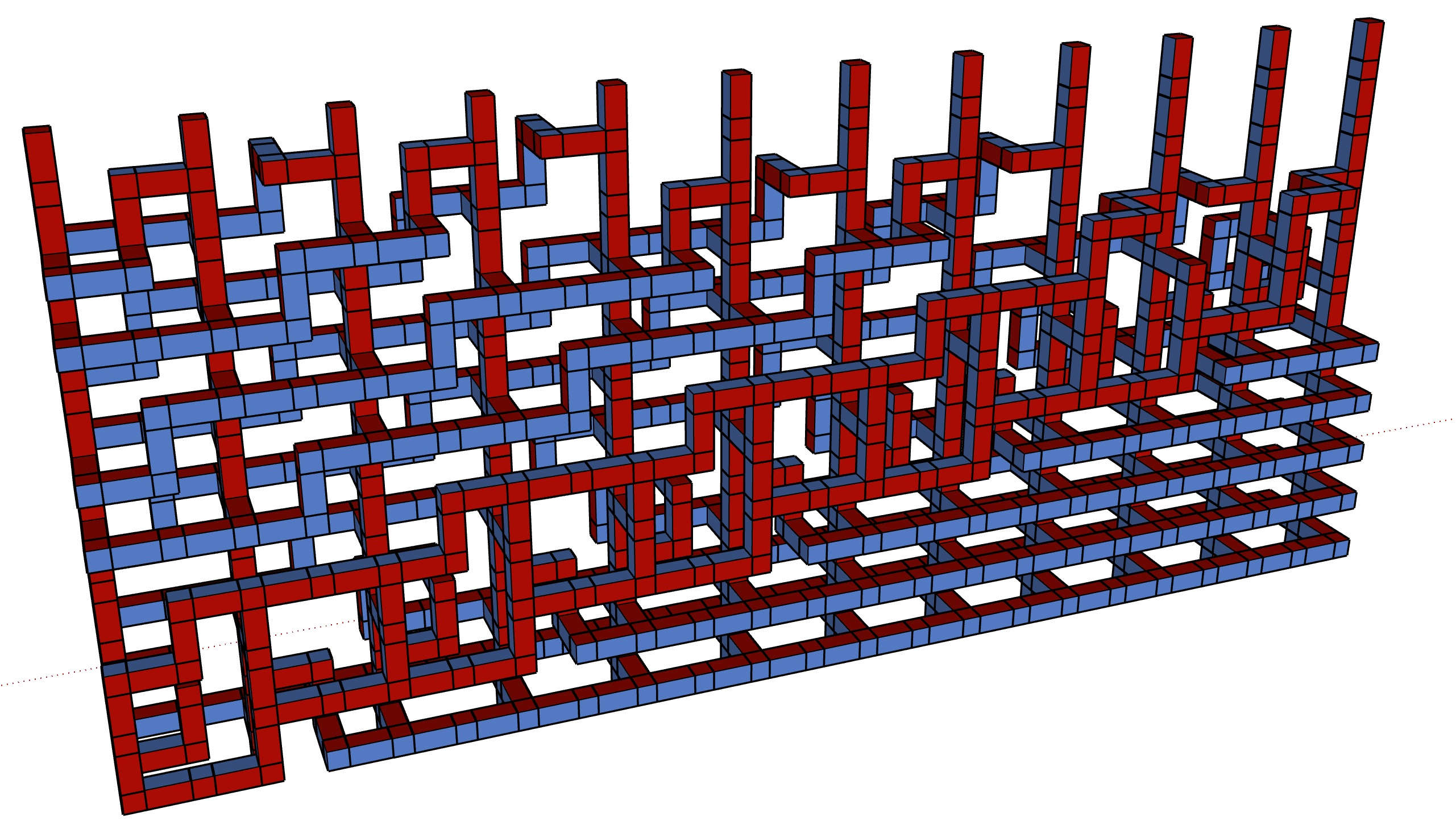}    
    }

    \caption{
        Spacetime layout for \lstinline{SynthesizeLR} - a primitive for synthesizing linear reversible circuits. For GF2 arithmetic, field operations like multiplication by a constant polynomial and squaring reduce to \lstinline{SynthesizeLR} \cite{maslov2025asymptotic}.
    }
    \label{fig:gf2_mul_k}
\end{figure}

\begin{figure}
    \centering
    \subfloat[
        ZX graph for a Parity Control Toffoli gate, which computes the AND of the parity of two sets of control qubits and flips one or more target qubits. 
        For the specific example, the Parity Control Toffoli flips 4 target qubits (on the right) based on $(x_1 \oplus x_2 \oplus x_3 \oplus x_4) \wedge (y_1 \oplus y_2 \oplus y_3 \oplus y_4)$. 
        Here is \href{https://algassert.com/quirk\#circuit=\%7B\%22cols\%22\%3A\%5B\%5B1\%2C1\%2C1\%2C\%22Y\%5Et\%22\%2C\%22Y\%5E-t\%22\%2C\%22Y\%5Et\%22\%2C\%22Y\%5E-t\%22\%2C\%22Y\%5Et\%22\%2C\%22Y\%5E-t\%22\%2C\%22Y\%5Et\%22\%2C\%22Y\%5E-t\%22\%2C\%22Y\%5Et\%22\%2C\%22Z\%5E\%C2\%BC\%22\%5D\%2C\%5B1\%2C1\%2C1\%2C\%22QFT9\%22\%5D\%2C\%5B1\%2C1\%2C1\%2C\%22Chance9\%22\%5D\%2C\%5B1\%2C\%22X\%22\%2C1\%2C\%22zpar\%22\%2C\%22zpar\%22\%2C\%22zpar\%22\%5D\%2C\%5B1\%2C1\%2C\%22X\%22\%2C1\%2C1\%2C1\%2C\%22zpar\%22\%2C\%22zpar\%22\%2C\%22zpar\%22\%5D\%2C\%5B1\%2C\%22\%E2\%80\%A2\%22\%2C\%22\%E2\%80\%A2\%22\%2C1\%2C1\%2C1\%2C1\%2C1\%2C1\%2C\%22X\%22\%2C\%22X\%22\%2C\%22X\%22\%5D\%2C\%5B1\%2C1\%2C\%22X\%22\%2C1\%2C1\%2C1\%2C\%22zpar\%22\%2C\%22zpar\%22\%2C\%22zpar\%22\%5D\%2C\%5B1\%2C\%22X\%22\%2C1\%2C\%22zpar\%22\%2C\%22zpar\%22\%2C\%22zpar\%22\%5D\%2C\%5B\%22\%E2\%80\%A6\%22\%2C\%22\%E2\%80\%A6\%22\%2C\%22\%E2\%80\%A6\%22\%2C\%22\%E2\%80\%A6\%22\%2C\%22\%E2\%80\%A6\%22\%2C\%22\%E2\%80\%A6\%22\%2C\%22\%E2\%80\%A6\%22\%2C\%22\%E2\%80\%A6\%22\%2C\%22\%E2\%80\%A6\%22\%2C\%22\%E2\%80\%A6\%22\%2C\%22\%E2\%80\%A6\%22\%2C\%22\%E2\%80\%A6\%22\%5D\%2C\%5B\%22X\%22\%2C1\%2C1\%2C\%22zpar\%22\%2C\%22zpar\%22\%2C\%22zpar\%22\%2C1\%2C1\%2C1\%2C1\%2C1\%2C1\%2C\%22zpar\%22\%5D\%2C\%5B\%22Z\%5E-\%C2\%BC\%22\%5D\%2C\%5B\%22X\%22\%2C1\%2C1\%2C\%22zpar\%22\%2C\%22zpar\%22\%2C\%22zpar\%22\%2C1\%2C1\%2C1\%2C1\%2C1\%2C1\%2C\%22zpar\%22\%5D\%2C\%5B\%22X\%22\%2C1\%2C1\%2C1\%2C1\%2C1\%2C\%22zpar\%22\%2C\%22zpar\%22\%2C\%22zpar\%22\%2C1\%2C1\%2C1\%2C\%22zpar\%22\%5D\%2C\%5B\%22Z\%5E-\%C2\%BC\%22\%5D\%2C\%5B\%22X\%22\%2C1\%2C1\%2C1\%2C1\%2C1\%2C\%22zpar\%22\%2C\%22zpar\%22\%2C\%22zpar\%22\%2C1\%2C1\%2C1\%2C\%22zpar\%22\%5D\%2C\%5B\%22X\%22\%2C1\%2C1\%2C\%22zpar\%22\%2C\%22zpar\%22\%2C\%22zpar\%22\%2C\%22zpar\%22\%2C\%22zpar\%22\%2C\%22zpar\%22\%2C1\%2C1\%2C1\%2C\%22zpar\%22\%5D\%2C\%5B\%22Z\%5E\%C2\%BC\%22\%5D\%2C\%5B\%22X\%22\%2C1\%2C1\%2C\%22zpar\%22\%2C\%22zpar\%22\%2C\%22zpar\%22\%2C\%22zpar\%22\%2C\%22zpar\%22\%2C\%22zpar\%22\%2C1\%2C1\%2C1\%2C\%22zpar\%22\%5D\%2C\%5B1\%2C1\%2C1\%2C1\%2C1\%2C1\%2C1\%2C1\%2C1\%2C1\%2C1\%2C1\%2C\%22H\%22\%5D\%2C\%5B1\%2C1\%2C1\%2C1\%2C1\%2C1\%2C1\%2C1\%2C1\%2C1\%2C1\%2C1\%2C\%22Z\%5E\%C2\%BD\%22\%5D\%2C\%5B1\%2C1\%2C1\%2C1\%2C1\%2C1\%2C1\%2C1\%2C1\%2C\%22X\%22\%2C\%22X\%22\%2C\%22X\%22\%2C\%22zpar\%22\%5D\%2C\%5B1\%2C1\%2C1\%2C1\%2C1\%2C1\%2C1\%2C1\%2C1\%2C1\%2C1\%2C1\%2C\%22H\%22\%5D\%2C\%5B1\%2C1\%2C1\%2C1\%2C1\%2C1\%2C1\%2C1\%2C1\%2C1\%2C1\%2C1\%2C\%22Measure\%22\%5D\%2C\%5B1\%2C1\%2C1\%2C\%22\%E2\%80\%A2\%22\%2C1\%2C1\%2C\%22Z\%22\%2C\%22Z\%22\%2C\%22Z\%22\%2C1\%2C1\%2C1\%2C\%22\%E2\%80\%A2\%22\%5D\%2C\%5B1\%2C1\%2C1\%2C1\%2C\%22\%E2\%80\%A2\%22\%2C1\%2C\%22Z\%22\%2C\%22Z\%22\%2C\%22Z\%22\%2C1\%2C1\%2C1\%2C\%22\%E2\%80\%A2\%22\%5D\%2C\%5B1\%2C1\%2C1\%2C1\%2C1\%2C\%22\%E2\%80\%A2\%22\%2C\%22Z\%22\%2C\%22Z\%22\%2C\%22Z\%22\%2C1\%2C1\%2C1\%2C\%22\%E2\%80\%A2\%22\%5D\%2C\%5B1\%2C1\%2C1\%2C\%22QFT\%E2\%80\%A09\%22\%5D\%2C\%5B1\%2C1\%2C1\%2C\%22Y\%5E-t\%22\%2C\%22Y\%5Et\%22\%2C\%22Y\%5E-t\%22\%2C\%22Y\%5Et\%22\%2C\%22Y\%5E-t\%22\%2C\%22Y\%5Et\%22\%2C\%22Y\%5E-t\%22\%2C\%22Y\%5Et\%22\%2C\%22Y\%5E-t\%22\%5D\%5D\%2C\%22init\%22\%3A\%5B0\%2C0\%2C0\%2C0\%2C0\%2C0\%2C0\%2C0\%2C0\%2C0\%2C0\%2C0\%2C\%22\%2B\%22\%5D\%7D}{quirk link}
        that shows how a Parity Control Toffoli can be implemented by consuming 4 $\ket{T}$ states.
    ]{
    \includegraphics[width=1.0\linewidth]{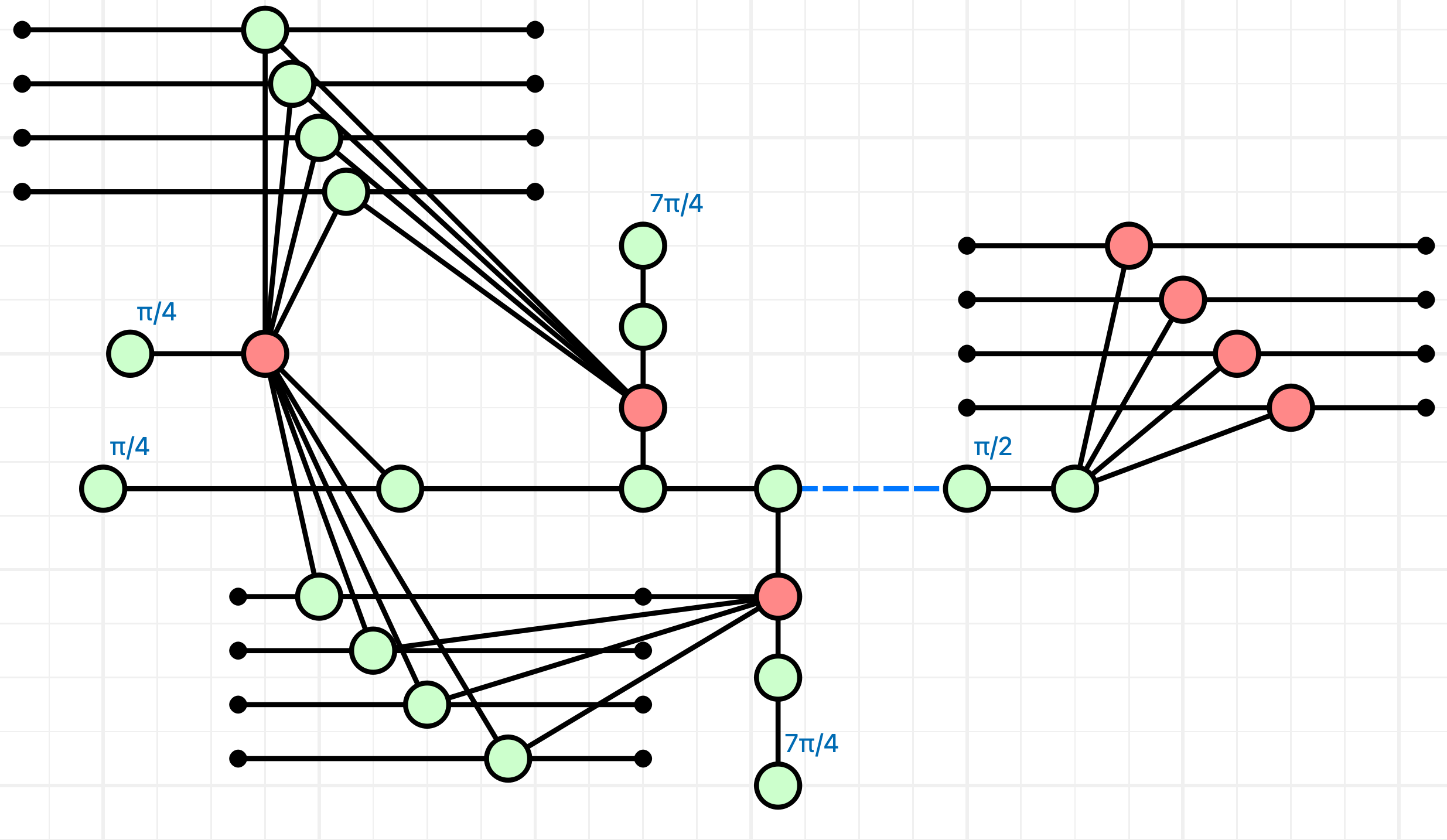}\label{fig:generalized_toffoli_zx}
    }

    \subfloat[
        Lattice surgery diagram for compiling two Parity Control Toffoli gates, each acting on two sets of $10$-qubit control registers (middle and right) and one $10$-qubit target register (on the left).
        Green boxes correspond to Y basis initialization or measurement \cite{Gidney_2024}. 
        Gray boxes correspond to the reaction time for decoder to process the measurement results.
        Pink boxes correspond to space available for magic state cultivation \cite{gidney2024magicstatecultivationgrowing}. 
        Every $\ket{T}$ state has $3\times2\times d^3$ spacetime available for cultivation. 
        With $d=21$, this is enough to cultivate $\ket{T}$ states with a logical error rate of $2\times 10^{-9}$.
        The space cost is $3 \times (2.5b + 9)$ and amortized time cost per Parity Control Toffoli is $1.5d$ rounds.
    ]{
    \includegraphics[width=\linewidth]{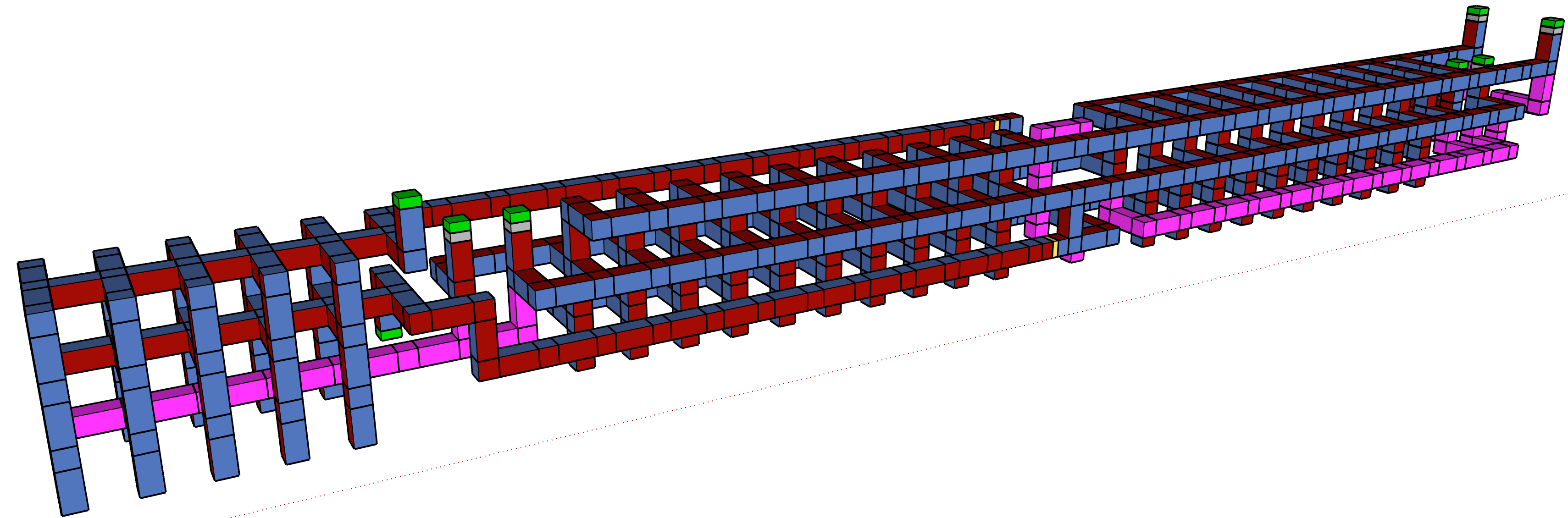}\label{fig:two_generalized_toffoli_pipes}
    }
    \caption{
    Lattice surgery compilation for two Parity Control Toffoli gates using magic state cultivation. 
    See \cite[Figure 1]{gidney2025classicalquantumadderconstantworkspace} and \cref{sec:improved_gf2_arithmetic} for more discussion on Parity Control Toffoli gates.
    Acts as a building block for compiling quantum-quantum multiplication (\lstinline{GF2Mul}) circuits for $\GF(2^b)$.
    }
    \label{fig:parity_control_toffoli_layout}
\end{figure}

\begin{figure}
    \centering
    \includegraphics[width=0.75\linewidth]{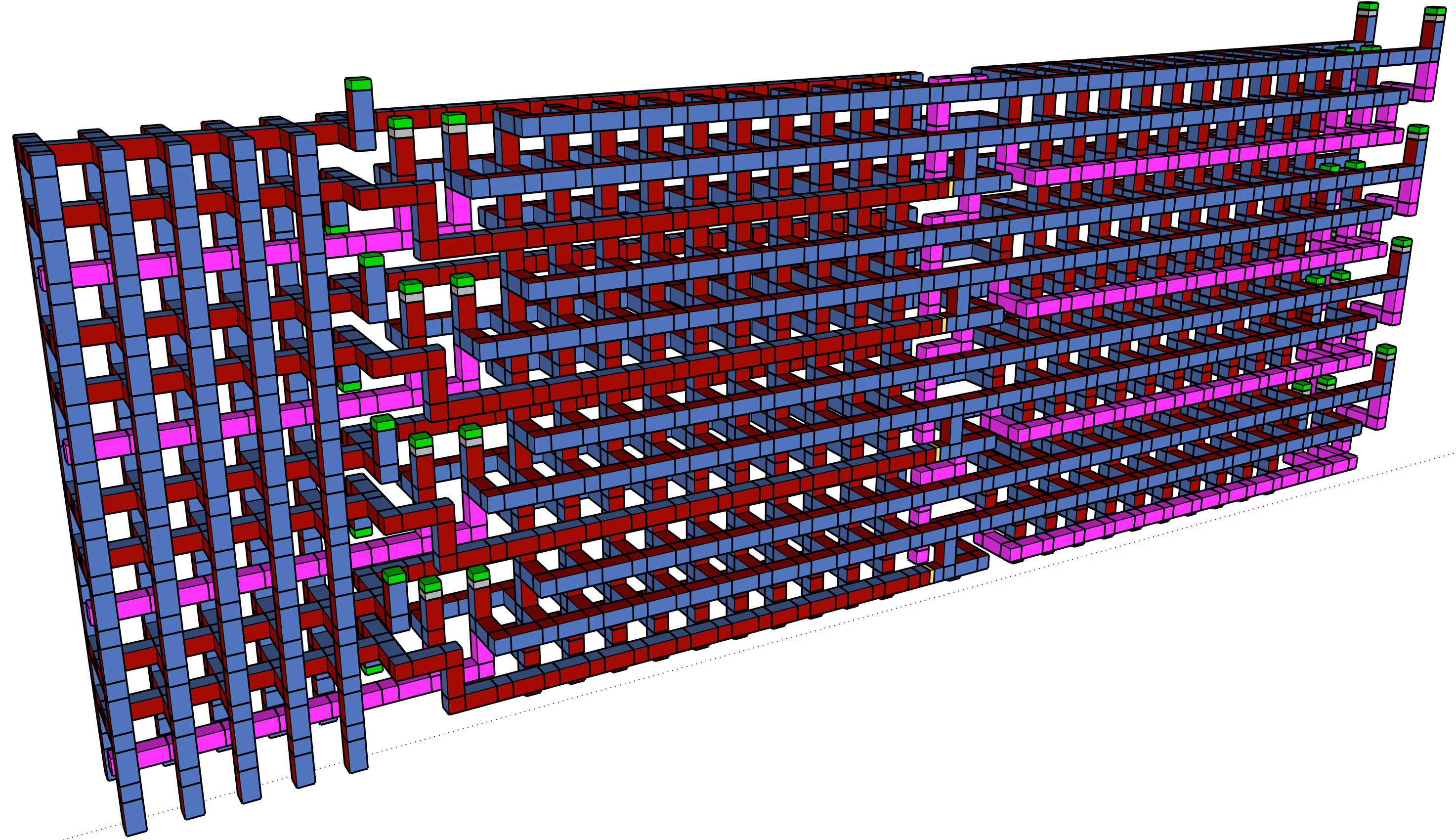}
    \caption{
    Lattice Surgery diagram for compiling GF2 Multiplication (\lstinline{GF2Mul(10)}) using Karatsuba algorithm by stacking the Parity Control Toffoli primitive from \cref{fig:two_generalized_toffoli_pipes}.   
    In \cref{sec:improved_gf2_arithmetic}, we show how the Karatsuba algorithm for GF2 multiplication can be viewed entirely as a sequence of Parity Control Toffoli gates. 
    For $b=10$, we use the modified Karatsuba quantum circuit \cite{maslov2025asymptotic}, which uses exactly $39$ (parity) Toffolis.
    }
    \label{fig:gf2_mul}
\end{figure}

\begin{figure}
    \centering
    \includegraphics[width=0.65\linewidth]{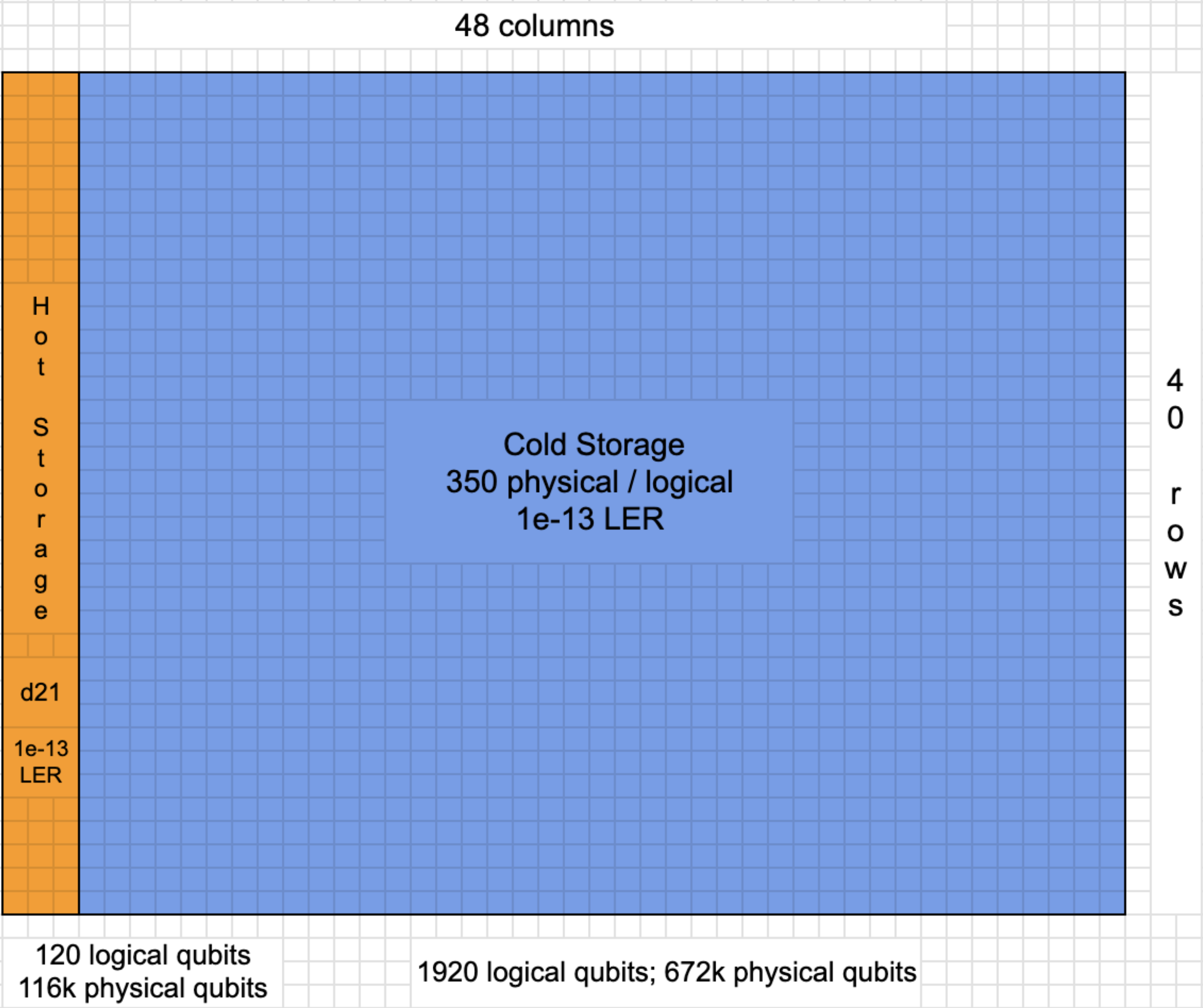}
    \caption{
    Mockup of a physical layout for $(m=4095, n=70, b=12, r=2016)$ OPI instance requiring $1885$ logical qubits and $5.72 \times 10^6$ Toffoli gates (see \cref{tbl:resource_estimates}). 
    The blue region stores logical qubits in \emph{cold storage}, where logical qubits are stored as densely as possible but are not operated upon.
    The orange region stores logical qubits in \emph{hot storage}, where each logical qubit is stored ``normally" using a surface code patch taking $2\times(d+1)^2$ physical qubits and is actively operated upon.
    For a logical error rate (LER) of $10^{-13}$, a $d=21$ surface code with $968$ physical qubits per logical qubit is used for hot storage and a 2D Yoked Surface Code \cite{gidney2025yoked} with $350$ physical qubits per logical qubit is used for cold storage \cite[Figure 6]{gidney2025factor2048bitrsa}.
    }
    \label{fig:qubit_storage_layout}
\end{figure}

\section{Discussion and Conclusion}\label{sec:conclusion}
This work establishes Decoded Quantum Interferometry (DQI) applied to the Optimal Polynomial Intersection (OPI) problem as a primary candidate for demonstrating verifiable quantum advantage. 
From a complexity-theoretic perspective, we have shown that DQI+OPI is the first known proposal for verifiable superpolynomial speedup that achieves optimal asymptotic efficiency, requiring only $\bigOtilde(N)$ quantum gates to solve instances with classical hardness $O(2^N)$, matching the theoretical lower bound.

The realization of this speedup relies on the concrete efficiency of the underlying reversible Reed-Solomon (RS) decoder. We introduced a suite of novel algorithmic and compilation techniques targeting the Extended Euclidean Algorithm (EEA), the critical bottleneck in RS decoding. 
Our innovations, including the \emph{Dialog} representation for implicit Bézout coefficient access and optimized \emph{in-place register sharing} architectures, rigorously reduce the leading-order space complexity to $2nb$ qubits while substantially lowering gate counts. 
These techniques are broadly applicable and promise significant improvements for other algebraic quantum algorithms, particularly Shor’s algorithm for Elliptic Curve Cryptography, as detailed in \cref{sec:broader_applicability} below.
We provided a comprehensive end-to-end compilation and resource analysis for DQI, focusing on OPI over binary extension fields $\GF(2^b)$. We analyzed the classical hardness against tailored attacks, such as the Extended Prange algorithm, and identified resilient instances based on bent functions. Our concrete resource estimates demonstrate that DQI can solve instances requiring $>10^{23}$ classical trials using approximately 5.72 million Toffoli gates per shot and 1885 logical qubits. This represents a reduction of two orders of magnitude in the gate count compared to breaking RSA-2048. 
Remarkably, these resource requirements are also lower than the most efficient interactive proofs of quantumness using contrived problems based on cryptographic trapdoors \cite{kahanamoku2022classically, kahanamokumeyer2024fastquantumintegermultiplication}. 
These findings challenge the prevailing assumption that practical, potentially useful optimization problems are inherently more expensive to implement than contrived verification tasks.
Furthermore, our physical resource analysis, supported by hand-optimized lattice surgery layouts for key primitives, estimates that such instances could be solved on a fault-tolerant architecture with $800,000$ physical qubits in $\approx 6$ hours of runtime.

\subsection{Applicability Beyond OPI}\label{sec:broader_applicability}
Although our end-to-end resource analysis is specific to DQI applied to OPI, most of the circuit and compilation constructions developed here rely on neither. Here we describe which algebraic or circuit properties each construction uses and, consequently, how these techniques apply to and optimize other quantum algorithms.

\paragraph{Extended Euclidean Algorithm and elliptic curve cryptography.}
In Shor's algorithm for the elliptic curve discrete logarithm problem (ECDLP) over an $n$-bit field (e.g., $n=256$ for \texttt{secp256k1}), the primary arithmetic bottleneck is windowed point addition, $\ket{w}\ket{Q} \mapsto \ket{w}\ket{Q + P[w]}$, where $\ket{Q}$ is a $2n$-qubit register storing the curve point $(Q_x, Q_y)$ and $P$ is a classical lookup table of $2^w$ precomputed points \cite{proos2004shorsdiscretelogarithmquantum, roetteler2017quantumresourceestimatescomputing, häner2020improvedquantumcircuitselliptic, litinski2023compute256bitellipticcurve}.
Computing the chord slope reduces the dominant cost of point addition to in-place modular division, $\ket{y}\ket{x} \mapsto \ket{y}\ket{(x / y) \bmod p}$ on two $n$-bit registers.
Historically, in-place modular division was decomposed via an invert-then-multiply sequence into two EEA-based modular inversions and two out-of-place multiplications---requiring $3n + \mathrm{Anc}(\text{inv})$ qubits, where $\mathrm{Anc}(\text{inv})$ is the ancilla overhead of modular inversion---while the underlying EEA faced a steep space--time tradeoff: Proos and Zalka's stepwise-reversible Euclid algorithm \cite{proos2004shorsdiscretelogarithmquantum, kaye2004optimizedquantumimplementationelliptic} shares two $n$-bit registers between remainders and Bézout coefficients in $2n + O(\sqrt{n})$ qubits (using a fixed $O(\sqrt{n})$ buffer across superposed trajectories of fluctuating length) at the cost of high Toffoli counts, whereas constant-time division-free GCD variants such as Binary GCD \cite{Stein1967} and Bernstein--Yang \cite{cryptoeprint:2019/266} achieve much lower Toffoli counts but accumulate $\Theta(n)$ bits of irreversible garbage.

Over binary extension fields ($\mathbb{F}_q[z]/P(z)$ with $q=2, b=1$), our constructions directly yield the reduced-space modular inversion and modular division circuits summarized in \cref{tbl:prior_art_eea,tbl:prior_art_eea_division} for binary elliptic curves \cite{kaye2004optimizedquantumimplementationelliptic, cryptoeprint:2020/1296, kim2023newspaceefficientquantumalgorithm}.
Replacing polynomial degrees with integer bit lengths transfers both techniques to prime-field ECDLP, where they have driven a rapid sequence of advances since our preprint appeared:
\begin{itemize}
    \item \emph{Dialog / record-and-replay regime:} The Dialog representation introduced in \cref{sec:dialog_representation} resolves the constant-time GCD space--time tradeoff by recognizing that the garbage emitted by division-free GCD steps is a linear trace of the EEA that fits in the space freed as the active operands shrink and can be consumed directly to apply in-place modular division, multiplication, and inversion. Adapting this approach from polynomials to integers in our subsequent work on 256-bit ECC \cite{babbush2026securingellipticcurvecryptocurrencies} improved the logical circuits by roughly $20\times$ in spacetime volume---down to $< 7 \times 10^7$ Toffoli gates and $1{,}450$ logical qubits (or $< 9 \times 10^7$ Toffoli gates and $< 1{,}200$ logical qubits in a space-optimized variant)---and follow-up works have further explained and refined these constructions \cite{Schrottenloher2026Optimized, Haner2026TrappedIon, Long2026RecordReplay}.
    \item \emph{Space-minimized Proos--Zalka regime:} Complementarily, our compilation of Zalka's EEA (\cref{sec:zalka_eea_optimized}), which tracks variable operand lengths coherently so quotients fit in-place within the $2n$-qubit shared register in $2n + O(\log_2 n)$ qubits (so $\mathrm{Anc}(\text{inv}) = n + O(\log_2 n)$ for a classical modulus) while merging redundant arithmetic across state-machine phases, was extended from polynomials to integers by Luo \emph{et al.}~\cite{Luo2026SpaceEfficient}. Crucially, Luo \emph{et al.}~\cite{Luo2026SpaceEfficient} also introduced an $X$-basis measurement-based uncomputation step on $\ket{x}$ before uncomputing $\ket{y^{-1}}$ in the invert-then-multiply sequence, reducing its space cost from $3n + \mathrm{Anc}(\text{inv})$ to $2n + \mathrm{Anc}(\text{inv}) = 3n + O(\log_2 n)$ and lowering 256-bit ECDLP to $835$ logical qubits.
\end{itemize}
Together, these developments define the current space--time Pareto frontier of 256-bit ECDLP point-addition circuits tracked on \href{https://www.ecdsa.fail/}{\texttt{ecdsa.fail}} \cite{Long2026RecordReplay}: our Dialog representation (combined below $1{,}000$ qubits with the $X$-basis MBU register-recycling trick of \cite{Luo2026SpaceEfficient}) underpins the low-gate ($\ge 840$-qubit) frontier, while optimized Proos--Zalka EEA circuits govern the ultra-low-space ($< 840$-qubit) regime.

\paragraph{Other instances of DQI and other error-correcting codes.}
The Sparse Dicke state preparation and the sequential syndrome computation with measurement-based uncomputation (Stages 1 and 2 of \cref{sec:dqi_quantum_circuit}) make no use of the Reed-Solomon structure: they only require that the columns of $B^T$ can be added into the syndrome register one at a time.
They therefore remove the $m$-qubit Dicke register and the $mb$-qubit error register for any max-LINSAT instance with $m \gg n$.
Likewise, the key equation \eqref{eq:key_equation} solved by our reversible decoders is not specific to Reed-Solomon codes: BCH codes are decoded using the same equation, and classical Goppa codes using the same equation with $z^{2\ell}$ replaced by the Goppa polynomial \cite{SUGIYAMA197587}, so the same $2nb$-qubit EEA circuits can serve as reversible decoders for these code families.

\paragraph{Binary-field arithmetic and combinatorial state preparation.}
The decomposition of $\GF(2^b)$ multiplication into mutually commuting Parity Control Toffoli gates, their lattice-surgery implementation with in-place magic state cultivation (\cref{fig:parity_control_toffoli_layout,fig:gf2_mul}), and the uncomputation of products using only measurements and $CZ$ gates (Appendix~\ref{sec:improved_gf2_arithmetic}) are properties of $\GF(2^b)$ multiplication itself.
They are available to any quantum algorithm whose cost is dominated by binary-field multiplications, such as Shor's algorithm for elliptic curves over binary fields \cite{kaye2004optimizedquantumimplementationelliptic, cryptoeprint:2020/1296, kim2023newspaceefficientquantumalgorithm}.
Finally, the unranking circuits of Appendix~\ref{sec:sparse_dicke_state}---providing both a near-linear $\bigOtilde(m)$-gate divide-and-conquer construction and a low-overhead iterative construction---prepare uniform superpositions over all $k$-element subsets of $m$ items (Dicke states in sparse or dense encodings), a primitive used in quantum chemistry for preparing spin-adapted configuration state functions and antisymmetrized fixed-particle-number states \cite{Sugisaki2016, Berry2018eigenstates, Babbush2019sublinear, Gard2020}.

\subsection{Outlook and Open Questions}\label{sec:outlook}
While our results position DQI as a compelling pathway for demonstrating quantum advantage in optimization, several avenues for future research remain.
Continued refinement of the classical hardness analysis is crucial. While we have demonstrated resilience against the best known classical attacks in the low-rate approximation regime, further investigation into the impact of sophisticated algebraic attacks (\cite{briaud2025quantum}) and establishing formal average-case hardness guarantees remain important open problems.
On the quantum side, there are a few routes to improved performance as well. In \cite{chailloux2025quantum}, a few interesting ideas were proposed. Among these was the idea to use Guruswami-Sudan and Koetter-Vardy list decoding algorithms for RS codes \cite{guruswami1998improved,koetter2003algebraic}. This would allow us to find solutions that would take exponential time even with DQI and the Berlekamp-Massey decoder, with the cost of a higher (still polynomial) quantum circuit complexity.
In a concurrent work, \cite{GuJordan2025Algebraic} considers applying DQI to algebraic geometry codes, where our sequential MBU architecture likewise removes the $\mathcal{O}(mb)$ error-register overhead. These codes have some potentially favorable properties for verifiable quantum advantage demonstrations. Notably, they can maintain a constant rate and relative distance with considerably smaller (and in some cases constant) alphabet sizes than Reed-Solomon codes, which could improve the space footprint of the syndrome registers.
The Dialog representation introduced here opens new theoretical and practical questions regarding linear representations of numbers and polynomials. Key open questions include:
\vspace{-0.5em}
\begin{enumerate}
    \itemsep0em 
    \item Can the Dialog size be reduced from $\approx 2n$ to $\approx n$ field elements when one GCD input is known classically?
    \item Is there an efficient, in-place algorithm for transforming the Dialog for $x$ into the Dialog for $x^{-1}$, or for computing the Dialog of $x+y$ from the Dialogs of $x$ and $y$, without passing through the standard polynomial representation?
    \item Are there more efficient linear representations beyond those based on the GCD algorithm, potentially enabling parallelization?
\end{enumerate}
\vspace{-0.5em}
By bridging complexity theory with concrete algorithmic engineering, this work identifies and enables an asymptotically optimal and practically efficient route toward verifiable quantum advantage. \\

\noindent \textbf{Contributions:}
T.K. designed and implemented logical circuits for qubit efficient DQI, Extended Euclidean Algorithm, Reed Solomon decoder and Sparse Dicke State Preparation; created lattice surgery layouts and performed resource estimation.
N.S. analyzed classical attacks on OPI to determine classically intractable instances, designed the optimized Prange's algorithm, formalized Extended Prange's algorithm (XP) for binary extension fields and observed asymptotically optimal quantum speedup for OPI.
C.G. came up with the Dialog representation for implicit access to Bézout coefficients. 
D.M. and N.Y. designed the optimized GF2 arithmetic circuits.
N.Y. implemented the knapsack solver for XP allocation strategy.
A.Z. proved upper bounds on affine intersections of Maiorana-McFarland target sets.
R.B. and S.J. supervised the project.
All authors contributed to writing and editing the paper. \\

\noindent \textbf{Acknowledgments:} We thank Mary Wootters for helpful discussions on decoding algorithms. We thank Alexandru Gheorghiu for teaching us about bent functions such as in eq~\eqref{eq:quadraticVariety} at a Simons workshop organized by Umesh Vazirani. We thank Hyunjoo Lee for helpful questions prompting us to clarify finite-size effects in \cref{tbl:resource_estimates}.

\clearpage
\bibliography{references}

\begin{thebibliography}{101}%
\makeatletter
\providecommand \@ifxundefined [1]{%
 \@ifx{#1\undefined}
}%
\providecommand \@ifnum [1]{%
 \ifnum #1\expandafter \@firstoftwo
 \else \expandafter \@secondoftwo
 \fi
}%
\providecommand \@ifx [1]{%
 \ifx #1\expandafter \@firstoftwo
 \else \expandafter \@secondoftwo
 \fi
}%
\providecommand \natexlab [1]{#1}%
\providecommand \enquote  [1]{``#1''}%
\providecommand \bibnamefont  [1]{#1}%
\providecommand \bibfnamefont [1]{#1}%
\providecommand \citenamefont [1]{#1}%
\providecommand \href@noop [0]{\@secondoftwo}%
\providecommand \href [0]{\begingroup \@sanitize@url \@href}%
\providecommand \@href[1]{\@@startlink{#1}\@@href}%
\providecommand \@@href[1]{\endgroup#1\@@endlink}%
\providecommand \@sanitize@url [0]{\catcode `\\12\catcode `\$12\catcode
  `\&12\catcode `\#12\catcode `\^12\catcode `\_12\catcode `\%12\relax}%
\providecommand \@@startlink[1]{}%
\providecommand \@@endlink[0]{}%
\providecommand \url  [0]{\begingroup\@sanitize@url \@url }%
\providecommand \@url [1]{\endgroup\@href {#1}{\urlprefix }}%
\providecommand \urlprefix  [0]{URL }%
\providecommand \Eprint [0]{\href }%
\providecommand \doibase [0]{https://doi.org/}%
\providecommand \selectlanguage [0]{\@gobble}%
\providecommand \bibinfo  [0]{\@secondoftwo}%
\providecommand \bibfield  [0]{\@secondoftwo}%
\providecommand \translation [1]{[#1]}%
\providecommand \BibitemOpen [0]{}%
\providecommand \bibitemStop [0]{}%
\providecommand \bibitemNoStop [0]{.\EOS\space}%
\providecommand \EOS [0]{\spacefactor3000\relax}%
\providecommand \BibitemShut  [1]{\csname bibitem#1\endcsname}%
\let\auto@bib@innerbib\@empty
\bibitem [{\citenamefont {Khattar}\ \emph {et~al.}(2026)\citenamefont
  {Khattar}, \citenamefont {Shutty}, \citenamefont {Gidney}, \citenamefont
  {Zalcman}, \citenamefont {Yosri}, \citenamefont {Maslov}, \citenamefont
  {Babbush},\ and\ \citenamefont
  {Jordan}}]{https://doi.org/10.5281/zenodo.24325102}%
  \BibitemOpen
  \bibfield  {author} {\bibinfo {author} {\bibfnamefont {T.}~\bibnamefont
  {Khattar}}, \bibinfo {author} {\bibfnamefont {N.}~\bibnamefont {Shutty}},
  \bibinfo {author} {\bibfnamefont {C.}~\bibnamefont {Gidney}}, \bibinfo
  {author} {\bibfnamefont {A.}~\bibnamefont {Zalcman}}, \bibinfo {author}
  {\bibfnamefont {N.}~\bibnamefont {Yosri}}, \bibinfo {author} {\bibfnamefont
  {D.}~\bibnamefont {Maslov}}, \bibinfo {author} {\bibfnamefont
  {R.}~\bibnamefont {Babbush}},\ and\ \bibinfo {author} {\bibfnamefont
  {S.}~\bibnamefont {Jordan}},\ }\href
  {https://doi.org/10.5281/ZENODO.24325102} {\bibinfo {title} {Data for
  ``{Verifiable Quantum Advantage via Optimized DQI Circuits}''}} (\bibinfo
  {year} {2026})\BibitemShut {NoStop}%
\bibitem [{\citenamefont {Aaronson}(2004)}]{aaronson2004multilinear}%
  \BibitemOpen
  \bibfield  {author} {\bibinfo {author} {\bibfnamefont {S.}~\bibnamefont
  {Aaronson}},\ }\bibfield  {title} {\bibinfo {title} {Multilinear formulas and
  skepticism of quantum computing},\ }in\ \href@noop {} {\emph {\bibinfo
  {booktitle} {Proceedings of the thirty-sixth annual ACM symposium on Theory
  of computing}}}\ (\bibinfo {year} {2004})\ pp.\ \bibinfo {pages}
  {118--127}\BibitemShut {NoStop}%
\bibitem [{\citenamefont
  {Gidney}(2025{\natexlab{a}})}]{gidney2025factor2048bitrsa}%
  \BibitemOpen
  \bibfield  {author} {\bibinfo {author} {\bibfnamefont {C.}~\bibnamefont
  {Gidney}},\ }\href {https://arxiv.org/abs/2505.15917} {\bibinfo {title} {How
  to factor 2048 bit {RSA} integers with less than a million noisy qubits}}
  (\bibinfo {year} {2025}{\natexlab{a}}),\ \Eprint
  {https://arxiv.org/abs/2505.15917} {arXiv:2505.15917 [quant-ph]} \BibitemShut
  {NoStop}%
\bibitem [{\citenamefont {Kahanamoku-Meyer}\ \emph {et~al.}(2025)\citenamefont
  {Kahanamoku-Meyer}, \citenamefont {Ragavan}, \citenamefont {Vaikuntanathan},\
  and\ \citenamefont {Van~Kirk}}]{kahanamoku2025jacobi}%
  \BibitemOpen
  \bibfield  {author} {\bibinfo {author} {\bibfnamefont {G.~D.}\ \bibnamefont
  {Kahanamoku-Meyer}}, \bibinfo {author} {\bibfnamefont {S.}~\bibnamefont
  {Ragavan}}, \bibinfo {author} {\bibfnamefont {V.}~\bibnamefont
  {Vaikuntanathan}},\ and\ \bibinfo {author} {\bibfnamefont {K.}~\bibnamefont
  {Van~Kirk}},\ }\bibfield  {title} {\bibinfo {title} {The {J}acobi factoring
  circuit: Quantum factoring with near-linear gates and sublinear space and
  depth},\ }in\ \href@noop {} {\emph {\bibinfo {booktitle} {Proceedings of the
  57th Annual ACM Symposium on Theory of Computing}}}\ (\bibinfo {year}
  {2025})\ pp.\ \bibinfo {pages} {1496--1507}\BibitemShut {NoStop}%
\bibitem [{\citenamefont {Babbush}\ \emph {et~al.}(2026)\citenamefont
  {Babbush}, \citenamefont {Zalcman}, \citenamefont {Gidney}, \citenamefont
  {Broughton}, \citenamefont {Khattar}, \citenamefont {Neven}, \citenamefont
  {Bergamaschi}, \citenamefont {Drake},\ and\ \citenamefont
  {Boneh}}]{babbush2026securingellipticcurvecryptocurrencies}%
  \BibitemOpen
  \bibfield  {author} {\bibinfo {author} {\bibfnamefont {R.}~\bibnamefont
  {Babbush}}, \bibinfo {author} {\bibfnamefont {A.}~\bibnamefont {Zalcman}},
  \bibinfo {author} {\bibfnamefont {C.}~\bibnamefont {Gidney}}, \bibinfo
  {author} {\bibfnamefont {M.}~\bibnamefont {Broughton}}, \bibinfo {author}
  {\bibfnamefont {T.}~\bibnamefont {Khattar}}, \bibinfo {author} {\bibfnamefont
  {H.}~\bibnamefont {Neven}}, \bibinfo {author} {\bibfnamefont
  {T.}~\bibnamefont {Bergamaschi}}, \bibinfo {author} {\bibfnamefont
  {J.}~\bibnamefont {Drake}},\ and\ \bibinfo {author} {\bibfnamefont
  {D.}~\bibnamefont {Boneh}},\ }\bibfield  {title} {\bibinfo {title} {Securing
  elliptic curve cryptocurrencies against quantum vulnerabilities: Resource
  estimates and mitigations},\ }\href {https://doi.org/10.1103/j3xf-bw18}
  {\bibfield  {journal} {\bibinfo  {journal} {PRX Quantum}\ }\textbf {\bibinfo
  {volume} {7}},\ \bibinfo {pages} {031001} (\bibinfo {year}
  {2026})}\BibitemShut {NoStop}%
\bibitem [{\citenamefont {Brakerski}\ \emph {et~al.}(2018)\citenamefont
  {Brakerski}, \citenamefont {Christiano}, \citenamefont {Mahadev},
  \citenamefont {Vazirani},\ and\ \citenamefont
  {Vidick}}]{brakerski2018certifiable}%
  \BibitemOpen
  \bibfield  {author} {\bibinfo {author} {\bibfnamefont {Z.}~\bibnamefont
  {Brakerski}}, \bibinfo {author} {\bibfnamefont {P.}~\bibnamefont
  {Christiano}}, \bibinfo {author} {\bibfnamefont {U.}~\bibnamefont {Mahadev}},
  \bibinfo {author} {\bibfnamefont {U.}~\bibnamefont {Vazirani}},\ and\
  \bibinfo {author} {\bibfnamefont {T.}~\bibnamefont {Vidick}},\ }\bibfield
  {title} {\bibinfo {title} {Certifiable randomness from a single quantum
  device},\ }\href@noop {} {\bibfield  {journal} {\bibinfo  {journal} {arXiv
  preprint arXiv:1804.00640}\ } (\bibinfo {year} {2018})}\BibitemShut {NoStop}%
\bibitem [{\citenamefont {Kahanamoku-Meyer}\ \emph {et~al.}(2022)\citenamefont
  {Kahanamoku-Meyer}, \citenamefont {Choi}, \citenamefont {Vazirani},\ and\
  \citenamefont {Yao}}]{kahanamoku2022classically}%
  \BibitemOpen
  \bibfield  {author} {\bibinfo {author} {\bibfnamefont {G.~D.}\ \bibnamefont
  {Kahanamoku-Meyer}}, \bibinfo {author} {\bibfnamefont {S.}~\bibnamefont
  {Choi}}, \bibinfo {author} {\bibfnamefont {U.~V.}\ \bibnamefont {Vazirani}},\
  and\ \bibinfo {author} {\bibfnamefont {N.~Y.}\ \bibnamefont {Yao}},\
  }\bibfield  {title} {\bibinfo {title} {Classically verifiable quantum
  advantage from a computational bell test},\ }\href@noop {} {\bibfield
  {journal} {\bibinfo  {journal} {Nature Physics}\ }\textbf {\bibinfo {volume}
  {18}},\ \bibinfo {pages} {918} (\bibinfo {year} {2022})}\BibitemShut
  {NoStop}%
\bibitem [{\citenamefont {Kahanamoku-Meyer}\ and\ \citenamefont
  {Yao}(2024)}]{kahanamokumeyer2024fastquantumintegermultiplication}%
  \BibitemOpen
  \bibfield  {author} {\bibinfo {author} {\bibfnamefont {G.~D.}\ \bibnamefont
  {Kahanamoku-Meyer}}\ and\ \bibinfo {author} {\bibfnamefont {N.~Y.}\
  \bibnamefont {Yao}},\ }\href {https://arxiv.org/abs/2403.18006} {\bibinfo
  {title} {Fast quantum integer multiplication with zero ancillas}} (\bibinfo
  {year} {2024}),\ \Eprint {https://arxiv.org/abs/2403.18006} {arXiv:2403.18006
  [quant-ph]} \BibitemShut {NoStop}%
\bibitem [{\citenamefont {Jones}(2013)}]{Jones2013}%
  \BibitemOpen
  \bibfield  {author} {\bibinfo {author} {\bibfnamefont {C.}~\bibnamefont
  {Jones}},\ }\bibfield  {title} {\bibinfo {title} {Low-overhead constructions
  for the fault-tolerant toffoli gate},\ }\bibfield  {journal} {\bibinfo
  {journal} {Physical Review A}\ }\textbf {\bibinfo {volume} {87}},\ \href
  {https://doi.org/10.1103/physreva.87.022328} {10.1103/physreva.87.022328}
  (\bibinfo {year} {2013})\BibitemShut {NoStop}%
\bibitem [{\citenamefont {Kliuchnikov}\ \emph {et~al.}(2023)\citenamefont
  {Kliuchnikov}, \citenamefont {Lauter}, \citenamefont {Minko}, \citenamefont
  {Paetznick},\ and\ \citenamefont {Petit}}]{kliuchnikov2023shorter}%
  \BibitemOpen
  \bibfield  {author} {\bibinfo {author} {\bibfnamefont {V.}~\bibnamefont
  {Kliuchnikov}}, \bibinfo {author} {\bibfnamefont {K.}~\bibnamefont {Lauter}},
  \bibinfo {author} {\bibfnamefont {R.}~\bibnamefont {Minko}}, \bibinfo
  {author} {\bibfnamefont {A.}~\bibnamefont {Paetznick}},\ and\ \bibinfo
  {author} {\bibfnamefont {C.}~\bibnamefont {Petit}},\ }\bibfield  {title}
  {\bibinfo {title} {Shorter quantum circuits via single-qubit gate
  approximation},\ }\href@noop {} {\bibfield  {journal} {\bibinfo  {journal}
  {Quantum}\ }\textbf {\bibinfo {volume} {7}},\ \bibinfo {pages} {1208}
  (\bibinfo {year} {2023})}\BibitemShut {NoStop}%
\bibitem [{\citenamefont
  {Litinski}(2023)}]{litinski2023compute256bitellipticcurve}%
  \BibitemOpen
  \bibfield  {author} {\bibinfo {author} {\bibfnamefont {D.}~\bibnamefont
  {Litinski}},\ }\href {https://arxiv.org/abs/2306.08585} {\bibinfo {title}
  {How to compute a 256-bit elliptic curve private key with only 50 million
  toffoli gates}} (\bibinfo {year} {2023}),\ \Eprint
  {https://arxiv.org/abs/2306.08585} {arXiv:2306.08585 [quant-ph]} \BibitemShut
  {NoStop}%
\bibitem [{\citenamefont {Bouland}\ \emph {et~al.}(2019)\citenamefont
  {Bouland}, \citenamefont {Fefferman}, \citenamefont {Nirkhe},\ and\
  \citenamefont {Vazirani}}]{bouland2019complexity}%
  \BibitemOpen
  \bibfield  {author} {\bibinfo {author} {\bibfnamefont {A.}~\bibnamefont
  {Bouland}}, \bibinfo {author} {\bibfnamefont {B.}~\bibnamefont {Fefferman}},
  \bibinfo {author} {\bibfnamefont {C.}~\bibnamefont {Nirkhe}},\ and\ \bibinfo
  {author} {\bibfnamefont {U.}~\bibnamefont {Vazirani}},\ }\bibfield  {title}
  {\bibinfo {title} {On the complexity and verification of quantum random
  circuit sampling},\ }\href {https://doi.org/10.1038/s41567-018-0318-2}
  {\bibfield  {journal} {\bibinfo  {journal} {Nature Physics}\ }\textbf
  {\bibinfo {volume} {15}},\ \bibinfo {pages} {159} (\bibinfo {year}
  {2019})}\BibitemShut {NoStop}%
\bibitem [{\citenamefont {Arute}\ \emph {et~al.}(2019)\citenamefont {Arute},
  \citenamefont {Arya}, \citenamefont {Babbush}, \citenamefont {Bacon},
  \citenamefont {Bardin}, \citenamefont {Barends}, \citenamefont {Biswas},
  \citenamefont {Boixo}, \citenamefont {Brandao}, \citenamefont {Buell} \emph
  {et~al.}}]{arute2019quantum}%
  \BibitemOpen
  \bibfield  {author} {\bibinfo {author} {\bibfnamefont {F.}~\bibnamefont
  {Arute}}, \bibinfo {author} {\bibfnamefont {K.}~\bibnamefont {Arya}},
  \bibinfo {author} {\bibfnamefont {R.}~\bibnamefont {Babbush}}, \bibinfo
  {author} {\bibfnamefont {D.}~\bibnamefont {Bacon}}, \bibinfo {author}
  {\bibfnamefont {J.~C.}\ \bibnamefont {Bardin}}, \bibinfo {author}
  {\bibfnamefont {R.}~\bibnamefont {Barends}}, \bibinfo {author} {\bibfnamefont
  {R.}~\bibnamefont {Biswas}}, \bibinfo {author} {\bibfnamefont
  {S.}~\bibnamefont {Boixo}}, \bibinfo {author} {\bibfnamefont {F.~G.}\
  \bibnamefont {Brandao}}, \bibinfo {author} {\bibfnamefont {D.~A.}\
  \bibnamefont {Buell}}, \emph {et~al.},\ }\bibfield  {title} {\bibinfo {title}
  {Quantum supremacy using a programmable superconducting processor},\
  }\href@noop {} {\bibfield  {journal} {\bibinfo  {journal} {Nature}\ }\textbf
  {\bibinfo {volume} {574}},\ \bibinfo {pages} {505} (\bibinfo {year}
  {2019})}\BibitemShut {NoStop}%
\bibitem [{\citenamefont {Zhu}\ \emph {et~al.}(2022)\citenamefont {Zhu},
  \citenamefont {Cao}, \citenamefont {Chen}, \citenamefont {Chen},
  \citenamefont {Chen}, \citenamefont {Chung}, \citenamefont {Deng},
  \citenamefont {Du}, \citenamefont {Fan}, \citenamefont {Gong} \emph
  {et~al.}}]{zhu2022quantum}%
  \BibitemOpen
  \bibfield  {author} {\bibinfo {author} {\bibfnamefont {Q.}~\bibnamefont
  {Zhu}}, \bibinfo {author} {\bibfnamefont {S.}~\bibnamefont {Cao}}, \bibinfo
  {author} {\bibfnamefont {F.}~\bibnamefont {Chen}}, \bibinfo {author}
  {\bibfnamefont {M.-C.}\ \bibnamefont {Chen}}, \bibinfo {author}
  {\bibfnamefont {X.}~\bibnamefont {Chen}}, \bibinfo {author} {\bibfnamefont
  {T.-H.}\ \bibnamefont {Chung}}, \bibinfo {author} {\bibfnamefont
  {H.}~\bibnamefont {Deng}}, \bibinfo {author} {\bibfnamefont {Y.}~\bibnamefont
  {Du}}, \bibinfo {author} {\bibfnamefont {D.}~\bibnamefont {Fan}}, \bibinfo
  {author} {\bibfnamefont {M.}~\bibnamefont {Gong}}, \emph {et~al.},\
  }\bibfield  {title} {\bibinfo {title} {Quantum computational advantage via
  60-qubit 24-cycle random circuit sampling},\ }\href@noop {} {\bibfield
  {journal} {\bibinfo  {journal} {Science bulletin}\ }\textbf {\bibinfo
  {volume} {67}},\ \bibinfo {pages} {240} (\bibinfo {year} {2022})}\BibitemShut
  {NoStop}%
\bibitem [{\citenamefont {Morvan}\ \emph {et~al.}(2024)\citenamefont {Morvan},
  \citenamefont {Villalonga}, \citenamefont {Mi}, \citenamefont {Mandr{\`a}},
  \citenamefont {Bengtsson}, \citenamefont {Klimov}, \citenamefont {Chen},
  \citenamefont {Hong}, \citenamefont {Erickson}, \citenamefont {Drozdov} \emph
  {et~al.}}]{morvan2024phase}%
  \BibitemOpen
  \bibfield  {author} {\bibinfo {author} {\bibfnamefont {A.}~\bibnamefont
  {Morvan}}, \bibinfo {author} {\bibfnamefont {B.}~\bibnamefont {Villalonga}},
  \bibinfo {author} {\bibfnamefont {X.}~\bibnamefont {Mi}}, \bibinfo {author}
  {\bibfnamefont {S.}~\bibnamefont {Mandr{\`a}}}, \bibinfo {author}
  {\bibfnamefont {A.}~\bibnamefont {Bengtsson}}, \bibinfo {author}
  {\bibfnamefont {P.}~\bibnamefont {Klimov}}, \bibinfo {author} {\bibfnamefont
  {Z.}~\bibnamefont {Chen}}, \bibinfo {author} {\bibfnamefont {S.}~\bibnamefont
  {Hong}}, \bibinfo {author} {\bibfnamefont {C.}~\bibnamefont {Erickson}},
  \bibinfo {author} {\bibfnamefont {I.}~\bibnamefont {Drozdov}}, \emph
  {et~al.},\ }\bibfield  {title} {\bibinfo {title} {Phase transitions in random
  circuit sampling},\ }\href@noop {} {\bibfield  {journal} {\bibinfo  {journal}
  {Nature}\ }\textbf {\bibinfo {volume} {634}},\ \bibinfo {pages} {328}
  (\bibinfo {year} {2024})}\BibitemShut {NoStop}%
\bibitem [{\citenamefont {Zhong}\ \emph {et~al.}(2020)\citenamefont {Zhong},
  \citenamefont {Wang}, \citenamefont {Deng}, \citenamefont {Chen},
  \citenamefont {Peng}, \citenamefont {Luo}, \citenamefont {Qin}, \citenamefont
  {Wu}, \citenamefont {Ding}, \citenamefont {Hu} \emph
  {et~al.}}]{zhong2020quantum}%
  \BibitemOpen
  \bibfield  {author} {\bibinfo {author} {\bibfnamefont {H.-S.}\ \bibnamefont
  {Zhong}}, \bibinfo {author} {\bibfnamefont {H.}~\bibnamefont {Wang}},
  \bibinfo {author} {\bibfnamefont {Y.-H.}\ \bibnamefont {Deng}}, \bibinfo
  {author} {\bibfnamefont {M.-C.}\ \bibnamefont {Chen}}, \bibinfo {author}
  {\bibfnamefont {L.-C.}\ \bibnamefont {Peng}}, \bibinfo {author}
  {\bibfnamefont {Y.-H.}\ \bibnamefont {Luo}}, \bibinfo {author} {\bibfnamefont
  {J.}~\bibnamefont {Qin}}, \bibinfo {author} {\bibfnamefont {D.}~\bibnamefont
  {Wu}}, \bibinfo {author} {\bibfnamefont {X.}~\bibnamefont {Ding}}, \bibinfo
  {author} {\bibfnamefont {Y.}~\bibnamefont {Hu}}, \emph {et~al.},\ }\bibfield
  {title} {\bibinfo {title} {Quantum computational advantage using photons},\
  }\href@noop {} {\bibfield  {journal} {\bibinfo  {journal} {Science}\ }\textbf
  {\bibinfo {volume} {370}},\ \bibinfo {pages} {1460} (\bibinfo {year}
  {2020})}\BibitemShut {NoStop}%
\bibitem [{\citenamefont {Zhong}\ \emph {et~al.}(2021)\citenamefont {Zhong},
  \citenamefont {Deng}, \citenamefont {Qin}, \citenamefont {Wang},
  \citenamefont {Chen}, \citenamefont {Peng}, \citenamefont {Luo},
  \citenamefont {Wu}, \citenamefont {Gong}, \citenamefont {Su} \emph
  {et~al.}}]{zhong2021phase}%
  \BibitemOpen
  \bibfield  {author} {\bibinfo {author} {\bibfnamefont {H.-S.}\ \bibnamefont
  {Zhong}}, \bibinfo {author} {\bibfnamefont {Y.-H.}\ \bibnamefont {Deng}},
  \bibinfo {author} {\bibfnamefont {J.}~\bibnamefont {Qin}}, \bibinfo {author}
  {\bibfnamefont {H.}~\bibnamefont {Wang}}, \bibinfo {author} {\bibfnamefont
  {M.-C.}\ \bibnamefont {Chen}}, \bibinfo {author} {\bibfnamefont {L.-C.}\
  \bibnamefont {Peng}}, \bibinfo {author} {\bibfnamefont {Y.-H.}\ \bibnamefont
  {Luo}}, \bibinfo {author} {\bibfnamefont {D.}~\bibnamefont {Wu}}, \bibinfo
  {author} {\bibfnamefont {S.-Q.}\ \bibnamefont {Gong}}, \bibinfo {author}
  {\bibfnamefont {H.}~\bibnamefont {Su}}, \emph {et~al.},\ }\bibfield  {title}
  {\bibinfo {title} {Phase-programmable gaussian boson sampling using
  stimulated squeezed light},\ }\href@noop {} {\bibfield  {journal} {\bibinfo
  {journal} {Physical review letters}\ }\textbf {\bibinfo {volume} {127}},\
  \bibinfo {pages} {180502} (\bibinfo {year} {2021})}\BibitemShut {NoStop}%
\bibitem [{\citenamefont {Madsen}\ \emph {et~al.}(2022)\citenamefont {Madsen},
  \citenamefont {Laudenbach}, \citenamefont {Askarani}, \citenamefont
  {Rortais}, \citenamefont {Vincent}, \citenamefont {Bulmer}, \citenamefont
  {Miatto}, \citenamefont {Neuhaus}, \citenamefont {Helt}, \citenamefont
  {Collins} \emph {et~al.}}]{madsen2022quantum}%
  \BibitemOpen
  \bibfield  {author} {\bibinfo {author} {\bibfnamefont {L.~S.}\ \bibnamefont
  {Madsen}}, \bibinfo {author} {\bibfnamefont {F.}~\bibnamefont {Laudenbach}},
  \bibinfo {author} {\bibfnamefont {M.~F.}\ \bibnamefont {Askarani}}, \bibinfo
  {author} {\bibfnamefont {F.}~\bibnamefont {Rortais}}, \bibinfo {author}
  {\bibfnamefont {T.}~\bibnamefont {Vincent}}, \bibinfo {author} {\bibfnamefont
  {J.~F.}\ \bibnamefont {Bulmer}}, \bibinfo {author} {\bibfnamefont {F.~M.}\
  \bibnamefont {Miatto}}, \bibinfo {author} {\bibfnamefont {L.}~\bibnamefont
  {Neuhaus}}, \bibinfo {author} {\bibfnamefont {L.~G.}\ \bibnamefont {Helt}},
  \bibinfo {author} {\bibfnamefont {M.~J.}\ \bibnamefont {Collins}}, \emph
  {et~al.},\ }\bibfield  {title} {\bibinfo {title} {Quantum computational
  advantage with a programmable photonic processor},\ }\href@noop {} {\bibfield
   {journal} {\bibinfo  {journal} {Nature}\ }\textbf {\bibinfo {volume}
  {606}},\ \bibinfo {pages} {75} (\bibinfo {year} {2022})}\BibitemShut
  {NoStop}%
\bibitem [{\citenamefont {Deng}\ \emph {et~al.}(2023)\citenamefont {Deng},
  \citenamefont {Gu}, \citenamefont {Liu}, \citenamefont {Gong}, \citenamefont
  {Su}, \citenamefont {Zhang}, \citenamefont {Tang}, \citenamefont {Jia},
  \citenamefont {Xu}, \citenamefont {Chen} \emph {et~al.}}]{deng2023gaussian}%
  \BibitemOpen
  \bibfield  {author} {\bibinfo {author} {\bibfnamefont {Y.-H.}\ \bibnamefont
  {Deng}}, \bibinfo {author} {\bibfnamefont {Y.-C.}\ \bibnamefont {Gu}},
  \bibinfo {author} {\bibfnamefont {H.-L.}\ \bibnamefont {Liu}}, \bibinfo
  {author} {\bibfnamefont {S.-Q.}\ \bibnamefont {Gong}}, \bibinfo {author}
  {\bibfnamefont {H.}~\bibnamefont {Su}}, \bibinfo {author} {\bibfnamefont
  {Z.-J.}\ \bibnamefont {Zhang}}, \bibinfo {author} {\bibfnamefont {H.-Y.}\
  \bibnamefont {Tang}}, \bibinfo {author} {\bibfnamefont {M.-H.}\ \bibnamefont
  {Jia}}, \bibinfo {author} {\bibfnamefont {J.-M.}\ \bibnamefont {Xu}},
  \bibinfo {author} {\bibfnamefont {M.-C.}\ \bibnamefont {Chen}}, \emph
  {et~al.},\ }\bibfield  {title} {\bibinfo {title} {Gaussian boson sampling
  with pseudo-photon-number-resolving detectors and quantum computational
  advantage},\ }\href@noop {} {\bibfield  {journal} {\bibinfo  {journal}
  {Physical review letters}\ }\textbf {\bibinfo {volume} {131}},\ \bibinfo
  {pages} {150601} (\bibinfo {year} {2023})}\BibitemShut {NoStop}%
\bibitem [{\citenamefont {{Google Quantum AI and
  Collaborators}}(2025{\natexlab{a}})}]{google2025constructive}%
  \BibitemOpen
  \bibfield  {author} {\bibinfo {author} {\bibnamefont {{Google Quantum AI and
  Collaborators}}},\ }\bibfield  {title} {\bibinfo {title} {Constructive
  interference at the edge of quantum ergodic dynamics},\ }\href@noop {}
  {\bibfield  {journal} {\bibinfo  {journal} {Nature}\ }\textbf {\bibinfo
  {volume} {646}},\ \bibinfo {pages} {825} (\bibinfo {year}
  {2025}{\natexlab{a}})}\BibitemShut {NoStop}%
\bibitem [{\citenamefont {King}\ \emph {et~al.}(2025)\citenamefont {King},
  \citenamefont {Kothari}, \citenamefont {Babbush}, \citenamefont {Boixo},
  \citenamefont {Kechedzhi}, \citenamefont {O'Brien},\ and\ \citenamefont
  {Smelyanskiy}}]{king2025simplified}%
  \BibitemOpen
  \bibfield  {author} {\bibinfo {author} {\bibfnamefont {R.}~\bibnamefont
  {King}}, \bibinfo {author} {\bibfnamefont {R.}~\bibnamefont {Kothari}},
  \bibinfo {author} {\bibfnamefont {R.}~\bibnamefont {Babbush}}, \bibinfo
  {author} {\bibfnamefont {S.}~\bibnamefont {Boixo}}, \bibinfo {author}
  {\bibfnamefont {K.}~\bibnamefont {Kechedzhi}}, \bibinfo {author}
  {\bibfnamefont {T.~E.}\ \bibnamefont {O'Brien}},\ and\ \bibinfo {author}
  {\bibfnamefont {V.}~\bibnamefont {Smelyanskiy}},\ }\bibfield  {title}
  {\bibinfo {title} {A simplified version of the quantum {$\textrm{OTOC}^2$}
  problem},\ }\href@noop {} {\bibfield  {journal} {\bibinfo  {journal}
  {arXiv:2510.19751}\ } (\bibinfo {year} {2025})}\BibitemShut {NoStop}%
\bibitem [{\citenamefont {Aaronson}\ and\ \citenamefont
  {Zhang}(2024)}]{aaronson2024verifiable}%
  \BibitemOpen
  \bibfield  {author} {\bibinfo {author} {\bibfnamefont {S.}~\bibnamefont
  {Aaronson}}\ and\ \bibinfo {author} {\bibfnamefont {Y.}~\bibnamefont
  {Zhang}},\ }\bibfield  {title} {\bibinfo {title} {On verifiable quantum
  advantage with peaked circuit sampling},\ }\href@noop {} {\bibfield
  {journal} {\bibinfo  {journal} {arXiv:2404.14493}\ } (\bibinfo {year}
  {2024})}\BibitemShut {NoStop}%
\bibitem [{\citenamefont {Jordan}\ \emph {et~al.}(2025)\citenamefont {Jordan},
  \citenamefont {Shutty}, \citenamefont {Wootters}, \citenamefont {Zalcman},
  \citenamefont {Schmidhuber}, \citenamefont {King}, \citenamefont {Isakov},
  \citenamefont {Khattar},\ and\ \citenamefont
  {Babbush}}]{jordan2024optimizationdecodedquantuminterferometry}%
  \BibitemOpen
  \bibfield  {author} {\bibinfo {author} {\bibfnamefont {S.~P.}\ \bibnamefont
  {Jordan}}, \bibinfo {author} {\bibfnamefont {N.}~\bibnamefont {Shutty}},
  \bibinfo {author} {\bibfnamefont {M.}~\bibnamefont {Wootters}}, \bibinfo
  {author} {\bibfnamefont {A.}~\bibnamefont {Zalcman}}, \bibinfo {author}
  {\bibfnamefont {A.}~\bibnamefont {Schmidhuber}}, \bibinfo {author}
  {\bibfnamefont {R.}~\bibnamefont {King}}, \bibinfo {author} {\bibfnamefont
  {S.~V.}\ \bibnamefont {Isakov}}, \bibinfo {author} {\bibfnamefont
  {T.}~\bibnamefont {Khattar}},\ and\ \bibinfo {author} {\bibfnamefont
  {R.}~\bibnamefont {Babbush}},\ }\bibfield  {title} {\bibinfo {title}
  {Optimization by decoded quantum interferometry},\ }\href@noop {} {\bibfield
  {journal} {\bibinfo  {journal} {Nature}\ }\textbf {\bibinfo {volume} {646}},\
  \bibinfo {pages} {831} (\bibinfo {year} {2025})}\BibitemShut {NoStop}%
\bibitem [{\citenamefont {Guruswami}\ and\ \citenamefont
  {Sudan}(1998)}]{guruswami1998improved}%
  \BibitemOpen
  \bibfield  {author} {\bibinfo {author} {\bibfnamefont {V.}~\bibnamefont
  {Guruswami}}\ and\ \bibinfo {author} {\bibfnamefont {M.}~\bibnamefont
  {Sudan}},\ }\bibfield  {title} {\bibinfo {title} {Improved decoding of
  {R}eed--{S}olomon and algebraic-geometric codes},\ }in\ \href@noop {} {\emph
  {\bibinfo {booktitle} {Proceedings 39th Annual Symposium on Foundations of
  Computer Science (Cat. No. 98CB36280)}}}\ (\bibinfo {organization} {IEEE},\
  \bibinfo {year} {1998})\ pp.\ \bibinfo {pages} {28--37}\BibitemShut {NoStop}%
\bibitem [{\citenamefont {Naor}\ and\ \citenamefont {Pinkas}(1999)}]{NP99}%
  \BibitemOpen
  \bibfield  {author} {\bibinfo {author} {\bibfnamefont {M.}~\bibnamefont
  {Naor}}\ and\ \bibinfo {author} {\bibfnamefont {B.}~\bibnamefont {Pinkas}},\
  }\bibfield  {title} {\bibinfo {title} {Oblivious transfer and polynomial
  evaluation},\ }in\ \href@noop {} {\emph {\bibinfo {booktitle} {Proceedings of
  the thirty-first annual ACM symposium on Theory of computing}}}\ (\bibinfo
  {year} {1999})\ pp.\ \bibinfo {pages} {245--254}\BibitemShut {NoStop}%
\bibitem [{\citenamefont {Bleichenbacher}\ and\ \citenamefont
  {Nguyen}(2000)}]{BN00}%
  \BibitemOpen
  \bibfield  {author} {\bibinfo {author} {\bibfnamefont {D.}~\bibnamefont
  {Bleichenbacher}}\ and\ \bibinfo {author} {\bibfnamefont {P.~Q.}\
  \bibnamefont {Nguyen}},\ }\bibfield  {title} {\bibinfo {title} {Noisy
  polynomial interpolation and noisy {C}hinese remaindering},\ }in\ \href@noop
  {} {\emph {\bibinfo {booktitle} {International Conference on the Theory and
  Applications of Cryptographic Techniques}}}\ (\bibinfo {organization}
  {Springer},\ \bibinfo {year} {2000})\ pp.\ \bibinfo {pages}
  {53--69}\BibitemShut {NoStop}%
\bibitem [{\citenamefont {Garcia-Morchon}\ \emph {et~al.}(2014)\citenamefont
  {Garcia-Morchon}, \citenamefont {Rietman}, \citenamefont {Shparlinski},\ and\
  \citenamefont {Tolhuizen}}]{GM14}%
  \BibitemOpen
  \bibfield  {author} {\bibinfo {author} {\bibfnamefont {O.}~\bibnamefont
  {Garcia-Morchon}}, \bibinfo {author} {\bibfnamefont {R.}~\bibnamefont
  {Rietman}}, \bibinfo {author} {\bibfnamefont {I.~E.}\ \bibnamefont
  {Shparlinski}},\ and\ \bibinfo {author} {\bibfnamefont {L.}~\bibnamefont
  {Tolhuizen}},\ }\bibfield  {title} {\bibinfo {title} {Interpolation and
  approximation of polynomials in finite fields over a short interval from
  noisy values},\ }\href@noop {} {\bibfield  {journal} {\bibinfo  {journal}
  {Experimental Mathematics}\ }\textbf {\bibinfo {volume} {23}},\ \bibinfo
  {pages} {241} (\bibinfo {year} {2014})}\BibitemShut {NoStop}%
\bibitem [{\citenamefont {Shparlinski}\ and\ \citenamefont
  {Winterhof}(2005)}]{SW05}%
  \BibitemOpen
  \bibfield  {author} {\bibinfo {author} {\bibfnamefont {I.}~\bibnamefont
  {Shparlinski}}\ and\ \bibinfo {author} {\bibfnamefont {A.}~\bibnamefont
  {Winterhof}},\ }\bibfield  {title} {\bibinfo {title} {Noisy interpolation of
  sparse polynomials in finite fields},\ }\href@noop {} {\bibfield  {journal}
  {\bibinfo  {journal} {Applicable Algebra in Engineering, Communication and
  Computing}\ }\textbf {\bibinfo {volume} {16}},\ \bibinfo {pages} {307}
  (\bibinfo {year} {2005})}\BibitemShut {NoStop}%
\bibitem [{\citenamefont {Shparlinski}(2005)}]{S05}%
  \BibitemOpen
  \bibfield  {author} {\bibinfo {author} {\bibfnamefont {I.~E.}\ \bibnamefont
  {Shparlinski}},\ }\bibfield  {title} {\bibinfo {title} {Playing
  ``hide-and-seek'' with numbers},\ }\href@noop {} {\bibfield  {journal}
  {\bibinfo  {journal} {Public-Key Cryptography: American Mathematical Society
  Short Course, January 13-14, 2003, Baltimore, Maryland}\ }\textbf {\bibinfo
  {volume} {62}},\ \bibinfo {pages} {153} (\bibinfo {year} {2005})}\BibitemShut
  {NoStop}%
\bibitem [{\citenamefont {{Google Quantum AI and
  Collaborators}}(2025{\natexlab{b}})}]{google2025quantum}%
  \BibitemOpen
  \bibfield  {author} {\bibinfo {author} {\bibnamefont {{Google Quantum AI and
  Collaborators}}},\ }\bibfield  {title} {\bibinfo {title} {Quantum error
  correction below the surface code threshold},\ }\href@noop {} {\bibfield
  {journal} {\bibinfo  {journal} {Nature}\ }\textbf {\bibinfo {volume} {638}},\
  \bibinfo {pages} {920} (\bibinfo {year} {2025}{\natexlab{b}})}\BibitemShut
  {NoStop}%
\bibitem [{\citenamefont {Gidney}\ and\ \citenamefont
  {Ekerå}(2021)}]{Gidney2021}%
  \BibitemOpen
  \bibfield  {author} {\bibinfo {author} {\bibfnamefont {C.}~\bibnamefont
  {Gidney}}\ and\ \bibinfo {author} {\bibfnamefont {M.}~\bibnamefont
  {Ekerå}},\ }\bibfield  {title} {\bibinfo {title} {How to factor 2048 bit rsa
  integers in 8 hours using 20 million noisy qubits},\ }\href
  {https://doi.org/10.22331/q-2021-04-15-433} {\bibfield  {journal} {\bibinfo
  {journal} {Quantum}\ }\textbf {\bibinfo {volume} {5}},\ \bibinfo {pages}
  {433} (\bibinfo {year} {2021})}\BibitemShut {NoStop}%
\bibitem [{\citenamefont {Low}\ \emph {et~al.}(2025)\citenamefont {Low},
  \citenamefont {King}, \citenamefont {Berry}, \citenamefont {Han},
  \citenamefont {DePrince~III}, \citenamefont {White}, \citenamefont {Babbush},
  \citenamefont {Somma},\ and\ \citenamefont {Rubin}}]{low2025fast}%
  \BibitemOpen
  \bibfield  {author} {\bibinfo {author} {\bibfnamefont {G.~H.}\ \bibnamefont
  {Low}}, \bibinfo {author} {\bibfnamefont {R.}~\bibnamefont {King}}, \bibinfo
  {author} {\bibfnamefont {D.~W.}\ \bibnamefont {Berry}}, \bibinfo {author}
  {\bibfnamefont {Q.}~\bibnamefont {Han}}, \bibinfo {author} {\bibfnamefont
  {A.~E.}\ \bibnamefont {DePrince~III}}, \bibinfo {author} {\bibfnamefont
  {A.~F.}\ \bibnamefont {White}}, \bibinfo {author} {\bibfnamefont
  {R.}~\bibnamefont {Babbush}}, \bibinfo {author} {\bibfnamefont {R.~D.}\
  \bibnamefont {Somma}},\ and\ \bibinfo {author} {\bibfnamefont {N.~C.}\
  \bibnamefont {Rubin}},\ }\bibfield  {title} {\bibinfo {title} {Fast quantum
  simulation of electronic structure by spectral amplification},\ }\href@noop
  {} {\bibfield  {journal} {\bibinfo  {journal} {Physical Review X}\ }\textbf
  {\bibinfo {volume} {15}},\ \bibinfo {pages} {041016} (\bibinfo {year}
  {2025})}\BibitemShut {NoStop}%
\bibitem [{\citenamefont {Cain}\ \emph {et~al.}(2026)\citenamefont {Cain},
  \citenamefont {Xu}, \citenamefont {King}, \citenamefont {Picard},
  \citenamefont {Levine}, \citenamefont {Endres}, \citenamefont {Preskill},
  \citenamefont {Huang},\ and\ \citenamefont
  {Bluvstein}}]{cain2026shorsalgorithmpossible10000}%
  \BibitemOpen
  \bibfield  {author} {\bibinfo {author} {\bibfnamefont {M.}~\bibnamefont
  {Cain}}, \bibinfo {author} {\bibfnamefont {Q.}~\bibnamefont {Xu}}, \bibinfo
  {author} {\bibfnamefont {R.}~\bibnamefont {King}}, \bibinfo {author}
  {\bibfnamefont {L.~R.~B.}\ \bibnamefont {Picard}}, \bibinfo {author}
  {\bibfnamefont {H.}~\bibnamefont {Levine}}, \bibinfo {author} {\bibfnamefont
  {M.}~\bibnamefont {Endres}}, \bibinfo {author} {\bibfnamefont
  {J.}~\bibnamefont {Preskill}}, \bibinfo {author} {\bibfnamefont {H.-Y.}\
  \bibnamefont {Huang}},\ and\ \bibinfo {author} {\bibfnamefont
  {D.}~\bibnamefont {Bluvstein}},\ }\href {https://arxiv.org/abs/2603.28627}
  {\bibinfo {title} {Shor's algorithm is possible with as few as 10,000
  reconfigurable atomic qubits}} (\bibinfo {year} {2026}),\ \Eprint
  {https://arxiv.org/abs/2603.28627} {arXiv:2603.28627 [quant-ph]} \BibitemShut
  {NoStop}%
\bibitem [{\citenamefont {Webster}\ \emph {et~al.}(2026)\citenamefont
  {Webster}, \citenamefont {Berent}, \citenamefont {Chandra}, \citenamefont
  {Hockings}, \citenamefont {Baspin}, \citenamefont {Thomsen}, \citenamefont
  {Smith},\ and\ \citenamefont
  {Cohen}}]{webster2026pinnaclearchitecturereducingcost}%
  \BibitemOpen
  \bibfield  {author} {\bibinfo {author} {\bibfnamefont {P.}~\bibnamefont
  {Webster}}, \bibinfo {author} {\bibfnamefont {L.}~\bibnamefont {Berent}},
  \bibinfo {author} {\bibfnamefont {O.}~\bibnamefont {Chandra}}, \bibinfo
  {author} {\bibfnamefont {E.~T.}\ \bibnamefont {Hockings}}, \bibinfo {author}
  {\bibfnamefont {N.}~\bibnamefont {Baspin}}, \bibinfo {author} {\bibfnamefont
  {F.}~\bibnamefont {Thomsen}}, \bibinfo {author} {\bibfnamefont {S.~C.}\
  \bibnamefont {Smith}},\ and\ \bibinfo {author} {\bibfnamefont {L.~Z.}\
  \bibnamefont {Cohen}},\ }\href {https://arxiv.org/abs/2602.11457} {\bibinfo
  {title} {The pinnacle architecture: Reducing the cost of breaking rsa-2048 to
  100 000 physical qubits using quantum ldpc codes}} (\bibinfo {year} {2026}),\
  \Eprint {https://arxiv.org/abs/2602.11457} {arXiv:2602.11457 [quant-ph]}
  \BibitemShut {NoStop}%
\bibitem [{\citenamefont {Tripier}\ \emph {et~al.}(2026)\citenamefont
  {Tripier}, \citenamefont {Chung}, \citenamefont {Young}, \citenamefont
  {Alam}, \citenamefont {Bjork}, \citenamefont {Brodutch}, \citenamefont
  {Buessen}, \citenamefont {Coble}, \citenamefont {Dellaert}, \citenamefont
  {Maslov}, \citenamefont {Roetteler}, \citenamefont {Tham}, \citenamefont
  {Webster}, \citenamefont {Ye}, \citenamefont {Gamble}, \citenamefont
  {Maksymov}, \citenamefont {Marceaux},\ and\ \citenamefont
  {Delfosse}}]{tripier2026faulttolerantquantumcomputingtrapped}%
  \BibitemOpen
  \bibfield  {author} {\bibinfo {author} {\bibfnamefont {F.}~\bibnamefont
  {Tripier}}, \bibinfo {author} {\bibfnamefont {W.~C.}\ \bibnamefont {Chung}},
  \bibinfo {author} {\bibfnamefont {J.}~\bibnamefont {Young}}, \bibinfo
  {author} {\bibfnamefont {S.}~\bibnamefont {Alam}}, \bibinfo {author}
  {\bibfnamefont {B.}~\bibnamefont {Bjork}}, \bibinfo {author} {\bibfnamefont
  {A.}~\bibnamefont {Brodutch}}, \bibinfo {author} {\bibfnamefont {F.~L.}\
  \bibnamefont {Buessen}}, \bibinfo {author} {\bibfnamefont {N.~J.}\
  \bibnamefont {Coble}}, \bibinfo {author} {\bibfnamefont {T.}~\bibnamefont
  {Dellaert}}, \bibinfo {author} {\bibfnamefont {D.}~\bibnamefont {Maslov}},
  \bibinfo {author} {\bibfnamefont {M.}~\bibnamefont {Roetteler}}, \bibinfo
  {author} {\bibfnamefont {E.}~\bibnamefont {Tham}}, \bibinfo {author}
  {\bibfnamefont {M.}~\bibnamefont {Webster}}, \bibinfo {author} {\bibfnamefont
  {M.}~\bibnamefont {Ye}}, \bibinfo {author} {\bibfnamefont {J.}~\bibnamefont
  {Gamble}}, \bibinfo {author} {\bibfnamefont {A.}~\bibnamefont {Maksymov}},
  \bibinfo {author} {\bibfnamefont {J.~P.}\ \bibnamefont {Marceaux}},\ and\
  \bibinfo {author} {\bibfnamefont {N.}~\bibnamefont {Delfosse}},\ }\href
  {https://arxiv.org/abs/2604.19481} {\bibinfo {title} {Fault-tolerant quantum
  computing with trapped ions: The walking cat architecture}} (\bibinfo {year}
  {2026}),\ \Eprint {https://arxiv.org/abs/2604.19481} {arXiv:2604.19481
  [quant-ph]} \BibitemShut {NoStop}%
\bibitem [{\citenamefont {Barker}\ \emph {et~al.}(2009)\citenamefont {Barker},
  \citenamefont {Chen}, \citenamefont {Regenscheid},\ and\ \citenamefont
  {Smid}}]{barker2009sp}%
  \BibitemOpen
  \bibfield  {author} {\bibinfo {author} {\bibfnamefont {E.~B.}\ \bibnamefont
  {Barker}}, \bibinfo {author} {\bibfnamefont {L.}~\bibnamefont {Chen}},
  \bibinfo {author} {\bibfnamefont {A.~R.}\ \bibnamefont {Regenscheid}},\ and\
  \bibinfo {author} {\bibfnamefont {M.~E.}\ \bibnamefont {Smid}},\ }\href@noop
  {} {\bibinfo {title} {Sp 800-56b. recommendation for pair-wise key
  establishment schemes using integer factorization cryptography}} (\bibinfo
  {year} {2009})\BibitemShut {NoStop}%
\bibitem [{\citenamefont {Chen}\ \emph {et~al.}(2023)\citenamefont {Chen},
  \citenamefont {Moody}, \citenamefont {Randall}, \citenamefont {Regenscheid},\
  and\ \citenamefont {Robinson}}]{chen2023recommendations}%
  \BibitemOpen
  \bibfield  {author} {\bibinfo {author} {\bibfnamefont {L.}~\bibnamefont
  {Chen}}, \bibinfo {author} {\bibfnamefont {D.}~\bibnamefont {Moody}},
  \bibinfo {author} {\bibfnamefont {K.}~\bibnamefont {Randall}}, \bibinfo
  {author} {\bibfnamefont {A.}~\bibnamefont {Regenscheid}},\ and\ \bibinfo
  {author} {\bibfnamefont {A.}~\bibnamefont {Robinson}},\ }\href
  {https://doi.org/10.6028/NIST.SP.800-186} {\emph {\bibinfo {title}
  {Recommendations for Discrete Logarithm-Based Cryptography: Elliptic Curve
  Domain Parameters}}},\ \bibinfo {type} {Tech. Rep.}\ \bibinfo {number} {NIST
  SP 800-186}\ (\bibinfo  {institution} {National Institute of Standards and
  Technology},\ \bibinfo {year} {2023})\BibitemShut {NoStop}%
\bibitem [{\citenamefont {Bravyi}\ and\ \citenamefont
  {Gosset}(2016)}]{Bravyi_2016}%
  \BibitemOpen
  \bibfield  {author} {\bibinfo {author} {\bibfnamefont {S.}~\bibnamefont
  {Bravyi}}\ and\ \bibinfo {author} {\bibfnamefont {D.}~\bibnamefont
  {Gosset}},\ }\bibfield  {title} {\bibinfo {title} {Improved classical
  simulation of quantum circuits dominated by {C}lifford gates},\ }\bibfield
  {journal} {\bibinfo  {journal} {Physical Review Letters}\ }\textbf {\bibinfo
  {volume} {116}},\ \href {https://doi.org/10.1103/physrevlett.116.250501}
  {10.1103/physrevlett.116.250501} (\bibinfo {year} {2016})\BibitemShut
  {NoStop}%
\bibitem [{\citenamefont {Shor}(1999)}]{Shor1999}%
  \BibitemOpen
  \bibfield  {author} {\bibinfo {author} {\bibfnamefont {P.~W.}\ \bibnamefont
  {Shor}},\ }\bibfield  {title} {\bibinfo {title} {Polynomial-time algorithms
  for prime factorization and discrete logarithms on a quantum computer},\
  }\href {https://doi.org/10.1137/s0036144598347011} {\bibfield  {journal}
  {\bibinfo  {journal} {SIAM Review}\ }\textbf {\bibinfo {volume} {41}},\
  \bibinfo {pages} {303–332} (\bibinfo {year} {1999})}\BibitemShut {NoStop}%
\bibitem [{\citenamefont {Regev}(2025)}]{regev2025efficient}%
  \BibitemOpen
  \bibfield  {author} {\bibinfo {author} {\bibfnamefont {O.}~\bibnamefont
  {Regev}},\ }\bibfield  {title} {\bibinfo {title} {An efficient quantum
  factoring algorithm},\ }\href@noop {} {\bibfield  {journal} {\bibinfo
  {journal} {Journal of the ACM}\ }\textbf {\bibinfo {volume} {72}},\ \bibinfo
  {pages} {1} (\bibinfo {year} {2025})}\BibitemShut {NoStop}%
\bibitem [{\citenamefont {Berlekamp}(2015)}]{berlekamp2015algebraic}%
  \BibitemOpen
  \bibfield  {author} {\bibinfo {author} {\bibfnamefont {E.~R.}\ \bibnamefont
  {Berlekamp}},\ }\href@noop {} {\emph {\bibinfo {title} {Algebraic coding
  theory (revised edition)}}}\ (\bibinfo  {publisher} {World Scientific},\
  \bibinfo {year} {2015})\BibitemShut {NoStop}%
\bibitem [{\citenamefont {Sugiyama}\ \emph {et~al.}(1975)\citenamefont
  {Sugiyama}, \citenamefont {Kasahara}, \citenamefont {Hirasawa},\ and\
  \citenamefont {Namekawa}}]{SUGIYAMA197587}%
  \BibitemOpen
  \bibfield  {author} {\bibinfo {author} {\bibfnamefont {Y.}~\bibnamefont
  {Sugiyama}}, \bibinfo {author} {\bibfnamefont {M.}~\bibnamefont {Kasahara}},
  \bibinfo {author} {\bibfnamefont {S.}~\bibnamefont {Hirasawa}},\ and\
  \bibinfo {author} {\bibfnamefont {T.}~\bibnamefont {Namekawa}},\ }\bibfield
  {title} {\bibinfo {title} {A method for solving key equation for decoding
  {G}oppa codes},\ }\href {https://doi.org/10.1016/S0019-9958(75)90090-X}
  {\bibfield  {journal} {\bibinfo  {journal} {Information and Control}\
  }\textbf {\bibinfo {volume} {27}},\ \bibinfo {pages} {87} (\bibinfo {year}
  {1975})}\BibitemShut {NoStop}%
\bibitem [{\citenamefont {Prange}(1962)}]{prange1962use}%
  \BibitemOpen
  \bibfield  {author} {\bibinfo {author} {\bibfnamefont {E.}~\bibnamefont
  {Prange}},\ }\bibfield  {title} {\bibinfo {title} {The use of information
  sets in decoding cyclic codes},\ }\href@noop {} {\bibfield  {journal}
  {\bibinfo  {journal} {IRE Transactions on Information Theory}\ }\textbf
  {\bibinfo {volume} {8}},\ \bibinfo {pages} {5} (\bibinfo {year}
  {1962})}\BibitemShut {NoStop}%
\bibitem [{\citenamefont {Gidney}(2018)}]{Gidney2018}%
  \BibitemOpen
  \bibfield  {author} {\bibinfo {author} {\bibfnamefont {C.}~\bibnamefont
  {Gidney}},\ }\bibfield  {title} {\bibinfo {title} {Halving the cost of
  quantum addition},\ }\href {https://doi.org/10.22331/q-2018-06-18-74}
  {\bibfield  {journal} {\bibinfo  {journal} {Quantum}\ }\textbf {\bibinfo
  {volume} {2}},\ \bibinfo {pages} {74} (\bibinfo {year} {2018})}\BibitemShut
  {NoStop}%
\bibitem [{\citenamefont {Kaye}\ and\ \citenamefont
  {Zalka}(2004)}]{kaye2004optimizedquantumimplementationelliptic}%
  \BibitemOpen
  \bibfield  {author} {\bibinfo {author} {\bibfnamefont {P.}~\bibnamefont
  {Kaye}}\ and\ \bibinfo {author} {\bibfnamefont {C.}~\bibnamefont {Zalka}},\
  }\href {https://arxiv.org/abs/quant-ph/0407095} {\bibinfo {title} {Optimized
  quantum implementation of elliptic curve arithmetic over binary fields}}
  (\bibinfo {year} {2004}),\ \Eprint {https://arxiv.org/abs/quant-ph/0407095}
  {arXiv:quant-ph/0407095 [quant-ph]} \BibitemShut {NoStop}%
\bibitem [{\citenamefont {Proos}\ and\ \citenamefont
  {Zalka}(2004)}]{proos2004shorsdiscretelogarithmquantum}%
  \BibitemOpen
  \bibfield  {author} {\bibinfo {author} {\bibfnamefont {J.}~\bibnamefont
  {Proos}}\ and\ \bibinfo {author} {\bibfnamefont {C.}~\bibnamefont {Zalka}},\
  }\href {https://arxiv.org/abs/quant-ph/0301141} {\bibinfo {title} {Shor's
  discrete logarithm quantum algorithm for elliptic curves}} (\bibinfo {year}
  {2004}),\ \Eprint {https://arxiv.org/abs/quant-ph/0301141}
  {arXiv:quant-ph/0301141 [quant-ph]} \BibitemShut {NoStop}%
\bibitem [{\citenamefont {Stein}(1967)}]{Stein1967}%
  \BibitemOpen
  \bibfield  {author} {\bibinfo {author} {\bibfnamefont {J.}~\bibnamefont
  {Stein}},\ }\bibfield  {title} {\bibinfo {title} {Computational problems
  associated with racah algebra},\ }\href
  {https://doi.org/10.1016/0021-9991(67)90047-2} {\bibfield  {journal}
  {\bibinfo  {journal} {Journal of Computational Physics}\ }\textbf {\bibinfo
  {volume} {1}},\ \bibinfo {pages} {397–405} (\bibinfo {year}
  {1967})}\BibitemShut {NoStop}%
\bibitem [{\citenamefont {Bernstein}\ and\ \citenamefont
  {Yang}(2019)}]{cryptoeprint:2019/266}%
  \BibitemOpen
  \bibfield  {author} {\bibinfo {author} {\bibfnamefont {D.~J.}\ \bibnamefont
  {Bernstein}}\ and\ \bibinfo {author} {\bibfnamefont {B.-Y.}\ \bibnamefont
  {Yang}},\ }\href {https://eprint.iacr.org/2019/266} {\bibinfo {title} {Fast
  constant-time {GCD} computation and modular inversion}},\ \bibinfo
  {howpublished} {Cryptology {ePrint} Archive, Paper 2019/266} (\bibinfo {year}
  {2019})\BibitemShut {NoStop}%
\bibitem [{\citenamefont {Chien}(1964)}]{Chien1964}%
  \BibitemOpen
  \bibfield  {author} {\bibinfo {author} {\bibfnamefont {R.}~\bibnamefont
  {Chien}},\ }\bibfield  {title} {\bibinfo {title} {Cyclic decoding procedures
  for bose- chaudhuri-hocquenghem codes},\ }\href
  {https://doi.org/10.1109/tit.1964.1053699} {\bibfield  {journal} {\bibinfo
  {journal} {IEEE Transactions on Information Theory}\ }\textbf {\bibinfo
  {volume} {10}},\ \bibinfo {pages} {357–363} (\bibinfo {year}
  {1964})}\BibitemShut {NoStop}%
\bibitem [{\citenamefont {Forney}(1965)}]{Forney1965}%
  \BibitemOpen
  \bibfield  {author} {\bibinfo {author} {\bibfnamefont {G.}~\bibnamefont
  {Forney}},\ }\bibfield  {title} {\bibinfo {title} {On decoding bch codes},\
  }\href {https://doi.org/10.1109/tit.1965.1053825} {\bibfield  {journal}
  {\bibinfo  {journal} {IEEE Transactions on Information Theory}\ }\textbf
  {\bibinfo {volume} {11}},\ \bibinfo {pages} {549–557} (\bibinfo {year}
  {1965})}\BibitemShut {NoStop}%
\bibitem [{\citenamefont {Harrigan}\ \emph {et~al.}(2024)\citenamefont
  {Harrigan}, \citenamefont {Khattar}, \citenamefont {Yuan}, \citenamefont
  {Peduri}, \citenamefont {Yosri}, \citenamefont {Malone}, \citenamefont
  {Babbush},\ and\ \citenamefont
  {Rubin}}]{harrigan2024expressinganalyzingquantumalgorithms}%
  \BibitemOpen
  \bibfield  {author} {\bibinfo {author} {\bibfnamefont {M.~P.}\ \bibnamefont
  {Harrigan}}, \bibinfo {author} {\bibfnamefont {T.}~\bibnamefont {Khattar}},
  \bibinfo {author} {\bibfnamefont {C.}~\bibnamefont {Yuan}}, \bibinfo {author}
  {\bibfnamefont {A.}~\bibnamefont {Peduri}}, \bibinfo {author} {\bibfnamefont
  {N.}~\bibnamefont {Yosri}}, \bibinfo {author} {\bibfnamefont {F.~D.}\
  \bibnamefont {Malone}}, \bibinfo {author} {\bibfnamefont {R.}~\bibnamefont
  {Babbush}},\ and\ \bibinfo {author} {\bibfnamefont {N.~C.}\ \bibnamefont
  {Rubin}},\ }\href {https://arxiv.org/abs/2409.04643} {\bibinfo {title}
  {Expressing and analyzing quantum algorithms with qualtran}} (\bibinfo {year}
  {2024}),\ \Eprint {https://arxiv.org/abs/2409.04643} {arXiv:2409.04643
  [quant-ph]} \BibitemShut {NoStop}%
\bibitem [{\citenamefont {Gidney}\ \emph {et~al.}(2025)\citenamefont {Gidney},
  \citenamefont {Newman}, \citenamefont {Brooks},\ and\ \citenamefont
  {Jones}}]{gidney2025yoked}%
  \BibitemOpen
  \bibfield  {author} {\bibinfo {author} {\bibfnamefont {C.}~\bibnamefont
  {Gidney}}, \bibinfo {author} {\bibfnamefont {M.}~\bibnamefont {Newman}},
  \bibinfo {author} {\bibfnamefont {P.}~\bibnamefont {Brooks}},\ and\ \bibinfo
  {author} {\bibfnamefont {C.}~\bibnamefont {Jones}},\ }\bibfield  {title}
  {\bibinfo {title} {Yoked surface codes},\ }\href@noop {} {\bibfield
  {journal} {\bibinfo  {journal} {Nature Communications}\ }\textbf {\bibinfo
  {volume} {16}},\ \bibinfo {pages} {4498} (\bibinfo {year}
  {2025})}\BibitemShut {NoStop}%
\bibitem [{\citenamefont {Häner}\ \emph {et~al.}(2020)\citenamefont {Häner},
  \citenamefont {Jaques}, \citenamefont {Naehrig}, \citenamefont {Roetteler},\
  and\ \citenamefont {Soeken}}]{häner2020improvedquantumcircuitselliptic}%
  \BibitemOpen
  \bibfield  {author} {\bibinfo {author} {\bibfnamefont {T.}~\bibnamefont
  {Häner}}, \bibinfo {author} {\bibfnamefont {S.}~\bibnamefont {Jaques}},
  \bibinfo {author} {\bibfnamefont {M.}~\bibnamefont {Naehrig}}, \bibinfo
  {author} {\bibfnamefont {M.}~\bibnamefont {Roetteler}},\ and\ \bibinfo
  {author} {\bibfnamefont {M.}~\bibnamefont {Soeken}},\ }\href
  {https://arxiv.org/abs/2001.09580} {\bibinfo {title} {Improved quantum
  circuits for elliptic curve discrete logarithms}} (\bibinfo {year} {2020}),\
  \Eprint {https://arxiv.org/abs/2001.09580} {arXiv:2001.09580 [quant-ph]}
  \BibitemShut {NoStop}%
\bibitem [{\citenamefont {Gosset}\ \emph {et~al.}(2024)\citenamefont {Gosset},
  \citenamefont {Kothari},\ and\ \citenamefont
  {Wu}}]{gosset2024quantumstatepreparationoptimal}%
  \BibitemOpen
  \bibfield  {author} {\bibinfo {author} {\bibfnamefont {D.}~\bibnamefont
  {Gosset}}, \bibinfo {author} {\bibfnamefont {R.}~\bibnamefont {Kothari}},\
  and\ \bibinfo {author} {\bibfnamefont {K.}~\bibnamefont {Wu}},\ }\href
  {https://arxiv.org/abs/2411.04790} {\bibinfo {title} {Quantum state
  preparation with optimal {T}-count}} (\bibinfo {year} {2024}),\ \Eprint
  {https://arxiv.org/abs/2411.04790} {arXiv:2411.04790 [quant-ph]} \BibitemShut
  {NoStop}%
\bibitem [{\citenamefont {Low}\ \emph {et~al.}(2024)\citenamefont {Low},
  \citenamefont {Kliuchnikov},\ and\ \citenamefont {Schaeffer}}]{Low2024}%
  \BibitemOpen
  \bibfield  {author} {\bibinfo {author} {\bibfnamefont {G.~H.}\ \bibnamefont
  {Low}}, \bibinfo {author} {\bibfnamefont {V.}~\bibnamefont {Kliuchnikov}},\
  and\ \bibinfo {author} {\bibfnamefont {L.}~\bibnamefont {Schaeffer}},\
  }\bibfield  {title} {\bibinfo {title} {Trading t gates for dirty qubits in
  state preparation and unitary synthesis},\ }\href
  {https://doi.org/10.22331/q-2024-06-17-1375} {\bibfield  {journal} {\bibinfo
  {journal} {Quantum}\ }\textbf {\bibinfo {volume} {8}},\ \bibinfo {pages}
  {1375} (\bibinfo {year} {2024})}\BibitemShut {NoStop}%
\bibitem [{\citenamefont {Berry}\ \emph {et~al.}(2025)\citenamefont {Berry},
  \citenamefont {Tong}, \citenamefont {Khattar}, \citenamefont {White},
  \citenamefont {Kim}, \citenamefont {Low}, \citenamefont {Boixo},
  \citenamefont {Ding}, \citenamefont {Lin}, \citenamefont {Lee} \emph
  {et~al.}}]{berry2025rapid}%
  \BibitemOpen
  \bibfield  {author} {\bibinfo {author} {\bibfnamefont {D.~W.}\ \bibnamefont
  {Berry}}, \bibinfo {author} {\bibfnamefont {Y.}~\bibnamefont {Tong}},
  \bibinfo {author} {\bibfnamefont {T.}~\bibnamefont {Khattar}}, \bibinfo
  {author} {\bibfnamefont {A.}~\bibnamefont {White}}, \bibinfo {author}
  {\bibfnamefont {T.~I.}\ \bibnamefont {Kim}}, \bibinfo {author} {\bibfnamefont
  {G.~H.}\ \bibnamefont {Low}}, \bibinfo {author} {\bibfnamefont
  {S.}~\bibnamefont {Boixo}}, \bibinfo {author} {\bibfnamefont
  {Z.}~\bibnamefont {Ding}}, \bibinfo {author} {\bibfnamefont {L.}~\bibnamefont
  {Lin}}, \bibinfo {author} {\bibfnamefont {S.}~\bibnamefont {Lee}}, \emph
  {et~al.},\ }\bibfield  {title} {\bibinfo {title} {Rapid initial-state
  preparation for the quantum simulation of strongly correlated molecules},\
  }\href@noop {} {\bibfield  {journal} {\bibinfo  {journal} {PRX Quantum}\
  }\textbf {\bibinfo {volume} {6}},\ \bibinfo {pages} {020327} (\bibinfo {year}
  {2025})}\BibitemShut {NoStop}%
\bibitem [{\citenamefont {Babbush}\ \emph {et~al.}(2018)\citenamefont
  {Babbush}, \citenamefont {Gidney}, \citenamefont {Berry}, \citenamefont
  {Wiebe}, \citenamefont {McClean}, \citenamefont {Paler}, \citenamefont
  {Fowler},\ and\ \citenamefont {Neven}}]{Babbush_2018}%
  \BibitemOpen
  \bibfield  {author} {\bibinfo {author} {\bibfnamefont {R.}~\bibnamefont
  {Babbush}}, \bibinfo {author} {\bibfnamefont {C.}~\bibnamefont {Gidney}},
  \bibinfo {author} {\bibfnamefont {D.~W.}\ \bibnamefont {Berry}}, \bibinfo
  {author} {\bibfnamefont {N.}~\bibnamefont {Wiebe}}, \bibinfo {author}
  {\bibfnamefont {J.}~\bibnamefont {McClean}}, \bibinfo {author} {\bibfnamefont
  {A.}~\bibnamefont {Paler}}, \bibinfo {author} {\bibfnamefont
  {A.}~\bibnamefont {Fowler}},\ and\ \bibinfo {author} {\bibfnamefont
  {H.}~\bibnamefont {Neven}},\ }\bibfield  {title} {\bibinfo {title} {Encoding
  electronic spectra in quantum circuits with linear {T} complexity},\
  }\bibfield  {journal} {\bibinfo  {journal} {Physical Review X}\ }\textbf
  {\bibinfo {volume} {8}},\ \href {https://doi.org/10.1103/physrevx.8.041015}
  {10.1103/physrevx.8.041015} (\bibinfo {year} {2018})\BibitemShut {NoStop}%
\bibitem [{\citenamefont {Khattar}\ and\ \citenamefont
  {Gidney}(2025)}]{Khattar_2025}%
  \BibitemOpen
  \bibfield  {author} {\bibinfo {author} {\bibfnamefont {T.}~\bibnamefont
  {Khattar}}\ and\ \bibinfo {author} {\bibfnamefont {C.}~\bibnamefont
  {Gidney}},\ }\bibfield  {title} {\bibinfo {title} {Rise of conditionally
  clean ancillae for efficient quantum circuit constructions},\ }\href
  {https://doi.org/10.22331/q-2025-05-21-1752} {\bibfield  {journal} {\bibinfo
  {journal} {Quantum}\ }\textbf {\bibinfo {volume} {9}},\ \bibinfo {pages}
  {1752} (\bibinfo {year} {2025})}\BibitemShut {NoStop}%
\bibitem [{\citenamefont {Sarwate}\ and\ \citenamefont
  {Yan}(2009)}]{sarwate2009modifiedeuclideanalgorithmsdecoding}%
  \BibitemOpen
  \bibfield  {author} {\bibinfo {author} {\bibfnamefont {D.~V.}\ \bibnamefont
  {Sarwate}}\ and\ \bibinfo {author} {\bibfnamefont {Z.}~\bibnamefont {Yan}},\
  }\href {https://arxiv.org/abs/0906.3778} {\bibinfo {title} {Modified
  euclidean algorithms for decoding {R}eed-{S}olomon codes}} (\bibinfo {year}
  {2009}),\ \Eprint {https://arxiv.org/abs/0906.3778} {arXiv:0906.3778 [cs.IT]}
  \BibitemShut {NoStop}%
\bibitem [{\citenamefont {Cormen}\ \emph {et~al.}(2022)\citenamefont {Cormen},
  \citenamefont {Leiserson}, \citenamefont {Rivest},\ and\ \citenamefont
  {Stein}}]{cormen2022introduction}%
  \BibitemOpen
  \bibfield  {author} {\bibinfo {author} {\bibfnamefont {T.~H.}\ \bibnamefont
  {Cormen}}, \bibinfo {author} {\bibfnamefont {C.~E.}\ \bibnamefont
  {Leiserson}}, \bibinfo {author} {\bibfnamefont {R.~L.}\ \bibnamefont
  {Rivest}},\ and\ \bibinfo {author} {\bibfnamefont {C.}~\bibnamefont
  {Stein}},\ }\href@noop {} {\emph {\bibinfo {title} {Introduction to
  algorithms}}}\ (\bibinfo  {publisher} {MIT press},\ \bibinfo {year}
  {2022})\BibitemShut {NoStop}%
\bibitem [{\citenamefont {Roetteler}\ \emph {et~al.}(2017)\citenamefont
  {Roetteler}, \citenamefont {Naehrig}, \citenamefont {Svore},\ and\
  \citenamefont {Lauter}}]{roetteler2017quantumresourceestimatescomputing}%
  \BibitemOpen
  \bibfield  {author} {\bibinfo {author} {\bibfnamefont {M.}~\bibnamefont
  {Roetteler}}, \bibinfo {author} {\bibfnamefont {M.}~\bibnamefont {Naehrig}},
  \bibinfo {author} {\bibfnamefont {K.~M.}\ \bibnamefont {Svore}},\ and\
  \bibinfo {author} {\bibfnamefont {K.}~\bibnamefont {Lauter}},\ }\href
  {https://arxiv.org/abs/1706.06752} {\bibinfo {title} {Quantum resource
  estimates for computing elliptic curve discrete logarithms}} (\bibinfo {year}
  {2017}),\ \Eprint {https://arxiv.org/abs/1706.06752} {arXiv:1706.06752
  [quant-ph]} \BibitemShut {NoStop}%
\bibitem [{\citenamefont {Banegas}\ \emph {et~al.}(2020)\citenamefont
  {Banegas}, \citenamefont {Bernstein}, \citenamefont {van Hoof},\ and\
  \citenamefont {Lange}}]{cryptoeprint:2020/1296}%
  \BibitemOpen
  \bibfield  {author} {\bibinfo {author} {\bibfnamefont {G.}~\bibnamefont
  {Banegas}}, \bibinfo {author} {\bibfnamefont {D.~J.}\ \bibnamefont
  {Bernstein}}, \bibinfo {author} {\bibfnamefont {I.}~\bibnamefont {van
  Hoof}},\ and\ \bibinfo {author} {\bibfnamefont {T.}~\bibnamefont {Lange}},\
  }\href {https://eprint.iacr.org/2020/1296} {\bibinfo {title} {Concrete
  quantum cryptanalysis of binary elliptic curves}},\ \bibinfo {howpublished}
  {Cryptology {ePrint} Archive, Paper 2020/1296} (\bibinfo {year}
  {2020})\BibitemShut {NoStop}%
\bibitem [{\citenamefont {Kim}\ and\ \citenamefont
  {Hong}(2023)}]{kim2023newspaceefficientquantumalgorithm}%
  \BibitemOpen
  \bibfield  {author} {\bibinfo {author} {\bibfnamefont {H.}~\bibnamefont
  {Kim}}\ and\ \bibinfo {author} {\bibfnamefont {S.}~\bibnamefont {Hong}},\
  }\href {https://arxiv.org/abs/2303.06570} {\bibinfo {title} {New
  space-efficient quantum algorithm for binary elliptic curves using the
  optimized division algorithm}} (\bibinfo {year} {2023}),\ \Eprint
  {https://arxiv.org/abs/2303.06570} {arXiv:2303.06570 [quant-ph]} \BibitemShut
  {NoStop}%
\bibitem [{\citenamefont {Earle}(1967)}]{carrysaveadderpatent}%
  \BibitemOpen
  \bibfield  {author} {\bibinfo {author} {\bibfnamefont {J.~G.}\ \bibnamefont
  {Earle}},\ }\href@noop {} {\bibinfo {title} {Latched carry save adder circuit
  for multipliers}} (\bibinfo {year} {1967}),\ \bibinfo {note} {uS Patent
  3,340,388}\BibitemShut {NoStop}%
\bibitem [{\citenamefont {van
  Hoof}(2020)}]{vanhoof2020spaceefficientquantummultiplicationpolynomials}%
  \BibitemOpen
  \bibfield  {author} {\bibinfo {author} {\bibfnamefont {I.}~\bibnamefont {van
  Hoof}},\ }\href {https://arxiv.org/abs/1910.02849} {\bibinfo {title}
  {Space-efficient quantum multiplication of polynomials for binary finite
  fields with sub-quadratic toffoli gate count}} (\bibinfo {year} {2020}),\
  \Eprint {https://arxiv.org/abs/1910.02849} {arXiv:1910.02849 [quant-ph]}
  \BibitemShut {NoStop}%
\bibitem [{\citenamefont {Amento}\ \emph {et~al.}(2012)\citenamefont {Amento},
  \citenamefont {Steinwandt},\ and\ \citenamefont
  {Roetteler}}]{amento2012efficientquantumcircuitsbinary}%
  \BibitemOpen
  \bibfield  {author} {\bibinfo {author} {\bibfnamefont {B.}~\bibnamefont
  {Amento}}, \bibinfo {author} {\bibfnamefont {R.}~\bibnamefont {Steinwandt}},\
  and\ \bibinfo {author} {\bibfnamefont {M.}~\bibnamefont {Roetteler}},\ }\href
  {https://arxiv.org/abs/1209.6348} {\bibinfo {title} {Efficient quantum
  circuits for binary elliptic curve arithmetic: reducing t-gate complexity}}
  (\bibinfo {year} {2012}),\ \Eprint {https://arxiv.org/abs/1209.6348}
  {arXiv:1209.6348 [quant-ph]} \BibitemShut {NoStop}%
\bibitem [{\citenamefont {Itoh}\ and\ \citenamefont
  {Tsujii}(1989)}]{ITOH198921}%
  \BibitemOpen
  \bibfield  {author} {\bibinfo {author} {\bibfnamefont {T.}~\bibnamefont
  {Itoh}}\ and\ \bibinfo {author} {\bibfnamefont {S.}~\bibnamefont {Tsujii}},\
  }\bibfield  {title} {\bibinfo {title} {Structure of parallel multipliers for
  a class of fields {GF}$(2^m)$},\ }\href
  {https://doi.org/https://doi.org/10.1016/0890-5401(89)90045-X} {\bibfield
  {journal} {\bibinfo  {journal} {Information and Computation}\ }\textbf
  {\bibinfo {volume} {83}},\ \bibinfo {pages} {21} (\bibinfo {year}
  {1989})}\BibitemShut {NoStop}%
\bibitem [{\citenamefont {Facility}(2023)}]{frontier2023epyc}%
  \BibitemOpen
  \bibfield  {author} {\bibinfo {author} {\bibfnamefont {O.~R. L.~C.}\
  \bibnamefont {Facility}},\ }\href@noop {} {\bibinfo {title} {The {EPYC}™
  {CPU} and {INSTINCT}™ {MI}250{X} {GPU}s in frontier}},\ \bibinfo
  {howpublished} {Frontier Training Workshop} (\bibinfo {year} {2023}),\
  \bibinfo {note}
  {\url{https://www.olcf.ornl.gov/wp-content/uploads/Public-AMD-Instinct-MI-250X-Frontier-8.23.23.pdf}}\BibitemShut
  {NoStop}%
\bibitem [{\citenamefont {Chen}\ \emph {et~al.}(2022)\citenamefont {Chen},
  \citenamefont {Liu},\ and\ \citenamefont {Zhandry}}]{clz21}%
  \BibitemOpen
  \bibfield  {author} {\bibinfo {author} {\bibfnamefont {Y.}~\bibnamefont
  {Chen}}, \bibinfo {author} {\bibfnamefont {Q.}~\bibnamefont {Liu}},\ and\
  \bibinfo {author} {\bibfnamefont {M.}~\bibnamefont {Zhandry}},\ }\bibfield
  {title} {\bibinfo {title} {Quantum algorithms for variants of average-case
  lattice problems via filtering},\ }in\ \href@noop {} {\emph {\bibinfo
  {booktitle} {Annual international conference on the theory and applications
  of cryptographic techniques}}}\ (\bibinfo {organization} {Springer},\
  \bibinfo {year} {2022})\ pp.\ \bibinfo {pages} {372--401}\BibitemShut
  {NoStop}%
\bibitem [{\citenamefont {Mullen}\ and\ \citenamefont
  {Panario}(2013)}]{mullen2013handbook}%
  \BibitemOpen
  \bibfield  {author} {\bibinfo {author} {\bibfnamefont {G.~L.}\ \bibnamefont
  {Mullen}}\ and\ \bibinfo {author} {\bibfnamefont {D.}~\bibnamefont
  {Panario}},\ }\href@noop {} {\emph {\bibinfo {title} {Handbook of finite
  fields}}},\ Vol.~\bibinfo {volume} {17}\ (\bibinfo  {publisher} {CRC press
  Boca Raton},\ \bibinfo {year} {2013})\BibitemShut {NoStop}%
\bibitem [{\citenamefont {Briaud}\ \emph {et~al.}(2026)\citenamefont {Briaud},
  \citenamefont {Dinur}, \citenamefont {Ghosal}, \citenamefont {Jain},
  \citenamefont {Lou},\ and\ \citenamefont {Sahai}}]{briaud2025quantum}%
  \BibitemOpen
  \bibfield  {author} {\bibinfo {author} {\bibfnamefont {P.}~\bibnamefont
  {Briaud}}, \bibinfo {author} {\bibfnamefont {I.}~\bibnamefont {Dinur}},
  \bibinfo {author} {\bibfnamefont {R.}~\bibnamefont {Ghosal}}, \bibinfo
  {author} {\bibfnamefont {A.}~\bibnamefont {Jain}}, \bibinfo {author}
  {\bibfnamefont {P.}~\bibnamefont {Lou}},\ and\ \bibinfo {author}
  {\bibfnamefont {A.}~\bibnamefont {Sahai}},\ }\bibinfo {title} {Quantum
  advantage via solving multivariate polynomials},\ in\ \href
  {https://doi.org/10.1137/1.9781611978971.52} {\emph {\bibinfo {booktitle}
  {Proceedings of the 2026 Annual ACM-SIAM Symposium on Discrete Algorithms
  (SODA)}}}\ (\bibinfo  {publisher} {Society for Industrial and Applied
  Mathematics},\ \bibinfo {year} {2026})\ pp.\ \bibinfo {pages}
  {1383--1423}\BibitemShut {NoStop}%
\bibitem [{\citenamefont {Hoeffding}(1963)}]{Hoeffding}%
  \BibitemOpen
  \bibfield  {author} {\bibinfo {author} {\bibfnamefont {W.}~\bibnamefont
  {Hoeffding}},\ }\bibfield  {title} {\bibinfo {title} {Probability
  inequalities for sums of bounded random variables},\ }\href
  {http://www.jstor.org/stable/2282952} {\bibfield  {journal} {\bibinfo
  {journal} {Journal of the American Statistical Association}\ }\textbf
  {\bibinfo {volume} {58}},\ \bibinfo {pages} {13} (\bibinfo {year}
  {1963})}\BibitemShut {NoStop}%
\bibitem [{\citenamefont {Von Zur~Gathen}\ and\ \citenamefont
  {Gerhard}(2003)}]{vonzurGathen2013moderncomputer}%
  \BibitemOpen
  \bibfield  {author} {\bibinfo {author} {\bibfnamefont {J.}~\bibnamefont {Von
  Zur~Gathen}}\ and\ \bibinfo {author} {\bibfnamefont {J.}~\bibnamefont
  {Gerhard}},\ }\href@noop {} {\emph {\bibinfo {title} {Modern computer
  algebra}}}\ (\bibinfo  {publisher} {Cambridge university press},\ \bibinfo
  {year} {2003})\BibitemShut {NoStop}%
\bibitem [{\citenamefont {Sch\"{o}nhage}\ and\ \citenamefont
  {Strassen}(1971)}]{Schnhage1971}%
  \BibitemOpen
  \bibfield  {author} {\bibinfo {author} {\bibfnamefont {A.}~\bibnamefont
  {Sch\"{o}nhage}}\ and\ \bibinfo {author} {\bibfnamefont {V.}~\bibnamefont
  {Strassen}},\ }\bibfield  {title} {\bibinfo {title} {Schnelle multiplikation
  großer zahlen},\ }\href {https://doi.org/10.1007/bf02242355} {\bibfield
  {journal} {\bibinfo  {journal} {Computing}\ }\textbf {\bibinfo {volume}
  {7}},\ \bibinfo {pages} {281–292} (\bibinfo {year} {1971})}\BibitemShut
  {NoStop}%
\bibitem [{\citenamefont {Cantor}\ and\ \citenamefont
  {Kaltofen}(1991)}]{cantor1991fast}%
  \BibitemOpen
  \bibfield  {author} {\bibinfo {author} {\bibfnamefont {D.~G.}\ \bibnamefont
  {Cantor}}\ and\ \bibinfo {author} {\bibfnamefont {E.}~\bibnamefont
  {Kaltofen}},\ }\bibfield  {title} {\bibinfo {title} {On fast multiplication
  of polynomials over arbitrary algebras},\ }\href@noop {} {\bibfield
  {journal} {\bibinfo  {journal} {Acta Informatica}\ }\textbf {\bibinfo
  {volume} {28}},\ \bibinfo {pages} {693} (\bibinfo {year} {1991})}\BibitemShut
  {NoStop}%
\bibitem [{\citenamefont {Harvey}\ \emph {et~al.}(2014)\citenamefont {Harvey},
  \citenamefont {van~der Hoeven},\ and\ \citenamefont
  {Lecerf}}]{harvey2014fasterpolynomialmultiplicationfinite}%
  \BibitemOpen
  \bibfield  {author} {\bibinfo {author} {\bibfnamefont {D.}~\bibnamefont
  {Harvey}}, \bibinfo {author} {\bibfnamefont {J.}~\bibnamefont {van~der
  Hoeven}},\ and\ \bibinfo {author} {\bibfnamefont {G.}~\bibnamefont
  {Lecerf}},\ }\href {https://arxiv.org/abs/1407.3361} {\bibinfo {title}
  {Faster polynomial multiplication over finite fields}} (\bibinfo {year}
  {2014}),\ \Eprint {https://arxiv.org/abs/1407.3361} {arXiv:1407.3361 [cs.CC]}
  \BibitemShut {NoStop}%
\bibitem [{\citenamefont {Justesen}(1976)}]{justesen2006complexity}%
  \BibitemOpen
  \bibfield  {author} {\bibinfo {author} {\bibfnamefont {J.}~\bibnamefont
  {Justesen}},\ }\bibfield  {title} {\bibinfo {title} {On the complexity of
  decoding {R}eed--{S}olomon codes ({C}orresp.)},\ }\href@noop {} {\bibfield
  {journal} {\bibinfo  {journal} {IEEE Transactions on Information Theory}\
  }\textbf {\bibinfo {volume} {22}},\ \bibinfo {pages} {237} (\bibinfo {year}
  {1976})}\BibitemShut {NoStop}%
\bibitem [{\citenamefont {Sarwate}(1977)}]{sarwate1977complexity}%
  \BibitemOpen
  \bibfield  {author} {\bibinfo {author} {\bibfnamefont {D.~V.}\ \bibnamefont
  {Sarwate}},\ }\bibfield  {title} {\bibinfo {title} {On the complexity of
  decoding {G}oppa codes ({C}orresp.)},\ }\href
  {https://doi.org/10.1109/TIT.1977.1055732} {\bibfield  {journal} {\bibinfo
  {journal} {IEEE Transactions on Information Theory}\ }\textbf {\bibinfo
  {volume} {23}},\ \bibinfo {pages} {515} (\bibinfo {year} {1977})}\BibitemShut
  {NoStop}%
\bibitem [{\citenamefont {Bostan}\ \emph {et~al.}(2003)\citenamefont {Bostan},
  \citenamefont {Lecerf},\ and\ \citenamefont {Schost}}]{bostan2003tellegen}%
  \BibitemOpen
  \bibfield  {author} {\bibinfo {author} {\bibfnamefont {A.}~\bibnamefont
  {Bostan}}, \bibinfo {author} {\bibfnamefont {G.}~\bibnamefont {Lecerf}},\
  and\ \bibinfo {author} {\bibfnamefont {{\'E}.}~\bibnamefont {Schost}},\
  }\bibfield  {title} {\bibinfo {title} {Tellegen's principle into practice},\
  }in\ \href {https://doi.org/10.1145/860854.860870} {\emph {\bibinfo
  {booktitle} {Proceedings of the 2003 International Symposium on Symbolic and
  Algebraic Computation}}}\ (\bibinfo  {publisher} {ACM},\ \bibinfo {year}
  {2003})\ pp.\ \bibinfo {pages} {37--44}\BibitemShut {NoStop}%
\bibitem [{\citenamefont {Bennett}(1973)}]{bennett1973logical}%
  \BibitemOpen
  \bibfield  {author} {\bibinfo {author} {\bibfnamefont {C.~H.}\ \bibnamefont
  {Bennett}},\ }\bibfield  {title} {\bibinfo {title} {Logical reversibility of
  computation},\ }\href {https://doi.org/10.1147/rd.176.0525} {\bibfield
  {journal} {\bibinfo  {journal} {IBM Journal of Research and Development}\
  }\textbf {\bibinfo {volume} {17}},\ \bibinfo {pages} {525} (\bibinfo {year}
  {1973})}\BibitemShut {NoStop}%
\bibitem [{\citenamefont {Nie}\ \emph {et~al.}(2023)\citenamefont {Nie},
  \citenamefont {Zhu}, \citenamefont {Li},\ and\ \citenamefont
  {Sun}}]{Nie2023}%
  \BibitemOpen
  \bibfield  {author} {\bibinfo {author} {\bibfnamefont {J.}~\bibnamefont
  {Nie}}, \bibinfo {author} {\bibfnamefont {Q.}~\bibnamefont {Zhu}}, \bibinfo
  {author} {\bibfnamefont {M.}~\bibnamefont {Li}},\ and\ \bibinfo {author}
  {\bibfnamefont {X.}~\bibnamefont {Sun}},\ }\bibfield  {title} {\bibinfo
  {title} {Quantum circuit design for integer multiplication based on
  {S}ch\"{o}nhage--{S}trassen algorithm},\ }\href
  {https://doi.org/10.1109/tcad.2023.3279300} {\bibfield  {journal} {\bibinfo
  {journal} {IEEE Transactions on Computer-Aided Design of Integrated Circuits
  and Systems}\ }\textbf {\bibinfo {volume} {42}},\ \bibinfo {pages} {4791}
  (\bibinfo {year} {2023})}\BibitemShut {NoStop}%
\bibitem [{\citenamefont {Bartschi}\ and\ \citenamefont
  {Eidenbenz}(2022)}]{Bartschi2022}%
  \BibitemOpen
  \bibfield  {author} {\bibinfo {author} {\bibfnamefont {A.}~\bibnamefont
  {Bartschi}}\ and\ \bibinfo {author} {\bibfnamefont {S.}~\bibnamefont
  {Eidenbenz}},\ }\bibfield  {title} {\bibinfo {title} {Short-depth circuits
  for dicke state preparation},\ }in\ \href
  {https://doi.org/10.1109/qce53715.2022.00027} {\emph {\bibinfo {booktitle}
  {2022 IEEE International Conference on Quantum Computing and Engineering
  (QCE)}}}\ (\bibinfo  {publisher} {IEEE},\ \bibinfo {year} {2022})\ p.\
  \bibinfo {pages} {87–96}\BibitemShut {NoStop}%
\bibitem [{\citenamefont {Koetter}\ and\ \citenamefont
  {Vardy}(2003)}]{koetter2003algebraic}%
  \BibitemOpen
  \bibfield  {author} {\bibinfo {author} {\bibfnamefont {R.}~\bibnamefont
  {Koetter}}\ and\ \bibinfo {author} {\bibfnamefont {A.}~\bibnamefont
  {Vardy}},\ }\bibfield  {title} {\bibinfo {title} {Algebraic soft-decision
  decoding of {R}eed--{S}olomon codes},\ }\href@noop {} {\bibfield  {journal}
  {\bibinfo  {journal} {IEEE Transactions on Information Theory}\ }\textbf
  {\bibinfo {volume} {49}},\ \bibinfo {pages} {2809} (\bibinfo {year}
  {2003})}\BibitemShut {NoStop}%
\bibitem [{\citenamefont {Arora}\ and\ \citenamefont
  {Ge}(2011)}]{arora2011new}%
  \BibitemOpen
  \bibfield  {author} {\bibinfo {author} {\bibfnamefont {S.}~\bibnamefont
  {Arora}}\ and\ \bibinfo {author} {\bibfnamefont {R.}~\bibnamefont {Ge}},\
  }\bibfield  {title} {\bibinfo {title} {New algorithms for learning in
  presence of errors},\ }in\ \href
  {https://doi.org/10.1007/978-3-642-22006-7_34} {\emph {\bibinfo {booktitle}
  {Automata, Languages and Programming --- {ICALP} 2011}}},\ \bibinfo {series}
  {Lecture Notes in Computer Science}, Vol.\ \bibinfo {volume} {6755}\
  (\bibinfo  {publisher} {Springer},\ \bibinfo {year} {2011})\ pp.\ \bibinfo
  {pages} {403--415}\BibitemShut {NoStop}%
\bibitem [{\citenamefont {Gidney}\ \emph {et~al.}(2024)\citenamefont {Gidney},
  \citenamefont {Shutty},\ and\ \citenamefont
  {Jones}}]{gidney2024magicstatecultivationgrowing}%
  \BibitemOpen
  \bibfield  {author} {\bibinfo {author} {\bibfnamefont {C.}~\bibnamefont
  {Gidney}}, \bibinfo {author} {\bibfnamefont {N.}~\bibnamefont {Shutty}},\
  and\ \bibinfo {author} {\bibfnamefont {C.}~\bibnamefont {Jones}},\ }\href
  {https://arxiv.org/abs/2409.17595} {\bibinfo {title} {Magic state
  cultivation: growing {T} states as cheap as {CNOT} gates}} (\bibinfo {year}
  {2024}),\ \Eprint {https://arxiv.org/abs/2409.17595} {arXiv:2409.17595
  [quant-ph]} \BibitemShut {NoStop}%
\bibitem [{\citenamefont
  {Gidney}(2025{\natexlab{b}})}]{gidney2025classicalquantumadderconstantworkspace}%
  \BibitemOpen
  \bibfield  {author} {\bibinfo {author} {\bibfnamefont {C.}~\bibnamefont
  {Gidney}},\ }\href {https://arxiv.org/abs/2507.23079} {\bibinfo {title} {A
  classical-quantum adder with constant workspace and linear gates}} (\bibinfo
  {year} {2025}{\natexlab{b}}),\ \Eprint {https://arxiv.org/abs/2507.23079}
  {arXiv:2507.23079 [quant-ph]} \BibitemShut {NoStop}%
\bibitem [{\citenamefont {Yoder}\ \emph {et~al.}(2025)\citenamefont {Yoder},
  \citenamefont {Schoute}, \citenamefont {Rall}, \citenamefont {Pritchett},
  \citenamefont {Gambetta}, \citenamefont {Cross}, \citenamefont {Carroll},\
  and\ \citenamefont {Beverland}}]{yoder2025tourgrossmodularquantum}%
  \BibitemOpen
  \bibfield  {author} {\bibinfo {author} {\bibfnamefont {T.~J.}\ \bibnamefont
  {Yoder}}, \bibinfo {author} {\bibfnamefont {E.}~\bibnamefont {Schoute}},
  \bibinfo {author} {\bibfnamefont {P.}~\bibnamefont {Rall}}, \bibinfo {author}
  {\bibfnamefont {E.}~\bibnamefont {Pritchett}}, \bibinfo {author}
  {\bibfnamefont {J.~M.}\ \bibnamefont {Gambetta}}, \bibinfo {author}
  {\bibfnamefont {A.~W.}\ \bibnamefont {Cross}}, \bibinfo {author}
  {\bibfnamefont {M.}~\bibnamefont {Carroll}},\ and\ \bibinfo {author}
  {\bibfnamefont {M.~E.}\ \bibnamefont {Beverland}},\ }\href
  {https://arxiv.org/abs/2506.03094} {\bibinfo {title} {Tour de gross: A
  modular quantum computer based on bivariate bicycle codes}} (\bibinfo {year}
  {2025}),\ \Eprint {https://arxiv.org/abs/2506.03094} {arXiv:2506.03094
  [quant-ph]} \BibitemShut {NoStop}%
\bibitem [{\citenamefont {Maslov}\ \emph {et~al.}(2025)\citenamefont {Maslov},
  \citenamefont {Yosri},\ and\ \citenamefont
  {Gavinsky}}]{maslov2025asymptotic}%
  \BibitemOpen
  \bibfield  {author} {\bibinfo {author} {\bibfnamefont {D.}~\bibnamefont
  {Maslov}}, \bibinfo {author} {\bibfnamefont {N.}~\bibnamefont {Yosri}},\ and\
  \bibinfo {author} {\bibfnamefont {D.}~\bibnamefont {Gavinsky}},\ }\bibfield
  {title} {\bibinfo {title} {Asymptotic yet practical optimization of quantum
  circuits implementing {GF}($2^m$) multiplication and division operations},\
  }in\ \href@noop {} {\emph {\bibinfo {booktitle} {7th International Workshop
  on Quantum Compilation, Helsinki, Finland}}}\ (\bibinfo {year}
  {2025})\BibitemShut {NoStop}%
\bibitem [{\citenamefont {Gidney}(2024)}]{Gidney_2024}%
  \BibitemOpen
  \bibfield  {author} {\bibinfo {author} {\bibfnamefont {C.}~\bibnamefont
  {Gidney}},\ }\bibfield  {title} {\bibinfo {title} {Inplace access to the
  surface code y basis},\ }\href {https://doi.org/10.22331/q-2024-04-08-1310}
  {\bibfield  {journal} {\bibinfo  {journal} {Quantum}\ }\textbf {\bibinfo
  {volume} {8}},\ \bibinfo {pages} {1310} (\bibinfo {year} {2024})}\BibitemShut
  {NoStop}%
\bibitem [{\citenamefont {Schrottenloher}(2026)}]{Schrottenloher2026Optimized}%
  \BibitemOpen
  \bibfield  {author} {\bibinfo {author} {\bibfnamefont {A.}~\bibnamefont
  {Schrottenloher}},\ }\href {https://doi.org/10.48550/arXiv.2606.02235}
  {\bibinfo {title} {Optimized point addition circuits for elliptic curve
  discrete logarithms}} (\bibinfo {year} {2026}),\ \Eprint
  {https://arxiv.org/abs/2606.02235} {arXiv:2606.02235 [quant-ph]} \BibitemShut
  {NoStop}%
\bibitem [{\citenamefont {H{\"a}ner}\ \emph {et~al.}(2026)\citenamefont
  {H{\"a}ner}, \citenamefont {Tripier}, \citenamefont {Young}, \citenamefont
  {Naehrig}, \citenamefont {Maksymov}, \citenamefont {Alam}, \citenamefont
  {Maslov}, \citenamefont {Parrott}, \citenamefont {de~Sereville},
  \citenamefont {Sullivan}, \citenamefont {Webster}, \citenamefont {Delfosse},
  \citenamefont {Gamble},\ and\ \citenamefont
  {Roetteler}}]{Haner2026TrappedIon}%
  \BibitemOpen
  \bibfield  {author} {\bibinfo {author} {\bibfnamefont {T.}~\bibnamefont
  {H{\"a}ner}}, \bibinfo {author} {\bibfnamefont {F.}~\bibnamefont {Tripier}},
  \bibinfo {author} {\bibfnamefont {J.}~\bibnamefont {Young}}, \bibinfo
  {author} {\bibfnamefont {M.}~\bibnamefont {Naehrig}}, \bibinfo {author}
  {\bibfnamefont {A.}~\bibnamefont {Maksymov}}, \bibinfo {author}
  {\bibfnamefont {S.}~\bibnamefont {Alam}}, \bibinfo {author} {\bibfnamefont
  {D.}~\bibnamefont {Maslov}}, \bibinfo {author} {\bibfnamefont
  {M.}~\bibnamefont {Parrott}}, \bibinfo {author} {\bibfnamefont
  {Y.}~\bibnamefont {de~Sereville}}, \bibinfo {author} {\bibfnamefont
  {J.}~\bibnamefont {Sullivan}}, \bibinfo {author} {\bibfnamefont
  {M.}~\bibnamefont {Webster}}, \bibinfo {author} {\bibfnamefont
  {N.}~\bibnamefont {Delfosse}}, \bibinfo {author} {\bibfnamefont
  {J.}~\bibnamefont {Gamble}},\ and\ \bibinfo {author} {\bibfnamefont
  {M.}~\bibnamefont {Roetteler}},\ }\href
  {https://doi.org/10.48550/arXiv.2609.05625} {\bibinfo {title} {Computing
  256-bit elliptic curve discrete logarithms in 26 days on a fault-tolerant
  trapped-ion quantum computer with 20,000 qubits}} (\bibinfo {year} {2026}),\
  \Eprint {https://arxiv.org/abs/2609.05625} {arXiv:2609.05625 [quant-ph]}
  \BibitemShut {NoStop}%
\bibitem [{\citenamefont {Long}\ \emph {et~al.}(2026)\citenamefont {Long},
  \citenamefont {Pender}, \citenamefont {Huang}, \citenamefont {Santos},
  \citenamefont {Singh}, \citenamefont {Naskr{\k{e}}cki}, \citenamefont
  {{BitWonka}}, \citenamefont {Doyle}, \citenamefont {Dallaire-Demers},
  \citenamefont {Giannicola}, \citenamefont {Paschoarelli}, \citenamefont
  {Freuler}, \citenamefont {Lee}, \citenamefont {Gnuchev}, \citenamefont
  {Kannappan}, \citenamefont {Boyer}, \citenamefont {Butler}, \citenamefont
  {Balasubramani}, \citenamefont {Newman}, \citenamefont {Dereje},
  \citenamefont {Hertlein}, \citenamefont {Kodra}, \citenamefont {Levy},
  \citenamefont {Patel}, \citenamefont {Rose}, \citenamefont {Zweil},
  \citenamefont {Wisdom}, \citenamefont {El-Eter}, \citenamefont {Lee},
  \citenamefont {Dong}, \citenamefont {Li}, \citenamefont {Joseph},
  \citenamefont {Nguyen}, \citenamefont {Naik}, \citenamefont {Anant},
  \citenamefont {Deb},\ and\ \citenamefont {Drake}}]{Long2026RecordReplay}%
  \BibitemOpen
  \bibfield  {author} {\bibinfo {author} {\bibfnamefont {J.}~\bibnamefont
  {Long}}, \bibinfo {author} {\bibfnamefont {T.}~\bibnamefont {Pender}},
  \bibinfo {author} {\bibfnamefont {Z.}~\bibnamefont {Huang}}, \bibinfo
  {author} {\bibfnamefont {M.~B.}\ \bibnamefont {Santos}}, \bibinfo {author}
  {\bibfnamefont {S.~K.}\ \bibnamefont {Singh}}, \bibinfo {author}
  {\bibfnamefont {B.}~\bibnamefont {Naskr{\k{e}}cki}}, \bibinfo {author}
  {\bibnamefont {{BitWonka}}}, \bibinfo {author} {\bibfnamefont
  {J.}~\bibnamefont {Doyle}}, \bibinfo {author} {\bibfnamefont {P.-L.}\
  \bibnamefont {Dallaire-Demers}}, \bibinfo {author} {\bibfnamefont
  {F.}~\bibnamefont {Giannicola}}, \bibinfo {author} {\bibfnamefont {R.~M.~L.}\
  \bibnamefont {Paschoarelli}}, \bibinfo {author} {\bibfnamefont
  {O.}~\bibnamefont {Freuler}}, \bibinfo {author} {\bibfnamefont {J.~C.-H.}\
  \bibnamefont {Lee}}, \bibinfo {author} {\bibfnamefont {V.}~\bibnamefont
  {Gnuchev}}, \bibinfo {author} {\bibfnamefont {G.}~\bibnamefont {Kannappan}},
  \bibinfo {author} {\bibfnamefont {J.}~\bibnamefont {Boyer}}, \bibinfo
  {author} {\bibfnamefont {X.}~\bibnamefont {Butler}}, \bibinfo {author}
  {\bibfnamefont {A.}~\bibnamefont {Balasubramani}}, \bibinfo {author}
  {\bibfnamefont {J.}~\bibnamefont {Newman}}, \bibinfo {author} {\bibfnamefont
  {B.}~\bibnamefont {Dereje}}, \bibinfo {author} {\bibfnamefont
  {A.}~\bibnamefont {Hertlein}}, \bibinfo {author} {\bibfnamefont
  {R.}~\bibnamefont {Kodra}}, \bibinfo {author} {\bibfnamefont
  {L.}~\bibnamefont {Levy}}, \bibinfo {author} {\bibfnamefont {S.}~\bibnamefont
  {Patel}}, \bibinfo {author} {\bibfnamefont {J.~T.}\ \bibnamefont {Rose}},
  \bibinfo {author} {\bibfnamefont {M.}~\bibnamefont {Zweil}}, \bibinfo
  {author} {\bibfnamefont {O.}~\bibnamefont {Wisdom}}, \bibinfo {author}
  {\bibfnamefont {T.}~\bibnamefont {El-Eter}}, \bibinfo {author} {\bibfnamefont
  {E.}~\bibnamefont {Lee}}, \bibinfo {author} {\bibfnamefont {M.}~\bibnamefont
  {Dong}}, \bibinfo {author} {\bibfnamefont {A.}~\bibnamefont {Li}}, \bibinfo
  {author} {\bibfnamefont {A.}~\bibnamefont {Joseph}}, \bibinfo {author}
  {\bibfnamefont {D.}~\bibnamefont {Nguyen}}, \bibinfo {author} {\bibfnamefont
  {G.}~\bibnamefont {Naik}}, \bibinfo {author} {\bibfnamefont {G.}~\bibnamefont
  {Anant}}, \bibinfo {author} {\bibfnamefont {S.}~\bibnamefont {Deb}},\ and\
  \bibinfo {author} {\bibfnamefont {J.}~\bibnamefont {Drake}},\ }\href
  {https://doi.org/10.48550/arXiv.2609.28882} {\bibinfo {title} {Efficient
  record-and-replay arithmetic for quantum elliptic-curve point addition}}
  (\bibinfo {year} {2026}),\ \Eprint {https://arxiv.org/abs/2609.28882}
  {arXiv:2609.28882 [quant-ph]} \BibitemShut {NoStop}%
\bibitem [{\citenamefont {Luo}\ \emph {et~al.}(2026)\citenamefont {Luo},
  \citenamefont {Yang}, \citenamefont {Luo}, \citenamefont {Wang},
  \citenamefont {Su}, \citenamefont {Sun}, \citenamefont {Li},\ and\
  \citenamefont {Li}}]{Luo2026SpaceEfficient}%
  \BibitemOpen
  \bibfield  {author} {\bibinfo {author} {\bibfnamefont {H.}~\bibnamefont
  {Luo}}, \bibinfo {author} {\bibfnamefont {Z.}~\bibnamefont {Yang}}, \bibinfo
  {author} {\bibfnamefont {J.}~\bibnamefont {Luo}}, \bibinfo {author}
  {\bibfnamefont {Z.}~\bibnamefont {Wang}}, \bibinfo {author} {\bibfnamefont
  {Y.}~\bibnamefont {Su}}, \bibinfo {author} {\bibfnamefont {X.}~\bibnamefont
  {Sun}}, \bibinfo {author} {\bibfnamefont {L.}~\bibnamefont {Li}},\ and\
  \bibinfo {author} {\bibfnamefont {T.}~\bibnamefont {Li}},\ }\href
  {https://doi.org/10.48550/arXiv.2607.13816} {\bibinfo {title} {Quantum
  algorithm for elliptic curve discrete logarithms with space-efficient point
  addition}} (\bibinfo {year} {2026}),\ \Eprint
  {https://arxiv.org/abs/2607.13816} {arXiv:2607.13816 [quant-ph]} \BibitemShut
  {NoStop}%
\bibitem [{\citenamefont {Sugisaki}\ \emph {et~al.}(2016)\citenamefont
  {Sugisaki}, \citenamefont {Yamamoto}, \citenamefont {Nakazawa}, \citenamefont
  {Toyota}, \citenamefont {Sato}, \citenamefont {Shiomi},\ and\ \citenamefont
  {Takui}}]{Sugisaki2016}%
  \BibitemOpen
  \bibfield  {author} {\bibinfo {author} {\bibfnamefont {K.}~\bibnamefont
  {Sugisaki}}, \bibinfo {author} {\bibfnamefont {S.}~\bibnamefont {Yamamoto}},
  \bibinfo {author} {\bibfnamefont {S.}~\bibnamefont {Nakazawa}}, \bibinfo
  {author} {\bibfnamefont {K.}~\bibnamefont {Toyota}}, \bibinfo {author}
  {\bibfnamefont {K.}~\bibnamefont {Sato}}, \bibinfo {author} {\bibfnamefont
  {D.}~\bibnamefont {Shiomi}},\ and\ \bibinfo {author} {\bibfnamefont
  {T.}~\bibnamefont {Takui}},\ }\bibfield  {title} {\bibinfo {title} {Quantum
  chemistry on quantum computers: A polynomial-time quantum algorithm for
  constructing the wave functions of open-shell molecules},\ }\href
  {https://doi.org/10.1021/acs.jpca.6b04932} {\bibfield  {journal} {\bibinfo
  {journal} {The Journal of Physical Chemistry A}\ }\textbf {\bibinfo {volume}
  {120}},\ \bibinfo {pages} {6459} (\bibinfo {year} {2016})}\BibitemShut
  {NoStop}%
\bibitem [{\citenamefont {Berry}\ \emph {et~al.}(2018)\citenamefont {Berry},
  \citenamefont {Kieferov{\'a}}, \citenamefont {Scherer}, \citenamefont
  {Sanders}, \citenamefont {Low}, \citenamefont {Wiebe}, \citenamefont
  {Gidney},\ and\ \citenamefont {Babbush}}]{Berry2018eigenstates}%
  \BibitemOpen
  \bibfield  {author} {\bibinfo {author} {\bibfnamefont {D.~W.}\ \bibnamefont
  {Berry}}, \bibinfo {author} {\bibfnamefont {M.}~\bibnamefont
  {Kieferov{\'a}}}, \bibinfo {author} {\bibfnamefont {A.}~\bibnamefont
  {Scherer}}, \bibinfo {author} {\bibfnamefont {Y.~R.}\ \bibnamefont
  {Sanders}}, \bibinfo {author} {\bibfnamefont {G.~H.}\ \bibnamefont {Low}},
  \bibinfo {author} {\bibfnamefont {N.}~\bibnamefont {Wiebe}}, \bibinfo
  {author} {\bibfnamefont {C.}~\bibnamefont {Gidney}},\ and\ \bibinfo {author}
  {\bibfnamefont {R.}~\bibnamefont {Babbush}},\ }\bibfield  {title} {\bibinfo
  {title} {Improved techniques for preparing eigenstates of fermionic
  {H}amiltonians},\ }\href {https://doi.org/10.1038/s41534-018-0071-5}
  {\bibfield  {journal} {\bibinfo  {journal} {npj Quantum Information}\
  }\textbf {\bibinfo {volume} {4}},\ \bibinfo {pages} {22} (\bibinfo {year}
  {2018})}\BibitemShut {NoStop}%
\bibitem [{\citenamefont {Babbush}\ \emph {et~al.}(2019)\citenamefont
  {Babbush}, \citenamefont {Berry}, \citenamefont {McClean},\ and\
  \citenamefont {Neven}}]{Babbush2019sublinear}%
  \BibitemOpen
  \bibfield  {author} {\bibinfo {author} {\bibfnamefont {R.}~\bibnamefont
  {Babbush}}, \bibinfo {author} {\bibfnamefont {D.~W.}\ \bibnamefont {Berry}},
  \bibinfo {author} {\bibfnamefont {J.~R.}\ \bibnamefont {McClean}},\ and\
  \bibinfo {author} {\bibfnamefont {H.}~\bibnamefont {Neven}},\ }\bibfield
  {title} {\bibinfo {title} {Quantum simulation of chemistry with sublinear
  scaling in basis size},\ }\href {https://doi.org/10.1038/s41534-019-0199-y}
  {\bibfield  {journal} {\bibinfo  {journal} {npj Quantum Information}\
  }\textbf {\bibinfo {volume} {5}},\ \bibinfo {pages} {92} (\bibinfo {year}
  {2019})}\BibitemShut {NoStop}%
\bibitem [{\citenamefont {Gard}\ \emph {et~al.}(2020)\citenamefont {Gard},
  \citenamefont {Zhu}, \citenamefont {Barron}, \citenamefont {Mayhall},
  \citenamefont {Economou},\ and\ \citenamefont {Barnes}}]{Gard2020}%
  \BibitemOpen
  \bibfield  {author} {\bibinfo {author} {\bibfnamefont {B.~T.}\ \bibnamefont
  {Gard}}, \bibinfo {author} {\bibfnamefont {L.}~\bibnamefont {Zhu}}, \bibinfo
  {author} {\bibfnamefont {G.~S.}\ \bibnamefont {Barron}}, \bibinfo {author}
  {\bibfnamefont {N.~J.}\ \bibnamefont {Mayhall}}, \bibinfo {author}
  {\bibfnamefont {S.~E.}\ \bibnamefont {Economou}},\ and\ \bibinfo {author}
  {\bibfnamefont {E.}~\bibnamefont {Barnes}},\ }\bibfield  {title} {\bibinfo
  {title} {Efficient symmetry-preserving state preparation circuits for the
  variational quantum eigensolver algorithm},\ }\href
  {https://doi.org/10.1038/s41534-019-0240-1} {\bibfield  {journal} {\bibinfo
  {journal} {npj Quantum Information}\ }\textbf {\bibinfo {volume} {6}},\
  \bibinfo {pages} {10} (\bibinfo {year} {2020})}\BibitemShut {NoStop}%
\bibitem [{\citenamefont {Chailloux}\ and\ \citenamefont
  {Tillich}(2025)}]{chailloux2025quantum}%
  \BibitemOpen
  \bibfield  {author} {\bibinfo {author} {\bibfnamefont {A.}~\bibnamefont
  {Chailloux}}\ and\ \bibinfo {author} {\bibfnamefont {J.-P.}\ \bibnamefont
  {Tillich}},\ }\bibfield  {title} {\bibinfo {title} {Quantum advantage from
  soft decoders},\ }in\ \href@noop {} {\emph {\bibinfo {booktitle} {Proceedings
  of the 57th Annual ACM Symposium on Theory of Computing}}}\ (\bibinfo {year}
  {2025})\ pp.\ \bibinfo {pages} {738--749}\BibitemShut {NoStop}%
\bibitem [{\citenamefont {Gu}\ and\ \citenamefont
  {Jordan}(2025)}]{GuJordan2025Algebraic}%
  \BibitemOpen
  \bibfield  {author} {\bibinfo {author} {\bibfnamefont {A.}~\bibnamefont
  {Gu}}\ and\ \bibinfo {author} {\bibfnamefont {S.~P.}\ \bibnamefont
  {Jordan}},\ }\href {https://arxiv.org/abs/2510.06603} {\bibinfo {title}
  {Algebraic geometry codes and decoded quantum interferometry}} (\bibinfo
  {year} {2025}),\ \Eprint {https://arxiv.org/abs/2510.06603} {arXiv:2510.06603
  [quant-ph]} \BibitemShut {NoStop}%
\bibitem [{\citenamefont {Knuth}(2005)}]{knuth2005generating}%
  \BibitemOpen
  \bibfield  {author} {\bibinfo {author} {\bibfnamefont {D.~E.}\ \bibnamefont
  {Knuth}},\ }\href@noop {} {\emph {\bibinfo {title} {The Art of Computer
  Programming, Volume 4, Fascicle 3: Generating All Combinations and
  Partitions}}}\ (\bibinfo  {publisher} {Addison-Wesley Professional},\
  \bibinfo {year} {2005})\BibitemShut {NoStop}%
\bibitem [{\citenamefont {Maslov}\ and\ \citenamefont
  {Zindorf}(2022)}]{maslov2022depth}%
  \BibitemOpen
  \bibfield  {author} {\bibinfo {author} {\bibfnamefont {D.}~\bibnamefont
  {Maslov}}\ and\ \bibinfo {author} {\bibfnamefont {B.}~\bibnamefont
  {Zindorf}},\ }\bibfield  {title} {\bibinfo {title} {Depth optimization of
  {CZ}, {CNOT}, and {C}lifford circuits},\ }\href@noop {} {\bibfield  {journal}
  {\bibinfo  {journal} {IEEE Transactions on Quantum Engineering}\ }\textbf
  {\bibinfo {volume} {3}},\ \bibinfo {pages} {1} (\bibinfo {year}
  {2022})}\BibitemShut {NoStop}%
\end{thebibliography}%

\appendix

\section{Optimizing the number of trials of Extended Prange} \label{sec:dynamic_programming}
Recall the optimization problem from \cref{sec:extendedPrange}
\begin{align}
    \mathrm{maximize}\,\,\,& \mathbb{P}\left(\sum_{i = 1}^m X_i \geq t\right) \label{eq:dynamicprogrammingobjective} \\
    \mathrm{subject\, to}\,\,\,&X_i \sim \mathrm{Ber}(P[s_i]) \nonumber \\ 
    &\sum_{i=1}^m s_i \leq B \nonumber \\
    &0 \leq s_i \leq b \nonumber\\
    &\mathbb{P}(1|s) \leq \mathbb{P}(1|s+1) \label{ineq:probability_constraint}
\end{align}

Eq~\eqref{eq:dynamicprogrammingobjective} is hard to solve exactly for large $m, b, $ and $B$, but using dynamic programming we can compute a tight approximation. To compute this approximation we note that eq \eqref{eq:dynamicprogrammingobjective} is invariant under permutations of $s_i$ which allows us to build the solution in any order we desire.

Before we describe our approach we note that any construction that chooses the $s_i$ values iteratively faces the problem of given $s_1 \cdots s_k$  choose the best $s_{k+1}$ which is a hard question since the sequence of $s_i$ induce a probability distribution and the set of probability distributions is intrinsically an unordered set. However, given our objective we can impose a few ordering relations with various degrees of effectiveness.

To solve the problem we employ a knapsack like dynamic programming approach. Dynamic programming is a constructive paradigm and since we are building an approximation algorithm the order we choose $s_i$ affects the quality of the result. Eq~\eqref{eq:probability_flow} shows that as we construct our solution the joint probability distribution induced by our choices of $s_i$ flows in the direction of increasing $\sum X_i$. We also note that the change in the joint distribution from $\mathbb{P}(\bullet|s_1\cdots s_k)$ to $\mathbb{P}(\bullet|s_1\cdots s_k s_{k+1})$ is controlled by $\mathbb{P}(1|s_{k+1})$ which from eq~\eqref{ineq:probability_constraint} is monotonic, from these two observation we find that it is advantageous to choose the $s_i$ values in descending order $s_1 \geq s_2 \geq \cdots \geq s_m$. The advantage of this ordering is twofold, the first is that after the first few choices of $s_i$ the joint probability distribution stabilizes from one choice to another,  the second is computational since it reduces the runtime of the algorithm.

\begin{multline}
    \mathbb{P}\left(\sum_{i = 1}^{k+1} X_i \geq t | s_1 \cdots s_k s_{k+1}\right) = \\
    \mathbb{P}\left(\sum_{i = 1}^k X_i \geq t | s_1 \cdots s_k\right) + \mathbb{P}\left(1 | s_{k+1}\right) \mathbb{P}\left(\sum_{i=1}^k X_i = t-1 | s_1 \cdots s_k\right) \label{eq:probability_flow} \\
\end{multline}

Our solution has a dynamic programming state of $(i, \mathrm{budget}, \mathrm{low})$ which means for the first $i$ variables and budget $\mathrm{budget} \leq B$ and $ \forall_{j \leq i} s_{j} \geq \mathrm{low}$ what is the best joint probability distribution. This state has only two transitions $(i, \mathrm{budget}, \mathrm{low}) \to (i, \mathrm{budget}, \mathrm{low} + 1)$ and $(i, \mathrm{budget}, \mathrm{low}) \to (i-1, \mathrm{budget} - \mathrm{low}, \mathrm{low})$. To compare between the up to two probability distributions we get from these transitions we used two comparison functions.

The first comparison function is eq~\eqref{eq:fast_sorting}, which is quick and for most test cases we generated was at most $5\%$ lower than the optimal result with some outliers with a bigger error. The second is eq~\eqref{eq:slow_sorting} which is slower since it involves computing the convolution between the probability distribution induced by our tuple $s_1\cdots s_k$ and the one induced by the sequence computed by eq~\eqref{eq:look_ahead}. The intuition behind $F$ which is computed using the classical dynamic programming knapsack algorithm is that we are doing a look ahead trying to estimate the best probability distribution given the remaining budget, but for computational reasons instead of solving the original problem we solve the related easier problem of maximizing the sum of probabilities, this gives us a better lower bound on the best final joint probability distribution starting from the given $S$. Note that, eq~\eqref{eq:fast_sorting} is a special case of eq~\eqref{eq:slow_sorting} since it can be interpreted as the convolution of the current probability distribution with the probability distribution $\mathbb{P}(0) = 1$.

\begin{equation}
    T(S) = (\mathbb{P}(\sum X_i \geq t | S), \mathbb{P}(\sum X_i = t-1 | S), \mathbb{P}(\sum X_i = t-2 | S), \cdots, \mathbb{P}(\sum X_i = 0 | S)) \label{eq:tuple_representation}
\end{equation}

\begin{equation}
    S_1 \geq S_2 \iff T(S_1) \geq T(S_2) \label{eq:fast_sorting}
\end{equation}

\begin{align}
    S_1 \geq S_2 &\iff T(S_1 + F(m - k, \mathrm{budget} - \sum_{s \in S_1} s_i,\ s_k^{(1)})) \geq T(S_2 + F(m - k, \mathrm{budget} - \sum_{s \in S_2} s_i,\ s_k^{(2)})) \label{eq:slow_sorting} \\
    S_1 + S_2 &= (s_1^{(1)}, \cdots, s_a^{(1)}) + (s_1^{(2)}, \cdots, s_b^{(2)}) \nonumber \\
    &= (s_1^{(1)}, \cdots, s_a^{(1)}, s_1^{(2)}, \cdots, s_b^{(2)}) \nonumber
\end{align}

\begin{align}
    F(m, \mathrm{budget}, \mathrm{low}) &= \mathrm{argmax}_{s_1 \cdots s_m} \sum_{i=1}^m P[s_i] \nonumber\\
    \mathrm{subject\, to}&\sum s_i \leq \mathrm{budget} \nonumber\\
    \mathrm{and\ }& s_i \geq \mathrm{low} \label{eq:look_ahead}
\end{align}

We were unable to find a case where the dynamic programming approach with the slow sorting relation eq~\eqref{eq:slow_sorting} disagrees with brute force solution. Although such cases theoretically exist we believe the error will be small. The time complexity of this solution is $\mathcal{O}(mbBt)$ with the fast sorting relation and $\mathcal{O}(mbBt^2)$ with the slow sorting relation.

\subsection{Upper Bound on the objective function} \label{subsection:objective_upper_bound}

The objective function of \eqref{eq:dynamicprogrammingobjective} is upper bounded through a direct application of the Hoeffding Inequality by

\begin{equation}
    \mathbb{P}\left(\sum_{i=1}^m X_i \geq t\right) \leq \exp\left(-2\frac{(t - \mathbb{E}[\sum_{i=1}^m X_i])^2}{m}\right)
\end{equation}

This upper bound is maximized by maximizing $\mathbb{E}[\sum_{i=1}^m X_i] = \sum_{i=1}^{m} P[s_i]$. The function $P[s]$ is derived from \eqref{eq:bounds_on_intersection} as 

\begin{equation}
    P[s] = \begin{cases} 
        \frac{1}{2} - \frac{1}{2^{k+1}} & s = 0 \\
        \frac{1}{2} & s = 1 \\
        \frac{1}{2} + \frac{1}{2^{k-s+2}} & 1 < s \leq k \\
        1 & s > k \\
    \end{cases} \label{eq:p_s_func}
\end{equation}

If we let $c_j$ be the number of times $s_i = j$ then maximizing the  expectation can be written as

\begin{align}
    \mathrm{maximize}\,\,\,& \mathbb{E}\left[\sum_{i=1}^m X_i\right] = \sum_{s=0}^{2k} c_s P[s]\\
    \mathrm{subject\, to}\,\,\,& \sum_{s=0}^{2k} c_s = m \nonumber\\
    \,&\sum_{s=0}^{2k} c_s s \leq B = bn = bRm\nonumber
\end{align}

Since the objective function is the expectation of the sum of $m$ Bernoulli variables, its value is a fraction of $m = hm, 0\leq h\leq 1$ which can be substituted into \eqref{eq:hoeffding} with $h$ controlling the exponent of \eqref{eq:asymptotic_scaling}. We can compute this expectation using a integer program or for large instances compute an upper bound by allowing $c_s$ to be real rather than integer and solving the linear program. For the specific $P[s]$ function in \eqref{eq:p_s_func}, a simple strategy that uses all of the budget to get $s=k+1$ also maximizes the objective. This can be seen as follows.
After spending $s$ points the expectation value is increased by $P[s] - P[0]$. To maximize our total expectation value we should spend in a way that maximizes our gain per unit cost, i.e. $(P[s]-P[0]) / s$. A brief analysis shows that this occurs for $s=k+1$, so the optimal solution to the LP should greedily allocate as much mass as possible onto $p_{k+1}$.

Using this greedy approach we get an upper bound on the expectation

\begin{align}
    \mathbb{E}\left[\sum_{i=1}^m X_i\right] &< P[k+1] c_{k+1} + P[0] c_0 \\
    &= \frac{B}{k+1} + \left(\frac{1}{2} - \frac{1}{2^{k+1}}\right) \left(m - \frac{B}{k+1}\right) \nonumber\\
    &< \frac{B}{k+1} + \frac{1}{2} \left(m - \frac{B}{k+1}\right) \nonumber\\
    &= \frac{m}{2} + \frac{1}{2} \frac{B}{k+1} \nonumber\\
    &< \frac{m}{2} + \frac{1}{2} \frac{B}{k} \nonumber\\
    &= \frac{m}{2} + \frac{1}{2} \frac{2kRm}{k} \nonumber\\
    &= \left(\frac{1}{2} + R\right)m\nonumber
\end{align}

\section{Sparse Dicke State Preparation} \label{sec:sparse_dicke_state}
The Dicke state $\ket{D^{m}_{k}}$ is an equal weight superposition of all $m$-qubit states with Hamming weight $k$ (i.e. all strings of length $m$ with exactly $k$ ones over a binary alphabet). As shown in \cref{sec:dqi_quantum_circuit}, to prepare the syndrome register for DQI, we need to prepare the state
$$
    \sum_{k=0}^{\ell}w_{k}\ket{D^{m}_{k}}
$$
where $w_{0}, w_{1}, \dots, w_{\ell}$ where $\ell=n/2$, are classically known coefficients and $\ket{D^m_k}$ is the Dicke state
$$
    \ket{D^m_k} = \frac{1}{\sqrt{{m}\choose{k}}}\sum_{|y|=k}\ket{y}
$$

In the context of DQI, the Dicke state encodes a superposition over all possible ${m \choose k}$ error locations for $0 \leq k \leq \ell$ errors in a codeword of length $m$. This encoding uses $m$ qubits. However, when $m \gg \ell = n/2$, we can define a sparse encoding that captures the same information using $\ell \times \log_2{m}$ qubits, and thus be more space efficient. 

A Sparse Dicke state $\ket{SD^{m}_{k}}$ can be defined as a uniform superposition over all $k-\text{combinations}$ of indices $[1, 2, \dots, m]$, stored using $k$ registers each of size $b=\lceil\log_2{(m + 1)}\rceil$, using a total of $k\cdot b$ qubits.

$$
    \ket{SD^m_k} = \frac{1}{\sqrt{\binom{m}{k}}} \sum_{1 \le c_1 < c_2 < \dots < c_k \le m} \ket{c_1}_b \ket{c_2}_b \cdots \ket{c_k}_b
$$

First, we describe a way to prepare the Sparse Dicke states $\ket{SD^m_k}$ using the Combinatorial Number System \cite{knuth2005generating}. Then, we give two different constructions for quantum circuits to unrank combinations, one that minimizes the asymptotic complexity and other that minimizes the constant factors. 

\subsection{The unranking strategy}

The Combinatorial Number System \cite{knuth2005generating} defines an ordering over all $\binom{m}{k}$ $k-$combinations such that a combination $[c_k, c_{k-1}, \dots, c_1]$ where $m - 1 \geq c_k > c_{k-1} > \dots > c_1\geq 0$ corresponds to a unique rank $0\leq r < \binom{m}{k}$ given by the bijection
$$
    r = \sum_{j=1}^{k}\binom{c_j}{j}
$$

Using this bijection, the state preparation strategy involves two steps:
\begin{enumerate}
    \item Prepare a uniform superposition of all ranks $r \in \{0, \dots, \binom{m}{k}-1\}$ using $\mathcal{O}(k\cdot b)$ gates \cite{Babbush_2018}. Here we use the fact that $\binom{m}{k} \leq m^k$, so $\log_2{\binom{m}{k}} \leq k\log_2{m} = k\cdot  b$. The state of the system is given as follows.
    $$
        \frac{1}{\sqrt{\binom{m}{k}}} \sum_{r=0}^{\binom{m}{k} - 1} \ket{r}_{k\cdot b}
    $$
    \item Apply a unitary transformation $U_\text{Unrank}$ that maps the rank $r$ to its corresponding combination $C(r)$:
    $$
    \frac{1}{\sqrt{\binom{m}{k}}} \sum_{r} \ket{r}_{k\cdot b} \xrightarrow{U_\text{Unrank}} \frac{1}{\sqrt{\binom{m}{k}}} \sum_{r} \ket{C(r)}_{k\cdot b} = \ket{SD^m_k}
    $$
\end{enumerate}
The efficiency of the state preparation is determined by the complexity of $U_\text{Unrank}$.

\subsection{Divide and Conquer Unranking for \texorpdfstring{$\bigOtilde(m)$}{Õ(m)} Dicke State Preparation} \label{sec:fast_dicke_state}
We first describe a Divide and Conquer strategy for fast unranking. 
Given inputs $(m, k, r)$, we split the $m$ elements into two halves, $M_1$ and $M_2$, of sizes $m_1 = \lfloor m/2 \rfloor$ and $m_2 = \lceil m/2 \rceil$, and then calculate:
\begin{enumerate}
    \item The number of elements $k_1$ and $k_2$ that the combination $C(r)$ selects from the first ($M_1$) and second halves ($M_2$) respectively, such that $k_1 + k_2 = k$.
    \item The residual ranks $r_1$ and $r_2$, such that $C_{m, k}(r) = C_{m_1, k_1}(r_1) \| C_{m_2, k_2}(r_2)$.
\end{enumerate}

The unranking problem can then be solved by recursively solving the left and right subproblems for $(m_1, k_1, r_1)$ and $(m_2, k_2, r_2)$. If the work done at one level of recursion is $W(m, k)$, then the overall complexity of $U_\text{Unrank}$ follows the following recurrence relation:
$$
T(m, k) = T(m/2, k_1) + T(m/2, k-k_1) + W(m, k)
$$
If we can show that $W(m, k) = \bigOtilde(m)$, then by Master Theorem (Case 2), the total complexity is $T(m, k) = \bigOtilde(m \log m)$, which is $\bigOtilde(m)$.

The number of ways to choose $k$ items such that exactly $i$ items come from $M_1$ is given by the Hypergeometric distribution:
$$
H(i) = \binom{m_1}{i} \binom{m_2}{k-i}
$$

To find the correct $k_1$ for rank $r$, we seek the value such that the prefix sum of $H(i)$ crosses $r$:
$$
PS(k_1) = \sum_{i=0}^{k_1-1} H(i) \le r < PS(k_1+1)
$$
We can find $k_1$ using a binary search in $O(\log k)$ steps, provided we can efficiently calculate the prefix sum $PS(x)$.
Once $k_1$ is found, we calculate the residual ranks $r_1$ and $r_2$ corresponding to the left and right subproblems (via division and modulo operations on $r - PS(k_1)$), and recursively solve them.

The efficiency of this approach hinges on the fast calculation of the Hypergeometric prefix sum $PS(x)$. The numbers involved (ranks and binomial coefficients) have bit length $L = O(k \log m)$. A naive summation of $O(k)$ terms is too slow.
We use a technique called \emph{Binary Splitting}. This method is effective for evaluating series where the ratio of consecutive terms $R(i) = H(i+1)/H(i)$ is a simple rational function of $i$:
$$
R(i) = \frac{(m_1-i)(k-i)}{(i+1)(m_2 - k + i + 1)}
$$
The prefix sum can be written as $PS(x) = H(0) \cdot (1 + R(0) + R(0)R(1) + \dots)$. Binary Splitting recursively computes this expression in a balanced tree structure. Crucially, it leverages fast integer multiplication algorithms (e.g., Schönhage–Strassen), which multiply $L$-bit numbers in time $\bigOtilde(L)$. The total time complexity for calculating $PS(x)$ using Binary Splitting is thus $\bigOtilde(L \log k) = \bigOtilde(k \log m)$.

The divide and conquer algorithm can be implemented as a reversible quantum circuit $U_\text{Unrank}$.

\begin{enumerate}
    \item \textbf{Reversible Fast Arithmetic:} We require efficient reversible quantum circuits for fast integer multiplication and division that maintain the $\bigOtilde(L)$ complexity for arithmetic over $L$ bit integers. 
    \item \textbf{Reversible Binary Splitting ($U_{BS}$):} The Binary Splitting procedure is implemented reversibly using the fast arithmetic circuits. The gate complexity of $U_{BS}$ is $\bigOtilde(k \log m)$.
    \item \textbf{Reversible Search:} The binary search for $k_1$ is implemented using the reversible Bit-wise Quantum Search strategy. This requires $O(\log k)$ coherent calls to $U_{BS}$.
\end{enumerate}

The work done at one level of the recursion, $W(m, k)$, is dominated by the reversible search:
$$
W(m, k) = O(\log k) \cdot \text{Cost}(U_{BS}) = \bigOtilde(k \log m) = \bigOtilde(m)
$$

\subsection{Iterative Unranking for low constant factor Dicke State Preparation}
The divide and conquer strategy described has optimal asymptotic complexity but would not be ideal for a constant factor analysis. Here, we describe a simple greedy algorithm to find the $k-$combination corresponding to rank $r$ with small constant factors: take $c_k$ maximal with $\binom{c_k}{k} \leq r$, then take $c_{k - 1}$ maximal with $\binom{c_{k-1}}{k-1} \leq r -  \binom{c_k}{k}$, and so on. We can translate this greedy algorithm into a reversible quantum circuit to prepare Sparse Dicke states using $2\cdot k\cdot b + b$ qubits and $\mathcal{O}(m.k + k^2b^2)$ Toffoli gates as follows:

\begin{enumerate}
    \item The unitary $U_\text{Unrank}(m, k)$ can be defined as a product of $k$ unitaries, each of which iteratively finds the $j$th coefficient $c_j$ in the sequence $C(r)$ as follows:
    $$ 
        U_\text{Unrank}(m, k) = U_{\text{search}}(m, 1) \cdot U_{\text{search}}(m, 2) \dots U_{\text{search}}(m, k - 1) \cdot U_{\text{search}}(m, k)
    $$
    Here $U_{\text{search}}(m, j)$ can be defined as:
    $$
        U_{\text{search}}(m, j) \ket{r}_{j\log_2{m}}\ket{0}_{\log_2{m}} \rightarrow \ket{r - \binom{c_j}{j}}_{(j - 1)\log_2{m}}\ket{c_j}_{\log_2{m}}
    $$

    \item The unitary $U_\text{search}(m, j)$ can be implemented using a bitwise binary search strategy, where we determine the largest $c_{j}$ such that $\binom{c_j}{j} \leq r$, iterating over $\log_{2}{m}$ bits of $c_{j}$ from most significant bit to least significant bit, and for each $\log_{2}{m} \geq i \geq 0$, set $c_j = c_j + 2^i $ if $\binom{c_j + 2^i}{j} \leq r$.
    A reversible quantum circuit would therefore require $\log_{2}(m)$ calls to a sparse QROM, the $i$th of which loads $2^i$ binomial coefficients of the form $\binom{c_j + 2^i}{j}$. The total Toffoli cost of the sparse QROMs across all iterations is therefore $k.m$. Once the Binomial coefficients are loaded using the QROM, we also need a $k.\log_2{m}$ comparator for each of the $\log_2{m}$ bits of $c_{j}$. The total Toffoli cost of the comparators across all iterations is therefore $\mathcal{O}((k\log_2{m})^2)$.
\end{enumerate}

\section{Improved arithmetic circuits for binary extension fields}\label{sec:improved_gf2_arithmetic}

\begin{table}[H]
    \centering
    \begin{tabular}{|c||c|c||c|c|}
    \hline
    \hline
    $m$ & Toffoli & CNOT  & PCTOF & CNOT\\
   \hline
   \hline
   10 & 39 & 738 & 39 & 0 \\
   \hline
   11 & 47 & 1278 & 46 & 0\\
   \hline
   12 & 51 & 1506 & 51 & 0\\
   \hline
   \hline
\end{tabular}
\caption{Resource counts for GF$(2^m)$ multiplication for $m=10, 11, 12$, expressed as gate counts over two libraries, \{Toffoli, CNOT\} and \{Parity Control Toffoli, CNOT\}.}
\label{tbl:gf101112mult}
\end{table}

\subsection{Parity Control Toffoli $PCTOF$ and Parity CNOT $PCNOT$}

The parity control Toffoli PCTOF is equivalent to a multi-target Toffoli gate with controls being the parity of a subset of qubits. The PCNOT gate is a CNOT gate that computes the XOR of its controls and is equivalent to a series of CNOTs with the same target. The main advantages of the PCNOT gate is that in the presence of enough ancillae, it can be performed in a single cycle using lattice surgery independent of the number of controls. 

When the two control sets of a PCTOF are disjoint, as is the case of GF$(2^m)$ multiplication, the two control parities can be computed in a single cycle. \cref{fig:pctof_pcnot} shows the first PCTOF of \cref{fig:mbuc_gf3} expanded in terms of PCNOT and Toffoli and in terms of CNOT and Toffoli gates. 

\begin{figure}
    \centering
    \includegraphics[width=0.5\linewidth]{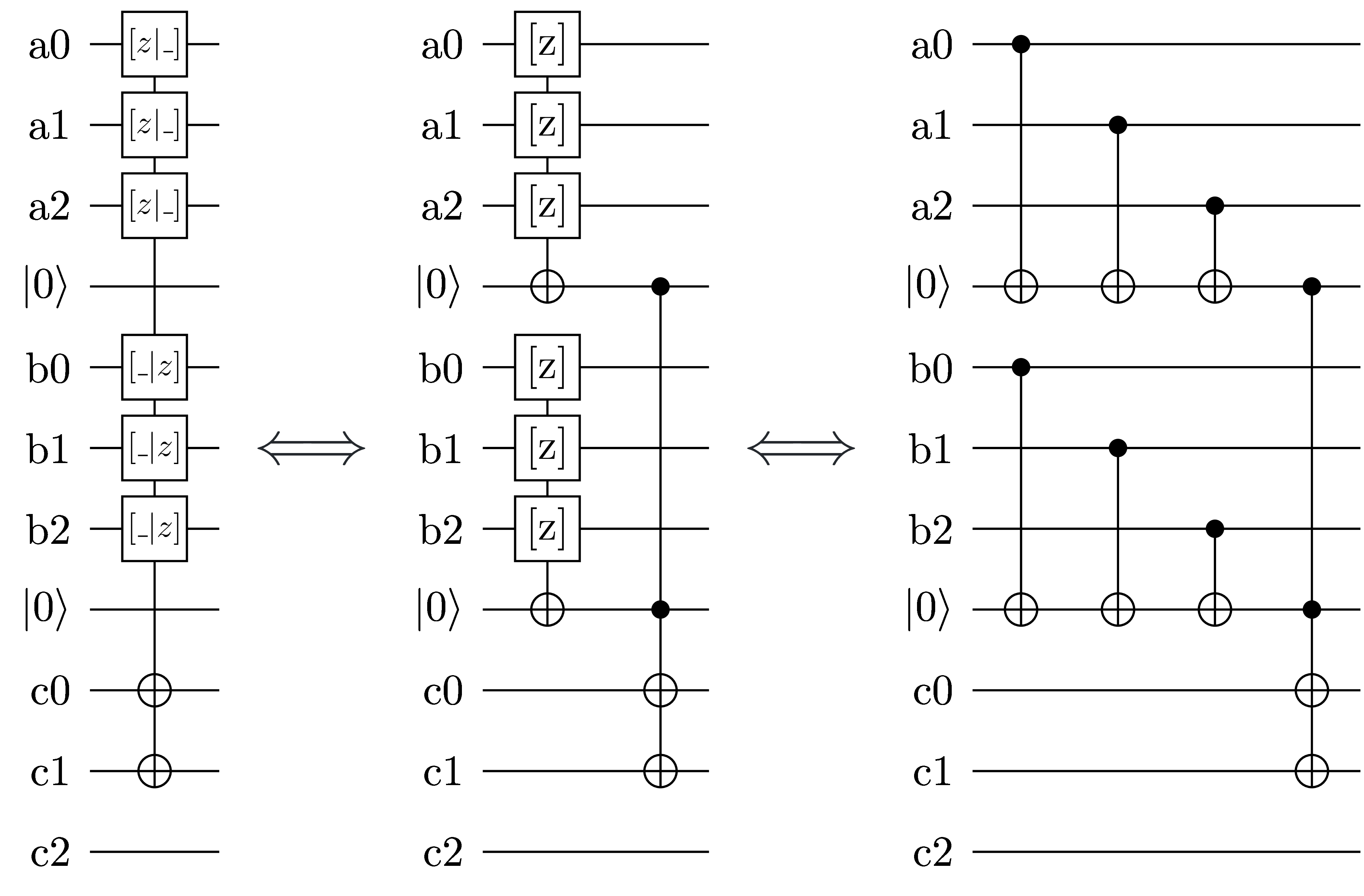}
    \caption{PCTOF in terms of a Toffoli and PCNOTs and then expanded in terms of a Toffoli and CNOTs. Note that we omitted the ancilla cleanup part.}
    \label{fig:pctof_pcnot}
\end{figure}

\subsection{\texorpdfstring{$\GF(2^m)$}{GF(2^m)} arithmetic circuits}
Our implementation of the DQI algorithm relies on arithmetic circuits over GF field sizes of $2^{10}$ through $2^{12}$.  The bottleneck is GF multiplication. 

To implement the respective multiplication (and division) circuits, we rely on \cite{maslov2025asymptotic}. 
This synthesis algorithm develops a modification of Karatsuba multiplication over a carefully selected irreducible polynomial, along with local optimizations using templates, allowing for a reduction in both Toffoli and CNOT gate counts compared to the state-of-the-art straight Karatsuba multiplication \cite{vanhoof2020spaceefficientquantummultiplicationpolynomials}. 
We tune the implementation \cite{maslov2025asymptotic} to focus on optimizing an $\infty$-to-1 weighted sum of Toffoli and CNOT gate counts to utilize the PCTOF gates better. 
To get constructions made of the PCTOF and CNOT gates, we start from the Toffoli and CNOT circuits from \cite{maslov2025asymptotic} and commute the CNOT gates to the left while updating the sets of the PCTOF gates, resulting in circuits that start with three disjoint CNOT stages acting on the two input registers and the resulting register followed by a circuit made of the PCTOF gates.  
The first two CNOT circuits are equivalent to the identity since we do not modify the input registers and the CNOT circuit acting on the target can be ignored since the target is always initialized in the zero state.  The resulting optimized circuits are summarized in \cref{tbl:gf101112mult}.  
\cref{fig:mbuc_gf3} shows the $\GF(2^3)$ multiplication circuit with irreducible polynomial $x^3+x+1$ in terms of PCTOF gates and explicit circuit files are available as a part of the Zenodo upload of this paper. 

In addition, we developed a circuit optimization method applicable to the kinds of circuits with PCTOF gates considered.  The optimization algorithm relies on two observations: first, any two PCTOF gates commute, and second, each PCTOF gate computes a Boolean polynomial of degree 2 and then EXORs it onto a separate output register.  This means that a single PCTOF gate can be thought of as a vector in the $m^2$-dimensional Boolean space, where each coordinate corresponds to a Boolean product of some two variables, and circuits with $k$ PCTOF gates (with targets on the bottom register such as in our case of GF multiplication circuits) can be thought of as $k \,{\times}\, m^2$ matrices.  This means that the matrix rank determines the minimal number of PCTOF gates necessary, in a given set that implements the desired functionality.  We implemented this optimization algorithm and found a small optimization compared to what is directly offered by \cite{maslov2025asymptotic} (compare Toffoli count to the PCTOF count in \cref{tbl:gf101112mult}).  We suspect that the improvement is small due to the large number of optimizations that already went into \cite{maslov2025asymptotic}.

\begin{figure}
    \centering
    \includegraphics[width=0.5\linewidth]{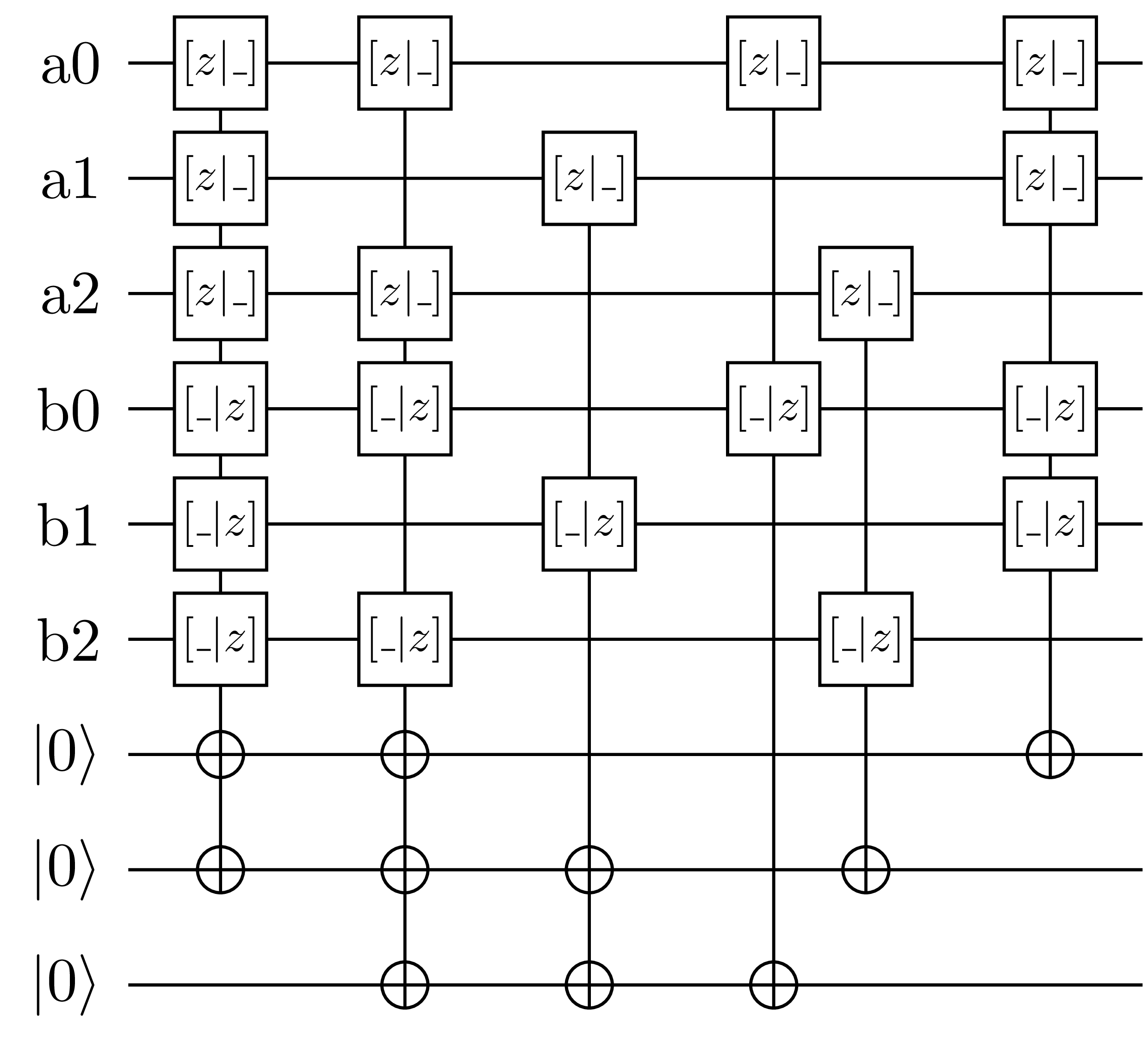}
    \caption{$\GF(2^3)$ multiplication with irreducible polynomial $ x^3+x+1$ in terms of PCTOF gates.}
    \label{fig:mbuc_gf3}
\end{figure}

\subsection{Measurement-based Uncomputation for GF2 Multiplication.}
The GF$(2^m)$ multiplication operation is an out-of-place operation that applies the transformation $\ket{x}\ket{y}\ket{0} \rightarrow \ket{x} \ket{y} \ket{xy}$. The uncomputation of this operation can be done using only measurements and $CZ$ gates.

The uncomputation starts by measuring the target register in the $X$ basis. This is equivalent to applying the $H^{\otimes m}$ to the target register and then measuring in the computational basis. The $H^{\otimes m}$ operations transform the state to $\ket{x} \ket{y} \ket{xy} = \sum a_{i,j} \ket{i} \ket{j} \ket{ij} \xrightarrow[]{H^{\otimes m}} \frac{1}{\sqrt{2^m}}\sum_{c=0}^{2^m-1} \big \{ \sum (-1)^{h(c, ij)} \ket{i}\ket{j} \big \} \ket{c}$ where $h(a, b)$ is the Hamming weight of $a \wedge b$. After the measurement, we end up with a classical bitstring $c$ and phase flips on some of the coefficients of the input superposition $\sum (-1)^{h(c, ij)} a_{i,j} \ket{i} \ket{j}$.

While it seems that we need to know the product to correct the phase, the linearity of $h$ and symmetry of $CZ$ allow us to correct the phase with only $CZ$s. For example if $c_t{=}1$ then we need to apply $CZ(x_i, y_j)$ for all $(i, j)$ pairs such that $x^{i+j} \mod p(x)$ has $x^t$ where $p(x)$ is the irreducible polynomial. This leads to the requirement to implement a certain $CZ$ gate circuit.  Note that any such circuit can be implemented in depth $\lfloor m/2 + 0.4993{\cdot}\log^2(m) + 3.0191{\cdot}\log(m) - 10.9139\rfloor$, using no more than $m^2/4 + O(m\log^2(m))$ Clifford gates \cite{maslov2022depth} or $O(m^2/\log(m))$ gates asymptotically.

Our implementation of measurement-based uncomputation of GF2 multiplication is available in Qualtran \cite{harrigan2024expressinganalyzingquantumalgorithms} as \verb|GF2MulMBUC| and \cref{fig:mbuc} shows the uncomputation of GF2 multiplication for $m{=}3$ and irreducible polynomial $x^3 + x + 1$.

\begin{figure}
    \centering
    \includegraphics[width=0.5\linewidth]{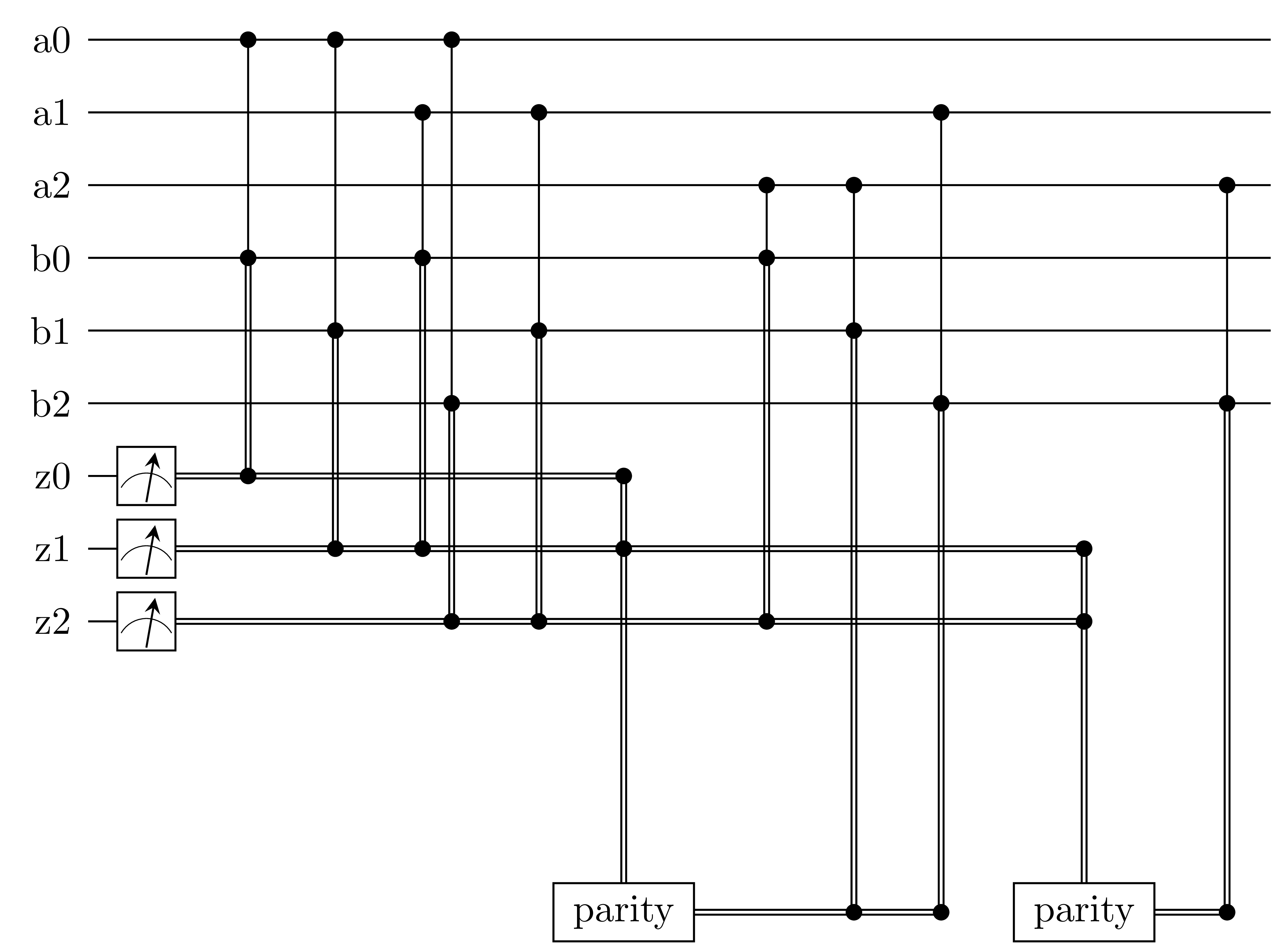}
    \caption{Construction of measurement based uncomputation of GF2 multiplication for irreducible polynomial $x^3 + x + 1$. Note that the measurement is in the X basis.}
    \label{fig:mbuc}
\end{figure}

\section{Bounding the number of iterations for Synchronized Reversible Polynomial EEA}\label{sec:improved_zalka_eea_iterations}
We analyze the worst-case number of iterations required for the synchronized, reversible implementation of the Polynomial Extended Euclidean Algorithm (PEEA) presented in \cref{sec:zalka_eea_optimized}. This implementation utilizes shared registers and a 4-stage synchronized architecture designed to allow computation branches in a superposition to proceed efficiently at their own pace.

Let the input polynomials be $A$ and $B$ over a finite field $\F_q$, with $n = \deg(A)$ and $\deg(B) < n$. We write $R_{-1} = A$ and $R_0 = B$, so that iteration $i \geq 1$ divides $R_{i-2}$ by $R_{i-1}$, and use the convention $\deg(0) = -1$.
\begin{definition}[Iteration Dynamics]
For iteration $i$ (which computes quotient $Q_i$, remainder $R_i$, and cofactor $U_i$):
\begin{itemize}
    \item $d_i = \deg(Q_i)$. We define $d_0=1$ (see Lemma \ref{lem:cost_per_iteration_opt}). Since $\deg(B) < \deg(A)$, $d_i \ge 1$ for $i \ge 1$.
    \item $s_i = \deg(R_{i-1}) - \deg(R_i)$: The degree reduction of the remainder. $s_i \ge 1$.
\end{itemize}
Let $k$ be the total number of iterations. Let $D_k = \sum_{i=1}^k d_i$ and $S_k = \sum_{i=1}^k s_i$.
\end{definition}
\begin{lemma}[PEEA Constraints]\label{lem:peea_constraints_opt}
The variables are constrained by $k \le D_k$ and $k \le S_k$. Moreover, $D_k = n - \deg(R_{k-1})$ and $S_k = \deg(B) - \deg(R_k)$. In particular, for Full PEEA (coprime inputs, unequal degrees), $D_k = n$ and $S_k \le n$.
\end{lemma}
We analyze the time complexity $T_i$ of iteration $i$ based on the 5 steps described in \cref{fig:optimized_zalka_eea}.

\begin{lemma}[Cost Per Iteration]\label{lem:cost_per_iteration_opt}
The number of cycles $T_i$ required for iteration $i$ in the optimized implementation is $T_i = 2d_i + s_i + d_{i-1} + 2$.
\begin{proof}
We analyze the time required for each stage in the 4-stage cycle, corresponding to Steps 1-4 in Figure 1:
\begin{enumerate}
    \item \textbf{Stage 1} (Step-1, Poly Long Division): Time taken is the length of $Q_i$ (number of terms), which is $d_i+1$.
    \item \textbf{Stage 2} (Step-2, Remainder Normalization): Time taken is the total degree reduction $s_i$.
    \item \textbf{Stage 3} (Step-3, B\'ezout Cofactor Alignment): Time taken is the degree growth from the prior iteration, $d_{i-1}$. In the first iteration, aligning the initial cofactors $U_{-1} = 0$ and $U_0 = 1$ takes one cycle, which is why we set $d_0 = 1$.
    \item \textbf{Stage 4} (Step-4, B\'ezout Cofactor Update): Time taken is the length of $Q_i$, which is $d_i+1$. The register swap (Step-5) occurs concurrently during the final cycle of this stage, incurring no extra time.
\end{enumerate}

Summing the costs for iteration $i$:
\begin{align*}
T_i &= (d_i+1) + s_i + d_{i-1} + (d_i+1) \\
    &= 2d_i + s_i + d_{i-1} + 2.
\end{align*}
\end{proof}
\end{lemma}

\begin{lemma}[Total Cost Function]\label{lem:total_cost_function_opt}
The total number of cycles $T$ required to complete $k$ iterations is $T = 3D_k + S_k - d_k + 2k + 1$.
\begin{proof}
The total time is the sum of $T_i$ from $i=1$ to $k$.
\begin{align*}
T &= \sum_{i=1}^{k} T_i = \sum_{i=1}^{k} (2d_i + s_i + d_{i-1} + 2) \\
  &= 2\sum_{i=1}^{k} d_i + \sum_{i=1}^{k} s_i + \sum_{i=1}^{k} d_{i-1} + 2k.
\end{align*}
We evaluate the third term: $\sum_{i=1}^{k} d_{i-1} = d_0 + d_1 + \dots + d_{k-1}$. Since $d_0=1$, this equals $D_k - d_k + 1$.
\[ T = 2D_k + S_k + (D_k - d_k + 1) + 2k = 3D_k + S_k - d_k + 2k + 1. \]
\end{proof}
\end{lemma}

\subsection{Worst-Case Analysis: Full PEEA (GCD)}

We analyze the maximum cycles required for the Full PEEA, assuming worst-case coprime inputs with degrees $(n, n-1)$.

\begin{theorem}[Cycle Bound for Full PEEA]\label{thm:full_peea_bound_opt}
The maximum number of cycles required for the optimized synchronized Full PEEA on inputs of unequal degrees (max degree $n$) is $T_{max} = 6n$.
\begin{proof}
We maximize the cost function $T = 3D_k + S_k - d_k + 2k + 1$.
We use the constraints for Full PEEA (Lemma \ref{lem:peea_constraints_opt}): $D_k = n$ and $S_k \le n$.

Substituting these values:
\[ T \le 3n + n - d_k + 2k + 1 = 4n - d_k + 2k + 1. \]
To maximize $T$, we maximize the number of iterations $k$ and minimize the final quotient degree $d_k$.
The constraints are $k \le n$ and $d_k \ge 1$.
We set $k_{max}=n$ and $d_{k,min}=1$.
\[ T_{max} = 4n - 1 + 2n + 1 = 6n. \]
This worst case is achieved by the generalized Fibonacci sequence where $d_i=1$ and $s_i=1$ for all $i=1 \dots n$.
\end{proof}
\end{theorem}

\subsection{Worst-Case Analysis: Half PEEA (Reed-Solomon Decoding)}

The Half-PEEA halts after the first iteration $k$ in which $\deg(R_k) < \lfloor n/2 \rfloor$. Since no division follows this iteration, its remainder normalization (Stage 2) is stopped as soon as the length of the remainder register reaches $\lfloor n/2 \rfloor$ (the parameter \texttt{n\_B\_min} in \cref{fig:optimized_zalka_eea}), even if further leading coefficients of $R_k$ vanish. Stage 2 of the last iteration then takes $\deg(R_{k-1}) - \lfloor n/2 \rfloor + 1$ cycles, so that $S_k$ in Lemma \ref{lem:total_cost_function_opt} is replaced by $S_k' = \deg(B) - \lfloor n/2 \rfloor + 1 \le \lceil n/2 \rceil$.

\begin{theorem}[Cycle Bound for Half PEEA]\label{thm:half_peea_bound_opt}
The maximum number of cycles required for the optimized synchronized Half-PEEA is $T_{max} = 6\lceil n/2 \rceil$.
\begin{proof}
We maximize the cost function $T = 3D_k + S_k' - d_k + 2k + 1$.
Since $\deg(R_{k-1}) \ge \lfloor n/2 \rfloor$, Lemma \ref{lem:peea_constraints_opt} gives $k \le D_k = n - \deg(R_{k-1}) \le \lceil n/2 \rceil$. Together with $S_k' \le \lceil n/2 \rceil$ and $d_k \ge 1$, this yields
\[ T \le 3\lceil n/2 \rceil + \lceil n/2 \rceil - 1 + 2\lceil n/2 \rceil + 1 = 6\lceil n/2 \rceil. \]
This worst case is achieved when $\deg(B) = n - 1$ and $d_i=1$ for all $i=1 \dots \lceil n/2 \rceil$.
\end{proof}
\end{theorem}

\section{Deferred proofs of upper bounds on affine intersections of Maiorana-McFarland target sets}
\label{sec:proofs_of_bounds_on_a_intersect_s}

Here, we state and prove a few lemmas that are omitted from \cref{sec:bounds_on_a_intersect_s}. Together the proofs here and in \cref{sec:bounds_on_a_intersect_s} establish the upper bounds on $|A\cap S_k|$ given in \eqref{eq:upper_bound_on_a_intersection_sk}.

\begin{lemma}[Upper bound on $|A\cap R_k|$ and $|A\cap S_k|$ when $\dim A < k$]
    \label{lm:case_1}
    \,\newline
    Let $k$ be a positive integer, $d$ a non-negative integer and $A$ a $d$-dimensional affine subspace of $\mathbb{F}_2^{2k}$. Then
    \begin{align}
        |A\cap R_k| \leq 2^d, \quad
        |A\cap S_k| \leq 2^d.
    \end{align}
\end{lemma}
\begin{proof}
    $|A\cap R_k| \leq |A| = 2^d$ and $|A\cap S_k| \leq |A| = 2^d$.
\end{proof}

\begin{lemma}[Size of $R_k$ and $S_k$]
    \label{lm:rk_and_sk}
    \begin{align}
        |R_k| = 2^{2k-1} + 2^{k-1}, \\
        |S_k| = 2^{2k-1} - 2^{k-1}.
    \end{align}
\end{lemma}
\begin{proof}
We can write
\begin{align}
    R_k = \bigsqcup_{x\in\mathbb{F}_2^k} \{x\}\times\{x\}^\perp
\end{align}
where $\{x\}^\perp \subset\mathbb{F}_2^k$ is the set of $k$-bit strings orthogonal to $x\in\mathbb{F}_2^k$ under the standard dot product $\langle.,.\rangle$. But
\begin{align}\label{eq:size_of_a_perp}
    |\{x\}^\perp| = \begin{cases}
        2^k & \text{if} \quad x=0\in\mathbb{F}_2^k \\
        2^{k-1} & \text{otherwise}
    \end{cases}
\end{align}
so $|R_k| = 2^{2k-1} + 2^{k-1}$ and $|S_k| = 2^{2k-1} - 2^{k-1}$.
\end{proof}

\begin{corollary}[Upper bound on $|A\cap R_k|$ and $|A\cap S_k|$ when $\dim A=2k$]
    \label{cr:case_4}
    \,\newline
    Let $k$ be a positive integer and $A = \mathbb{F}_2^{2k}$ the $2k$-dimensional affine subspace of $\mathbb{F}_2^{2k}$. Then
    \begin{align}
        |A \cap R_k| = 2^{2k-1} + 2^{k-1}, \\
        |A \cap S_k| = 2^{2k-1} - 2^{k-1}.
    \end{align}
\end{corollary}
\begin{proof}
    Immediate consequence of Lemma \ref{lm:rk_and_sk}.
\end{proof}

\begin{lemma}[Upper bound on $|A\cap R_k|$ and $|A\cap S_k|$ when $\dim A=2k-1$]
    \label{lm:case_3}
    \,\newline
    Let $k$ be a positive integer and let $A \subset \mathbb{F}_2^{2k}$ be a $(2k-1)$-dimensional affine subspace of $\mathbb{F}_2^{2k}$. Then
    \begin{align}
        |A\cap R_k| & \leq 2^{2k-2} + 2^{k-1} \\
        |A\cap S_k| & \leq 2^{2k-2}.
    \end{align}
\end{lemma}
\begin{proof}
If $k=1$, then $|A|=2$ and $|S_1|=1$, so $|A \cap R_1| \leq 2$ and $|A \cap S_1|\leq 1$. Assume $k>1$. The affine space $A$ contains $2^{2k-1}$ bit strings, but there are only $2^{2k-2}$ bit strings of each possible suffix. Therefore, $F$ has at least two elements. Thus, by Remark \ref{rm:five_cases} we have three cases to consider
\begin{align}
    \dim F = 1 &\quad\land\quad 11\notin F\\
    \dim F = 1 &\quad\land\quad 11\in F\\
    \dim F = 2 &.
\end{align}
We prove the first two cases directly and the third one by induction.

If $\dim F = 1$, then Lemma \ref{lm:recursive_structure_formulas} implies that $\dim W'=2k-2$ and therefore
\begin{align}
    A = \mathbb{F}_2^{2k-2} \otimes F.
\end{align}
If $11\notin F$, then
\begin{align}
    A \cap R_k &= R_{k-1} \otimes F \\
    A \cap S_k &= S_{k-1} \otimes F
\end{align}
and using Lemma \ref{lm:rk_and_sk}
\begin{align}
    |A \cap R_k| &= |F|\cdot|R_{k-1}| = 2\cdot(2^{2k-3} + 2^{k-2}) = 2^{2k-2} + 2^{k-1} \\
    |A \cap S_k| &= |F|\cdot|S_{k-1}| = 2\cdot(2^{2k-3} - 2^{k-2}) < 2^{2k-2}.
\end{align}
If $F=\{\sigma,11\}$ with $\sigma\ne 11$, then
\begin{align}
    A \cap R_k &= (R_{k-1} \otimes \sigma) \sqcup (S_{k-1} \otimes 11)\\
    A \cap S_k &= (S_{k-1} \otimes \sigma) \sqcup (R_{k-1} \otimes 11)
\end{align}
and using $R_{k-1} \sqcup S_{k-1} = \mathbb{F}_2^{2k-2}$, we find
\begin{align}
    |A \cap R_k| &= |\mathbb{F}_2^{2k-2}| < 2^{2k-2} + 2^{k-1} \\
    |A \cap S_k| &= |\mathbb{F}_2^{2k-2}| = 2^{2k-2}.
\end{align}
We now prove the case $\dim F = 2$ by induction. We established the base case at the opening of the proof. Assume now that
\begin{align}
    |A\cap R_{k-1}| & \leq 2^{2k-4} + 2^{k-2} \\
    |A\cap S_{k-1}| & \leq 2^{2k-4}.
\end{align}
By Lemma \ref{lm:recursive_structure_formulas}, the four affine subspaces $A_\sigma'$ of $\mathbb{F}_2^{2k-2}$ with $\sigma\in F$ in \eqref{eq:structure_of_a_cap_r} and \eqref{eq:structure_of_a_cap_s} arise as translations of the same linear subspace $W'$ of dimension $2k-3$, so at most two of them are distinct
\begin{align}
    A_1' := A_{11}'=A_\rho',\quad A_2':=A_\sigma'=A_\origtau'.
\end{align}
The equation $R_{k-1}\sqcup S_{k-1}=\mathbb{F}_2^{2k-2}$ implies then that
\begin{align}
    |A \cap R_k| &= |A_1'| + 2\cdot |A_2'\cap R_{k-1}| \\
    |A \cap S_k| &= |A_1'| + 2\cdot |A_2'\cap S_{k-1}|
\end{align}
and using our inductive hypothesis, we obtain
\begin{align}
    |A \cap R_k| & \leq 2^{2k-3} + 2\cdot(2^{2k-4} + 2^{k-2}) = 2^{2k-2} + 2^{k-1} \\
    |A \cap S_k| & \leq 2^{2k-3} + 2\cdot 2^{2k-4} = 2^{2k-2}
\end{align}
completing the proof of the Lemma.
\end{proof}

\begin{lemma}[Proxy expression for mixed recursive structure formulas with four terms]
    \label{lm:hard_rsf_4terms}
    \,\newline
    Let $k$ be a positive integer and $a,b,c\in\mathbb{F}_2^{2k}$. For any four disjoint affine subspaces of the form
    \begin{align}
        A_{00} &= a + W \\
        A_{01} &= a + b + W \\
        A_{10} &= a + c + W \\
        A_{11} &= a + b + c + W
    \end{align}
    arising as translations of the same linear subspace $W\subset\mathbb{F}_2^{2k}$, there exist two disjoint sets $B_1$ and $B_2$ each of which is either empty or affine and such that
    \begin{align}
        |A_{00}\cap S_k| + |A_{01}\cap S_k| + |A_{10}\cap S_k| + |A_{11}\cap R_k| = |B_1| + 2\cdot|B_2\cap S_k|
    \end{align}
    and $|B_1| + |B_2| = 2\cdot|W|$.
\end{lemma}
\begin{proof}
First note that for any $x\in W$
\begin{align}
    \mathbb{1}_{S_k}(a+x) &= \mathbb{1}_{S_k}(a) + \omega(a, x) + \mathbb{1}_{S_k}(x) \\
    \mathbb{1}_{S_k}(a+b+x) & = \mathbb{1}_{S_k}(a) + \mathbb{1}_{S_k}(b) + \omega(a, b) + \omega(a, x) + \omega(b, x) + \mathbb{1}_{S_k}(x) \\
    \mathbb{1}_{S_k}(a+c+x) & = \mathbb{1}_{S_k}(a) + \mathbb{1}_{S_k}(c) + \omega(a, c) + \omega(a, x) + \omega(c, x) + \mathbb{1}_{S_k}(x) \\
    \mathbb{1}_{S_k}(a+b+c+x) & = \mathbb{1}_{S_k}(a) + \mathbb{1}_{S_k}(b) + \mathbb{1}_{S_k}(c) + \omega(a, b) + \omega(b, c) + \omega(c, a) \notag \\
    &+ \omega(a, x) + \omega(b, x) + \omega(c, x) + \mathbb{1}_{S_k}(x)
\end{align}
which means that
\begin{align}
    \mathbb{1}_{S_k}(a+x) + \mathbb{1}_{S_k}(a + b + x) + \mathbb{1}_{S_k}(a + c + x) + \mathbb{1}_{S_k}(a + b + c + x) = \omega(b, c).
\end{align}
Thus, for any fixed $a,b,c\in\mathbb{F}_2^{2k}$, we can infer $\mathbb{1}_{S_k}(a+b+c+x)$ from $\mathbb{1}_{S_k}(a+x)$, $\mathbb{1}_{S_k}(a+b+x)$, and $\mathbb{1}_{S_k}(a+c+x)$. Indeed, if $\omega(b,c)=0$, then for every $x\in W$ we have the following eight equivalences
\begin{align}
    &(a+x\in R_k)\,\land\,(a+b+x\in R_k)\,\land\,(a+c+x\in R_k) \iff a+b+c+x\in R_k \\
    &(a+x\in R_k)\,\land\,(a+b+x\in R_k)\,\land\,(a+c+x\in S_k) \iff a+b+c+x\in S_k \\
    &(a+x\in R_k)\,\land\,(a+b+x\in S_k)\,\land\,(a+c+x\in R_k) \iff a+b+c+x\in S_k \\
    &(a+x\in R_k)\,\land\,(a+b+x\in S_k)\,\land\,(a+c+x\in S_k) \iff a+b+c+x\in R_k \\
    &\quad\quad\ldots\notag \\
    &(a+x\in S_k)\,\land\,(a+b+x\in S_k)\,\land\,(a+c+x\in S_k) \iff a+b+c+x\in S_k.
\end{align}
If on the other hand $\omega(b,c)=1$, then for every $x\in W$ we have the eight equivalences
\begin{align}
    &(a+x\in R_k)\,\land\,(a+b+x\in R_k)\,\land\,(a+c+x\in R_k) \iff a+b+c+x\in S_k \\
    &(a+x\in R_k)\,\land\,(a+b+x\in R_k)\,\land\,(a+c+x\in S_k) \iff a+b+c+x\in R_k \\
    &(a+x\in R_k)\,\land\,(a+b+x\in S_k)\,\land\,(a+c+x\in R_k) \iff a+b+c+x\in R_k \\
    &(a+x\in R_k)\,\land\,(a+b+x\in S_k)\,\land\,(a+c+x\in S_k) \iff a+b+c+x\in S_k \\
    &\quad\quad\ldots\notag \\
    &(a+x\in S_k)\,\land\,(a+b+x\in S_k)\,\land\,(a+c+x\in S_k) \iff a+b+c+x\in R_k.
\end{align}
In either case, $W$ can be partitioned into eight disjoint sets
\begin{align}
    W_{rrr} &:= \{w\in W\,|\,(a+w\in R_k)\,\land\,(a+b+w\in R_k)\,\land\,(a+c+w\in R_k)\} \\
    W_{rrs} &:= \{w\in W\,|\,(a+w\in R_k)\,\land\,(a+b+w\in R_k)\,\land\,(a+c+w\in S_k)\} \\
    W_{rsr} &:= \{w\in W\,|\,(a+w\in R_k)\,\land\,(a+b+w\in S_k)\,\land\,(a+c+w\in R_k)\} \\
    W_{rss} &:= \{w\in W\,|\,(a+w\in R_k)\,\land\,(a+b+w\in S_k)\,\land\,(a+c+w\in S_k)\} \\
    &\ldots \notag \\
    W_{sss} &:= \{w\in W\,|\,(a+w\in S_k)\,\land\,(a+b+w\in S_k)\,\land\,(a+c+w\in S_k)\}
\end{align}
so that $W = W_{rrr} \sqcup W_{rrs} \sqcup\ldots\sqcup W_{sss}$. Consider the following disjoint unions of these sets
\begin{align}
    \label{eq:w_b0_definition}
    W_{xxy} &:= W_{rrr} \sqcup W_{rrs} \sqcup W_{ssr} \sqcup W_{sss} \\
    \label{eq:w_b1_definition}
    W_{x\overline{x}y} &:= W_{rsr} \sqcup W_{rss} \sqcup W_{srr} \sqcup W_{srs} \\
    \label{eq:w_c0_definition}
    W_{xyy} &:= W_{rrr} \sqcup W_{rss} \sqcup W_{srr} \sqcup W_{sss} \\
    \label{eq:w_c1_definition}
    W_{xy\overline{y}} &:= W_{rrs} \sqcup W_{rsr} \sqcup W_{srs} \sqcup W_{ssr}.
\end{align}
The indicator function of $W_{x\overline{x}y}$ restricted to $W$ is
\begin{align}
    \mathbb{1}_{W_{x\overline{x}y}}(x) &= \mathbb{1}_{R_k}(a+x)\cdot\mathbb{1}_{S_k}(a+b+x)\cdot\mathbb{1}_{R_k}(a+c+x) + \mathbb{1}_{R_k}(a+x)\cdot\mathbb{1}_{S_k}(a+b+x)\cdot\mathbb{1}_{S_k}(a+c+x) \\
    &+ \mathbb{1}_{S_k}(a+x)\cdot\mathbb{1}_{R_k}(a+b+x)\cdot\mathbb{1}_{R_k}(a+c+x)  + \mathbb{1}_{S_k}(a+x)\cdot\mathbb{1}_{R_k}(a+b+x)\cdot\mathbb{1}_{S_k}(a+c+x) \\
    &= \mathbb{1}_{R_k}(a+x)\cdot\mathbb{1}_{S_k}(a+b+x) + \mathbb{1}_{S_k}(a+x)\cdot\mathbb{1}_{R_k}(a+b+x) \\
    &= \left(1+\mathbb{1}_{S_k}(a+x)\right)\cdot\mathbb{1}_{S_k}(a+b+x) + \mathbb{1}_{S_k}(a+x)\cdot\left(1+\mathbb{1}_{S_k}(a+b+x)\right) \\
    &= \mathbb{1}_{S_k}(a+x) + \mathbb{1}_{S_k}(a+b+x) \\
    &= \mathbb{1}_{S_k}(b) + \omega(a,b)+\omega(b,x)
\end{align}
which is an affine functional on $W$. Therefore, $W_{x\overline{x}y}$ is empty or an affine subspace of $W$. So is $W_{xxy}$ whose indicator function restricted to $W$ is $\mathbb{1}_{W_{xxy}}(x)=\mathbb{1}_{W_{x\overline{x}y}}(x)+1$. Moreover, \eqref{eq:w_b0_definition} and \eqref{eq:w_b1_definition} imply that $W_{xxy}\sqcup W_{x\overline{x}y} = W$.

Similarly, the indicator function of $W_{xy\overline{y}}$ restricted to $W$ is $\mathbb{1}_{W_{xy\overline{y}}}(x) = \mathbb{1}_{S_k}(b) + \mathbb{1}_{S_k}(c) + \omega(a,b+c)+\omega(b+c,x)$ which is also an affine functional on $W$. Therefore, $W_{xy\overline{y}}$ is empty or an affine subspace of $W$. So is $W_{xyy}$ whose indicator function restricted to $W$ is $\mathbb{1}_{W_{xyy}}(x)=\mathbb{1}_{W_{xy\overline{y}}}(x)+1$. Moreover, \eqref{eq:w_c0_definition} and \eqref{eq:w_c1_definition} imply that $W_{xyy}\sqcup W_{xy\overline{y}} = W$.

We can write
\begin{align}
    |A_{00}\cap S_k| = |W_{srr}| + |W_{srs}| + |W_{ssr}| + |W_{sss}| \\
    |A_{01}\cap S_k| = |W_{rsr}| + |W_{rss}| + |W_{ssr}| + |W_{sss}| \\
    |A_{10}\cap S_k| = |W_{rrs}| + |W_{rss}| + |W_{srs}| + |W_{sss}|
\end{align}
irrespective of $\omega(b,c)$, so that
\begin{align}
    \label{eq:a00s_a01s_a10s}
    &|A_{00}\cap S_k| + |A_{01}\cap S_k| + |A_{10}\cap S_k| \\
    &= |W_{rrs}| + |W_{rsr}| + 2\cdot|W_{rss}| + |W_{srr}| + 2\cdot|W_{srs}| + 2\cdot|W_{ssr}| + 3\cdot|W_{sss}| \notag
\end{align}
also irrespective of $\omega(b,c)$. The analogous expression for $|A_{11}\cap R_k|$ depends on the value of $\omega(b,c)$. We consider both cases constructing a pair of suitable affine spaces $B_1$ and $B_2$ for each one in turn.

First, if $\omega(b,c)=0$, then
\begin{align}
    |A_{11}\cap R_k| = |W_{rrr}| + |W_{rss}| + |W_{srs}| + |W_{ssr}|
\end{align}
and adding \eqref{eq:a00s_a01s_a10s} yields
\begin{align}
    &|A_{00}\cap S_k| + |A_{01}\cap S_k| + |A_{10}\cap S_k| + |A_{11}\cap R_k| \\
    &= |W_{rrr}| + |W_{rrs}| + |W_{rsr}| + 3\cdot|W_{rss}| + |W_{srr}| + 3\cdot|W_{srs}| + 3\cdot|W_{ssr}| + 3\cdot|W_{sss}| \\
    &= |W| + 2\cdot\left(|W_{rss}| + |W_{srs}| + |W_{ssr}| + |W_{sss}|\right).
\end{align}
Define
\begin{align}
    B_1 &:= (a + b + W_{x\overline{x}y}) \sqcup (a + c + W_{xxy}) \\
    B_2 &:= (a + b + W_{xxy}) \sqcup (a + c + W_{x\overline{x}y})
\end{align}
so that $B_1 \sqcup B_2$ is an affine space of dimension $\dim W + 1$ and
\begin{align}
    |B_1| &= |W| \\
    |B_2 \cap S_k| &= |W_{rss}| + |W_{srs}| + |W_{ssr}| + |W_{sss}|
\end{align}
and hence
\begin{align}
    |A_{00}\cap S_k| + |A_{01}\cap S_k| + |A_{10}\cap S_k| + |A_{11}\cap R_k| = |B_1| + 2\cdot|B_2 \cap S_k|.
\end{align}

Now, if $\omega(b,c)=1$, then
\begin{align}
    |A_{11}\cap R_k| = |W_{rrs}| + |W_{rsr}| + |W_{srr}| + |W_{sss}|
\end{align}
and adding \eqref{eq:a00s_a01s_a10s} yields
\begin{align}
    &|A_{00}\cap S_k| + |A_{01}\cap S_k| + |A_{10}\cap S_k| + |A_{11}\cap R_k| \\
    &= 2\cdot|W_{rrs}| + 2\cdot|W_{rsr}| + 2\cdot|W_{rss}| + 2\cdot|W_{srr}| + 2\cdot|W_{srs}| + 2\cdot|W_{ssr}| + 4\cdot|W_{sss}|. \notag
\end{align}
Redefine
\begin{align}
    B_1 &:= (a + W_{xy\overline{y}}) \sqcup (a + b + W_{xy\overline{y}}) \\
    B_2 &:= (a + W_{xyy}) \sqcup (a + b + W_{xyy})
\end{align}
so that $B_1 \sqcup B_2$ is again an affine space of dimension $\dim W + 1$ and
\begin{align}
    |B_1| &= 2\cdot|W_{rrs}| + 2\cdot|W_{rsr}| + 2\cdot|W_{srs}| + 2\cdot|W_{ssr}| \\
    |B_2 \cap S_k| &= |W_{rss}| + |W_{srr}| + 2\cdot|W_{sss}|
\end{align}
and hence
\begin{align}
    |A_{00}\cap S_k| + |A_{01}\cap S_k| + |A_{10}\cap S_k| + |A_{11}\cap R_k| = |B_1| + 2\cdot|B_2 \cap S_k|
\end{align}
completing the proof of the Lemma.
\end{proof}

We note the following conditional upper bound which is implied by Lemma \ref{lm:hard_rsf_4terms} and which will facilitate its use in an inductive proof.

\begin{corollary}[Bound for mixed recursive structure formulas with four terms]
    \label{cr:hard_rsf_4terms}
    \,\newline
    Let $e$ and $k>1$ be positive integers. If $|B\cap S_{k-1}| \leq 2^{\dim B-1} + 2^{k-3}$ for every affine subspace $B$ of $\mathbb{F}_2^{2k-2}$ with $\dim B\in\{e,e+1\}$, then
    \begin{align}
        |A_{00}\cap S_{k-1}| + |A_{01}\cap S_{k-1}| + |A_{10}\cap S_{k-1}| + |A_{11}\cap R_{k-1}| \leq 2^{e+1} + 2^{k-2}
    \end{align}
    for any four disjoint affine subspaces of the form
    \begin{align}
        A_{00} &= a + W \\
        A_{01} &= a + b + W \\
        A_{10} &= a + c + W \\
        A_{11} &= a + b + c + W
    \end{align}
    arising as translations of the same $e$-dimensional linear subspace $W\subset\mathbb{F}_2^{2k-2}$ with $a,b,c\in\mathbb{F}_2^{2k-2}$.
\end{corollary}
\begin{proof}
By Lemma \ref{lm:hard_rsf_4terms}, there exist two sets $B_1$ and $B_2$ each of which is either empty or affine and such that
\begin{align}
    \label{eq:ar_as<=b+2bs}
    |A_{00}\cap S_{k-1}| + |A_{01}\cap S_{k-1}| + |A_{10}\cap S_{k-1}| + |A_{11}\cap R_{k-1}| = |B_1| + 2\cdot|B_2\cap S_{k-1}|
\end{align}
and such that $|B_1| + |B_2|=2\cdot|W|$. Thus, if $B_1=\emptyset$, then $B_2$ is an affine subspace of $\mathbb{F}_2^{2k-2}$ of dimension $e+1$ and using the assumption in the Corollary, we obtain
\begin{align}
    |A_{00}\cap S_{k-1}| + |A_{01}\cap S_{k-1}| + |A_{10}\cap S_{k-1}| + |A_{11}\cap R_{k-1}| \leq 0 + 2\cdot(2^{e} + 2^{k-3}) = 2^{e+1} + 2^{k-2}.
\end{align}
If $B_2=\emptyset$, then $B_1$ is an affine subspace of $\mathbb{F}_2^{2k-2}$ of dimension $e+1$. Therefore,
\begin{align}
    |A_{00}\cap S_{k-1}| + |A_{01}\cap S_{k-1}| + |A_{10}\cap S_{k-1}| + |A_{11}\cap R_{k-1}| \leq 2^{e+1} + 2\cdot 0 < 2^{e+1} + 2^{k-2}.
\end{align}
Finally, if neither $B_1$ nor $B_2$ is empty, then both are affine and $\dim B_1=\dim B_2=e$. Therefore,
\begin{align}
    |A_{00}\cap S_{k-1}| + |A_{01}\cap S_{k-1}| + |A_{10}\cap S_{k-1}| + |A_{11}\cap R_{k-1}| \leq 2^e + 2\cdot(2^{e-1} + 2^{k-3}) = 2^{e+1} + 2^{k-2}
\end{align}
completing the proof of the Corollary.
\end{proof}

\section{Reference Python Implementation}\label{sec:ref_python_impl}

\subsection{Reed Solomon Decoding using Synchronized EEA for Explicit Bézout coefficients}
\cref{fig:optimized_zalka_eea} gives our optimized implementation of Zalka's synchronized reversible EEA \cite{kaye2004optimizedquantumimplementationelliptic}. 
\cref{fig:rs_decoder_eea_zalka} uses this as a subroutine and shows how to implement the Chien Search \cite{Chien1964} and Forney's algorithm \cite{Forney1965} steps for syndrome decoding of Reed Solomon codes, given the shared in-place register representation of the error locator polynomial $\sigma(z)$ and error evaluator polynomial $\Omega(z)$. 
 
\begin{figure}
    \pythonlisting{6pt}{assets/code/poly_eea_zalka_gcd.py}
    \caption{Our optimized implementation of Zalka's EEA for explicit access to Bézout coefficients.}
    \label{fig:optimized_zalka_eea}
\end{figure}

\begin{figure}
    \pythonlisting{7pt}{assets/code/rs_syndrome_decoder_eea_zalka.py}
    \caption{RS Decoding using our optimized implementation of Zalka's EEA \cite{kaye2004optimizedquantumimplementationelliptic, proos2004shorsdiscretelogarithmquantum} presented in \cref{fig:optimized_zalka_eea}}
    \label{fig:rs_decoder_eea_zalka}
\end{figure}

\subsection{Reed Solomon Decoding using Dialog based EEA for Implicit Bézout Coefficients}
\cref{fig:in-place-dialog-divstep} gives our optimized implementation of constructing a Dialog based on the Bernstein-Yang GCD algorithm \cite{cryptoeprint:2019/266} using a shared register representation to save space. We also show how one can use the Dialog to perform modular multiplication or modular division. 
\cref{fig:rs_decoder_eea_inplace} uses the Dialog based EEA for syndrome decoding of Reed Solomon codes, and shows how one can evaluate the error locator polynomial $\sigma(z)$ and its derivative $\sigma^\prime(z)$, given access to the Dialog of $\sigma(z)$.

\begin{figure}
    \pythonlisting{7pt}{assets/code/divstep_dialog_inplace.py}
    \caption{In Place Dialog construction using register sharing and its use to perform modular division and modular multiplication.}
    \label{fig:in-place-dialog-divstep}
\end{figure}

\begin{figure}
    \pythonlisting{7pt}{assets/code/rs_syndrome_decoder_eea_dialogs.py}
    \caption{Optimized implementation of EEA based syndrome decoder for RS codes using register sharing to efficiently compute the Dialog representation and use it to implicitly evaluate Bézout coefficient corresponding to polynomial $\sigma(z)$ and $\sigma^\prime(z)$.}
    \label{fig:rs_decoder_eea_inplace}
\end{figure}

\end{document}